\documentclass{amsart}

\usepackage[letterpaper, hmarginratio=1:1]{geometry}

\usepackage{amsmath, amsthm, amsfonts, thmtools, amssymb,tikz,hyperref,cleveref,bm,stmaryrd,bbm}
\usetikzlibrary{arrows}

\usepackage[mathscr]{euscript}

\numberwithin{equation}{section}

\DeclareMathOperator{\eig}{eig}

\DeclareMathOperator{\Imaginary}{Im}
\DeclareMathOperator{\Real}{Re}
\DeclareMathOperator{\dr}{dr}
\DeclareMathOperator{\eff}{eff}
\DeclareMathOperator{\Id}{Id}

\DeclareMathOperator{\Mat}{Mat}
\DeclareMathOperator{\SymMat}{SymMat}
\DeclareMathOperator{\dist}{dist}
\DeclareMathOperator{\supp}{supp}
\DeclareMathOperator{\Infl}{Infl}
\DeclareMathOperator{\sgn}{sgn}

\newtheorem{thm}{Theorem}[section]
\newtheorem{prop}[thm]{Proposition}
\newtheorem{lem}[thm]{Lemma}
\newtheorem{cor}[thm]{Corollary}

\theoremstyle{remark}
\newtheorem{rem}[thm]{Remark}
\theoremstyle{definition}
\newtheorem{definition}[thm]{Definition}

\newtheorem{example}[thm]{Example}
\newtheorem{assumption}[thm]{Assumption}

\crefname{thm}{Theorem}{Theorems}
\crefname{prop}{Proposition}{Propositions}
\crefname{lem}{Lemma}{Lemmas}
\crefname{cor}{Corollary}{Corollaries}
\crefname{conj}{Conjecture}{Conjectures}
\crefname{rem}{Remark}{Remarks}
\crefname{assumption}{Assumption}{Assumptions}
\crefname{definition}{Definition}{Definitions}
\crefname{que}{Question}{Questions}
\crefname{example}{Example}{Examples}

\makeatletter
\@addtoreset{equation}{section}

\makeatother



\title{Effective Velocities in the Toda Lattice}
\author{Amol Aggarwal} 

\begin{document}
	
	\begin{abstract}

		In this paper we consider the Toda lattice $(\bm{p}(t); \bm{q}(t))$ at thermal equilibrium, meaning that its variables $(p_i)$ and $(e^{q_i-q_{i+1}})$ are independent Gaussian and Gamma random variables, respectively. This model can be thought of a dense collection of many ``quasiparticles'' that act as solitons. We establish a law of large numbers for the trajectory of these quasiparticles, showing that they travel with approximately constant velocities, which are explicit. Our proof is based on a direct analysis of the asymptotic scattering relation, an equation (proven in \cite{LC}) that approximately governs the dynamics of quasiparticles locations. This makes use of a regularization argument that essentially linearizes this relation, together with concentration estimates for the Toda lattice's (random) Lax matrix. 
		
	\end{abstract}

	\maketitle 
	
	\tableofcontents

	\section{Introduction}  
	\label{Chain} 
	
	\subsection{Preface} 
	
	A basic tenet of integrable systems (essentially dating back Zakharov's study of dilute soliton gases \cite{KES}) is that, under natural random initial data, they can be thought of as dense collections of objects called ``quasiparticles'' that behave as \emph{solitons}; the latter are (loosely speaking) localized, wave-like functions that retain their shape as they propagate in time. As such, each quasiparticle possesses a time-independent spectral parameter $\lambda_j$ and a time-dependent location $Q_j (t)$. Under this interpretation, relevant quantities describing the system (such as local charges and currents) should be approximable by simple functions of the quasiparticle data.
	
	Therefore, once this putative quasiparticle description is made sense of, the immediate question that arises for one interested in the long-time behavior of the integrable system is to analyze the limiting trajectories of the associated quasiparticles. It has been broadly predicted in the physics literature that, under invariant initial data, the quasiparticle of spectral parameter $\lambda_j$ should travel with an approximately constant \emph{effective velocity} $v_{\eff} (\lambda_j)$, which satisfies an equation of the form 
	\begin{flalign}
		\label{fvlambda} 
		f(\lambda_j) = v_{\eff} (\lambda_j) + \displaystyle\int_{-\infty}^{\infty}  \big( v_{\eff} (\lambda_j) - v_{\eff} (\lambda) \big) \mathfrak{s} (\lambda_j, \lambda) \varrho (\lambda) d \lambda.
	\end{flalign}
	
	\noindent Here, $\varrho$ is a density function prescribing the relative proportions of quasiparticle spectral parameters in the system, and $\mathfrak{s}$ and $f$ are model-dependent parameters (namely, the scattering shift and bare soliton velocity, respectively). This was originally posited for the Korteweg--de Vries (KdV) equation, first when the initial data is dilute in \cite{KES} and later when it is dense by El in \cite{TLE}. More recently, it was proposed broadly for (classical and quantum) integrable systems, starting in the independent works of Bertini--Collura--de Nardis--Fagotti \cite{TEP} and Castro-Alvaredo--Doyon--Yoshimura \cite{EHSOE}. 
	
	The program of mathematically making sense of this quasiparticle description, and justifying these asymptotic predictions, had previously only been realized for two integrable systems. The first is the hard rods (one-dimensional hard spheres) model; the second is the box-ball system, a cellular automaton introduced by Takahashi--Satsuma \cite{SCA}. In the former, the quasiparticles are the rods; in the latter, they are more hidden but can be recovered through an elementary, combinatorial algorithm \cite{SCA}. For both systems, the linear quasiparticle trajectory\footnote{One might also ask about fluctuations for the quasiparticle trajectories. The physics works of De Nardis--Bernard--Doyon \cite{HDIS} and Gopalakrishnan--Huse--Khemani--Vasseur \cite{HOSQ} predict that they should be diffusive and scale to a Brownian motion. This was proven for the hard rods model by Boldrighini--Suhov \cite{OHRC}, Presutti--Wick \cite{MSF}, Ferrari-Franceschini--Grevino--Spohn \cite{HRHF}, and Ferrari--Olla \cite{MDFGHR}, and for the box-ball system by Olla--Sasada--Suda \cite{SLSBS}.} with an effective velocity satisfying \eqref{fvlambda} has been established; this was done for the hard rods model by Boldrighini--Dubroshin--Suhov \cite{OHR} and for the box-ball system by Ferrari--Nguyen--Rolla--Wang \cite{SDBS}.
	
	For other integrable systems under invariant initial data, even a precise definition for the quasiparticle locations $(Q_j)$ did not seem to appear in the literature until recently. Still, a signature of the above asymptotics was established by Girotti--Grava--Jenkins--McLaughlin--Minakov \cite{SGAR}, who studied a (deterministic) profile formed by a many-soliton solution of the modified KdV equation. By altering this solution to incorporate one large ``tracer'' soliton, they could track its location by the solution's maximum. Through a Riemann--Hilbert analysis, they proved that this tracer soliton, which should be thought of as a single ``large'' quasiparticle, has a linear limiting trajectory, with an effective velocity satisfying an equation of the form \eqref{fvlambda}.

	The integrable system that we study in this paper is the Toda lattice, whose quasparticle description under certain random initial data was recently introduced in \cite{LC}. This model is a Hamiltonian dynamical system $(\bm{p}(t); \bm{q}(t))$, where $\bm{p}(t) = (p_i(t))$ and $\bm{q}(t) = (q_i(t))$ are indexed by a one-dimensional integer lattice $i \in \mathscr{I}$ that could either be an interval $\mathscr{I} = [N_1, N_2]$, a torus $\mathscr{I} = \mathbb{Z} / N \mathbb{Z}$, or the full line $\mathscr{I} = \mathbb{Z}$ (we typically focus on the former here). Its Hamiltonian is given by
	\begin{flalign*}
		\mathfrak{H} (\bm{p}; \bm{q}) = \displaystyle\sum_{i \in \mathscr{I}} \bigg( \displaystyle\frac{p_i^2}{2} + e^{q_i - q_{i+1}} \bigg),
	\end{flalign*}
	
	\noindent so the dynamics $\partial_t q_i = \partial_{p_i} \mathfrak{H}(\bm{p}; \bm{q})$ and $\partial_t p_i = -\partial_{q_i} \mathfrak{H} (\bm{p}; \bm{q})$ are  
	\begin{flalign}
		\label{qtderivative} 
		\partial_t q_i (t) = p_i (t), \quad \text{and} \quad \partial_t p_i (t) = e^{q_{i-1} (t) - q_i (t)} - e^{q_i (t) - q_{i+1} (t)}.
	\end{flalign}
	
	\noindent This model may be thought of as a system of particles moving on the real line, with locations $(q_i)$ and momenta $(p_i)$. It was originally introduced by Toda \cite{VCNI} as a Hamiltonian dynamic that admits soliton solutions. Since the works of Flaschka \cite{LEI} and Manakov \cite{CIS} exhibiting its full set of conserved quantities, and that of Moser \cite{FMPL} determining its scattering shift, the Toda lattice has become recognized as an archetypal example of a completely integrable system. 
	
	We consider the Toda lattice under perhaps its most natural invariant measure, given by \emph{thermal equilibrium}. Given parameters $\beta, \theta > 0$, this means (see \Cref{mubeta2}) that we sample the $(p_i)$ and $(e^{(q_i - q_{i+1})/2})$ as independent random variables, with probability densities $C_{\beta} e^{-\beta x^2 / 2}$ and $C_{\beta,\theta} x^{2\theta-1} e^{-\beta x^2}$, respectively, where $C_{\beta}, C_{\beta,\theta} > 0$ are normalization constants. 
	
	Now we must explain what the associated quasiparticle spectral parameters and locations are. The quasiparticle spectral parameters $\lambda_j$ have long been understood. They are defined to be the conserved quantities for the Toda lattice, given by the eigenvalues of its \emph{Lax matrix}. This is the tridiagonal, symmetric matrix $\bm{L}(t) = [L_{ij}(t)]$ (where $i, j \in \mathscr{I}$), whose diagonal and off-diagonal entries are the $(p_i)$ and $(e^{(q_i - q_{i+1})/2})$, respectively \cite{LEI,CIS}. When $\mathscr{I}$ is large and $(\bm{p} (t); \bm{q} (t))$ is random, $\bm{L}(t)$ becomes a high-dimensional random matrix, whose eigenvalue density then prescribes the distribution of quasiparticle spectral parameters in the Toda lattice under thermal equilibrium; this (up to a normalization) was denoted by $\varrho$ in \eqref{fvlambda}. Its computation was addressed by Spohn \cite{GECC} (after initial work of Opper \cite{ASCEC}), who predicted formulas for its limiting density (and derived expectations for local currents), which were later verified in works of Mazzuca, Guionnet, and Memin \cite{DSME,LDC,EIST}.

	The definition for the location $Q_j (t)$, of the $j$-th quasiparticle at time $t$, is more recent; it appeared and was analyzed in the mathematics paper \cite{LC}, and it was earlier hypothesized in the physics work of Bulchandini--Cao--Moore \cite{KTQL}. To explain it (simplifying slightly; see \Cref{lbetaeta} below for its full description), let $\bm{u}_j (t) = (u_j (i;t))_{i \in \mathscr{I}}$ denote the unit eigenvector of $\bm{L}(t)$ with eigenvalue $\lambda_j$. Results on random tridiagonal matrices due to Kunz--Souillard \cite{SSO} and Schenker \cite{LRBM}   imply that, if $\bm{L}(t)$ is under thermal equilibrium, then $\bm{u}_j (t)$ is \emph{exponentially localized}. This means that it admits some ``center'' $\varphi_t (j) \in \mathscr{I}$ such that $|u_j (i; t)| \le C e^{-c|i-\varphi_t (j)|}$ likely holds for any $i \in \mathscr{I}$. We view this center $\varphi_t (j)$ as the index of the particle associated with the $j$-th quasiparticle. So, we define the $j$-th quasiparticle's location on $\mathbb{R}$ to be this particle's position $Q_j (t) = q_{\varphi_t (j)} (t)$; \cite{LC} showed that this is a viable definition for quasiparticle locations, in that it satisfies the postulates suggested in the physics literature.
	
	Under this notation, the prediction for the asymptotic quasiparticle trajectories admits a precise formulation, given by
	\begin{flalign}
		\label{qjt00} 
		Q_j (t) \approx Q_j (0) + t v_{\eff} (\lambda_j),
	\end{flalign} 
	
	\noindent where $v_{\eff}$ solves \eqref{fvlambda}, with the $\varrho$ there given by the Lax matrix density of states computed in the above-mentioned works \cite{ASCEC,GECC,DSME} (with $\mathfrak{s} (\lambda, \mu) = 2 \log |\lambda - \mu|$ and $f(\lambda) = \lambda$ for the Toda lattice).
	
	The purpose of this paper is to prove \eqref{qjt00} when the thermal equilibrium parameter $\theta$ is sufficiently small\footnote{This constraint on $\theta$ should be artificial. It arises for us since we only know that a certain matrix is (quantitatively) invertible for $\theta$ sufficiently small (though we suspect it is true for all $\theta>0$); see \Cref{Estimate2} and \Cref{smatrixchi} below for further information.}; see \Cref{vestimate} below. Our starting point is an approximation, established in \cite{LC} (and predicted in \cite{KES,TEP,EHSOE,SGGH,NCIS}), for the evolution of the quasiparticle locations $(Q_j (t))$; it states that 
		\begin{flalign}
		\label{qktqk01} 
		Q_k (t) \approx Q_k (0) + \lambda_k t -  2 \displaystyle\sum_{j : Q_j(t) < Q_k (t)} \log |\lambda_k - \lambda_j| + 2 \displaystyle\sum_{j : Q_j(0) < Q_k (0)} \log |\lambda_k - \lambda_j|. 
	\end{flalign}

	\noindent The relation \eqref{qktqk01} can be interpreted as follows. The $k$-th quasiparticle, initially at $Q_k (0)$, moves with velocity $\lambda_k$ until it meets another quasiparticle, say the $j$-th one. At that moment, the $k$-th quasiparticle instantaneously moves forward or backward by $2 \log |\lambda_k - \lambda_j|$, depending on whether it met the $j$-th quasiparticle from the right or left, respectively (this could make it pass another quasiparticle, producing a ``cascade'' of such interactions, but no two quasiparticles can interact with each other more than once in this way). Then, the $k$-th quasiparticle proceeds at velocity $\lambda_k$ until meeting another quasiparticle, when the procedure repeats. The reason for the choice $2 \log |\lambda_k - \lambda_j|$ is that it is the Toda lattice's scattering shift, describing the phase displacement of just two quasiparticles passing through each other.
	
	This is directly analogous to the behavior of solitons in integrable systems. Indeed, in for example the Korteweg--de Vries (KdV) equation, a soliton solution takes the form of a traveling wave. Its amplitude (shape) is conserved and determined by a spectral parameter; its position is the wave's peak, which moves linearly (``freely'') in time. When two distant solitons evolve under the KdV equation, they initially proceed independently, until they get close to each other. They then undergo a fairly intricate interaction where they interfere, during which their precise locations are ``blurred.'' After this, they emerge as independent solitons; their shapes are the same as before the interaction, but their positions are translated (relative to their free evolutions) by an overall scattering shift. In this way, there is a close parallel between the soliton and quasiparticle dynamics (so, the terms ``soliton'' and ``quasiparticle'' are sometimes used interchangeably). 

	 We refer to \eqref{qktqk01} as the \emph{asymptotic scattering relation}. It is also sometimes called the ``collision rate ansatz'' or ``flea-gas algorithm.'' The latter was originally termed in the context of certain quantum integrable systems by Doyon--Yoshimura--Caux \cite{SGGH} and was later studied in greater generality in the paper \cite{NCIS} Doyon--H\"{u}bner--Yoshimura (who also highlighted its resemblance to the dynamics of solitons, for classical integrable systems, and wave packets \cite{DGHG}, for quantum ones).
	
	The proof of \eqref{qjt00}, with $v_{\eff}$ satisfying \eqref{fvlambda}, is based on an analysis of the asymptotic scattering relation \eqref{qktqk01}, and requires little information about the Toda lattice \eqref{qtderivative} itself. In fact, if one were to assume \eqref{qjt00} for some function $v_{\eff}$ only dependent on $\lambda_j$, then a concise heuristic (see \cite[Appendix B]{LC}) would indicate that \eqref{qktqk01} is a discretization of \eqref{fvlambda}. Indeed, this intuition is what led to the predicted form of \eqref{fvlambda} in \cite{KES,TEP,EHSOE}.
	
	Unfortunately, we do not know how to verify this hypothesis directly. So we proceed differently, based on a regularization argument that can be used to approximately linearize \eqref{qktqk01}, with concentration estimates for the (random) Toda lattice Lax matrix $\bm{L}(t)$. To explain this further, it will be convenient to set up some additional notation (which is anyways needed to state our main results), so we defer the proof outline to \Cref{Analysis0} below.

	We next describe our results in more detail. Throughout, for any $a, b \in \mathbb{R}$, set $\llbracket a, b \rrbracket = [a, b] \cap \mathbb{Z}$. A vector $\bm{v} = (v_1, v_2, \ldots , v_n) \in \mathbb{C}^n$ is a unit vector if $\sum_{i=1}^n v_i^2 = 1$. For any real symmetric $n \times n$ matrix $\bm{M}$, let $\eig \bm{M} = (\lambda_1, \lambda_2, \ldots , \lambda_n)$ denote the eigenvalues of $\bm{M}$, counted with multiplicity and ordered so that $\lambda_1 \ge \lambda_2 \ge \cdots \ge \lambda_n$. 
	
	\subsection{Toda Lattice and Thermal Equilibrium} 
	\label{Open} 
	
	In this section we recall the Toda lattice on an interval, and its thermal equilibrium initial data. Throughout, we fix integers $N_1 \le N_2$ and set $N = N_2 - N_1 + 1$ (which will prescribe the interval's endpoints and length, respectively). 
	
	The state space of the Toda lattice on the interval $\llbracket N_1, N_2 \rrbracket$, also called the \emph{Toda lattice}, is given by a pair of $N$-tuples $( \bm{p} (t); \bm{q}(t) ) \in \mathbb{R}^N \times \mathbb{R}^N$, where $\bm{p}(t) = ( p_{N_1} (t), p_{N_1+1} (t), \ldots , p_{N_2} (t) )$ and $\bm{q}(t) = ( q_{N_1} (t), q_{N_1+1}(t), \ldots , q_{N_2} (t) )$; both are indexed by a real number $t \ge 0$ called the time. Given any \emph{initial data} $( \bm{p}(0); \bm{q}(0) ) \in \mathbb{R}^N \times \mathbb{R}^N$, the joint evolution of $( \bm{p}(t); \bm{q}(t) )$ for $t \ge 0$ is prescribed by the system of ordinary differential equations 
	\begin{flalign}
		\label{qtpt} 
		\partial_t q_j (t) = p_j (t), \quad \text{and} \quad \partial_t p_j (t) = e^{q_{j-1} (t) - q_j (t)} - e^{q_j (t) - q_{j+1} (t)},
	\end{flalign}
	
	\noindent for all $(j, t) \in \llbracket N_1, N_2 \rrbracket \times \mathbb{R}_{\ge 0}$; here, we set $q_{N_1-1}(t) = -\infty$ and $q_{N_2+1}(t) = \infty$ for all $t \ge 0$. One might interpret this as the dynamics for $N$ points (indexed by $\llbracket N_1, N_2 \rrbracket$) moving on the real line, whose locations and momenta at time $t \ge 0$ are given by the $(q_i(t))$ and $(p_i(t))$, respectively. 	
	
	The system of differential equations \eqref{qtpt} is equivalent to the Hamiltonian dynamics generated by the Hamiltonian $\mathfrak{H} : \mathbb{R}^N \times \mathbb{R}^N \rightarrow \mathbb{R}$ that is defined, for any $\bm{p} = (p_0, p_1, \ldots , p_{N-1}) \in \mathbb{R}^N$ and $\bm{q} = (q_0, q_1, \ldots , q_{N-1}) \in \mathbb{R}^N$, by setting
	\begin{flalign}
		\label{hpq}
		\mathfrak{H} (\bm{p}; \bm{q}) = \displaystyle\sum_{j=0}^{N-1} \bigg( \displaystyle\frac{p_j^2}{2} + e^{q_j - q_{j+1}} \bigg),
	\end{flalign} 
	
	\noindent where $q_N = \infty$. The existence and uniqueness of solutions to \eqref{qtpt} for all time $t \ge 0$, under arbitrary initial data $(\bm{p}; \bm{q}) \in \mathbb{R}^N \times \mathbb{R}^N$, is thus a consequence of the Picard--Lindel\"{o}f theorem (see, for example, the proof of \cite[Theorem 12.6]{OINL}). 
	
	It will often be useful to reparameterize the variables of the Toda lattice, following \cite{LEI}. To that end, for any $(i, t) \in \llbracket N_1, N_2 \rrbracket \times \mathbb{R}_{\ge 0}$, define 
	\begin{flalign}
		\label{abr} 
		r_i (t) = q_{i+1} (t) - q_i (t); \qquad a_i (t) = e^{-r_i(t)/2}; \qquad b_i (t) = p_i (t).
	\end{flalign}
	
	\noindent Denoting $\bm{a}(t) = ( a_{N_1} (t), a_{N_1+1} (t), \ldots , a_{N_2} (t) ) \in \mathbb{R}_{\ge 0}^N$ and $\bm{b}(t) = ( b_{N_1} (t), b_{N_1+1} (t), \ldots , b_{N_2} (t) ) \in \mathbb{R}^N$, the $( \bm{a} (t); \bm{b} (t) )$ are sometimes called \emph{Flaschka variables}; they satisfy $r_{N_2}(t) = q_{N_2+1}(t) - q_{N_2} (t) = \infty$ and $a_{N_2} (t) = e^{-r_{N_2}(t) / 2} = 0$. Then, \eqref{qtpt} is equivalent to the system  
	\begin{flalign}
		\label{derivativepa} 
		\partial_t a_j (t) = \displaystyle\frac{a_j (t)}{2} \cdot \big(b_j (t) - b_{j+1} (t) \big), \qquad \text{and} \qquad \partial_t b_j (t) = a_{j-1} (t)^2 - a_j (t)^2,
	\end{flalign} 
	
	\noindent for each $(j, t) \in \llbracket N_1, N_2 \rrbracket \times \mathbb{R}_{\ge 0}$. 
	
	It will at times be necessary to define the original Toda state space variables $( \bm{p}(t); \bm{q}(t) )$ from the Flaschka variables $( \bm{a}(t); \bm{b}(t) )$; it suffices to do this at $t = 0$, as $( \bm{p} (t); \bm{q}(t) )$ is determined from $( \bm{p}(0); \bm{q}(0) )$, by \eqref{qtpt}. We explain how to do this when $0 \in \llbracket N_1, N_2 \rrbracket$ (as otherwise we may translate $(N_1, N_2)$ to guarantee that this inclusion holds).\footnote{In this work, we will usually have $N_1$ and $N_2$ be large negative and large positive integers, respectively, and we will be interested in the $(p_i(t); q_i (t))$ for $i$ in some interval containing $0$.} By \eqref{abr}, the Flaschka variables $\bm{a}(0)$ only specifies the differences between consecutive entries in $\bm{q}(0)$, so the former only determines the latter up to an overall shift. We will fix this shift by setting $q_0 (0) = 0$, that is, we define $(\bm{p}(0); \bm{q}(0) )$ by imposing  
	\begin{flalign}
		\label{q00} 
		q_0 (0) = 0; \qquad q_{i+1} (0) - q_i (0) = - 2 \log a_i (0); \qquad p_i (0) = b_i (0),
	\end{flalign} 
	
	\noindent  for each $i \in \llbracket N_1, N_2 \rrbracket$. Then, $( \bm{p} (0); \bm{q} (0) )$ is called the Toda state space intial data associated with $(\bm{a} (0); \bm{b}(0) )$. The evolution $( \bm{p}(t), \bm{q}(t) )$ of this initial data under \eqref{qtpt} is called the Toda state space dynamics associated with $( \bm{a}(t), \bm{b}(t) )$; observe that we may have $q_0 (t) \ne 0$ if $t \ne 0$.	
	
	In this work, we mainly consider the Toda lattice under a specific class of random initial data; it is sometimes referred to as thermal equilibrium, and is given by independent Gamma random variables for the $\bm{a}$ Flaschka variables, and independent Gaussian random variables for the $\bm{b}$ ones.
	
	\begin{definition}
		\label{mubeta2} 
		
		Fix real numbers $\beta, \theta > 0$. The \emph{thermal equilibrium} with parameters $(\beta, \theta; N)$ is the product measure $\mu = \mu_{\beta, \theta} = \mu_{\beta, \theta; N-1,N}$ on $\mathbb{R}^{N-1} \times \mathbb{R}^N$ defined by 
		\begin{flalign*}
			\mu (d \bm{a}; d \bm{b}) = \bigg( \displaystyle\frac{2 \beta^{\theta}}{\Gamma(\theta)} \bigg)^{N-1}  (2 \pi \beta^{-1})^{-N/2} \cdot \displaystyle\prod_{j=0}^{N-2} a_j^{2\theta-1} e^{-\beta a_j^2} da_j  \displaystyle\prod_{j=0}^{N-1} e^{-\beta b_j^2/2} db_j,
		\end{flalign*}
		
		\noindent where $\bm{a} = (a_0, \ldots , a_{N-2}) \in \mathbb{R}_{\ge 0}^{N-1}$ and $\bm{b} = (b_0, b_1, \ldots , b_{N-1}) \in \mathbb{R}^N$. It will be convenient to view $\mu_{\beta,\theta;N-1,N}$ as a measure on $\mathbb{R}^N \times \mathbb{R}^N$ by, if we denote $\hat{\bm{a}}(0) = (a_0, a_1, \ldots , a_{N-2}, 0) \in \mathbb{R}_{\ge 0}^N$, then also saying $(\hat{\bm{a}}; \bm{b}) \in \mathbb{R}^N \times \mathbb{R}^N$ is sampled under $\mu_{\beta,\theta;N-1,N}$. 
		
	\end{definition} 

	Thermal equilibrium is related to invariant measures for the Toda lattice; the latter are measures on the Flaschka variable initial data $(\bm{a}(0);\bm{b}(0))$ such that, for any $t \ge 0$, the law of $(\bm{a}(t); \bm{b}(t))$ under the Toda lattice is the same as that of $(\bm{a}(0); \bm{b}(0))$. The Toda lattice on a finite interval $\llbracket N_1, N_2 \rrbracket$ admits no nontrivial invariant measures. However, the Toda lattice on the full line $\mathbb{Z}$ does, among which the thermal equilibrium product measure of \Cref{mubeta2} (extrapolated to when $N=\infty$) is perhaps the most natural one. By \cite[Proposition 2.5]{LC}, with high probability, the Toda lattice at thermal equilbrium on $\mathbb{Z}$ up to some time $T \ge 0$ can be closely approximated (with error decaying exponentially in $T$) by the Toda lattice on $\llbracket N_1, N_2 \rrbracket$, as long as $T \ll N$. Thus, asymptotic questions about the Toda lattice run for some large time $T$, on $\mathbb{Z}$ (or a large torus; see \cite[Proposition 4.3]{LC}) under thermal equilibrium, can be recovered from those about the Toda lattice on $\llbracket N_1, N_2 \rrbracket$. For this reason, we will focus on the latter throughout. 
	
	\subsection{Dressing Operator and Effective Velocities}
	\label{TOperator}
	
	In this section we introduce the dressing operator that will prescribe the effective velocities governing quasiparticle trajectories in the Toda lattice; this will follow \cite[Equation (3.57)]{HSIMS}. Throughout, we fix real numbers $\beta, \theta > 0$; the constants below may depend on $\beta$ and $\theta$, even when not explicitly stated.
	
	We begin by introducing certain quantities and functions. The $\mathfrak{F}$ and $\varrho_{\beta}$ below originally appeared in the study of high-temperature beta ensembles in \cite[Equation (16)]{IEC} and \cite[Theorem 1.1(ii)]{SMRME};\footnote{In those works, our $\beta$ is set to $1$, but it can taken to be arbitrary by rescaling $x$ to $\beta^{1/2} x$.} the $\varrho$ originally appeared in \cite[Equation (3.5)]{GECC}.
	
	\begin{definition}
		
		\label{frho} 
		
		Define the real number
		\begin{flalign}
			\label{alpha} 
			\alpha =  \log \beta - \displaystyle\frac{\Gamma'(\theta)}{\Gamma(\theta)}, \qquad \text{and assume that $\alpha \ne 0$}.
		\end{flalign} 

		\noindent For any $x \in \mathbb{R}$, also set   
		\begin{flalign*} 
			\mathfrak{F} (\theta; x) = \bigg( \displaystyle\frac{\theta}{\Gamma(\theta)} \bigg)^{1/2} \displaystyle\int_0^{\infty} y^{\theta-1} e^{\mathrm{i} xy - y^2/2} dy.
		\end{flalign*}
		
		\noindent Then, define the function $\varrho_{\beta} : \mathbb{R} \rightarrow \mathbb{R}$ and the \emph{density of states} $\varrho : \mathbb{R} \rightarrow \mathbb{R}$ by for any $x \in \mathbb{R}$ setting\footnote{That $\mathfrak{F}(\theta;x) \ne 0$ for all $x \in \mathbb{R}$ is a quick consequence of the fact (see \cite[Section 3.3]{SMRME}) that $\mathfrak{F}$ is smooth and that $\varrho_{\beta}$ has all its moments bounded.}
		\begin{flalign*}
			\varrho_{\beta} (x) = \varrho_{\beta; \theta} (x) = \Big( \displaystyle\frac{\beta}{2\pi} \Big)^{1/2} \cdot | \mathfrak{F} (\theta; \beta^{1/2} x) |^{-2} \cdot e^{- \beta x^2 / 2}; \qquad \varrho (x) = \partial_{\theta} \big( \theta \cdot \varrho_{\beta; \theta} (x) \big).
		\end{flalign*}
		
	\end{definition} 
	
	Let us make several comments about \Cref{frho}. First, the reason $\varrho$ is often called the density of states is that it denotes the empirical eigenvalue density for the Lax matrix of the Toda lattice under thermal equilibrium; see \Cref{lf} below. Second, there are several reasons for our assumption that $\alpha \ne 0$. One is that the dressing operator (defined by \Cref{trho} and \Cref{fdr}), which is needed to define the effective velocity (\Cref{v}) appearing in our results, is no longer well-defined if $\alpha = 0$; see \Cref{sigmaalpha0} below. Another is that it is natural in the context of the Toda lattice, as it is equivalent to the density of particles in the system being finite.

	\begin{definition}
			\label{tbeta} 
	
		For any functions $f, g : \mathbb{R} \rightarrow \mathbb{C}$, define the inner product
		\begin{flalign}
			\label{betaproduct} 
			\langle f, g \rangle_{\varrho} = \displaystyle\int_{-\infty}^{\infty} f(x) \overline{g(x)} \varrho (x) dx,
		\end{flalign}
		
		\noindent when it is finite. Let $\mathcal{H}$ be the Hilbert space associated with this inner product; denote the norm on $\mathcal{H}$ by $\| f \|_{\mathcal{H}} = \langle f, f \rangle_{\varrho}^{1/2}$ for any $f \in \mathcal{H}$. Observe that $\varrho_{\beta} \in \mathcal{H}$ and that $\varsigma_k \in \mathcal{H}$ for any integer $k \ge 0$, where $\varsigma_k : \mathbb{R} \rightarrow \mathbb{R}$ is the polynomial function defined by setting 
		\begin{flalign}
			\label{kx} 
			\varsigma_k (x) = x^k, \qquad \text{for all $x \in \mathbb{R}$}.
		\end{flalign}
		
		\noindent For any function $h \in \mathcal{H}$, denote the associated multiplication operator by $\bm{h}$, defined by setting $\bm{h} f = hf$ for any $f : \mathbb{R} \rightarrow \mathbb{R}$; if $h \in \mathcal{H}$ is constant (that is, of the form $h = a \varsigma_0$ for some $a \in \mathbb{R}$), we identify $h = \bm{h}$. Further define the integral operator $\bm{\mathrm{T}}$ by setting  
		\begin{flalign}
			\label{operatort}
			\bm{\mathrm{T}} f (x) = 2 \displaystyle\int_{-\infty}^{\infty} \log |x-y| f(y) dy,
		\end{flalign}
		
		\noindent for any function $f : \mathbb{R} \rightarrow \mathbb{R}$ such that the above integral is finite for almost all $x \in \mathbb{R}$. 
		
	\end{definition} 
	
	The following two lemmas will be shown in \Cref{ProofT0} below. 
	
	\begin{lem} 
		\label{trhobounded} 
		
		The operator $\bm{\mathrm{T} \varrho_{\beta}}$ is a bounded operator on $\mathcal{H}$. 
		
	\end{lem}

	\begin{lem} 
		\label{trho} 
		
		The operator $(\theta^{-1} - \bm{\mathrm{T}} \bm{\varrho_{\beta}}): \mathcal{H} \rightarrow \mathcal{H}$ is a bijection.
	\end{lem}

	\begin{definition}
		\label{fdr} 
		
		We call $(\theta^{-1} - \bm{\mathrm{T}} \bm{\varrho_{\beta}})^{-1} : \mathcal{H} \rightarrow \mathcal{H}$ the \emph{dressing operator}. For any $f \in \mathcal{H}$, set
		\begin{flalign*}
			f^{\dr} = (\theta^{-1} - \bm{\mathrm{T}} \bm{\varrho_{\beta}})^{-1} f \in \mathcal{H}.
		\end{flalign*}
		
	\end{definition} 
	
	Expressions describing the Toda lattice will often involve the function $(\varsigma_0^{\dr})^{-1}$, so we must ensure that $\varsigma_0^{\dr} \ne 0$. The following lemma confirms this; its proof is given in \Cref{TProperty} below.
	
	\begin{lem}
		\label{positive0}
		
		There exists a constant $c > 0$ such that $\varsigma_0^{\dr} (x) \cdot \sgn(\alpha) > c$ for all $x \in \mathbb{R}$.  
	\end{lem}
	
	We next define the effective velocity for quasiparticles in the Toda lattice under thermal equilibrium.
	
	\begin{definition}
		\label{v} 
		
		Define the \emph{effective velocity} $v_{\eff} \in \mathcal{H}$ by setting $v_{\eff} (x) = \varsigma_0^{\dr} (x)^{-1} \cdot \varsigma_1^{\dr} (x)$ for each $x \in \mathbb{R}$.
		
	\end{definition}
	
	This definition of the effective velocity is the standard one in the physics literature; see \cite[Equation (6.20)]{HSIMS}. The fact that it satisfies \eqref{fvlambda} is implied by \Cref{vt} below.

	\subsection{Results}
	
	\label{MatrixL} 
	
	In this section we state our primary results. To do so, we must first recall the Lax matrix and associated localization centers for the Toda lattice. Throughout, we fix integers $N_1 \le N_2$ and set $N = N_2 - N_1 + 1$. Let $( \bm{a}(t); \bm{b}(t) ) \in \mathbb{R}_{\ge 0}^N \times \mathbb{R}^N$ be a pair of $N$-tuples indexed by $t \in \mathbb{R}_{\ge 0}$, where $\bm{a}(t) = ( a_j (t) )$ and $\bm{b}(t) = ( b_j (t) )$ satisfies the system \eqref{derivativepa} for each $(j, t) \in \llbracket N_1, N_2 \rrbracket \times \mathbb{R}$; assume $a_{N_2}(t) = 0$ for each $t \in \mathbb{R}_{\ge 0}$. The associated Lax matrix (introduced in \cite{LEI,CIS}) is defined as follows.
	
	\begin{definition}
		\label{matrixl} 
		
		For any real number $t \ge 0$, the \emph{Lax matrix} $\bm{L}(t) = [ L_{ij}] = [L_{ij}(t)]$ is an $N \times N$ real symmetric matrix, with entries indexed by $i, j \in \llbracket N_1, N_2 \rrbracket$, defined as follows. Set
		\begin{flalign*} 
			L_{ii} = b_i (t), \quad \text{for each $i \in \llbracket N_1, N_2 \rrbracket$}; \qquad L_{i,i+1} = L_{i+1,i} = a_i (t), \quad \text{for each $i \in \llbracket N_1, N_2-1 \rrbracket$}.
		\end{flalign*} 
		
		\noindent Also set $L_{ij} = 0$ for any $i,j \in \llbracket N_1, N_2 \rrbracket$ with $|i-j| \ge 2$. 
		
	\end{definition} 
	
	A fundamental feature of the Lax matrix is that its eigenvalues are preserved under the Toda dynamics \eqref{derivativepa}. This was originally due to \cite{LEI}; see also \cite[Section 2]{FMPL}.
	
	\begin{lem}[\cite{LEI,FMPL}]
		\label{ltt}
		
		For any real numbers $t, t' \in \mathbb{R}_{\ge 0}$, we have $\eig \bm{L} (t) = \eig \bm{L}(t')$. 
		
	\end{lem} 
	
	Lemma \ref{ltt} provides a large family of conserved quantities for the Toda lattice, given by the eigenvalues of the Lax matrix. However, these are ``non-local,'' in the sense that they depend on all of the Flaschka variables $(\bm{a}(t);\bm{b}(t))$, as opposed to only the $(a_i(t),b_i(t))$ for $i$ in some (uniformly) bounded interval. Still, in certain cases, they are ``approximately local,'' in that they only depend on a few entries of $\bm{L}(t)$, up to a small error. These entries will correspond to those on which the associated eigenvectors of $\bm{L}(t)$ are mainly supported. Such entries are called localization centers, given (in a more general context) by the below definition. 
	
	\begin{definition}
		\label{ucenter}
		
		Let $\bm{u} = ( u(N_1), u(N_1+1), \ldots , u(N_2) ) \in \mathbb{R}^N$ be a unit vector. For any $\zeta \in \mathbb{R}_{\ge 0}$, we call an index $\varphi \in \llbracket N_1, N_2 \rrbracket$ a \emph{$\zeta$-localization center} for $\bm{u}$ if $| u(\varphi) | \ge \zeta$. 
		
		Next, let $\bm{M} = [M_{ij}]$ be a symmetric $N \times N$ matrix, with entries indexed by $i, j \in \llbracket N_1, N_2 \rrbracket$. If $\lambda \in \eig \bm{M}$, then we call $\varphi \in \llbracket N_1, N_2 \rrbracket$ a $\zeta$-localization center for $\lambda$ with respect to $\bm{M}$ if $\varphi$ is a $\zeta$-localization center for some unit eigenvector $\bm{u} \in \mathbb{R}^N$ of $\bm{M}$ with eigenvalue $\lambda$. Further let $(\bm{u}_1, \bm{u}_2, \ldots , \bm{u}_N)$ denote an orthonormal eigenbasis of $\bm{M}$. We call a bijection $\varphi : \llbracket 1, N \rrbracket \rightarrow \llbracket N_1, N_2 \rrbracket$ a \emph{$\zeta$-localization center bijection} for $\bm{M}$ if $\varphi(j)$ is a $\zeta$-localization center for $\bm{u}_j$ for each $j \in \llbracket 1, N \rrbracket$.
		
		By {\cite[Lemma 2.7]{LC}}, any symmetric $N \times N$ matrix admits a $(2N)^{-1}$-localization center bijection.
		
	\end{definition} 
	
	Localization centers are in general not always unique. However, they are ``approximately unique'' (that is, up to some small error, and with high probability) when the entries $(\bm{a};\bm{b})$ of the Lax matrix are sufficiently random, such as under thermal equilibrium; see \Cref{centerdistance} below. We next require some notation on the Toda lattice at thermal equilibrium, which will frequently be adopted throughout the remainder of this paper. 
		
		\begin{assumption}
			\label{lbetaeta} 
			
			Fix real numbers\footnote{Throughout this paper, constants may depend on $\beta$ and $\theta$, even when not explicitly stated.} $\beta, \theta > 0$, and assume that $\alpha \ne 0$ (recall \eqref{alpha}). For each real number $t \ge0 $, let $\bm{L}(t) = [L_{ij}(t)]$ denote the Lax matrix for the Toda lattice $(\bm{a}(t);\bm{b}(t))$ on $\llbracket N_1, N_2 \rrbracket$ (as in \Cref{matrixl}). Set $\eig \bm{L}(t) = (\lambda_1, \lambda_2, \ldots , \lambda_N)$, which does not depend on $t$ by \Cref{ltt}. At $t = 0$, abbreviate $\bm{L} = \bm{L}(0)$ and $(\bm{a}; \bm{b}) = ( \bm{a}(0); \bm{b}(0))$. Assume that the initial data $(\bm{a}; \bm{b})$ is sampled under the thermal equilibrium $\mu_{\beta, \theta; N-1,N}$ from \Cref{mubeta2}. Let $( \bm{p}(s); \bm{q}(s) )$, over $s \in \mathbb{R}_{\ge 0}$, denote the Toda state space dynamics associated with $( \bm{a}(t); \bm{b}(t))$, as in \Cref{Open}. Let $T \ge 1$ and $\zeta \ge 0$ be real numbers satisfying
			\begin{flalign}
				\label{n1n2zetat} 
				N_1 \le -N(\log N)^{-1} \le N (\log N)^{-1} \le N_2; \qquad 1 \le T \le N (\log N)^{-6}; \qquad  \zeta \ge e^{-100(\log N)^{3/2}}.
			\end{flalign}
			
			\noindent For each $s \in \mathbb{R}$, let $\bm{u}_j (s) \in \mathbb{R}^N$ denote a unit eigenvector of $\bm{L}(s)$ with eigenvalue $\lambda_j$. Under the orthonormal basis $(\bm{u}_1 (s), \bm{u}_2 (s), \ldots , \bm{u}_N (s))$ of $\bm{L}(s)$, let $\varphi_s: \llbracket 1, N \rrbracket \rightarrow \llbracket N_1, N_2 \rrbracket$ be a $\zeta$-localization center bijection for $\bm{L}(s)$, and denote 
			\begin{flalign}
				\label{qjs2}
				Q_j (s) = q_{\varphi_s (j)} (s), \qquad \text{for each $(j, s) \in \llbracket 1, N \rrbracket \times \mathbb{R}_{\ge 0}$}.
			\end{flalign}

		\end{assumption}
		
		Let us briefly explain the bounds in \eqref{n1n2zetat}. The first indicates that $0$ is in the ``bulk'' of the interval $\llbracket N_1, N_2 \rrbracket$ (that is, not too close to its endpoints), and the second indicates that the time scale $T$ is sublinear in $N$ (so as to guarantee that the boundary of $\llbracket N_1, N_2 \rrbracket$ does not asymptotically affect its bulk under the Toda lattice); although we will not impose this, it is beneficial to imagine that $T \in [N^{\delta}, N^{1-\delta}]$ for some small constant $\delta>0$. The third ensures that $\zeta$ is not too small.
		
		Under \Cref{lbetaeta}, we view $Q_j (s)$ as the ``location on $\mathbb{R}$'' of the eigenvalue $\lambda_j \in \eig \bm{L} (s)$; see \cite[Sections 1 and 2]{LC} for a justification as to why, and how this is in agreement with the notions from the physics literature. The following theorem, to be proven in \Cref{ProofV} below, then states the asymptotic velocity of this location $Q_j$ is with high probability given by the effective velocity $v_{\eff} (\lambda_j)$ from \Cref{v}, if $\theta \le \theta_0$ is sufficiently small.

	\begin{thm}
		\label{vestimate} 
		
		Adopt \Cref{lbetaeta}, and fix $\beta > 0$. There is a constant $\theta_0 = \theta_0 (\beta) > 0$ so that, whenever $\theta \in (0, \theta_0)$, there exists a real number $c = c(\beta, \theta) > 0$ such that the following holds with probability at least $1 - c^{-1} e^{-c(\log N)^2}$. For any integer $j \in \llbracket 1, N \rrbracket$ satisfying 
		\begin{flalign} 
			\label{j0} 
			N_1 + T (\log N)^5 \le \varphi_0 (j) \le N_2 - T(\log N)^5,
		\end{flalign} 
		
		\noindent we have 
		\begin{flalign}
			\label{qtq0} 
			\displaystyle\sup_{t \in [0,T]} \big| Q_j (t) - Q_j (0) - tv_{\eff} (\lambda_j) \big| \le T^{1/2} (\log N)^{35}.
		\end{flalign} 
		
	\end{thm}

	We have two comments on \Cref{vestimate}. First, the condition \eqref{j0} in \Cref{vestimate} indicates that the initial location of $\lambda_j$ (through its localization center) is in the bulk of $\llbracket N_1, N_2 \rrbracket$ (otherwise, boundary effects on the interval might become more visible and make \eqref{qtq0} invalid). Second, the error given by the right side of \eqref{qtq0} below is $T^{1/2+o(1)}$ (if $T \in [N^{\delta}, N^{1-\delta}]$). Since the fluctuations of the location $Q_j$ are believed to be diffusive\footnote{Our methods are also quite suggestive of this; see \Cref{ZkEstimate} below.} \cite{HDIS,HOSQ,HSIMS}, this should be essentially optimal.

	\subsection{Notation} 
	
	\label{Notation} 
	
	For any point $z \in \mathbb{C}$ and set $\mathcal{A} \subseteq \mathbb{C}$, denote $\dist (z, \mathcal{A}) = \inf_{s \in \mathcal{A}} |z-s|$. Denote the complement of any event $\mathsf{E}$ by $\mathsf{E}^{\complement}$. Denote the set of $n \times n$ real matrices by $\Mat_{n \times n}$. For any $\bm{M} \in \Mat_{n \times n}$, denote its transpose by $\bm{M}^{\mathsf{T}}$. Denote the set of $n \times n$ real symmetric matrices by $\SymMat_{n \times n} = \{ \bm{M} \in \Mat_{n \times n} : \bm{M} = \bm{M}^{\mathsf{T}} \}$. As in \Cref{matrixl}, it will often be convenient to index the rows and columns of $n \times n$ matrices by index sets different from $\llbracket 1, n \rrbracket$. Given a nonempty index set $\mathscr{I} \subset \mathbb{Z}$ of size $n = |\mathscr{I}|$, let $\Mat_{\mathscr{I}}$ denote the set of $n \times n$ real matrices $\bm{M} = [M_{ij}]_{i,j \in \mathscr{I}} \in \Mat_{n \times n}$, whose rows and columns are indexed by $\mathscr{I}$; also let $\SymMat_{\mathscr{I}} = \Mat_{\mathscr{I}} \cap \SymMat_{n \times n}$ denote the set of real symmetric matrices whose rows and columns are indexed by $\mathscr{I}$.

	Throughout, given some integer parameter $N \ge 1$ and event $\mathsf{E}_N$ depending on $N$, we say that $\mathsf{E}_N$ \emph{holds with overwhelming probability} if there exists a constant $c > 0$ such that $\mathbb{P} [\mathsf{E}_N^{\complement}] \le c^{-1} e^{-c(\log N)^2}$. In this case, we call $\mathsf{E}_N$ \emph{overwhelmingly probable}. Observe that, whenever proving that $\mathsf{E}_N$ is overwhelmingly probable, we may assume $N \ge N_0$ is sufficiently large; we will often do this implicitly (and without comment) throughout this work.

	\subsection*{Acknowledgements} 
	
	The author heartily thanks Alexei Borodin, Jeremy Quastel, and Herbert Spohn for valuable conversations. The author is also very grateful to the anonymous referees for their helpful suggestions. This work was partially supported by a Clay Research Fellowship and a Packard Fellowship for Science and Engineering.

	\section{Proof Outline} 
	
	\label{Analysis0} 
	
	In this section we outline the proof of \Cref{vestimate}, providing an outline for organization of the remainder of the paper in the process. We adopt the notation and assumptions of that theorem throughout. We also assume for notational convenience that $\alpha > 0$, and we do not carefully track the precise power of $\log N$ that appears in the errors below (writing this exponent as $C$ throughout, which might change between appearances).

	\subsection{Regularization and Concentration Bounds}
	
	\label{Estimate1} 
	
	We begin with the asymptotic scattering relation for the $Q_j$ (see \Cref{ztlambda2} below), which states that we likely have
	\begin{flalign}
		\label{qktqk0} 
		Q_k (t) - Q_k (0) + 2 \displaystyle\sum_{i=1}^N (\mathbbm{1}_{Q_i (t) < Q_k (t)} - \mathbbm{1}_{Q_i (0) < Q_k (0)}) \cdot \log |\lambda_k - \lambda_i| = \lambda_k t + \mathcal{O} ( (\log N)^C). 
	\end{flalign} 
	
	\noindent An issue with \eqref{qktqk0} is that its left side is not linear in the $(Q_i)$, due to the presence of the indicator functions there, so we will regularize the latter using a cutoff function. To implement this, fix a real number $\mathfrak{M} \in [T^{1/2}, T]$, and let $\chi (x)$ be a smooth approximation for $\mathbbm{1}_{x > 0}$ on scale $\mathfrak{M}$, such that $\chi'$ is even. More specifically, we will have that $\chi' (x)$ is even in $x$; that $\chi (x) = \mathbbm{1}_{x>0}$ for $|x| > \mathfrak{M}$; and that $\chi' (x) = \mathcal{O} (\mathfrak{M}^{-1})$ and $\chi'' (x) = \mathcal{O} (\mathfrak{M}^{-2})$ for all $x \in \mathbb{R}$. To avoid the singularity of $\log |\lambda_i - \lambda_k|$ at $i=k$, it will also be convenient to introduce the function $\mathfrak{l} (x) = (\log|x^2+\mathfrak{d}^2|)/2$, for some very small real number $\mathfrak{d} \ll 1$ (we will in particular take $\mathfrak{d} = e^{-5(\log N)^2}$).
	
	Under this notation, we will first show the concentration bound (\Cref{chiconcentration} below), which indicates that for any $s \in [0, t]$ we with high probability have
	\begin{flalign}
		\label{qchiq} 
		\displaystyle\sum_{i=1}^N \big( \chi (Q_k(s)-Q_i(s)) - \mathbbm{1}_{Q_k (s) - Q_i (s) > 0} \big) \cdot \mathfrak{l} (\lambda_k - \lambda_i) = \mathcal{O} (\mathfrak{M}^{1/2} (\log N)^C).
	\end{flalign}
	
	\noindent This is will essentially be a consequence of the more general concentration bound \Cref{concentrationh3}, to be proven in \Cref{Localization2} and \Cref{Proof12} below, indicating that we likely have
	\begin{flalign}
		\label{qgq2} 
		\begin{aligned} 
		\displaystyle\sum_{i=1}^N F(& \lambda_i) \cdot G(Q_k (s) - Q_i(s)) \cdot \mathfrak{l} (\lambda_k-\lambda_i) \\
		& = \displaystyle\int_{-\infty}^{\infty} F(\lambda) \log |\lambda_k - \lambda| \varrho(\lambda) d \lambda \displaystyle\int_{-\infty}^{\infty} G(\alpha q) dq + \mathcal{O} (\mathfrak{M}^{1/2} (\log N)^C),
		\end{aligned} 
	\end{flalign}
	
	\noindent for functions $F$ and $G$ satisfying $\supp G \subseteq [-\mathfrak{M}, \mathfrak{M}]$ (and other properties listed in \Cref{fgab} below). Indeed, taking $F = 1$ and $G (x) = \chi (x) - \mathbbm{1}_{x>0}$ in \eqref{qgq2} yields \eqref{qchiq}, as this $G$ is odd (so that the second integral on the right side of \eqref{qgq2} is equal to $0$). 
	
	Before continuing to analyze \eqref{qktqk0}, let us briefly explain why \eqref{qgq2} should hold (as it will be used again below). Set $Q_i = Q_i (s)$. We first use the ``approximate locality'' of the Lax matrix eigenvalues (\Cref{lleigenvalues} below); this indicates that, up to small error and with high probability, a given $\lambda_i$ only depends on the entries of $\bm{L}(s) $ with indices close to the associated localization center $\varphi_s (i)$. In particular, $\lambda_i$ essentially depends only on the increments $q_{j+1} - q_j$ for $j \approx \varphi_s (i)$, or equivalently on $Q_j - Q_i$ for $Q_j \approx Q_i$ (\Cref{estimateik} below). Hence, the $(\lambda_i)$-dependent parts on the left side of \eqref{qgq2} should approximately decouple from the $(Q_i)$-dependent parts, namely,
	\begin{flalign}
		\label{sum3} 
		\displaystyle\sum_{i=1}^N F(& \lambda_i) \cdot G(Q_k - Q_i) \cdot \mathfrak{l} (\lambda_k-\lambda_i) \approx \displaystyle\frac{1}{N} \displaystyle\sum_{i=1}^N F(\lambda_i) \cdot \mathfrak{l} (\lambda_k - \lambda_i) \cdot \displaystyle\sum_{i=1}^N G(Q_k - Q_i).
	\end{flalign}
	
	\noindent Next, the average spacing between consecutive particles in the Toda lattice in thermal equilbrium is given by $\mathbb{E}[q_{i+1}-q_i] = \alpha$ (\eqref{qiqjs4} below). This indicates that
	\begin{flalign}
		\label{sum4} 
		\displaystyle\sum_{i=1}^N G(Q_k - Q_i) \approx \displaystyle\int_{-\infty}^{\infty} G(\alpha q) dq.
	\end{flalign}
	
	\noindent Moreover, it was shown in \cite{DSME} that $\varrho$ is the limiting empirical spectral distribution for the Lax matrix $\bm{L}$ (\Cref{lf} below). As such (and using the fact that $\mathfrak{l}(x) \approx \log |x|$),
	\begin{flalign}
		\label{sum5} 
		\displaystyle\frac{1}{N} \displaystyle\sum_{i=1}^N F(\lambda_i) \cdot \mathfrak{l} (\lambda_k - \lambda_i) \approx \displaystyle\int_{-\infty}^{\infty} F(\lambda) \log |\lambda_k - \lambda| \varrho (\lambda) d\lambda.
	\end{flalign}
	
	\noindent Combining \eqref{sum3}, \eqref{sum4}, and \eqref{sum5} yields \eqref{qgq2}. The error of about $\mathfrak{M}^{1/2}$ on the right side of \eqref{qgq2} arises from the fact that the left side of \eqref{qgq2} likely constitutes $\mathcal{O}(\mathfrak{M})$ nonzero terms (as $\supp G \subseteq [-\mathfrak{M}, \mathfrak{M}]$), which are nearly independent (by the approximate locality of the $(\lambda_i)$, with the independence of the entries in $\bm{L}$).

	\subsection{Proxy Dynamics and Their Analysis}
	
	\label{Estimate2} 
	
	 Inserting \eqref{qchiq} into \eqref{qktqk0} (and replacing $\log |\lambda_k - \lambda_i|$ there with $\mathfrak{l}(\lambda_k-\lambda_i)$) yields 
	\begin{flalign}
		\label{qktqk02} 
		\begin{aligned} 
		Q_{k} (t) - Q_{k} (0) + 2 \displaystyle\sum_{i=1}^N \big( \chi \big(Q_{k} (t) - Q_i (t) \big) - \chi \big(Q_{k} (0) - Q_{i} (& 0) \big) \big) \cdot \mathfrak{l} (\lambda_k - \lambda_i) \\
		& = \lambda_k t + \mathcal{O} ( \mathfrak{M}^{1/2} (\log N)^C).
		\end{aligned} 
	\end{flalign} 
	
	\noindent Next, we would like to differentiate both sides of \eqref{qktqk02} in $t$, which would yield for $s \in [0, t]$ that
	\begin{flalign*}
		Q_{k}' (s) + 2 \displaystyle\sum_{i=1}^N \big( Q_{k}' (s) - Q_{i}' (s) \big) \cdot \chi' \big( Q_{k} (s) - Q_{i} (s) \big) \cdot \mathfrak{l} (\lambda_k - \lambda_i) \approx \lambda_k.
	\end{flalign*} 
	
	\noindent However, this is not quite accurate, as the error $\mathcal{O}( \mathfrak{M}^{1/2} (\log N)^C)$ on the right side of \eqref{qktqk02} is not differentiable in $t$ (in fact, it is not continuous in $t$, since the $\varphi_t (i)$ and thus $Q_i (t)$ are not). To circumvent this, we instead introduce a ``proxy dynamic'' $(\mathfrak{Q}_k (s))$ satisfying the equation (that is essentially the above one that we would have liked for $(Q_k(s))$ to satisfy, but without the error)
	\begin{flalign}
		\label{qk2} 
		\mathfrak{Q}_k' (s) + 2 \displaystyle\sum_{i=1}^N \big( \mathfrak{Q}_k' (s) - \mathfrak{Q}_i' (s) \big) \cdot \chi' \big( \mathfrak{Q}_k (s) - \mathfrak{Q}_i (s) \big) \cdot \mathfrak{l} (\lambda_k - \lambda_i) = \lambda_k;
	\end{flalign}
	
	\noindent see \Cref{qjt2} below.\footnote{The definition there is in fact slightly different, since it involves the reindexing $j = \varphi_0 (k)$ that orders the Lax matrix eigenvalues by their initial positions (it also involves additional boundary terms when $\mathfrak{Q}_k$ is too close to leftmost or rightmost particles, which are asymptotically irrelevant but convenient for the proofs). For simplicity, we ignore that reindexing here.} We then show as \Cref{q2q} that these proxy dynamics are indeed close to the original ones, namely, 
	\begin{flalign}
		\label{qkqk} 
		\mathfrak{Q}_k (s) = Q_k (s) + \mathcal{O} (\mathfrak{M}^{1/2} (\log N)^C).
	\end{flalign}  
	
	The benefit of \eqref{qk2} is that it is linear in $(\mathfrak{Q}_k')$ if we view the $(\mathfrak{Q}_k)$ as fixed. Let us explain why $v_{\eff} (\lambda_k)$ should be an approximate solution (for $\mathfrak{Q}_k' (s)$) of \eqref{qk2}. Replacing $\mathfrak{Q}_k' (s)$ with $v_{\eff} (\lambda_k)$ (and $\chi' (\mathfrak{Q}_k (s) - \mathfrak{Q}_i (s))$ with $\chi' (Q_k (s) - Q_i (s))$) on the left side of \eqref{qk2} gives 
	\begin{flalign}
		\label{v2} 
		\begin{aligned}
		& v_{\eff} (\lambda_k) + 2 \displaystyle\sum_{i=1}^N \big( v_{\eff} (\lambda_k) - v_{\eff} (\lambda_i) \big) \cdot \chi' \big( Q_k (s) - Q_i (s) \big) \cdot \mathfrak{l} (\lambda_k - \lambda_i) \\
		& \quad = v_{\eff} (\lambda_k) + 2 \displaystyle\int_{-\infty}^{\infty} \big( v_{\eff} (\lambda_k) - v_{\eff} (\lambda) \big)  \log |\lambda_k - \lambda| \varrho (\lambda) d \lambda \displaystyle\int_{-\infty}^{\infty} \chi' (\alpha q) dq + \mathcal{O} (\mathfrak{M}^{-1/2} (\log N)^C),
		\end{aligned} 
	\end{flalign}
	
	\noindent where the approximation holds with high probability due to the ($F \in \{ v_{\eff}, 1 \} $ and $G = \chi'$ case of the) concentration bound \eqref{qgq2}; observe that the error on the right side of \eqref{v2} improves on that in \eqref{qgq2} by a factor of $\mathfrak{M}^{-1}$ since $\chi' = \mathcal{O}(\mathfrak{M}^{-1})$. Since the total integral of $\chi' (\alpha q)$ is equal to $\alpha^{-1}$ (as $\chi (x) = \mathbbm{1}_{x>0}$ for $|x| > \mathfrak{M}$), we deduce recalling the definition \eqref{operatort} of $\bm{\mathrm{T}}$ that 
	\begin{flalign*}
		v_{\eff} (\lambda_k &) + 2 \displaystyle\sum_{i=1}^N \big( v_{\eff} (\lambda_k) - v_{\eff} (\lambda_i) \big) \cdot \chi' \big( Q_k (s) - Q_i (s) \big) \cdot \mathfrak{l} (\lambda_k - \lambda_i) \\
		& \quad = v_{\eff} (\lambda_k) + \alpha^{-1} \cdot \big( v_{\eff} (\lambda_k) \cdot \bm{\mathrm{T}} \varrho (\lambda_k) - \bm{\mathrm{T}} \bm{\varrho} v_{\eff} (\lambda_k) \big) + \mathcal{O} (\mathfrak{M}^{-1/2} (\log N)^C),
	\end{flalign*} 
	
	\noindent where $\bm{\mathrm{T}} \bm{\varrho} v_{\eff}$ denotes the composition of the operators $\bm{\mathrm{T}}$ and $\bm{\varrho}$ applied to the function $v_{\eff}$. Using the identity (see \Cref{rho0} and \Cref{vt} below) for $v_{\eff}$ given by 
	\begin{flalign}
		\label{vlambda01} 
		v_{\eff} (\lambda_k) + \alpha^{-1} \cdot \big( v_{\eff} (\lambda_k) \cdot \bm{\mathrm{T}} \varrho (\lambda_k) - \bm{\mathrm{T}} \bm{\varrho} v_{\eff} (\lambda_k) \big) = \lambda_k,
	\end{flalign} 
	
	\noindent it follows that 
	\begin{flalign}
		\label{v3} 
		v_{\eff} (\lambda_k) + 2 \displaystyle\sum_{i=1}^N \big( v_{\eff} (\lambda_k) - v_{\eff} (\lambda_i) \big) \cdot \chi' \big( Q_k (s) - Q_i (s) \big) \cdot \mathfrak{l} (\lambda_k - \lambda_i) = \lambda_k + \mathcal{O} (\mathfrak{M}^{-1/2} (\log N)^C),
	\end{flalign}
	
	\noindent which indeed verifies that $v_{\eff}$ is an approximate solution for $\mathfrak{Q}_k'$ of \eqref{qk2}. 
	
	We would like to use this to deduce that $\mathfrak{Q}_k' (s) \approx v_{\eff} (\lambda_k)$. To that end, denote $\mathfrak{w}_i = \mathfrak{Q}_i' (s) - v_{\eff} (\lambda_i)$. Then, subtracting \eqref{v3} from \eqref{qk2} (and again replacing $\chi' (\mathfrak{Q}_k (s) - \mathfrak{Q}_i (s))$ with $\chi' (Q_k (s) - Q_i (s))$) yields
	\begin{flalign*} 
	& \mathfrak{w}_k + 2 \displaystyle\sum_{i=1}^N ( \mathfrak{w}_k - \mathfrak{w}_i) \cdot \chi' \big( Q_k (s) - Q_i (s) \big) \cdot \mathfrak{l} (\lambda_k - \lambda_i) = \mathcal{O}(\mathfrak{M}^{-1/2} (\log N)^C).
	\end{flalign*} 
	
	\noindent Viewing $(Q_k)$ as fixed and denoting $\bm{\mathfrak{w}} = (\mathfrak{w}_k)$, this is a matrix equation of the form $\bm{S} \bm{\mathfrak{w}} = \mathcal{O} (\mathfrak{M}^{-1/2} (\log N)^C)$ for some explicit matrix $\bm{S}$ (see \eqref{sij2} below). We would like for $\bm{S}$ to be (quantitatively) invertible, which we believe to be true, at least for some choice of $\chi$ (indeed, its definition allows quite a bit of freedom in fixing $\chi$). While we do not know how to prove this in full generality, we do if $\theta \le \theta_0$ is sufficiently small, in which case $\bm{S}$ is in fact strictly diagonally dominant (see \Cref{sk1k2inverse} and \Cref{qjt}). So, for $\theta \le \theta_0$, it follows that $|\mathfrak{Q}_k' (s) - v_{\eff}(\lambda_k)| = |\mathfrak{w}_k| = \mathcal{O} (\mathfrak{M}^{-1/2} (\log N)^C)$; see \Cref{qv}. Integrating this over $s \in [0, t]$ yields 
	\begin{flalign*} 
		|\mathfrak{Q}_k (t) - \mathfrak{Q}_k (0) - t v_{\eff} (\lambda_k)| \le \mathcal{O} (t\mathfrak{M}^{-1/2} (\log N)^C),
	\end{flalign*} 	
	
	\noindent which with \eqref{qkqk} implies 
	\begin{flalign*}
		|Q_k (t) - Q_k (t) - t v_{\eff} (\lambda_k)| \le \mathcal{O} ( \mathfrak{M}^{1/2} (\log N)^C + t \mathfrak{M}^{-1/2} (\log N)^C).
	\end{flalign*}
	
	\noindent Taking $\mathfrak{M} = t$ then yields \Cref{vestimate}.

	\subsection{Heuristics for Fluctuations}

	\label{ZkEstimate} 
	
	Although we will not pursue a mathematical justification in this work, in this section we briefly provide some heuristic commentary on the fluctuations for the $Q_k$; this will suggest that they are diffusive, as also predicted in \cite{HDIS,HOSQ} and \cite[Chapter 15]{HSIMS}. In what follows, we will write $A \approx B$ if $A$ and $B$ agree up to the diffusive scale, namely, if $A = B + o(t^{1/2})$. Instead of fixing $\mathfrak{M} = t$, we will take $t^{1/2} \ll \mathfrak{M} \ll t$, so that $t^{-1} \mathfrak{M}^{1/2} (\log N)^C \approx 0$.
	 
	 Now let $Z_k = Z_k (t)$ denote the fluctuations of $Q_k (t)$, that is, define it to satisfy 
	 \begin{flalign*} 
	 	Q_k (t) = Q_k (0) + t v_{\eff} (\lambda_k) + t^{1/2} Z_k;
	 \end{flalign*} 
	 
	 \noindent we would like to see that $Z_k$ is a random variable of order $1$. Inserting this into \eqref{qktqk02}, we find 
	\begin{flalign*}
		t v_{\eff} (\lambda_k) + t^{1/2} Z_k \approx t \lambda_k + 2 \displaystyle\sum_{i = N_1}^{N_2} \mathfrak{l} (\lambda_k - \lambda_i) \cdot \big( \chi(Q_k (0) - Q_i (0)) - \chi (Q_k(t)-Q_i(t)) \big).
	\end{flalign*} 
	
	\noindent Taylor expanding gives
	\begin{flalign*}
		\chi (Q_k (t) - Q_i (t)) & = \chi \big( Q_k (0) - Q_i (0) + t \big( v_{\eff} (\lambda_k) -  v_{\eff} (\lambda_i) \big) \big) \\
		& \qquad + t^{1/2} (Z_k - Z_i) \cdot \chi' \big( Q_k (0) - Q_i (0) + t \big( v_{\eff} (\lambda_k) - v_{\eff} (\lambda_i) \big) \big) + \mathcal{O} (t \mathfrak{M}^{-2}),
	\end{flalign*}
	
	\noindent and thus (using the facts that the sums over $i$ are supported on $\mathcal{O}(\mathfrak{M})$ terms and $t \mathfrak{M}^{-1} = o(t^{1/2})$)
	\begin{flalign}
		\label{vzequation} 
		\begin{aligned} 
		t & \big( v_{\eff} (\lambda_k) - \lambda_k \big) + t^{1/2} Z_k \\
		& \approx 2 \displaystyle\sum_{i=N_1}^{N_2} \mathfrak{l}(\lambda_k-\lambda_i) \cdot \Big( \chi (Q_k (0) - Q_i (0)) - \chi \big(Q_k (0) - Q_i(0) + t \big( v_{\eff}(\lambda_k)-v_{\eff}(\lambda_i) \big) \big) \Big) \\
		& \qquad + 2t^{1/2} \displaystyle\sum_{i=N_1}^{N_2} (Z_i - Z_k) \cdot \mathfrak{l} (\lambda_k - \lambda_i) \cdot \chi'  \big( Q_k(0)-Q_i(0) + t \big(v_{\eff}(\lambda_k)-v_{\eff}(\lambda_i) \big) \big).
		\end{aligned}
	\end{flalign}
	
	 Let us Taylor expand the second term of the above statement. For an index $i \in \llbracket N_1, N_2 \rrbracket$ to contribute to this sum, it must hold that $\chi' (Q_k (0) - Q_i (0) + t(v_{\eff} (\lambda_k) - v_{\eff} (\lambda_i))) \ne 0$, meaning that $|k-i| = \mathcal{O} (t)$, in which case we will have $Q_k (0) - Q_i (0) = \alpha ( \varphi(k)-\varphi(i)) + \mathcal{O} (t^{1/2})$ (see \eqref{qiqjs4}, ignoring the logarithmic corrections). Therefore,
	\begin{flalign*}
		\chi' \big( Q_k (0 & ) - Q_i (0) + t\big( v_{\eff} (\lambda_k) - v_{\eff} (\lambda_i) \big) \big) \\
		& = \chi' \big( \alpha(\varphi(k) - \varphi(i)) + t \big( v_{\eff} (\lambda_k) - v_{\eff} (\lambda_i) \big) \big) + \mathcal{O} ( t^{1/2} \mathfrak{M}^{-2}),
	\end{flalign*}
	
	\noindent by a Taylor expansion. Upon insertion into the second sum on the right side of \eqref{vzequation}, this yields (again using the facts that this sum is supported on $\mathcal{O}(\mathfrak{M})$ terms and $t \mathfrak{M}^{-1} = o(t^{1/2})$)
		\begin{flalign}
		\label{vzequation2} 
		\begin{aligned} 
			t & \big( v_{\eff} (\lambda_k) - \lambda_k \big) + t^{1/2} Z_k \\
			& \approx 2 \displaystyle\sum_{i=N_1}^{N_2} \mathfrak{l}(\lambda_k-\lambda_i) \cdot \Big( \chi (Q_k (0) - Q_i (0)) - \chi \big(Q_k (0) - Q_i(0) + t \big( v_{\eff}(\lambda_k)-v_{\eff}(\lambda_i) \big) \big) \Big) \\
			& \qquad + 2t^{1/2} \displaystyle\sum_{i=N_1}^{N_2} (Z_i - Z_k) \cdot \mathfrak{l} (\lambda_k - \lambda_i) \cdot \chi'  \big( \alpha (\varphi(k) - \varphi(i) ) + t \big(v_{\eff}(\lambda_k)-v_{\eff}(\lambda_i) \big) \big).
		\end{aligned}
	\end{flalign}
	
	\noindent The benefit of \eqref{vzequation2} is that it provides a linear equation for the $(Z_i)$ that only depends on the Toda lattice through its initial data $(Q_i (0))$ (and not on its evolution); in this way, to analyze the $(Z_i)$ one requires purely static (as opposed to dynamical) information about the Lax matrix $\bm{L}$. 
		
	Now let us analyze the first sum on the right side of \eqref{vzequation2}, which we denote by $\Psi$. The calculations around \eqref{vlambda01} suggest that $\mathbb{E}[\Psi] \approx t v_{\eff} (\lambda_k) - t\lambda_k$, and the approximate independence between the $(\lambda_i)$ mentioned at the end of \Cref{Estimate1} suggest that its fluctuations should be diffusive and converge in the scaling limit to a Gaussian process\footnote{More specifically, it is plausible by the explicit form of $\Psi$ that its flucutations converge to (a variant of) a L\'{e}vy--Chentsov field \cite[Equation (30)]{HRHF}, which also appears in the fluctuations of the hard rods model.} $\Xi$, so that 
	\begin{flalign*} 
		\Psi \approx tv_{\eff} (\lambda_k) - t\lambda_k + t^{1/2} \cdot \Xi_k. 
	\end{flalign*} 
	
	\noindent Inserting this into \eqref{vzequation2} gives
	\begin{flalign}
		\label{zkxik} 
		Z_k + 2\displaystyle\sum_{i=N_1}^{N_2} (Z_k - Z_i) \cdot \mathfrak{l} (\lambda_k - \lambda_i) \cdot \chi' \big( \alpha (\varphi(k) - \varphi(i) ) + t \big( v_{\eff} (\lambda_k) - v_{\eff} (\lambda_i) \big) \big) = \Xi_k + o(1).
	\end{flalign}
	
	\noindent This provides a linear equation for the $(Z_k)$ in terms of the $(\Xi_k)$, thereby indicating that $Z_k$ should be a random variable of order $1$ (as $\Xi_k$ is). So, the fluctuations of $Q_k$ should indeed be diffusive. One might use \eqref{zkxik} further to predict an expression for the fluctuations $(Z_k)$ through the random field $(\Xi_k)$, analogously to what was done in \Cref{Estimate2} (or \cite[Appendix B]{LC}), but we will not pursue this further here.

	\section{Miscellaneous Preliminaries} 
	
	\label{MatrixLattice}

	\subsection{Properties of $\bm{\mathrm{T}}$} 
	
	\label{TProperty}
	
	Here, we state various properties of the operator $\bm{\mathrm{T}}$ from \Cref{TOperator}. Throughout, we recall the constant $\alpha$ and functions $\varrho$ and $\varrho_{\beta}$ from \Cref{frho}, as well as the Hilbert space $\mathcal{H}$, associated inner product, and integral operator $\bm{\mathrm{T}}$ from \Cref{tbeta}. 
	
	We begin with the following lemma that bounds $\varrho_{\beta}$, its derivative, and $\varrho$. The first two estimates in \eqref{estimate0} are due to \cite[Lemma 2.2]{RET}; we establish the third in \Cref{ProofIntegral} below. 
	
	\begin{lem}
		
		\label{rhoexponential} 
		
		There exists a constant $C > 1$ such that 
		\begin{flalign}
			\label{estimate0} 
			\begin{aligned} 
			\varrho_{\beta} (x) \le C (|x &|+1)^{2 \theta} e^{-\beta x^2/2}; \qquad \varrho_{\beta}' (x) \le C(|x|+1) \varrho_{\beta} (x); \\
			& \varrho (x) \le C (|x|+1)^{2\theta+1} e^{-\beta x^2/2}.
			\end{aligned} 
		\end{flalign} 
	\end{lem} 
	
	We establish the following relation between $\varrho$ and $\varrho_{\beta}$ in \Cref{ProofIntegral} below.

	\begin{lem}
		\label{rhorho} 
		
		For any $x \in \mathbb{R}$, we have $\varrho(x) = \theta \cdot ( \bm{\mathrm{T}} \varrho (x) + \alpha ) \cdot \varrho_{\beta} (x)$.
		
	\end{lem}

	The next lemma states that $\bm{\mathrm{T}} \varrho (x) + \alpha$ is uniformly bounded away from $0$. We provide its proof in \Cref{ProofBound} below.
	
	\begin{lem}
		\label{alphat}
		
		There exists a constant $c > 0$ such that $\bm{\mathrm{T}} \varrho (x) + \alpha > c$, for any $x \in \mathbb{R}$.
		
	\end{lem}
	
	The following corollary follows quickly from \Cref{rhorho}; we establish it in \Cref{Proofrho} below.
	
	\begin{cor}
		\label{rho0}
		
		For any $x \in \mathbb{R}$, we have 
		\begin{flalign} 
			\label{rho2} 
			\varrho (x) = \alpha \cdot \varsigma_0^{\dr} (x) \cdot \varrho_{\beta} (x), \qquad \text{and} \qquad \varsigma_0^{\dr} (x) = \displaystyle\frac{\theta}{\alpha} \cdot \big( \bm{\mathrm{T}} \varrho (x) + \alpha \big).
		\end{flalign} 
		
	\end{cor}
	
	\begin{proof}[Proof of \Cref{positive0}] 
		
		This follows from the second statement in \eqref{rho2} and \Cref{alphat}.
	\end{proof}
	
	The following lemma provides an alternative expression for the effective velocity $v_{\eff}$ from \Cref{v}. It was originally shown as \cite[Equations (6.20) and (6.21)]{HSIMS}, though we provide its quick proof in \Cref{Proofrho} below.

	\begin{lem}[{\cite[Equations (6.20) and (6.21)]{HSIMS}}] 
		\label{vt} 
		
		We have $(  \theta^{-1} \cdot \bm{\varsigma_0^{\dr}} - \alpha^{-1} \cdot \bm{\mathrm{T}} \bm{\varrho}) v_{\eff} = \varsigma_1$.
	
	\end{lem}

	Given a function $f \in \mathcal{H}$, we next have the following pointwise estimates on $f^{\dr}$ and its derivative (in terms of $f$). Their proofs will be given in \Cref{DerivativefProof} below.
	
	\begin{lem} 
		\label{xf2} 
		
		There exists a constant $C>1$ such that the following holds. For any function $f \in \mathcal{H}$ and real number $x \in \mathbb{R}$, we have 
		\begin{flalign*} 
			|f^{\dr} (x)| \le C \cdot |f(x)| + C \| f \|_{\mathcal{H}} \cdot \log (|x|+2).
		\end{flalign*} 
		
	\end{lem}

	\begin{lem} 
		\label{derivativefg}
		
		There exists a constant $C>1$ such that the following holds. For any differentiable function $f \in \mathcal{H}$ such that $f' \in \mathcal{H}$, and real number $x \in \mathbb{R}$, we have 
		\begin{flalign*}
			|\partial_x f^{\dr} (x)| \le C \cdot |f'(x)| + C (  \| f' \|_{\mathcal{H}} + \| f \|_{\mathcal{H}} ) \cdot \log (|x|+2).
		\end{flalign*} 
	\end{lem} 
	
	\begin{cor}
		\label{derivativev}
		
		There exists a constant $C>1$ such that, for any real number $A \ge 2$, we have
		\begin{flalign*}
			\displaystyle\sup_{|x| \le A} | v_{\eff} (x) | \le CA; \qquad \displaystyle\sup_{|x| \le A} | \partial_x v_{\eff} (x) | \le CA \log A. 
		\end{flalign*}
		
	\end{cor} 
	
	The following lemma, to be shown in \Cref{IntegralEstimate} below, lower bounds a particular integral if $\theta$ is sufficiently small (and will be used to verify strict diagonal dominance of a certain matrix; see \Cref{qjt} below).
	
	\begin{lem}
		
		\label{beta0theta}
		
		Fix $\beta > 0$. There exists a real number $\theta_0 = \theta_0 (\beta) > 0$ such that the following holds whenever $\theta \in (0, \theta_0)$. Let $\mathfrak{d} \in [0,1)$ be a real number, and define the function $\mathfrak{l} = \mathfrak{l}_{\mathfrak{d}}: \mathbb{R} \rightarrow \mathbb{R}$ by setting $\mathfrak{l}(x) = \log (x^2+\mathfrak{d}^2)/2$ for each $x \in \mathbb{R}$. Then, for any $\lambda \in \mathbb{R}$, we have 
		\begin{flalign*}
			\Bigg| 2\alpha^{-1} \displaystyle\int_{-\infty}^{\infty} \mathfrak{l} (x - \lambda)  \varrho (x) dx + 1 \Bigg| \ge 2 |\alpha|^{-1} \displaystyle\int_{-\infty}^{\infty} | \mathfrak{l}(x-\lambda) |  \varrho (x) dx + \displaystyle\frac{1}{2}. 
		\end{flalign*}
		
	\end{lem}

	\subsection{Random Lax Matrices} 
	
	\label{Localization}
	
	In this section we describe various properties of Lax matrices whose Flaschka variables are sampled from thermal equilibrium. The following two lemmas approximate the distance between Toda particles $(q_i (t))$ under thermal equilibrium initial data. The first does this for $t=0$; the second does this for general $t \ge 0$ (in which case one requires a restriction on the particle indices $i$).
	
	\begin{lem}[{\cite[Lemma 3.12]{LC}}]
		\label{qij} 
		
		Adopt \Cref{lbetaeta}. There exists a constant $c > 0$ such that the following holds. For any distinct indices $i, j \in \llbracket N_1, N_2 \rrbracket$ and real number $R \ge 1$, we have 
		\begin{flalign*}
			\mathbb{P} \big[ | q_j (0) - q_i (0) - \alpha(j-i) | \ge R \big] \le 2 (e^{-cR^2/|i-j|} + e^{-cR}).
		\end{flalign*}
	\end{lem} 
	
	\begin{lem}[{\cite[Lemma 7.2]{LC}}]
	
		\label{qijsalpha} 
		
		Adopt \Cref{lbetaeta}. The following two statements hold with overwhelming probability.

		\begin{enumerate} 
			
			\item For any $s \in [0,T]$ and $i, j \in \llbracket N_1 + T(\log N)^3, N_2 - T (\log N)^3 \rrbracket$, we have
			\begin{flalign}
				\label{qiqjs4}
				\big| q_i (s) - q_j (s) - \alpha (i-j) \big| \le |i-j|^{1/2} (\log N)^2.
			\end{flalign}
			
			\item For any $s \in [0,T]$ and $i \in \llbracket N_1, N_2 \rrbracket$ with $|i-j| \ge T (\log N)^5$, we have 
			\begin{flalign}
				\label{qiqjs5}
				\big(  q_i (s) - q_j (s) \big) \cdot \sgn (\alpha i - \alpha j) \ge \displaystyle\frac{|\alpha|}{2} \cdot |i-j|.
			\end{flalign}
			
		\end{enumerate} 
		
	\end{lem} 	
	
	Next, we will frequently require that the eigenvalues of a Lax matrix are bounded or separated from each other. The following definition provides notation for these two events, and the lemma below it states that a Lax matrix under thermal equilibrium likely satisfies both. 
	
	\begin{definition}
		\label{adelta}
		
		Fix real numbers $A, \delta > 0$; let $\mathscr{I}$ denote an index set; and let $\bm{M} = [M_{ij}] \in \SymMat_{\mathscr{I}}$. Define the events
		\begin{flalign*}
			\mathsf{BND}_{\bm{M}} (A) = \Bigg\{ \displaystyle\max_{i,j \in \mathscr{I}} |M_{ij}| \le A \Bigg\} \cap \Bigg\{ \displaystyle\max_{\lambda \in \eig \bm{M}} |\lambda| \le A \Bigg\}; \quad \mathsf{SEP}_{\bm{M}} (\delta) = \Bigg\{ \displaystyle\min_{\substack{\nu, \nu' \in \eig \bm{M} \\ \nu \ne \nu'}} |\nu - \nu'| \ge \delta \Bigg\}. 
		\end{flalign*}
	\end{definition} 
	
	\begin{lem}[{\cite[Lemmas 3.15 and 3.18]{LC}}]
		\label{l0eigenvalues}
		
		Adopt \Cref{lbetaeta}. There exists a constant $c > 0$ such that the following two statements hold. 
		
		\begin{enumerate} 
			\item For any real number $A \ge 1$, we have 
		\begin{flalign*} 
			\mathbb{P} \Bigg[ \bigcap_{t \in \mathbb{R}_{\ge 0}} \mathsf{BND}_{\bm{L}(t)} (A) \Bigg] \ge 1 - c^{-1} N e^{-cA^2}.
		\end{flalign*} 
		
		\item For any real number $\delta > 0$, we have $\mathbb{P} [ \mathsf{SEP}_{\bm{L}} (\delta) ] \ge 1 - c^{-1} (\delta N^3 + e^{-cN^2})$.
		\end{enumerate} 
		
	\end{lem} 
	
	The next lemma realizes the function $\varrho$ from \Cref{frho} as the limiting spectral distribution of a random Lax matrix under thermal equilibrium; it follows from \cite[Lemma 4.3]{DSME}.\footnote{In fact, \cite{DSME} considers the Lax matrix $\hat{\bm{L}} = [\hat{L}_{ij}] \in \SymMat_{N \times N}$ for the periodic Toda lattice, which differs from $\bm{L}$ only in its $(1,N)$ and $(N,1)$ entries. There, we instead have $\hat{L}_{1N} = \hat{L}_{N1} = b_N$, where $b_N$ has the same law as (and is independent from) the $(b_i(0))$. It is quickly seen from the Weyl interlacing inequality that this alteration does not affect the convergence \eqref{rhofconverge}.}

	\begin{lem}[{\cite[Lemma 4.3]{DSME}}]
		
		\label{lf} 
		
		Adopt \Cref{lbetaeta}, and denote $\bm{L} = \bm{L}_N$. For any bounded, continuous function $f : \mathbb{R} \rightarrow \mathbb{R}$, we have
		\begin{flalign}
			\label{rhofconverge}
			\displaystyle\lim_{N\rightarrow \infty} \mathbb{E} \Bigg[ \displaystyle\frac{1}{N} \cdot \displaystyle\sum_{\lambda \in \eig \bm{L}_N} f(\lambda) \Bigg] = \displaystyle\int_{-\infty}^{\infty} f(\lambda) \varrho(\lambda) d \lambda.
		\end{flalign}
		
	\end{lem} 
	
	\begin{rem} 
		
		\label{fcn} 
		
		By \cite[Corollary 3.2]{DSME}, \eqref{rhofconverge} also holds for any polynomial $f$. This, \Cref{lf}, and the dominated convergence theorem together imply that \eqref{rhofconverge} holds (and both of its sides are finite) if $f : \mathbb{R} \rightarrow \mathbb{R}$ is of polynomial growth, meaning that there is a constant $C > 1$ so that $|f(x)| \le C (x^2+1)^C$ for all $x \in \mathbb{R}$.
	\end{rem}

	\subsection{Comparison Estimates}

	\label{LCompare} 	
	
	In this section we state four comparison results between different Toda lattices or Lax matrices. The first compares two Toda lattices on different intervals that initially coincide on a subinterval.

	\begin{lem}[{\cite[Proposition 4.5]{LC}}] 
	
	\label{a2p2}
	
		Let $\tilde{N}_1 \le N_1 \le N_2 \le \tilde{N}_2$ be integers; set $\tilde{N} = \tilde{N}_2 - \tilde{N}_1 + 1$ and $N = N_2 - N_1 + 1$. For each $t \in \mathbb{R}_{\ge 0}$, fix $\tilde{N}$-tuples $\tilde{\bm{a}}(t), \tilde{\bm{b}} (t) \in \mathbb{R}^{\tilde{N}}$ and $N$-tuples $\bm{a}(t), \bm{b}(t) \in \mathbb{R}^N$, indexed as 
		\begin{flalign*} 
			& \tilde{\bm{a}} (t) = ( \tilde{a}_{\tilde{N}_1} (t), \tilde{a}_{\tilde{N}_1+1} (t), \ldots , \tilde{a}_{\tilde{N}_2} (t) ); \qquad \tilde{\bm{b}} (t) = ( \tilde{b}_{\tilde{N}_1} (t), \tilde{b}_{\tilde{N}_1+1} (t), \ldots , \tilde{b}_{\tilde{N}_2} (t) ); \\
			&\bm{a} (t) = ( a_{N_1} (t), a_{N_1+1} (t), \ldots , a_{N_2} (t) ); \qquad \bm{b}(t) = ( b_{N_1} (t), b_{N_1+1}(t), \ldots , b_{N_2} (t) ).
		\end{flalign*} 
		
		\noindent For each $s \in \mathbb{R}_{\ge 0}$, also set $\tilde{a}_i (s) = 0 = \tilde{b}_i (s)$ if $i \in \mathbb{Z} \setminus \llbracket \tilde{N}_1, \tilde{N}_2 - 1\rrbracket$, and set $a_i (s) = 0 = b_i (s)$ if $i \in \mathbb{Z} \setminus \llbracket N_1, N_2 - 1 \rrbracket$. Assume $( \tilde{\bm{a}} (t); \tilde{\bm{b}} (t) )$ satisfies \eqref{derivativepa} for each $(j, t) \in \llbracket \tilde{N}_1, \tilde{N}_2 \rrbracket \times \mathbb{R}_{\ge 0}$, and $( \bm{a}(t), \bm{b}(t) )$ satisfies \eqref{derivativepa} for each $(j, t) \in \llbracket N_1, N_2 \rrbracket \times \mathbb{R}_{\ge 0}$. Let $A \ge 1$ be a real number satisfying
		\begin{flalign*}
			A \ge \displaystyle\max_{i \in \llbracket N_1, N_2 \rrbracket} \big( | a_i (0) | + |\tilde{a}_i (0) | +  | b_i (0) | + | \tilde{b}_i (0) | \big).
		\end{flalign*}
		
		\noindent Now fix a real number $T \ge 1$ and integers $N_1' \le N_2'$ and $K \ge 0$ with
		\begin{flalign*} 
			K \ge 200AT,
		\end{flalign*} 
		
		\noindent and $N_1 \le N_1'  \le N_2' \le N_2$ and $N_1' + K \le N_2' - K$. If $(a_j (0), b_j (0)) = (\tilde{a}_j (0), \tilde{b}_j (0))$ for each $j \in \llbracket N_1', N_2'  \rrbracket$, then 
		\begin{flalign*}
			 \displaystyle\sup_{s \in [0, T]} \bigg( \displaystyle\max_{i \in \llbracket N_1' + K, N_2' - K \rrbracket} | a_i (s) - \tilde{a}_i (s) | + \displaystyle\max_{i \in \llbracket N_1'+K, N_2'-K \rrbracket} | b_i (s) - \tilde{b}_i (s) | \bigg) \le e^{-K/4}.
		\end{flalign*}

	\end{lem}
	
	The second indicates that thermal equilibrium is ``approximately invariant'' for the Toda lattice, by comparing it to the Toda lattice run for some time $t \ge 0$ initialized under thermal equilibrium (at sites sufficiently far from the endpoints of its domain).

	\begin{lem}[{\cite[Proposition 4.4]{LC}}] 
		
		\label{ltl0}

		Adopt \Cref{lbetaeta}, and fix $t \in [0, T]$. There exists a random matrix $\bm{M} = [M_{ij}] \in \SymMat_{\llbracket N_1, N_2 \rrbracket}$, whose law coincides with that of $\bm{L}(0)$, such that the following holds with overwhelming probability. For any real number $K \ge T \log N $, we have that 
		\begin{flalign}
			\label{estimatelm} 
			\displaystyle\max_{i,j \in \llbracket N_1+K, N_2-K \rrbracket} | L_{ij} (t) - M_{ij} | \le e^{-K/5}.
		\end{flalign}
				
	\end{lem}
	
	The third and fourth estimate the effect of perturbing a random Lax matrix, under thermal equilibrium, on its eigenvalues. To state them, we require some notation.  
	
	\begin{assumption} 
		
		\label{lmatrixl} 
		
		Sample $(\bm{a}; \bm{b})$ under the thermal equilibrium $\mu_{\beta, \theta; N-1,N}$ from \Cref{mubeta2}, where $\bm{a} = (a_{N_1}, a_{N_1+1}, \ldots , a_{N_2-1})$ and $\bm{b} = (b_{N_1}, b_{N_1+1}, \ldots , b_{N_2})$. Let $\bm{L} = [L_{ij}] \in \SymMat_{\llbracket N_1, N_2 \rrbracket}$ denote the associated Lax matrix (as in \Cref{matrixl}), and let $\tilde{\bm{L}} = [\tilde{L}_{ij}] \in \SymMat_{\llbracket N_1, N_2 \rrbracket}$ be another (random) tridiagonal matrix. Assume that there is an index set $\mathcal{D} \subseteq \llbracket N_1, N_2 \rrbracket$ and a real number $\delta \in (0, 1)$ satisfying 
		\begin{flalign}
			\label{ll} 		
			\displaystyle\max_{i,j \notin (\llbracket N_1, N_2 \rrbracket \setminus \mathcal{D})^2}  |\tilde{L}_{ij}| \le 2 \log N; \qquad \displaystyle\max_{i,j \in \llbracket N_1, N_2 \rrbracket \setminus \mathcal{D}} |L_{ij} - \tilde{L}_{ij}| \le \delta.
		\end{flalign}

	\end{assumption}
	
	We then have the following two lemmas, indicating that the eigenvalues of $\bm{L}$ (or $\tilde{\bm{L}}$) with localization centers sufficiently distant from $\mathcal{D}$ are also nearly eigenvalues of $\tilde{\bm{L}}$ (or $\bm{L}$, respectively). In this way, eigenvalues of $\bm{L}$ are ``approximately local,'' in that up to small error they likely only depend on the entries of $\bm{L}$ close to their localization centers.

	\begin{lem}[{\cite[Corollary 5.5]{LC}}]
		
		\label{lleigenvalues2}

		There exists a constant $c>0$ such that the following holds with overwhelming probability. Adopt \Cref{lmatrixl}; assume $\delta \le e^{-10(\log N)^2}$, and let $\zeta \ge e^{-150(\log N)^{3/2}}$ be a real number. Fix $\lambda \in \eig \bm{L}$, and let $\varphi \in \llbracket N_1, N_2 \rrbracket$ denote a $\zeta$-localization center of $\lambda$ with respect to $\bm{L}$. Suppose that $\dist (\varphi, \mathcal{D}) \ge (\log N)^3$. Then there exists an eigenvalue $\tilde{\lambda} \in \eig \tilde{\bm{L}}$ such that
		\begin{flalign*}
			|\lambda - \tilde{\lambda}| \le e^{(\log N)^2} (\delta^{1/8} + e^{-c \dist (\varphi, \mathcal{D})}), 
		\end{flalign*}
		
		\noindent and $\varphi$ is an $N^{-1} \zeta$-localization center for $\tilde{\lambda}$ with respect to $\tilde{\bm{L}}$. 
		
	\end{lem}
	
	\begin{lem}[{\cite[Corollary 5.6]{LC}}]
		\label{lleigenvalues} 
		
		There exists a constant $c \in (0, 1)$ such that the following holds with overwhelming probability. Adopt \Cref{lmatrixl}; assume that $\delta \le e^{-10(\log N)^2}$; and let $\zeta \ge e^{-150(\log N)^{3/2}}$ be a real number. Fix $\tilde{\lambda} \in \eig \tilde{\bm{L}}$, and let $\tilde{\varphi} \in \llbracket N_1, N_2 \rrbracket$ denote a $\zeta$-localization center of $\tilde{\lambda}$ with respect to $\tilde{\bm{L}}$. Suppose that $\dist (\tilde{\varphi}, \mathcal{D}) \ge (\log N)^3$.
		
		\begin{enumerate} 
			\item There exists a unique eigenvalue $\lambda \in \eig \bm{L}$ such that $|\lambda - \tilde{\lambda}| \le e^{(\log N)^2} (\delta^{1/8} + e^{-c \dist (\tilde{\varphi}, \mathcal{D})})$.
			\item We have that $\tilde{\varphi}$ is an $N^{-1} \zeta$-localization center of $\lambda$ with respect to $\bm{L}$, and any $N^{-1} \zeta$-localization center $\varphi \in \llbracket N_1, N_2 \rrbracket$ satisfies $|\varphi - \tilde{\varphi}| \le (\log N)^2 / 2$. 
		\end{enumerate} 
		
	\end{lem}

 	\subsection{Localization Centers of Random Lax Matrices}
 	
 	\label{Center}
	
	In this section we discuss properties of localization centers (recall \Cref{ucenter}) of Lax matrices under thermal equilibrium. The following lemma bounds the speed at which any localization center can move under the Toda lattice.

	\begin{lem}[{\cite[Lemma 5.2]{LC}}]
		
	\label{centert} 
	
	The following holds with overwhelming probability. Adopt \Cref{lbetaeta}, but assume more generally that $\zeta \ge e^{-200(\log N)^{3/2}}$. Fix any eigenvalue $\lambda \in \eig \bm{L}(0)$ and any $\zeta$-localization center $\varphi \in \llbracket N_1, N_2 \rrbracket$ of $\lambda$ with respect to $\bm{L}(0)$. Then, for each real number $s \in [0, T]$, there does not exist an index $m \in \llbracket N_1, N_2 \rrbracket$ satisfying $|m - \varphi| \ge T(\log N)^2$ that is a localization center for $\lambda$ with respect to $\bm{L}(s)$.

	\end{lem} 

	The next lemma is an approximate continuity bound in $t$ for localization centers of $\bm{L}(t)$, that reside in the bulk of $\llbracket N_1, N_2 \rrbracket$; we show it in \Cref{Proofst} below.
	
	\begin{lem} 
		
		\label{centerdistance}

		The following holds with overwhelming probability. Adopt \Cref{lbetaeta}, but assume more generally that $\zeta \ge e^{-150 (\log N)^{3/2}}$. Fix real numbers $t, t' \in [0, T]$; an eigenvalue $\lambda \in \eig \bm{L}$; and $\zeta$-localization centers $\varphi \in \llbracket N_1, N_2 \rrbracket$ and $\varphi' \in \llbracket N_1, N_2 \rrbracket$ of $\lambda$ with respect to $\bm{L}(t)$ and $\bm{L}(t')$, respectively. Assume that $N_1 + T (\log N)^4 \le \varphi \le N_2 - T (\log N)^4$. 
		
		\begin{enumerate}
			\item We have $|\varphi - \varphi'| \le (|t-t'| + 2) (\log N)^3$. 
			\item We have $|q_{\varphi} (t) - q_{\varphi'} (t')| \le (|t-t'| + 1) (\log N)^4$. 
		\end{enumerate}
		
	\end{lem}

	The following proposition provides the asymptotic scattering relation for the Toda lattice at thermal equilibrium. It will serve as the starting point for our proof of \Cref{vestimate}. 
	
	\begin{prop}[{\cite[Theorem 2.11]{LC}}]
		\label{ztlambda2} 
		
		Adopt \Cref{lbetaeta}. The following holds with overwhelming probability. Let $k \in \llbracket 1, N \rrbracket$ satisfy
		\begin{flalign}
			\label{n1n2k02} 
			N_1 + T (\log N)^6 \le \varphi_0 (k) \le N_2 - T (\log N)^6.
		\end{flalign}
		
		\noindent Then, for each $t \in [0, T]$, we have
		\begin{flalign}
			\label{lambdak2}
			\begin{aligned} 
			\Bigg| \lambda_k t - Q_k(t) + Q_k (0& ) - 2 \sgn (\alpha) \displaystyle\sum_{i: Q_t (i) < Q_t (k)}  \log |\lambda_{k} - \lambda_{i}| \\
			& \qquad \quad + 2 \sgn (\alpha) \displaystyle\sum_{i: Q_0 (i) < Q_0 (k)} \log |\lambda_k - \lambda_i| \Bigg| \le (\log N)^{15}.
			\end{aligned}
		\end{flalign} 
		
	\end{prop}

	\section{Concentration Estimates for Random Lax Matrices} 
	
	\label{Localization2}
	
	In this section we establish concentration bounds for functionals of Lax matrices, given by \Cref{concentrationh} and \Cref{concentrationh3}. We state these results and prove the latter given the former in \Cref{MatrixEstimate}. We then reduce the former to two estimates (given by \Cref{concentrationh2} and \Cref{expectationh}, to be established in \Cref{Proof12}) in \Cref{ProofHs0} below. 
	
	\subsection{Concentration Bounds}
	
	\label{MatrixEstimate}
	
	In this section we provide concentration estimates for the random Lax matrices from \Cref{lbetaeta}, which will involve both the Lax matrix eigenvalues $(\lambda_i)$ and their locations $(Q_i)$. We begin by imposing the following assumption on functions involved in these concentration bounds. 
	
	\begin{assumption}

	\label{fgab} 
	
	Let $A, B \ge 0$ and $S \in [1, T]$ be real numbers. Further let $G : \mathbb{R} \rightarrow \mathbb{R}$ be a function and $F: \mathbb{R} \rightarrow \mathbb{R}$ be a continuous function satisfying the following properties.
	
	\begin{enumerate} 
		
		\item For each $x \in \mathbb{R}$, we have  $| F(x) | \le A e^{|x|^{1/2}}$.
		\item For each $x \in [-\log N, \log N]$, we have $| F(x) | \le A$. 
		\item For any $x, y \in [-\log N, \log N]$ with $|x-y| \le e^{-(\log N)^{5/2}}$, we have $| F(x) - F(y) | \le A e^{-(\log N)^2}$.
		\item We have $\supp G \subseteq [-S, S]$, and $| G (x) - G(y) |\le BS^{-1} (x-y) + B \cdot \mathbbm{1}_{x \ge 0 \ge y}$ for any $x \ge y$.
		
	\end{enumerate} 
	
	\end{assumption} 

	\begin{example} 
		
		\label{examplef} 
		
		Let us provide two functions $F$ satisfying \Cref{fgab}. One is the constant function $F(x) = 1$, which satisfies \Cref{fgab} at $A=1$. Another is $F(x) = v_{\eff} (x)$ (recall \Cref{v}), which satisfies \Cref{fgab} at $A = \mathcal{O}(\log N)$, due to \Cref{derivativev}. As indicated after \eqref{v2}, we will apply the below concentration bound \Cref{concentrationh3} in both of these cases.
	
	\end{example} 
	
	We can now state the following concentration estimate for functions of $\eig \bm{L}$ and the $Q_j (s)$; its proof will appear in \Cref{ProofHs0} below. In what follows, we recall the density $\varrho$ from \Cref{frho}.

	\begin{prop}
		
		\label{concentrationh} 
		
		Adopt \Cref{lbetaeta} and \Cref{fgab}, and fix $s \in [0, T]$. The following holds with overwhelming probability. For any index $j \in \llbracket 1, N \rrbracket$ such that 
		\begin{flalign} 
			\label{js} 
			N_1 + T (\log N)^5 \le \varphi_s (j) \le N_2 - T (\log N)^5,
		\end{flalign} 
		
		\noindent we have 
		\begin{flalign}
			\label{lambdaqjs}
			\Bigg| \displaystyle\sum_{i=1}^N F(\lambda_i) \cdot G( Q_i (s) - Q_j (s))  - \displaystyle\int_{-\infty}^{\infty} F(\lambda) \varrho(\lambda) d\lambda  \displaystyle\int_{-\infty}^{\infty} G(\alpha q) dq \Bigg| \le AB S^{1/2} (\log N)^{12}.
		\end{flalign}
		
	\end{prop}
	
	The following proposition is a modification of \Cref{concentrationh}, in which $F(\lambda)$ is replaced by $F(\lambda) \cdot f (\lambda - \lambda_j)$, for a function $f$ satisfying certain properties.

	\begin{prop}
		
		\label{concentrationh3} 
		
		Adopt \Cref{lbetaeta} and \Cref{fgab}, and fix $s \in [0, T]$. Let $D \ge 1$ be a real number; $j \in \llbracket 1, N \rrbracket$ be an index satisfying \eqref{js}; and $f:\mathbb{R} \rightarrow \mathbb{R}$ be a continuous function satisfying 
		\begin{flalign}
			\label{f1d} 
			\displaystyle\sup_{x \in \mathbb{R}} \displaystyle\sup_{|y| \le \log N}& e^{-|x|^{1/2}} \cdot | F(x) \cdot f(x-y) | \le AD; \qquad \displaystyle\sup_{|x| \le 2 \log N} | f(x) | \le D,
		\end{flalign} 
		
		\noindent and 
		\begin{flalign}
			\label{f2d}
			\begin{aligned}
			 & | f(x) - f(y) | \le D|x-y| \cdot \min \{  e^{100(\log N)^2}, |x|^{-1} + |y|^{-1} \}, \quad \text{if $x, y \in [-2 \log N, 2 \log N]$}. 
			 \end{aligned}
		\end{flalign}
		
		\noindent With overwhelming probability, we have
		\begin{flalign}
			\label{lambdaqjs3}
			\begin{aligned} 
				\Bigg| \displaystyle\sum_{i = 1}^N F(\lambda_i) \cdot f( \lambda_i - \lambda_j) \cdot G ( Q_i (s) - Q_j (s)) - \displaystyle\int_{-\infty}^{\infty} F(\lambda) f( & \lambda - \lambda_j) \varrho (\lambda) d \lambda  \displaystyle\int_{-\infty}^{\infty} G(\alpha q) dq \Bigg| \\
				& \qquad \le ABD S^{1/2} (\log N)^{13}.
			\end{aligned}
		\end{flalign}
		
	\end{prop} 
	
	\begin{proof} 
		
		This result would follow from \Cref{concentrationh}, upon formally replacing $F(\lambda)$ there by the function $F(\lambda) \cdot f(\lambda-\lambda_j)$, except that the latter depends on $\bm{L}$ (through $\lambda_j$). To circumvent this, we will instead apply \Cref{concentrationh} upon replacing the function $F(\lambda)$ there by $F(\lambda) \cdot f(\lambda - \nu)$, where $\nu$ ranges over some fine mesh. 
		
		To implement this, let $c_0$ be such that \Cref{concentrationh} holds with probability at least $1 - c_0^{-1} e^{-c_0 (\log N)^2}$, and set $\mathfrak{c} = \min \{ c_0 / 5, 1 \}$. Further let $(\nu_1, \nu_2, \ldots , \nu_K)$ denote an $e^{-3\mathfrak{c}(\log N)^2}$-mesh of $[ -(\log N) / 2, (\log N) / 2 ]$, so that $K \le  e^{3\mathfrak{c} (\log N)^2} \cdot \log N$. For each $k \in \llbracket 1, K \rrbracket$, define the function $H_k : \mathbb{R} \rightarrow \mathbb{R}$ by setting $H_k (\lambda) = F(\lambda) \cdot f(\lambda - \nu_k)$, for any $\lambda \in \mathbb{R}$. Then, it is quickly verified using \eqref{f1d} and \eqref{f2d} that $H_k$ satisfies the first three conditions for $F$ in \Cref{fgab}, with the $A$ there replaced by $2AD$ here. Thus, for each $k \in \llbracket 1, K \rrbracket$, letting $\mathsf{E}_{1,k}$ denote the event on which 
		\begin{flalign}
			\label{hkintegral} 
			\Bigg| \displaystyle\sum_{i=1}^N H_k (\lambda_i) \cdot G ( Q_i (s) - Q_j (s) ) - \displaystyle\int_{-\infty}^{\infty} H_k (\lambda) \varrho (\lambda) d \lambda \displaystyle\int_{-\infty}^{\infty} G(\alpha q) dq \Bigg| \le 2ABD S^{1/2} (\log N)^{12},
		\end{flalign}
	
		\noindent \Cref{concentrationh} (with the $(F, G)$ there equal to $(H_k, G)$ here) yields $\mathbb{P} [\mathsf{E}_{1,k}^{\complement}] \le (5\mathfrak{c})^{-1} e^{-5 \mathfrak{c} (\log N)^2}$. Set $\mathsf{E}_1 = \bigcap_{k=1}^K \mathsf{E}_{1,k}$, so that \eqref{hkintegral} holds for each $k \in \llbracket 1, K \rrbracket$ on the event $\mathsf{E}_1$. Using a union bound to bound $\mathbb{P}[\mathsf{E}_1^{\complement}]$ by the sum of the $\mathbb{P}[\mathsf{E}_{1,k}^{\complement}]$, we deduce that $\mathbb{P}[ \mathsf{E}_1^{\complement} ] \le K \cdot (5\mathfrak{c})^{-1} e^{-5 \mathfrak{c} (\log N)^2} \le e^{-\mathfrak{c} (\log N)^2}$.
		
		 Recalling \Cref{adelta}, also define the event $\mathsf{E}_2 = \mathsf{BND}_{\bm{L}(0)} ( (\log N) / 2 ) \cap \mathsf{SEP}_{\bm{L}(0)} (e^{-\mathfrak{c} (\log N)^2})$, which is overwhelmingly probable, by \Cref{l0eigenvalues}. In what follows, we restrict to the event $\mathsf{E} = \mathsf{E}_1 \cap \mathsf{E}_2$ and verify \eqref{lambdaqjs3} on it.
		
		To that end, let $k_0 \in \llbracket 1, K \rrbracket$ denote the index such that $|\lambda_j - \nu_{k_0}| \le e^{-3\mathfrak{c} (\log N)^2}$; we abbreviate $\nu = \nu_{k_0}$. Therefore, \eqref{hkintegral} implies
		\begin{flalign}
			\label{flambdai2}
			\begin{aligned} 
				\Bigg| \displaystyle\sum_{i = 1}^N F(\lambda_i) \cdot f(\lambda_i - \nu) \cdot G ( Q_i (s) - Q_j (s) ) - \displaystyle\int_{-\infty}^{\infty} H_{k_0} (\lambda) \varrho (\lambda) d \lambda & \displaystyle\int_{-\infty}^{\infty} G(\alpha q) dq \Bigg| \\
				& \le 2 ABDS^{1/2} (\log N)^{12}. 
			\end{aligned}
		\end{flalign}
		
		\noindent Next, observe for any $i \in \llbracket 1, N \rrbracket$ that 
		\begin{flalign}
			\label{falambdai} 
			\begin{aligned} 
				\big| F(\lambda) \cdot f(\lambda_i - \nu) - F(\lambda) \cdot f(\lambda_i - \lambda_j) \big| & \le 2A D |\nu - \lambda_j| \cdot ( |\lambda_i - \nu|^{-1} + |\lambda_i - \lambda_j|^{-1} ) \\
				& \le 4AD e^{-\mathfrak{c} (\log N)^2},
			\end{aligned} 
		\end{flalign}
		
		\noindent where in the first inequality we used the second statement in \Cref{fgab} (with our restriction to $\mathsf{E}_2$) with \eqref{f2d}, and in the second we used the facts (from the definition of $\nu_{k_0}$ and our restriction to $\mathsf{E}_2$) that $|\lambda_j - \nu| \le e^{-3\mathfrak{c} (\log N)^2}$; that $|\lambda_i - \lambda_j| \ge e^{-\mathfrak{c} (\log N)^2}$; and that $|\lambda_i - \nu| \ge e^{-\mathfrak{c} (\log N)^2} - e^{-2 \mathfrak{c} (\log N)^2} \ge e^{-2\mathfrak{c} (\log N)^2}$. Moreover, 
		\begin{flalign}
			\label{hk0lambdaf} 
			\begin{aligned} 
				\displaystyle\int_{-\infty}^{\infty} \big| H_{k_0} (\lambda) - F(\lambda) \cdot f(\lambda - \lambda_j) \big| \varrho (\lambda) d \lambda & \le A \displaystyle\int_{-\log N}^{\log N}\big| f(\lambda - \nu) - f(\lambda - \lambda_j) \big| \cdot \varrho (\lambda) d\lambda \\
				& \qquad + AD \displaystyle\int_{|\lambda| > \log N} e^{2 |\lambda|^{1/2}} \varrho (\lambda) d\lambda \le 2 AD,
			\end{aligned} 
		\end{flalign}
		
		\noindent where in the first inequality we used \Cref{fgab} and \eqref{f1d}, and in the second we used \eqref{f2d} and the fact from \Cref{frho} that there exists a constant $c_2 > 0$ such that $\varrho (\lambda) \le c_2^{-1} e^{-c_2 \lambda^2}$. Inserting \eqref{falambdai} and \eqref{hk0lambdaf} into \eqref{flambdai2}, and using our restriction to $\mathsf{E}_2$ with the facts that $| G(q) | \le 2B$ for all $q \in \mathbb{R}$ and that $\supp G \subseteq [-S, S] \subseteq [-N, N]$ (by the fourth statement in \Cref{fgab}), yields \eqref{lambdaqjs3}.
	\end{proof}

	\subsection{Proof of \Cref{concentrationh}} 
	
	\label{ProofHs0} 
	
	In this section we establish \Cref{concentrationh}, to which end, we will reduce (under a change of variables) to the case when $s=0$.
	
	\begin{prop}
		
		\label{concentrationhs0} 
		
		Adopt \Cref{lbetaeta} and \Cref{fgab}. Set $\Lambda_{i} = \lambda_{\varphi_0^{-1} (i)}$ and $q_i = q_i (0)$ for each $i \in \llbracket N_1, N_2 \rrbracket$. With overwhelming probability, we have for any index $j \in \llbracket N_1 + T (\log N)^5,  N_2 - T (\log N)^5 \rrbracket$ that  
		\begin{flalign*}
			\Bigg| \displaystyle\sum_{i=N_1}^{N_2} F(\Lambda_i) \cdot G( q_i - q_j)  - \displaystyle\int_{-\infty}^{\infty} F(\lambda) \varrho(\lambda) d\lambda \displaystyle\int_{-\infty}^{\infty} G(\alpha q) dq \Bigg| \le 7 AB S^{1/2} (\log N)^{11}.
		\end{flalign*}

	\end{prop} 
	
	Proposition \ref{concentrationhs0} follows quickly from the following two lemmas. The first is established in \Cref{ProofH2} and the second in \Cref{ProofH}.
	
	\begin{lem}
		
		\label{concentrationh2} 
		
		Adopt the notation and assumptions of \Cref{concentrationhs0}. With overwhelming probability, we have 
		\begin{flalign*}
			\Bigg| \displaystyle\sum_{i=N_1}^{N_2} F(\Lambda_{i}) \cdot G ( q_i - q_j  ) - \mathbb{E} \bigg[ \displaystyle\sum_{i=N_1}^{N_2} F(\Lambda_{i}) \cdot G ( q_i  - q_j ) \bigg] \Bigg| \le AB S^{1/2} (\log N)^8.
		\end{flalign*}
		
	\end{lem}
	
	\begin{lem}
		
		\label{expectationh} 
		
		Adopt the notation and assumptions of \Cref{concentrationhs0}. For sufficiently large $N$, we have 
		\begin{flalign*}
			\Bigg| \mathbb{E} \bigg[ \displaystyle\sum_{i=N_1}^{N_2} F(\Lambda_{i}) \cdot G ( q_i - q_j ) \bigg] - \displaystyle\int_{\infty}^{\infty} F(\lambda) \varrho(\lambda) d \lambda \displaystyle\int_{-\infty}^{\infty} G(\alpha q) dq \Bigg| \le 6AB S^{1/2} (\log N)^{11}.
		\end{flalign*}
	\end{lem}

	\begin{proof}[Proof of \Cref{concentrationhs0}]
		
		This follows from \Cref{concentrationh2} and \Cref{expectationh}. 
	\end{proof}

	\begin{proof}[Proof of \Cref{concentrationh}]
		
		Throughout this proof, for each $i \in \llbracket N_1, N_2 \rrbracket$, we abbreviate $a_i = a_i (s)$ and $q_i = q_i (s)$, and denote $\Lambda_i = \lambda_{\varphi_s^{-1} (i)}$. It then suffices to show that, for any fixed index $j \in \llbracket N_1 + T (\log N)^5, N_2 - T (\log N)^5 \rrbracket$, we have with overwhelming probability that  
		\begin{flalign}
			\label{lambdai2} 
			\Bigg| \displaystyle\sum_{i=N_1}^{N_2} F(\Lambda_i) \cdot G (q_i - q_j) - \displaystyle\int_{-\infty}^{\infty} F(\lambda) \varrho (\lambda) d \lambda \displaystyle\int_{-\infty}^{\infty} G(\alpha q) dq \Bigg| \le ABS^{1/2} (\log N)^{12}.
		\end{flalign}
		
		\noindent To do so, we apply \Cref{ltl0} and \Cref{lleigenvalues} to compare the Toda lattice $\bm{q}(s)$ at time $s$ to a Toda lattice at thermal equilibrium, and then use \Cref{concentrationhs0}. 
		
		So, set $K = T (\log N)^{9/2}$, and observe that $j \in \llbracket N_1 + 5K, N_2 - 5K \rrbracket$. Then \Cref{ltl0} yields a random matrix $\tilde{\bm{L}} = [\tilde{L}_{ij}] \in \SymMat_{\llbracket N_1, N_2 \rrbracket}$ with the same law as $\bm{L}(0)$, and an overwhelmingly probable event $\mathsf{E}_1$, on which 
		\begin{flalign}
			\label{l2ij} 
			\displaystyle\max_{i, i' \in \llbracket N_1 + K, N_2 - K \rrbracket} | L_{ii'} (s) - \tilde{L}_{ii'} | \le e^{-c_1 (\log N)^4}.
		\end{flalign}
		
		 We restrict to the event $\mathsf{E}_1$ in what follows. Analogously to in \Cref{matrixl}, define the Flaschka variables $\tilde{\bm{a}} = (\tilde{a}_{N_1}, \tilde{a}_{N_1+1}, \ldots , \tilde{a}_{N_2}) \in \mathbb{R}^{N-1}$ and $\tilde{\bm{b}} = (\tilde{b}_{N_1}, \tilde{b}_{N_1+1}, \ldots , \tilde{b}_{N_2}) \in \mathbb{R}^N$ associated with $\tilde{\bm{L}}$ by setting $\tilde{a}_i = \tilde{L}_{i,i+1}$ for each $i \in \llbracket N_1, N_2-1 \rrbracket$; setting $\tilde{a}_{N_2} = 0$; and setting $\tilde{b}_i = \tilde{L}_{i,i}$ for each $i \in \llbracket N_1, N_2 \rrbracket$. Let $(\tilde{\bm{p}}; \tilde{\bm{q}}) \in \mathbb{R}^N \times \mathbb{R}^N$ denote the Toda state space initial data associated with $(\tilde{\bm{a}}; \tilde{\bm{b}})$ (as described in \Cref{Open}), and denote $\tilde{\bm{p}} = (\tilde{p}_{N_1}, \tilde{p}_{N_1+1}, \ldots , \tilde{p}_{N_2})$ and $\tilde{\bm{q}} = (\tilde{q}_{N_1}, \tilde{q}_{N_1+1}, \ldots , \tilde{q}_{N_2})$. Set $\eig \tilde{\bm{L}} = (\tilde{\lambda}_1, \tilde{\lambda}_2, \ldots , \tilde{\lambda}_N)$; let $\tilde{\varphi} : \llbracket 1, N \rrbracket \rightarrow \llbracket N_1, N_2 \rrbracket$ denote a $\zeta$-localization center bijection for $\tilde{\bm{L}}$; and denote $\tilde{\Lambda}_i = \tilde{\lambda}_{\tilde{\varphi}^{-1} (i)}$ for each $i \in \llbracket N_1, N_2 \rrbracket$. 
		
		By \Cref{concentrationhs0} there exists an overwhelmingly probable event $\mathsf{E}_2$, on which 
		\begin{flalign*}
			\Bigg| \displaystyle\sum_{i=N_1}^{N_2} F(\tilde{\Lambda}_i) \cdot G(\tilde{q}_i - \tilde{q}_j) - \displaystyle\int_{-\infty}^{\infty} F(\lambda) \varrho (\lambda) d \lambda \displaystyle\int_{-\infty}^{\infty} G(\alpha q) dq \Bigg| \le 7 ABS^{1/2} (\log N)^{11}.
		\end{flalign*}
		
		\noindent We further restrict to $\mathsf{E}_2$ in what follows. To verify \eqref{lambdai2}, it therefore suffices to show with overwhelming probability that 
		\begin{flalign}
			\label{flambda2} 
			 \displaystyle\sum_{i=N_1}^{N_2} \big| F(\Lambda_i) \cdot G(q_i - q_j) - F(\tilde{\Lambda}_i) \cdot G(\tilde{q}_i - \tilde{q}_j) \big| < ABS^{1/2} (\log N)^{11}.
		\end{flalign}
		
		To that end, we define three additional events on which we will be able to compare $\bm{L}(s)$ and $\tilde{\bm{L}}$. The first is that on which $\bm{L}(s)$ and $\tilde{\bm{L}}$ are bounded in a particular way, namely, 
		\begin{flalign}
			\label{mie} 
				\mathsf{E}_3 = \mathsf{BND}_{\tilde{\bm{L}}} (\log N) \cap \bigcap_{r \ge 0} \mathsf{BND}_{\bm{L}(r)} (\log N) \cap \bigcap_{i = N_1}^{N_2-1} \{ \tilde{L}_{i,i+1} \ge e^{-(\log N)^2} \}.
		\end{flalign} 
		
		\noindent Observe that $\mathsf{E}_3$ is overwhelmingly probable, where the probability estimate on the first two events in \eqref{mie} follows from \Cref{l0eigenvalues}, and that on the third event in \eqref{mie} follows from the explicit form (\Cref{mubeta2}) for the thermal equilibrium $\mu_{\beta,\theta;N-1,N}$.
		
		To define the second event, observe by \eqref{l2ij} and \Cref{lleigenvalues} that there exists an overwhelmingly probable event $\mathsf{E}_4$, on which the following holds. There exists a bijection $\psi : \llbracket N_1, N_2 \rrbracket \rightarrow \llbracket N_1, N_2 \rrbracket$ such that, for each $i \in \llbracket N_1 + 2K, N_2 - 2 K  \rrbracket$, we have 
		\begin{flalign}
			\label{lambda2psi} 
			|\Lambda_i - \tilde{\Lambda}_{\psi(i)} | < c_4 e^{-c_4 (\log N)^3}, \qquad \text{and} \qquad | \psi(i) - i | \le (\log N)^2. 
		\end{flalign} 
		
		\noindent The third event is that on which the $\tilde{q}_i$ are separated as indicated by \Cref{qij}; specifically, let 
		\begin{flalign*}
			\mathsf{E}_5 = \bigcap_{i = N_1}^{N_2} & \{ | q_i (0) - q_j (0) - \alpha(i-j) | < |i-j|^{1/2} (\log N)^2 \} \\ 
			& \quad \cap \{ |\tilde{q}_i - \tilde{q}_j - \alpha (i-j) | < |i-j|^{1/2} (\log N)^2 \}.
		\end{flalign*} 
		
		\noindent  By \Cref{qij} and a union bound, $\mathsf{E}_5$ is overwhelmingly probable. Thus, setting $\mathsf{E} = \bigcap_{i=1}^5 \mathsf{E}_i$, it suffices to verify \eqref{flambda2} upon restricting to $\mathsf{E}$. So, we restrict to $\mathsf{E}$ in what follows.
		
		 First observe by our restriction to $\mathsf{E}_5$ that for $|i-j| \ge K \ge S(\log N)^{9/2}$ we have $|\tilde{q}_i - \tilde{q}_j| > 3S$ and hence $G(\tilde{q}_i - \tilde{q}_j) = 0$ by the fourth part of \Cref{fgab}. Moreover, by \eqref{l2ij}; our restriction to $\mathsf{E}_1$, with the third event in \eqref{mie}; and \eqref{q00}, we have for each $i \in \llbracket N_1 + K, N_2 + K \rrbracket$ that 
		\begin{flalign}
			\label{q4} 
			| (\tilde{q}_i - \tilde{q}_j) - (q_i - q_j) | \le 2 \displaystyle\sum_{k=i}^{j-1} | \log a_k - \log \tilde{a}_k| = 2 \displaystyle\sum_{k=i}^{j-1} | \log L_{k,k+1} (s) - \log \tilde{L}_{k,k+1} | \le e^{-c_6 (\log N)^3},
		\end{flalign}
		
		\noindent for some constant $c_6 > 0$. In particular, it follows since $j \in \llbracket N_1 +5K, N_2 - 5K \rrbracket$ that for $i \in \llbracket N_1 + K, N_2 - K \rrbracket$ with $|i-j| \ge K$ we have $|q_i - q_j| \ge |\tilde{q}_i - \tilde{q}_j| - e^{-c_6 (\log N)^3} > 2S$. Moreover, for $i \in \llbracket N_1, N_2 \rrbracket \setminus \llbracket N_1 + K, N_2 - K \rrbracket$, we have 
		\begin{flalign*}
			|q_i - q_j| & \ge | q_i (0) - q_j(0) | - | q_i (s) - q_i (0) | - | q_j (s) - q_j (0) | \\
			& \ge | q_i (0) - q_j(0) | - 2T \cdot \displaystyle\sup_{s' \in \mathbb{R}} \big( | L_{ii} (s) | + | L_{jj} (s) | \big) \\ 
			& \ge | q_i (0) - q_j (0) | - 2T \log N \ge \displaystyle\frac{|\alpha|}{2} \cdot |i-j| - \displaystyle\frac{K}{\log N} > 2S,
		\end{flalign*} 
		
		\noindent where in the first and second statements we used \eqref{qtpt} and the fact that each $p_k (s) = b_k (s) = L_{kk} (s)$ by \eqref{abr} and \Cref{matrixl}; in the third we used our restriction to $\mathsf{E}_3$; and in the fourth and fifth we used our restriction to $\mathsf{E}_5$ and the definition of $K$.
				
		Hence, for any $i \in \llbracket N_1, N_2 \rrbracket$ with $|i-j| \ge K$, we have $|q_i - q_j| > 2S$ and thus $G(q_i - q_j) = 0$ (again by the fourth part of \Cref{fgab}). Therefore, 
		\begin{flalign}
			\label{flambdagq0} 
			\begin{aligned} 
			\displaystyle\sum_{i=N_1}^{N_2} & \Big| F(\Lambda_i) \cdot G(q_i - q_j) - F(\tilde{\Lambda}_i) \cdot G(\tilde{q}_i - \tilde{q}_j) \Big| \\
			& \le \displaystyle\sum_{i=j-K}^{j+K} \big( | F(\Lambda_i) - F(\tilde{\Lambda}_i) | \cdot | G(\tilde{q}_i - \tilde{q}_j) | + | F(\Lambda_i) | \cdot | G(q_i - q_j) - G(\tilde{q}_i - \tilde{q}_j) | \big) \\
			& \le 2BN \cdot A e^{-(\log N)^2} + A \displaystyle\sum_{i=j-K}^{j+K} \big| G(q_i - q_j) - G (\tilde{q}_i - \tilde{q}_j) \big|.
			\end{aligned} 
		\end{flalign}
		
		\noindent Here, in the first bound, we used the above fact that $G(q_i - q_j) = G(\tilde{q}_i - \tilde{q}_j) = 0$ whenever $|i-j| \ge K$. In the second, we used the facts that $| G(q) | \le 2B$ for all $q \in \mathbb{R}$ (by the fourth part of \Cref{fgab}); that $| F(\Lambda_i) - F(\tilde{\Lambda}_i) | \le e^{-(\log N)^2}$ by the third part of \Cref{fgab}, the fact that $|\Lambda_i|, |\tilde{\Lambda}_i| \le \log N$ by our restriction to $\mathsf{E}_3$, and \eqref{lambda2psi}; and the fact that $| F(\Lambda_i) | \le A$ by the second part of \Cref{fgab} (and again our restriction to $\mathsf{E}_3$). 
		
		Next observe from \eqref{q4} and the fourth part of \Cref{fgab} that 
		\begin{flalign}
			\label{gq2} 
			\begin{aligned}  
			\displaystyle\sum_{i=j-K}^{j+K} \big| G( & q_i - q_j) - G(\tilde{q}_i - \tilde{q}_j) \big| \\ 
			& \le  BS^{-1} N \cdot e^{-c_6 (\log N)^3} + B \displaystyle\sum_{i=j-K}^{j+K} (\mathbbm{1}_{q_i - q_j \ge 0 \ge \tilde{q}_i - \tilde{q}_j} + \mathbbm{1}_{\tilde{q}_i - \tilde{q}_j \ge 0 \ge q_i - q_j}).
			\end{aligned}
		\end{flalign}
		
		\noindent Due to \eqref{q4} and our restriction to $\mathsf{E}_5$, if $|i-j| \ge (\log N)^5$, both $q_i - q_j$ and $\tilde{q}_i - \tilde{q}_j$ are nonzero and have the same sign as $\alpha (i-j)$ (since they are both within $2 |i-j|^{1/2} (\log N)^2 < |\alpha i - \alpha j|$ of $\alpha(i-j)$). Together with \eqref{gq2}, this gives 
		\begin{flalign*}
			\displaystyle\sum_{i=j-K}^{j+K} \big| G(q_i - q_j) - G(\tilde{q}_i - \tilde{q}_j) \big| \le BN^{-1} + B \cdot 4(\log N)^5 \le 5B(\log N)^5.
		\end{flalign*}	
		
		\noindent Upon insertion into \eqref{flambdagq0}, this implies \eqref{flambda2} and thus the proposition.		
	\end{proof} 
	
	\section{Proofs of \Cref{concentrationh2} and \Cref{expectationh}} 
	
	\label{Proof12} 
	
	In this section we establish \Cref{concentrationh2} and \Cref{expectationh}. We begin in \Cref{EstimatesProbability} by stating a general concentration bound (\Cref{xfsum}), similar to the McDiarmid inequality, that bounds the fluctuations of multivariate function in terms of how much each of its arguments ``influences'' it. We then prove \Cref{concentrationh2} assuming certain influence bounds given by \Cref{estimateik}, which are shown in \Cref{ProofI}. We next show \Cref{expectationh} in \Cref{ProofH}.
	
	\subsection{Concentration Estimates} 
	
	\label{EstimatesProbability}
	
	In this section we state a concentration bound, which is similar to the McDiarmid inequality, that we will use to show \Cref{concentrationh2}. To that end, the following definition provides the notion of how a random variable ``influences'' a multivariate function.
	
	\begin{definition}
		
	\label{xf} 
	
	Let $\mathcal{I}$ be an index set, let $\bm{x} = (x_i)_{i \in \mathcal{I}} \in \mathbb{R}^{\mathcal{I}}$ be a sequence of mutually independent real random variables, and let $F : \mathbb{R}^{\mathcal{I}} \rightarrow \mathbb{R}$ be a function. For nonempty subset $\mathcal{J} \subseteq \mathcal{I}$, define the set of variables $\bm{x}(\mathcal{J}) = (x_j)_{j \in \mathcal{J}}$. Then, for any real number $p \ge 0$, define the \emph{influence} $\Infl_{\bm{x}(\mathcal{J})} (F; p) = \Infl_{\bm{x}(\mathcal{J})} (F; p; \bm{x})$ of $\bm{x}(\mathcal{J})$ on $F$ by 
	\begin{flalign*}
		\Infl_{\bm{x}(\mathcal{J})} (F; p) = \displaystyle\inf \{ A \ge 0 : \mathbb{P} [ | F(\bm{y}) - F(\bm{x}) | \ge A ] \le p \}.
	\end{flalign*}
	
	\noindent Here, the sequence $\bm{y} = (y_i)_{i \in \mathcal{I}}$ is a family of mutually independent random variables, obtained by setting $y_j = x_j$ if $j \ne \mathcal{J}$, and setting $y_j$ to be a random variable with the same law as $x_j$ that is independent from $\bm{x}$ if $j \in \mathcal{J}$.
	
	\end{definition} 
	
	The next lemma is a variant of the McDiarmid inequality providing a concentration result for functions of random variables, in terms of their influences; it is in a similar direction as, but slightly different from, \cite[Proposition 2]{EI} and \cite[Theorem 1.2]{MTB}. Its proof is given in \Cref{Proofxf} below.

	\begin{lem} 
		
	\label{xfsum} 
	
	Adopt the notation of \Cref{xf}. Let $m \ge 1$ be an integer and $\mathcal{J}_1 \cup \mathcal{J}_2 \cup \cdots \cup \mathcal{J}_m = \mathcal{I}$ be a partition of $\mathcal{I}$ into $m$ disjoint, nonempty subsets. Denote 
	\begin{flalign}
		\label{sdelta} 
		\qquad S = \sum_{k=1}^m \Infl_{\bm{x}(\mathcal{J}_k)} (F; p)^2; \qquad U = \mathbb{E} [ F(\bm{x})^2 ]^{1/2}. 
	\end{flalign}
	
	\noindent Then, for any real number $R \ge 0$, we have
	\begin{flalign*}
		\mathbb{P} \big[ \big| F(\bm{x}) - \mathbb{E} [ F(\bm{x}) ] \big| \ge RS^{1/2} + 2m^{1/2} p^{1/4} U \big] < 2mp^{1/2} + 2 e^{-R^2/4}.
	\end{flalign*}
	
	\end{lem}

	\subsection{Proof of \Cref{concentrationh2}} 
	
	\label{ProofH2} 
	
	In this section we establish \Cref{concentrationh2}, adopting the notation of that proposition throughout. In the below, we abbreviate the Flaschka variables $\bm{a}(0) = \bm{a} = (a_{N_1}, a_{N_1+1}, \ldots , a_{N_2-1})$ and $\bm{b}(0) = \bm{b} = (b_{N_1}, b_{N_1+1}, \ldots , b_{N_2})$. Define the function $\mathfrak{F} = \mathfrak{F}_j: \mathbb{R}^{N-1} \times \mathbb{R}^N \rightarrow \mathbb{R}$ by  
		\begin{flalign}
			\label{f} 
			 \mathfrak{F} (\bm{a}; \bm{b}) = \displaystyle\sum_{i=N_1}^{N_2} F(\Lambda_i) \cdot G (q_i - q_j),
		\end{flalign} 
		
		\noindent where we observe that $\mathfrak{F}$ can indeed be viewed as a function of the random variables $(\bm{a}; \bm{b})$, since $\bm{L}$ and the $(q_j)$ can be. Next, define the variable sets (recalling $j \in \llbracket N_1 + T(\log N)^5, N_2 - T(\log N)^5 \rrbracket$ is fixed)
		\begin{flalign*}
			\mathcal{S} = \{ a_i, b_i : |i-j| > S (\log N)^{9/2} \}, \qquad \text{and} \qquad \mathcal{S}_k = (a_k, b_k),
		\end{flalign*} 
		
		\noindent  for any integer $k \in \llbracket N_1, N_2 \rrbracket$ with $|k-j| \le S (\log N)^{9/2}$. Recalling the notation from \Cref{xf}, abbreviate for any such $k$ the influences 
		\begin{flalign*} 
			I = \Infl_{\mathcal{S}} (\mathfrak{F}; \mathfrak{p}; \bm{a} \cup \bm{b}), \qquad \text{and} \qquad I_k = \Infl_{\mathcal{S}_k} (\mathfrak{F}; \mathfrak{p}; \bm{a} \cup \bm{b}), \qquad \text{where} \quad \mathfrak{p} = e^{-\mathfrak{c} (\log N)^2},
		\end{flalign*} 
	
		\noindent for some sufficiently small constant $\mathfrak{c} > 0$ that we will fix later. We will deduce \Cref{concentrationh2} from \Cref{xfsum} and the following lemma bounding these influences, which we establish in \Cref{ProofI} below. 
		
		\begin{lem} 
			
		\label{estimateik} 
		
		The following hold if $N \ge 1$ is sufficiently large and $\mathfrak{c} > 0$ is sufficiently small. 
		
		\begin{enumerate} 
			\item For each $k \in \llbracket N_1, N_2 \rrbracket$ with $|k-j| \le S (\log N)^{9/2}$, we have $I_k \le AB(\log N)^5$.
			\item We have $I \le AB(\log N)^5$.
		\end{enumerate} 
		
		\end{lem} 
		
		Before proving \Cref{concentrationh2}, we require the following quick lemma bounding the expectation of the maximal value of $| F(\lambda)|$ over $\lambda \in \eig \bm{L}$.

		\begin{lem} 
			
			\label{expectationflambda} 
			
			Fix $v > 0$. For $N$ sufficiently large, we have that 
				\begin{flalign*}
				\mathbb{E} \bigg[ \displaystyle\max_{\lambda \in \eig \bm{L}} | F(\lambda) |^v \bigg] \le 2A^v.
			\end{flalign*}
			
		\end{lem} 
		
		\begin{proof}

			Recalling \Cref{adelta}, define the event $\mathsf{E}_0 = \mathsf{BND}_{\bm{L}} (\log N)$. Then, there exists a constant $c > 0$ such that
			\begin{flalign*}
				\mathbb{E} \bigg[ \displaystyle\max_{\lambda \in \eig \bm{L}} | F(\lambda) |^v \bigg] & \le A^v + \mathbb{E} \bigg[ \displaystyle\max_{\lambda \in \eig \bm{L}} A^v e^{v|\lambda|^{1/2}} \cdot \mathbbm{1}_{\mathsf{E}_0^{\complement}} \bigg] \\
				& \le A^v + A^v \displaystyle\int_{\log N}^{\infty} e^{v |x|^{1/2}} \cdot \mathbb{P} \bigg[ \displaystyle\max_{\lambda \in \eig \bm{L}} |\lambda| \ge x \bigg] dx \\ 
				& \le A^v + c^{-1} A^v N \displaystyle\int_{\log N}^{\infty} e^{v |x|^{1/2} - cx^2} dx \le 2A^v,
			\end{flalign*}
			
			\noindent where in the first bound we used the first and second statements in \Cref{fgab}; in the second we used the definition of $\mathsf{E}_0$; in the third we used \Cref{l0eigenvalues}; and in the fourth we used that $N$ is sufficiently large. 
		\end{proof}

		\begin{proof}[Proof of \Cref{concentrationh2}]

			This will follow from \Cref{xfsum} and \Cref{estimateik}. Set $U = \mathbb{E} [ \mathfrak{F}(\bm{a};\bm{b})^2 ]^{1/2}$. Then applying \Cref{xfsum}, with the $(R, p)$ there equal to $(\log N, \mathfrak{p})$ here and the $\bigcup_{j=1}^m \bm{x}(\mathcal{J}_k)$ here equal to $\mathcal{S} \cup \bigcup_{k: |k - j| \le S(\log N)^{9/2}} \mathcal{S}_k$ here, we obtain 
			\begin{flalign}
				\label{fsumab}
				\mathbb{P} \big[ \big| \mathfrak{F} (\bm{a};\bm{b}) - \mathbb{E} [ \mathfrak{F}(\bm{a};\bm{b}) ] \big| \ge ABS^{1/2} (\log N)^{15/2} + 2N^{1/2} \mathfrak{p}^{1/4} U \big] \le 2N\mathfrak{p}^{1/4} + 2e^{-(\log N)^2},
			\end{flalign}
			
			\noindent where we used the fact from \Cref{estimateik} that 
			\begin{flalign*}
				\displaystyle\sum_{k = -\lfloor S(\log N)^{9/2} \rfloor}^{\lfloor S(\log N)^{9/2} \rfloor} \Infl_{\mathcal{S}_k} (\mathfrak{F}; \mathfrak{p})^2  + \Infl_{\mathcal{S}} (\mathfrak{F};\mathfrak{p})^2 \le 3S (\log N)^{9/2} \cdot (AB)^2 (\log N)^{10} \le A^2 B^2 S (\log N)^{15}. 
			\end{flalign*}	
			
			\noindent Thus, it suffices to show that $U \le 4ABN$ for sufficiently large $N$, as then insertion into \eqref{fsumab} (and using the definition of $\mathfrak{p} = e^{-\mathfrak{c}(\log N)^2}$) would yield the lemma. Since $| G(q) | \le 2B$ for all $q \in \mathbb{R}$ (by the fourth statement in \eqref{fgab}), this follows from (the $v=1$ case of) \Cref{expectationflambda}. 		
			\end{proof}

	\subsection{Proof of \Cref{estimateik}}
	
	\label{ProofI} 
	
	In this section we establish \Cref{estimateik}; we adopt the notation of \Cref{ProofH2} throughout.
	
	To address both parts of \Cref{estimateik} simultaneously, we fix an index $k \in \llbracket N_1, N_2-1 \rrbracket$ with $|k-j| \le S(\log N)^{9/2}$, and define the subset $\mathcal{D} \subseteq \llbracket N_1, N_2-1 \rrbracket$ by either setting $\mathcal{D} = \{ k \}$ or setting $\mathcal{D} = \{ i \in \llbracket N_1, N_2 \rrbracket : |i-j| > S(\log N)^{9/2} \}$; in the first case we set $\mathfrak{I} = I_k$, and in the second we set $\mathfrak{I} = I$. We must estimate $\mathfrak{I}$, to which end we set notation for replacing the random variable $a_i \in \bm{a}$ and $b_i \in \bm{b}$ with independent copies of them, whenever $i \in \mathcal{D}$. 
	
	For each such $i \in \mathcal{D}$, let $a_i'$ and $b_i'$ be mutually independent random variables with the same laws as $a_i$ and $b_i$, respectively, that are independent from $\bm{a} \cup \bm{b}$. Define $\tilde{\bm{a}} = (\tilde{a}_{N_1}, \tilde{a}_{N_1+1}, \ldots , \tilde{a}_{N_2-1})$ and $\tilde{\bm{b}} = (\tilde{b}_{N_1}, \tilde{b}_{N_1+1}, \ldots , \tilde{b}_{N_2})$ by setting $\tilde{a}_i = a_i$ and $\tilde{b}_i = b_i$ if $i \notin \mathcal{D}$, and by setting $\tilde{a}_i = a_i'$ and $\tilde{b}_i = b_i'$ if $i \in \mathcal{D}$. Then, by \Cref{xf},
	\begin{flalign}
		\label{ak} 
		\mathfrak{I} = \inf \big\{ A \ge 0 : \mathbb{P} \big[ | \mathfrak{F} (\bm{a}; \bm{b}) - \mathfrak{F} (\tilde{\bm{a}};\tilde{\bm{b}})| \ge A \big] \le \mathfrak{p} \big\}.
	\end{flalign}
	
	Let us set some additional notation parallel to that in \Cref{lbetaeta} for $(\tilde{\bm{a}}; \tilde{\bm{b}})$. Let $\tilde{\bm{L}} = [\tilde{L}_{ii'}] \in \SymMat_{\llbracket N_1, N_2 \rrbracket}$ denote the tridiagonal matrix associated with $(\tilde{\bm{a}};\tilde{\bm{b}})$, as in \Cref{matrixl} (so $\tilde{L}_{i,i+1} = \tilde{L}_{i+1,i} = \tilde{a}_i$ for $i \in \llbracket N_1, N_2-1\rrbracket$ and $\tilde{L}_{i,i} = \tilde{b}_i$ for $i \in \llbracket N_1, N_2 \rrbracket$); in this way, we have 
	\begin{flalign}
		\label{lij2} 
		\tilde{L}_{ii'} = L_{ii'}, \qquad \text{if either $\dist (i, \mathcal{D}) \ge 2$ or $\dist (i', \mathcal{D}) \ge 2$}.
	\end{flalign} 
	
	\noindent Set $\eig \tilde{\bm{L}}  = (\tilde{\lambda}_1, \tilde{\lambda}_2, \ldots , \tilde{\lambda}_N)$; let $\tilde{\varphi} : \llbracket 1, N \rrbracket \rightarrow \llbracket N_1, N_2 \rrbracket$ denote an arbitrary $\zeta$-localization center bijection for $\tilde{\bm{L}}$; and set $\tilde{\Lambda}_i = \tilde{\lambda}_{\tilde{\varphi}^{-1} (i)}$ for each $i \in \llbracket N_1, N_2 \rrbracket$. Let $(\tilde{\bm{p}}; \tilde{\bm{q}})$ denote the Toda state space initial data associated with the Flaschka variables $(\tilde{\bm{a}}; \tilde{\bm{b}})$, as in \Cref{Open}. Set $\tilde{\bm{q}} = (\tilde{q}_{N_1}, \tilde{q}_{N_1+1}, \ldots , \tilde{q}_{N_2})$. 
	
	By \eqref{ak}, to estimate $\mathfrak{I}$ we must bound $| \mathfrak{F}(\bm{a};\bm{b}) - \mathfrak{F} (\tilde{\bm{a}};\tilde{\bm{b}}) |$ with high probability; we next introduce several events on which such a bound will hold. The first and second are those on which $\bm{a}$, $\eig \bm{L}$, $\tilde{\bm{a}}$, and $\eig \tilde{\bm{L}}$ are bounded. Recalling \Cref{adelta}, set
	\begin{flalign*} 
		\mathsf{E}_1 = \mathsf{BND}_{\bm{L}} (\log N) \cap \mathsf{BND}_{\bm{\tilde{L}}} (\log N); \quad \mathsf{E}_2 = \bigcap_{i=N_1}^{N_2}  \{ a_k \ge e^{-(\log N)^2} \} \cap \{ \tilde{a}_k \ge e^{-(\log N)^2} \}.
	\end{flalign*} 
	
	\noindent Then $\mathsf{E}_1$ is overwhelmingly probable by \Cref{l0eigenvalues}, and $\mathsf{E}_2$ is overwhelmingly probable by the explicit densities of $a_k$ and $\tilde{a}_k$ (from \Cref{mubeta2}). 
	
	The third event is that on which $\eig \bm{L}$ and $\eig \tilde{\bm{L}}$ are close to each other. Specifically, by \Cref{lleigenvalues} (using \eqref{lij2} and our restriction to $\mathsf{E}_1$ to verify its hypothesis \eqref{ll} at $\delta=0$), there is a constant $c_1>0$ and an overwhelmingly probable event $\mathsf{E}_3$, on which the following holds. There exists a bijection $\psi : \llbracket N_1, N_2 \rrbracket \rightarrow \llbracket N_1, N_2 \rrbracket$ such that, for each $i \in \llbracket N_1, N_2 \rrbracket$ with either $\dist ( i, \mathcal{D} ) \ge (\log N)^3 + 2$ or $\dist ( \psi(i), \mathcal{D} ) \ge (\log N)^3 + 2$, we have
	\begin{flalign}
		\label{lambdapsi} 
		|\Lambda_i - \tilde{\Lambda}_{\psi(i)}| \le e^{-c_1 (\log N)^3}, \qquad \text{and} \qquad |   \psi(i)  - i | \le (\log N)^2.
	\end{flalign}
	
	The fourth event is that on which consecutive $q_i$ and $\tilde{q}_i$ are not too close or far, namely, 
	\begin{flalign*}
		\mathsf{E}_4 =  \bigcap_{\substack{i,i' \in \llbracket N_1, N_2 \rrbracket \\ i-i' \ge (\log N)^2}} \bigg\{ \big| q_i - q_{i'} - \alpha(i-i') \big| + \big| \tilde{q}_i - \tilde{q}_{i'} - \alpha(i-i') \big| \le \displaystyle\frac{|\alpha|}{2} \cdot (i-i') \bigg\}.
	\end{flalign*} 
	
	\noindent By \Cref{qij} with a union bound, $\mathsf{E}_4$ is overwhelmingly probable.

	Set $\mathsf{E} = \mathsf{E}_1 \cap \mathsf{E}_2 \cap \mathsf{E}_3 \cap \mathsf{E}_4$, which by a union bound is overwhelmingly probable. In particular, $\mathbb{P} [ \mathsf{E}^{\complement} ] < \mathfrak{p}$ for sufficiently small $\mathfrak{c}>0$, so we will restrict to $\mathsf{E}$ in what follows. By \eqref{ak}, it then suffices to show that 
	\begin{flalign}
		\label{fba4} 
		\big| \mathfrak{F}(\bm{a}; \bm{b}) - \mathfrak{F}(\tilde{\bm{a}}; \tilde{\bm{b}}) \big| \le AB(\log N)^5.
	\end{flalign} 
	
	\noindent To that end, observe from the definition \eqref{f} of $\mathfrak{F}$ and the fact that $\psi$ is a bijection that  
	\begin{flalign}
		\label{fba2} 
		\begin{aligned}
			| \mathfrak{F} (\bm{a};\bm{b}) - \mathfrak{F} (\tilde{\bm{a}};\tilde{\bm{b}}) | &= \Bigg| \displaystyle\sum_{i=N_1}^{N_2} \big( F(\Lambda_i) \cdot G(q_i - q_j) - F(\tilde{\Lambda}_i) \cdot G(\tilde{q}_{i} - \tilde{q}_j) \big) \Bigg| \\
			& \le \displaystyle\sum_{i=N_1}^{N_2} \big| F (\Lambda_i) \cdot G(q_{i} - q_{j}) - F(\tilde{\Lambda}_{\psi(i)}) \cdot G (\tilde{q}_{\psi(i)} -  \tilde{q}_{j}) \big|.
		\end{aligned}
	\end{flalign} 
	
	\noindent The following lemma restricts the sum on the right side of \eqref{fba2} to $i$ satisfying $\dist (i, \mathcal{D}) \ge 2(\log N)^3$. 
	
	\begin{lem} 
		
		\label{dflambdaigestimate}
		
		On $\mathsf{E}$, we have for sufficiently large $N$ that
		\begin{flalign}
			\label{fba3}
			\begin{aligned}  
				\displaystyle\sum_{i=N_1}^{N_2} & \big| F (\Lambda_i) \cdot G(q_{i} - q_{j}) - F(\tilde{\Lambda}_{\psi(i)}) \cdot G (\tilde{q}_{\psi(i)} -  \tilde{q}_{j}) \big| \\
				& \le \displaystyle\sum_{i:\dist(i,\mathcal{D}) \ge 2 (\log N)^3}\big| F (\Lambda_i) \cdot G(q_{i} - q_{j}) - F(\tilde{\Lambda}_{\psi(i)}) \cdot G (\tilde{q}_{\psi(i)} -  \tilde{q}_{j}) \big| + 16 AB(\log N)^3.
			\end{aligned} 
		\end{flalign}
		
	\end{lem} 
	
	\begin{proof}

		First assume that $\mathcal{D} = \{ k \}$ for some integer $k \in \llbracket N_1, N_2 \rrbracket$. Observe for any $i \in \llbracket N_1, N_2 \rrbracket$ that 
		\begin{flalign*} 
			\big| F(\Lambda_i) \cdot G(q_i - q_j) - F(\tilde{\Lambda}_{\psi(i)}) \cdot G(\tilde{q}_{\psi(i)}-\tilde{q}_j) \big| \le 2 \displaystyle\sup_{|\lambda| \le \log N} | F(\lambda) | \cdot \displaystyle\sup_{q \in \mathbb{R}} | G(q) | \le 2 \cdot A \cdot 2B \le 4AB,
		\end{flalign*}
		
		\noindent where the former bound holds by our restriction to the event $\mathsf{E}_1$ and the latter by the second and fourth statements in \Cref{fgab}. This, together with the fact  that there are at most $4(\log N)^3$ indices $i \in \llbracket N_1, N_2 \rrbracket$ with $\dist (i, \mathcal{D}) < 2 (\log N)^3$, yields \eqref{fba3}.
		
		Thus, assume instead that $\mathcal{D} = \{ i \in \llbracket N_1, N_2 \rrbracket : |i-j| > S(\log N)^{9/2} \}$. To verify \eqref{fba3}, it suffices to show that $G(q_i - q_j) = 0 = G(\tilde{q}_{\psi(i)}-\tilde{q}_j)$ if $\dist (i, \mathcal{D}) < 2 (\log N)^3$. To do so, by the fourth statement in \Cref{fgab}, it suffices to verify that 
		\begin{flalign} 
			\label{qiqj2s} 
			|q_i-q_j| \ge 2S, \qquad \text{and} \qquad |\tilde{q}_{\psi(i)} - \tilde{q}_j| \ge 2S, \qquad \text{if $\dist (i, \mathcal{D}) \le 2(\log N)^3$}.
		\end{flalign}
		
		\noindent If $\dist (i, \mathcal{D}) \le 2(\log N)^3$, then since $|i'-j| \ge S(\log N)^{9/2}$ for any $i' \in \mathcal{D}$ we have 
		\begin{flalign}
			\label{ij2s}
			|i-j| \ge \displaystyle\frac{1}{2} \cdot S(\log N)^{9/2}.
		\end{flalign} 
		
		\noindent Hence, by our restriction to the event $\mathsf{E}_4$, we have $|q_i - q_j| \ge |\alpha| \cdot |i-j| / 2 \ge 2 S$; this confirms the first statement in \eqref{qiqj2s}. To verify the second, first observe that entirely analogous reasoning to that above confirms it if $\dist (\psi(i), \mathcal{D}) \le 2 (\log N)^3$. If instead $\dist ( \psi(i), \mathcal{D} ) > 2 (\log N)^3$, then \eqref{ij2s} and \eqref{lambdapsi} give $| \psi(i)-j | \ge |i-j| - (\log N)^2 \ge S(\log N)^{9/2} / 4$. So, by our restriction to $\mathsf{E}_4$, we have $|\tilde{q}_{\psi(i)} - \tilde{q}_j| \ge |\alpha| \cdot | \psi(i)-j| / 2 \ge 2S$. This shows \eqref{qiqj2s} and thus the lemma.
	\end{proof} 
	
	Now we can establish \Cref{estimateik}.
	
	\begin{proof}[Proof of \Cref{estimateik}]
		
		First observe that 
		\begin{flalign}
			\label{sumd} 
			\begin{aligned}
				& \displaystyle\sum_{i:\dist (i, \mathcal{D}) \ge 2 (\log N)^3} \big| F(\Lambda_i) \cdot G(q_i - q_j) - F(\tilde{\Lambda}_{\psi(i)}) \cdot G(\tilde{q}_{\psi(i)} - \tilde{q}_j) \big| \\
				& \qquad \qquad \le \displaystyle\sum_{i: \dist(i, \mathcal{D}) \ge 2 (\log N)^3} \big( A \cdot | G(q_i - q_j) - G(\tilde{q}_{\psi(i)}-\tilde{q}_j) | + 2B \cdot | F(\Lambda_i) - F(\tilde{\Lambda}_{\psi(i)}) | \big) \\
				& \qquad \qquad \le A \displaystyle\sum_{i:\dist(i,\mathcal{D}) \ge 2(\log N)^3} | G(q_i-q_j) - G(\tilde{q}_{\psi(i)} -\tilde{q}_j) | + 2AB N e^{-(\log N)^2}.
			\end{aligned} 
		\end{flalign}
		
		\noindent Here, in the first bound we used the fact that for each $i \in \llbracket N_1, N_2 \rrbracket$ we have $| F(\Lambda_i) | \le A$ (by the second statement in \eqref{fgab} with our restriction to $\mathsf{E}_1$) and that $\big| G(\tilde{q}_{\psi(i)} - \tilde{q}_j) \big| \le 2B$ (by the fourth statement in \eqref{fgab}), and in the second we used the third statement of \Cref{fgab} with \eqref{lambdapsi} and our restriction to $\mathsf{E}_1$. 
		
		To bound the right side of \eqref{sumd}, we claim when $\dist (i, \mathcal{D}) \ge 2(\log N)^3$ that 
		\begin{flalign}
			\label{q2}
			|\tilde{q}_i - \tilde{q}_{\psi(i)}| \le (\log N)^{5/2}, \qquad \text{and} \qquad |q_i - q_j - \tilde{q}_i + \tilde{q}_j| \le 12 (\log N)^2.
		\end{flalign}
		
		\noindent Since $| i - \psi(i) | \le (\log N)^2$ by \eqref{lambdapsi}, the first bound in \eqref{q2} follows from our restriction to the event $\mathsf{E}_4$ (by using it to bound $|\tilde{q}_i - \tilde{q}_k|$ and $|\tilde{q}_k - \tilde{q}_{\psi(i)}|$ for $k = \max \{ i + (\log N)^2, \psi (i) + (\log N)^2 \}$). To verify the second bound in \eqref{q2}, observe that 
		\begin{flalign}
			\label{qij2} 
			|q_i - q_j + \tilde{q}_i - \tilde{q}_j| = 2 \displaystyle\sum_{m=i}^{j-1} |\log a_m - \log \tilde{a}_m| \le 4 (\log N)^2 \cdot \# \{ m \in \llbracket i, j-1 \rrbracket: \dist (m, \mathcal{D}) \le 1 \},
		\end{flalign}
		
		\noindent where in the first statement we used \eqref{q00}, and in the second we used the fact that $a_m = \tilde{a}_m$ unless $\dist (m,  \mathcal{D}) \le 1$ and our restriction to the event $\mathsf{E}_2$ (on which each $|\log a_m - \log \tilde{a}_m| \le 2(\log N)^2$). If $\mathcal{D} = \{ k \}$ for some $k \in \llbracket N_1, N_2 \rrbracket$, then the number of $m \in \llbracket i,j-1 \rrbracket$ with $\dist (m, \mathcal{D}) \le 1$ is at most $3$. If instead $\mathcal{D} = \{ h \in \llbracket N_1, N_2 \rrbracket : |h-j| \ge S(\log N)^{9/2} \}$, then this number is equal to $0$, as $\mathcal{D}$ then does not intersect $[i-1,j]$ (since $\dist (i, \mathcal{D}) \ge 2(\log N)^3$). This with \eqref{qij2} confirms the second bound in \eqref{q2}. 
		
		Using \eqref{q2}, we estimate the right side of \eqref{sumd} through the fourth statement in \Cref{fgab}. Due to presence of the term $B \cdot \mathbbm{1} _{x \ge 0} \cdot \mathbbm{1} _{y \le 0}$ there, it will be useful to define the set $\mathcal{I}_1$ of indices $i \in \llbracket N_1, N_2 \rrbracket$ with $\dist (i, \mathcal{D}) \ge 2 (\log N)^3$, such that either $q_i - q_j \ge 0 \ge \tilde{q}_i - \tilde{q}_j$ or $q_i - q_j \le 0 \le \tilde{q}_i - \tilde{q}_j$. Similarly, we define the set $\mathcal{I}_2$ of indices $i \in \llbracket N_1, N_2 \rrbracket$ with $\dist (i, \mathcal{D}) \ge 2(\log N)^3$, such that either $\tilde{q}_i - \tilde{q}_j \ge 0 \ge \tilde{q}_{\psi(i)} - \tilde{q}_j$ or $\tilde{q}_i - \tilde{q}_j \le 0 \le \tilde{q}_{\psi(i)} - \tilde{q}_j$. Then, due to \eqref{q2} and our restriction to the event $\mathsf{E}_4$, it is quickly verified that $|\mathcal{I}_1| \le (\log N)^3$ and $|\mathcal{I}_2| \le (\log N)^3$. Hence, denoting $\mathcal{I} = \mathcal{I}_1 \cup \mathcal{I}_2$, we have 
		\begin{flalign}
			\label{dsumg}
			\begin{aligned} 
				& \displaystyle\sum_{i:\dist(i,\mathcal{D}) \ge 2(\log N)^3} | G(q_i - q_j) - G(\tilde{q}_{\psi(i)} - \tilde{q}_j) | \\ 
				& \qquad \le \displaystyle\sum_{i:\dist(i,\mathcal{D}) \ge 2(\log N)^3} \big( | G(q_i - q_j) - G(\tilde{q}_i - \tilde{q}_j) | + | G(\tilde{q}_i - \tilde{q}_j) - G(\tilde{q}_{\psi(i)} - \tilde{q}_j) | \big) \\
				& \qquad \le BS^{-1} \displaystyle\sum_{i:\dist(i,\mathcal{D}) \ge 2(\log N)^3} \big( | q_i - q_j -\tilde{q}_i + \tilde{q}_j| \cdot (\mathbbm{1}_{|q_i - q_j| \le 2S} + \mathbbm{1}_{|\tilde{q}_i - \tilde{q}_j| \le 2S}) \\ 
				& \qquad \qquad \qquad \qquad \qquad \qquad \qquad + | \tilde{q}_i - \tilde{q}_{\psi(i)} | \cdot (\mathbbm{1}_{|\tilde{q}_i - \tilde{q}_j| \le 2S} + \mathbbm{1}_{|\tilde{q}_{\psi(i)} - \tilde{q}_j| \le 2S}) \big) + 2B \cdot |\mathcal{I}| \\
				& \qquad \le 4BS^{-1} \cdot (\log N)^{5/2} \cdot \displaystyle\sum_{i = N_1}^{N_2} ( \mathbbm{1}_{|q_i-q_j| \le 2S} + \mathbbm{1}_{|\tilde{q}_i-\tilde{q}_j| \le 2S} + \mathbbm{1}_{|\tilde{q}_{\psi(i)} - \tilde{q}_j| \le 2S})  + 4B(\log N)^3,
			\end{aligned} 
		\end{flalign}
		
		\noindent where in the first bound we decomposed the sum; in the second we used the fourth statement in \eqref{fgab} (with the definitions of $\mathcal{I}_1$ and $\mathcal{I}_2$); and in the third we used \eqref{q2} with the fact that that $|\mathcal{I}| \le |\mathcal{I}_1| + |\mathcal{I}_2| \le 2(\log N)^3$ (and that $N$ is sufficiently large). Now, due to our restriction to the event $\mathsf{E}_4$, there exists a constant $C > 1$ such that there are at most $CS (\log N)^2$ indices $i \in \llbracket N_1, N_2 \rrbracket$ such that either $|q_i - q_j| \le 2S$, $|\tilde{q}_i - \tilde{q}_j| \le 2S$, or $|\tilde{q}_{\psi(i)} - \tilde{q}_j| \le 2S$. Inserting this into \eqref{dsumg} yields
		\begin{flalign*}
			\displaystyle\sum_{i:\dist(i,\mathcal{D}) \ge 2(\log N)^3} | G(q_i - q_j) - G(\tilde{q}_{\psi(i)} - \tilde{q}_j) | \le 8CB (\log N)^{9/2},
		\end{flalign*} 
		
		\noindent which with \eqref{sumd}, \Cref{dflambdaigestimate}, and \eqref{fba2} shows \eqref{fba4} and thus the lemma.			
	\end{proof}

	\subsection{Proof of \Cref{expectationh}} 
	
	\label{ProofH}
	
	In this section we establish \Cref{expectationh}. To that end, it will be useful to define a bounded variant of $F$; we therefore define the function $H: \mathbb{R} \rightarrow \mathbb{R}$ by for each $\lambda \in \mathbb{R}$ setting 
	\begin{flalign*}
		H(\lambda) = F(\lambda) \cdot \mathbbm{1}_{|F(\lambda)| \le A}  -A \cdot \mathbbm{1}_{F(\lambda) < -A} + A \cdot \mathbbm{1}_{F(\lambda) >A}.
	\end{flalign*}
	
	\noindent Observe that $| H(\lambda) | \le A$ for all $\lambda \in \mathbb{R}$, and (by the third property in \Cref{fgab}) that 
	\begin{flalign}
		\label{hxhy2} 
		| H(x) - H(y) | \le e^{-(\log N)^2}, \quad \text{for any $x, y \in [-\log N, \log N]$ with $|x-y| \le e^{-(\log N)^{5/2}}$}.
	\end{flalign} 
	
	We first compare the expectation of the sum of $F(\Lambda_i) \cdot G(q_i - q_j)$ to that of $H(\Lambda_i) \cdot G(q_i -q_j)$. 
	
	\begin{lem}
		
		\label{fhlambda}
		
		There exists a constant $c>0$ such that 
		\begin{flalign}
			\label{fhlambda2} 
			\begin{aligned} 
			\Bigg| \mathbb{E} \bigg[ \displaystyle\sum_{i=N_1}^{N_2} F(\Lambda_i) \cdot G(& q_i - q_j) \bigg] - \mathbb{E} \bigg[ \displaystyle\sum_{i=N_1}^{N_2} H(\Lambda_i) \cdot G(q_i - q_j) \bigg] \Bigg| \le c^{-1} AB e^{-c(\log N)^2}; \\
			& \displaystyle\int_{-\infty}^{\infty} |F(\lambda)-H(\lambda)| \varrho (\lambda) d \lambda \le c^{-1} A e^{-c(\log N)^2}.
			\end{aligned} 
		\end{flalign}
		
	\end{lem} 
	
	\begin{proof} 
		
		Recalling \Cref{adelta}, let $\mathsf{E}_0 = \mathsf{BND}_{\bm{L}} (\log N)$. To show the first statement of the \eqref{fhlambda2}, observe for some $c_1>0$ that
		\begin{flalign*}
			\Bigg|  \mathbb{E} \bigg[ & \displaystyle\sum_{i=N_1}^{N_2} \big( F(\Lambda_i) - H(\Lambda_i) \big) \cdot G(q_i - q_j) \bigg] \Bigg| \\ 
			& = 2B \cdot \mathbb{E} \Bigg[ \displaystyle\sum_{i=1}^N | F(\lambda_i) | \cdot \mathbbm{1}_{\mathsf{E}_0^{\complement}} \Bigg] \le 2BN \cdot \mathbb{E} \bigg[ \displaystyle\max_{\lambda \in \eig \bm{L}} | F(\lambda) |^2 \Big]^{1/2} \cdot \mathbb{P} [ \mathsf{E}_0^{\complement} ]^{1/2} \le c_1^{-1} AB e^{-c_1 (\log N)^2},
		\end{flalign*}
		
		\noindent where the first statement follows from the facts that $| G(q) | \le 2B$ for all $q \in \mathbb{R}$ and \eqref{hxhy2} (which in particular implies that $F(\lambda) = H(\lambda)$ whenever $|\lambda| \le \log N$, by \Cref{fgab}), with the definition of $\mathsf{BND}$; the second from bounding each term in the sum over $\lambda$ by its maximum; and the third from \Cref{l0eigenvalues} and (the $v=2$ case of) \Cref{expectationflambda}. This shows the first bound in \eqref{fhlambda2}. 
		
		To confirm the second, observe for some $c_2, c_3 > 0$ that
		\begin{flalign*}
				\displaystyle\int_{-\infty}^{\infty} |F(\lambda) - H(\lambda)| \varrho (\lambda) d \lambda & \le \displaystyle\int_{\lambda: |F(\lambda)| > A} |F(\lambda)| \varrho (\lambda) d \lambda \\
				&  \le \displaystyle\int_{|\lambda| > \log N} |F(\lambda)| \varrho (\lambda) d \lambda \\
				& \le A \displaystyle\int_{|\lambda| > \log N} e^{|\lambda|^{1/2} - c_2 |\lambda|^2} d \lambda \le A e^{-c_3 (\log N)^2},
		\end{flalign*}
		
		\noindent where the first inequality follows from \eqref{hxhy2}; the second from the second statement in \Cref{fgab}; the third from the first statement in \Cref{fgab}, with \Cref{rhoexponential}; and the fourth from performing the integration.
	\end{proof}

	The following lemma approximates the expectation of the sum of $H(\Lambda_i)$ over some interval $i \in \llbracket n_1, n_2 \rrbracket$; if $\llbracket n_1, n_2 \rrbracket = \llbracket N_1, N_2 \rrbracket$, it may be thought of as a variant of \Cref{lf} with an effective error. We establish it in \Cref{ProofLambdaH} below.
	
	\begin{lem}
	
	\label{sumlambdaiexpectation} 
	
	Fix integers $n_1, n_2 \in \llbracket N_1, N_2 \rrbracket$ with $n_2 \ge n_1$; denote $n = n_2 - n_1 + 1$, and assume that $n \ge (\log N)^5$. For $N$ sufficiently large, we have 
	\begin{flalign*}
		\Bigg| \mathbb{E} \bigg[ \displaystyle\sum_{i=n_1}^{n_2} H (\Lambda_i) \bigg] - n \displaystyle\int_{-\infty}^{\infty} H (\lambda) \varrho(\lambda) d \lambda \Bigg| \le A (\log N)^6.
	\end{flalign*} 
	
	\end{lem} 
	
	The following lemma approximates ``local averages'' of $H(\Lambda_i) \cdot G(q_i - q_j)$, that is, over intervals $i \in \llbracket n_1, n_2 \rrbracket$ of essentially arbitrary size. 
	
	\begin{lem}
	
	\label{n1n2h} 
	
	Fix integers $n_1, n_2 \in \llbracket N_1, N_2 \rrbracket$ with $n_2 \ge n_1$; set $n = n_2 - n_1 + 1$, and assume that $n \ge (\log N)^5$. For $N$ sufficiently large, we have 
	\begin{flalign*}
		\Bigg| \mathbb{E} \bigg[ \displaystyle\sum_{i=n_1}^{n_2} H (\Lambda_i) \cdot G(& q_i - q_j ) \bigg] - G(\alpha n_1 - \alpha j) \cdot  \mathbb{E} \bigg[ \displaystyle\sum_{i=n_1}^{n_2} H (\Lambda_i) \bigg] \Bigg| \\
		& \quad \le A B \big( S^{-1/2} n (\log N)^{9/2} + S^{-1} n^2 \log N + 2\cdot \mathbbm{1}_{|n_1-j| \le 2n} \big).
	\end{flalign*}
	
	\end{lem} 
	
	\begin{proof} 
	
	First observe that 
	\begin{flalign}
		\label{qgq} 
		\Bigg| \displaystyle\sum_{i=n_1}^{n_2} H(\Lambda_i) \cdot G(q_i - q_j) - G(\alpha n_1 - \alpha j) \cdot \displaystyle\sum_{i=n_1}^{n_2} H(\Lambda_i) \Bigg| \le A  \displaystyle\sum_{i=n_1}^{n_2} | G(q_i - q_j) - G(\alpha n_1 - \alpha j) |,
	\end{flalign}
	
	\noindent where we used the fact that $| H(\lambda) | \le A$ for each $\lambda \in \mathbb{R}$. To estimate the right side of \eqref{qgq}, we define the event 
	\begin{flalign}
		\label{eventf0} 
		\mathsf{F} = \bigcap_{i, i' \in \llbracket N_1, N_2 \rrbracket} \{ |q_i - q_{i'} - \alpha (i-i')| \le |i-i'|^{1/2} (\log N)^2 \}.
	\end{flalign}
	
	\noindent By \Cref{qij} and a union bound, $\mathsf{F}$ holds with overwhelming probability. 
	
	Let us observe two inequalities on the event $\mathsf{F}$. The first is 
	\begin{flalign*} 
		\mathbbm{1}_{\mathsf{F}} \cdot (\mathbbm{1}_{|q_i - q_j| \le S} + \mathbbm{1}_{|\alpha i - \alpha j| \le S}) \cdot |i-j| \le S (\log N)^5,
	\end{flalign*} 
	
	\noindent which holds since if $|i-j| \ge S (\log N)^5$ and $\mathsf{F}$ holds then $|q_i - q_j| \ge |\alpha i - \alpha j| - |i-j|^{1/2} (\log N)^2 \ge |\alpha| S (\log N)^5 - S^{1/2} (\log N)^{9/2} > S$ and $|\alpha i - \alpha j| \ge |\alpha| S (\log N)^5 > S$. The second is that, on $\mathsf{F}$, the bounds $q_i \le q_j$ and $\alpha i \ge \alpha j$ can both hold only if $|i-j| \le (\log N)^5$. Indeed, suppose to the contrary that on $\mathsf{F}$ the former two inequalities held, in addition to $|i-j| > (\log N)^5$. Then, we would have $0 \ge q_i - q_j \ge \alpha (i-j) - |i-j|^{1/2} (\log N)^2 > 0$, as $|i-j| > (\log N)^5$, which is a contradiction. Similarly, on $\mathsf{F}$, the bounds $q_i \ge q_j$ and $\alpha i \le \alpha j$ can both hold only if $|i-j| \le (\log N)^5$. 
	
	Therefore, for any $i \in \llbracket n_1, n_2 \rrbracket$, we have 
	\begin{flalign*} 
		| G(q_i - q_j) - G(\alpha i - \alpha j) | \cdot \mathbbm{1}_{\mathsf{F}} & \le BS^{-1} \cdot |i-j|^{1/2} (\log N)^2 \cdot (\mathbbm{1}_{|q_i - q_j| \le S} + \mathbbm{1}_{|\alpha i - \alpha j| \le S}) \cdot \mathbbm{1}_{\mathsf{F}} \\ 
		& \qquad + B \cdot (\mathbbm{1}_{q_i \ge q_j} \cdot \mathbbm{1}_{\alpha i \le \alpha j} + \mathbbm{1}_{q_i \le q_j} \cdot \mathbbm{1}_{\alpha i \ge \alpha j}) \\
		& \le BS^{-1/2} (\log N)^{9/2} + B \cdot \mathbbm{1}_{|i - j| \le (\log N)^5},
	\end{flalign*} 
	
	\noindent where the first statement holds by the fourth part of \Cref{fgab} (with the definition of $\mathsf{F}$), and the second holds by the above two inequalities that hold on $\mathsf{F}$. Hence, 
	\begin{flalign*}
		 | G(q_i - q_j) - G & (\alpha n_1 - \alpha j) | \cdot \mathbbm{1}_{\mathsf{F}}   \\ 
		& \le BS^{-1/2} (\log N)^{9/2} + B \cdot \mathbbm{1}_{|i - j| \le (\log N)^5}  + | G(\alpha i - \alpha j) - G(\alpha n_1 - \alpha j) | \\ 
		& \le BS^{-1/2} (\log N)^{9/2} + B \cdot \mathbbm{1}_{|i - j| \le (\log N)^5} + |\alpha| BS^{-1} n + B \cdot \mathbbm{1}_{i \ge j \ge n_1},
	\end{flalign*} 
	
	\noindent where we again used the fourth statement of \Cref{fgab}, with the fact that $|i-n_1| \le n$. Summing over $i \in \llbracket n_1, n_2 \rrbracket$, this gives (since and $i \ge j \ge n_1$ implies $|n_1-j| \le n$, and $|i-j| \le (\log N)^5$ implies $|n_1-j| \le n + (\log N)^5 \le 2n$, as $n \ge (\log N)^5$) that
	\begin{flalign*}
		\displaystyle\sum_{i=n_1}^{n_2} | G(q_i - q_j) - G(\alpha n_1 - \alpha j) | \cdot \mathbbm{1}_{\mathsf{F}} \le BS^{-1/2} n (\log N)^{9/2} + |\alpha| BS^{-1} n^2  + 2B \cdot \mathbbm{1}_{|n_1-j| \le 2n}.
	\end{flalign*} 
	
	\noindent The lemma then follows from inserting this into \eqref{qgq}; taking expectations; using the fact that $| G(q_i - q_j) - G(\alpha n_1 - \alpha j) | \le 2B$ deterministically (by the fourth part of \Cref{fgab}); and using the fact that $\mathsf{F}$ holds with overwhelming probability. 
	\end{proof}
	
	We can now establish \Cref{expectationh}.

	\begin{proof}[Proof of \Cref{expectationh}]
		
		By \Cref{n1n2h} and \Cref{sumlambdaiexpectation} (with the fact that $|G(q)| \le 2B$ for all $q \in \mathbb{R}$), we find for any $n_1, n_2 \in \llbracket N_1, N_2 \rrbracket$ with $n = n_2 - n_1 + 1 \le \lceil S^{1/2} (\log N)^5 \rceil$ that
		\begin{flalign}
			\label{n1n2h2} 
			\begin{aligned}
			\Bigg| \mathbb{E} \bigg[ \displaystyle\sum_{i=n_1}^{n_2} H(\Lambda_i) \cdot G(q_i - q_j) \bigg] - n \cdot G(\alpha n_1 - \alpha j &) \cdot \displaystyle\int_{-\infty}^{\infty} H(\lambda) \varrho(\lambda) d \lambda \Bigg| \\
			& \le 2AB \big( (\log N)^{11} + \mathbbm{1}_{|n_1-j| \le 2n} \big).
			\end{aligned} 
		\end{flalign}
		
		\noindent We will first apply \eqref{n1n2h2} for a family of intervals $\llbracket n_1, n_2 \rrbracket$ covering a neighborhood of $j$ (of size about $2S (\log N)^{9/2}$). Thus, let $r \le S^{1/2}$ and $n_{1,1} < n_{1,2} < \cdots < n_{1,r}$ be integers with $n_{1,i+1} - n_{1,i} = n = \lceil S^{1/2} (\log N)^5 \rceil$ for each $i \in \llbracket 1, r-1 \rrbracket$ and $n_{1,1} \le j - S(\log N)^{9/2} < j + S(\log N)^{9/2} \le n_{1,r}$. Applying \eqref{n1n2h2}, with the $(n_1, n_2)$ there equal to $(n_{1,m}, n_{1,m+1})$ here, and summing over $m \in \llbracket 1, r-1 \rrbracket$, we obtain
		\begin{flalign}
			\label{hgn11n2r} 
			\Bigg| \mathbb{E} \bigg[ \displaystyle\sum_{i=n_{1,1}}^{n_{1,r}} H(\Lambda_i) \cdot G(q_i - q_j) \bigg] - n \displaystyle\sum_{m=1}^r G(\alpha n_{1,m} - \alpha j) \cdot \displaystyle\int_{-\infty}^{\infty} H(\lambda) \varrho (\lambda) d \lambda \Bigg| \le 3AB S^{1/2} (\log N)^{11},
		\end{flalign}	
		
		\noindent where we used the facts that $r \le S^{1/2}$ and that there are at most $5$ values of $m \in \llbracket 1, r \rrbracket$ for which $|n_{1,m} - j| \le 2n$.

		Next, we must estimate the expectation of $H(\Lambda_i) \cdot G(q_i - q_j)$ when $i \notin \llbracket n_{1,1}, n_{1,r} \rrbracket$; observe for such $i$ that we have $|i-j| \ge S(\log N)^{9/2}$. To that end, recall the event $\mathsf{F}$ from \eqref{eventf0}, which holds with overwhelming probability, by \Cref{qij} (and a union bound). On $\mathsf{F}$, we have that $|q_i - q_j| \ge 2S$ whenever $|i-j| \ge S (\log N)^{9/2}$; in particular, we have $G(q_i - q_j) = 0$ on $\mathsf{F}$ if $i \notin \llbracket n_{1,1}, n_{1,r} \rrbracket$. Together with the deterministic bounds $| H(\Lambda) |\le A$ for all $\Lambda \in \mathbb{R}$ and $| G(q) | \le 2B$ (applied off of $\mathsf{F}$), we find that 
		\begin{flalign*}
			\Bigg| \mathbb{E} \bigg[ \displaystyle\sum_{i \notin \llbracket n_{1,1}, n_{2,r} \rrbracket} H(\Lambda_i) \cdot G(q_i - q_j) \bigg] \Bigg| \le 2c_1^{-1} AB N e^{-c_1 (\log N)^2} \le ABN^{-1}.
		\end{flalign*}
		
		\noindent Together with \eqref{hgn11n2r} and \Cref{fhlambda}, this gives
		\begin{flalign*}
			\Bigg| \mathbb{E} \bigg[ \displaystyle\sum_{i=N_1}^{N_2} F(\Lambda_i) \cdot G(q_i - q_j) \bigg] - n \displaystyle\sum_{m=1}^r G(\alpha n_{1,m} - \alpha j) \cdot \displaystyle\int_{-\infty}^{\infty} F(\lambda) \varrho (\lambda) d \lambda \Bigg| \le 4 AB S^{1/2} (\log N)^{11}.
		\end{flalign*} 
		
		\noindent Since 
		\begin{flalign*}
			\displaystyle\int_{-\infty}^{\infty} |F(\lambda)| \varrho (\lambda) d \lambda \le \displaystyle\int_{-\infty}^{\infty} |H(\lambda)| \varrho (\lambda) d\lambda + A \le 2A,
		\end{flalign*} 
		
		\noindent where the first inequality holds from \Cref{fhlambda}, and the second holds from the facts that $| H(\lambda) | \le A$ and that $\varrho$ is a probability measure, we must show that 
		\begin{flalign*} 
			\Bigg| n \displaystyle\sum_{m=1}^r G(\alpha n_{1,m} - \alpha j) - \displaystyle\int_{-\infty}^{\infty} G(\alpha q) dq \Bigg| \le BS^{1/2} (\log N)^{11}.
		\end{flalign*}
		
		 Since $r \le S^{1/2}$; since $n \le 2 S^{1/2} (\log N)^5$; since $G(\alpha q - \alpha j) = 0$ for $q \notin [n_{1,1}, n_{1,r}]$ (by \Cref{fgab}); and since there are at most $2$ indices $m \in \llbracket 1, r-1 \rrbracket$ for which $n_{1,m+1} \ge j \ge n_{1,m}$, it suffices to show for any $m \in \llbracket 1, r-1 \rrbracket$ that
		\begin{flalign*}
			\Bigg| n \cdot G(\alpha n_{1,m} - \alpha j) - \displaystyle\int_{n_{1,m}}^{n_{1,m+1}} G(\alpha q - \alpha j) dq \Bigg| \le \displaystyle\frac{B}{10} \cdot (\log N)^{11} + Bn \cdot \mathbbm{1}_{n_{1,m+1} \ge j \ge n_{1,m}}.
		\end{flalign*}		
		
		\noindent This follows from the bounds
		\begin{flalign*}
			\Bigg| n \cdot G(\alpha n_{1,m} - \alpha j & ) - \displaystyle\int_{n_{1,m}}^{n_{1,m+1}} G(\alpha q - \alpha j) dq \Bigg| \\ 
			& \le \Bigg| (n_{1,m+1} - n_{1,m}) \cdot G(\alpha n_{1,m} - \alpha j) - \displaystyle\int_{n_{1,m}}^{n_{1,m+1}} G(\alpha q - \alpha j) dq \Bigg| + 2B \\
			& \le n \cdot \displaystyle\max_{q \in [n_{1,m}, n_{1,m+1}]} | G(\alpha n_{1,m} - \alpha j) - G(\alpha q - \alpha j) | + 2B \\ 
			& \le BS^{-1} n \cdot |\alpha n_{1,m+1} - \alpha n_{1,m}| + Bn \cdot \mathbbm{1}_{n_{1,m+1} \ge j \ge n_{1,m}} + 2B \\
			& \le 4 |\alpha| B (\log N)^{10} + Bn \cdot \mathbbm{1}_{n_{1,m+1} \ge j \ge n_{1,m}} + 2B, 
		\end{flalign*}
		
		\noindent where in the first inequality we used the facts that $| G(q) | \le 2B$ for all $q \in \mathbb{R}$ and that $n = n_{1,m+1} - n_{1,m} + 1$; in the second we bounded the integral by its maximum (and used the fact that $n \ge n_{1,m+1} - n_{1,m}$); in the third we used the fourth part of \Cref{fgab}; and in the fourth we used the fact that $n_{1,m+1} - n_{1,m} \le n \le 2S^{1/2} (\log N)^5$.
	\end{proof}

	\section{Regularization and Matrix Bounds}
	
	\label{InverseAsymptotic}
	
	In \Cref{YEstimate} we use the concentration bound \Cref{concentrationh3} to derive a ``regularized'' variant of the asymptotic scattering relation \Cref{ztlambda2} (along the lines of \eqref{qktqk02}), given by \Cref{ztlambda}. In \Cref{EstimateMatrix} we discuss properties of a matrix $\bm{S}$ (defined in \eqref{sij2}) that will eventually arise from formally differentiating this regularized relation.

	\subsection{Regularized Asymptotic Scattering Relation} 
	
	\label{YEstimate}
	
	In this section we show a variant of the asymptotic scattering relation \Cref{ztlambda2}, in which the restrictions on $i$ and logarithms in the sums in \eqref{lambdak2} are incorporated through more regular functions $\chi$ and $\mathfrak{l}$, respectively. To state this more precisely, we require some notation.

	\begin{assumption} 
	
	\label{mchi} 
	
	Adopt \Cref{lbetaeta}, and fix real numbers $B \ge 10$ and $\mathfrak{M} \in [1, T]$. Let $\chi = \chi_{\mathfrak{M}} : \mathbb{R} \rightarrow \mathbb{R}$ denote a smooth function with $\chi'$ even and nonnegative, satisfying the following two properties.
	
	\begin{enumerate} 
		\item If $|x| \ge \mathfrak{M}$ then $\chi (x) = \sgn (\alpha) \cdot \mathbbm{1}_{x > 0}$. 
		\item For each $k \in \{ 0, 1, 2 \}$, we have $| \partial_x^k \chi (x) | \le B \mathfrak{M}^{-k}$ for all $x \in \mathbb{R}$.
	\end{enumerate} 
	
	\noindent Further set $\mathfrak{d} = e^{-5 (\log N)^2}$, and define the function $\mathfrak{l} = \mathfrak{l}_{\mathfrak{d}} : \mathbb{R}$ by for any $x \in \mathbb{R}$ setting 
	\begin{flalign}
		\label{functionl} 
		\mathfrak{l}(x) = \displaystyle\frac{1}{2} \cdot \log |x^2 + \mathfrak{d}^2|.
	\end{flalign} 
	
	\end{assumption} 
	
	We then have the following regularized version of \Cref{ztlambda2}.

	\begin{prop}
		
	\label{ztlambda} 
	
	Adopt \Cref{mchi}, and fix $t \in [0, T]$. The following holds with overwhelming probability. For any index $k \in \llbracket 1, N \rrbracket$ satisfying \eqref{n1n2k02}, we have
		\begin{flalign}
			\label{lambdaq2} 
			\begin{aligned} 
		\Bigg| \lambda_k t - Q_k(t) + Q_k (0) - 2 \displaystyle\sum_{i = 1}^N \mathfrak{l} (\lambda_{k} - \lambda_{i}) \cdot \big( \chi ( Q_k (t) - Q_i (t) ) - \chi ( & Q_k (0) - Q_i (0) ) \big) \Bigg| \\
		& \le B \mathfrak{M}^{1/2} (\log N)^{16}.
		\end{aligned} 
	\end{flalign} 
	
	\end{prop}

	We establish \Cref{ztlambda} as a quick consequence of \Cref{ztlambda2} and the following lemma. 
	
	\begin{lem} 
		
		\label{chiconcentration} 
		
		Fix $t \in [0, T]$. The following holds with overwhelming probability. For any $k \in \llbracket 1, N \rrbracket$ satisfying \eqref{n1n2k02}, we have
		\begin{flalign*}
			\Bigg| \sgn (\alpha) & \displaystyle\sum_{i: Q_t (i) < Q_t (k)} \log | \lambda_k - \lambda_i| - \sgn (\alpha) \displaystyle\sum_{i: Q_0 (i) < Q_0 (k)} \log |\lambda_k - \lambda_i| \\
			& - \displaystyle\sum_{i = 1}^N \mathfrak{l} (\lambda_k - \lambda_i) \cdot \big( \chi ( Q_k (t) - Q_i (t) ) - \chi ( Q_k (0) - Q_i (0) ) \big) \Bigg| \le 11 B \mathfrak{M}^{1/2} (\log N)^{15}.
		\end{flalign*}
		
	\end{lem} 
	
	\begin{proof}[Proof of \Cref{ztlambda}] 
		
		This follows from \Cref{ztlambda2} and \Cref{chiconcentration}.
	\end{proof}

	Now we prove \Cref{chiconcentration} using \Cref{concentrationh3}.

	\begin{proof}[Proof of \Cref{chiconcentration}]
		
		Throughout, we assume for notational convenience that $\alpha > 0$, as the proof when $\alpha < 0$ is entirely analogous. We apply \Cref{concentrationh3} with the $F, G : \mathbb{R} \rightarrow \mathbb{R}$ there defined by setting $F(\lambda) = 1$ for all $\lambda \in \mathbb{R}$, setting $f(x) = \mathfrak{l} (x)$, and setting $G(q) = \chi (-q) - \mathbbm{1} _{q < 0}$ (if $\alpha < 0$, then we instead set $G(q) = \chi (-q) + \mathbbm{1}_{q < 0}$). Observe in this way that $F$, $f$, and $G$ satisfy \Cref{fgab}, \eqref{f1d}, and \eqref{f2d} with the $(A, B, D, S)$ there equal to $(1, B, 5(\log N)^2, \mathfrak{M} )$ here, by \Cref{mchi}. Moreover, since $\chi'$ is even and $G(q) = \chi(-q) - \mathbbm{1} _{q < 0}$ is compactly supported, $G(q)$ is odd in $q$ (away from $q=0$), which means that $\int_{-\infty}^{\infty} G(q) dq = 0$. Therefore \Cref{concentrationh3} yields an overwhelmingly probable event $\mathsf{E}_1$, on which we have 
		\begin{flalign*}
		\Bigg|   \displaystyle\sum_{i: Q_i (s) < Q_k (s)} \mathfrak{l} (\lambda_i - \lambda_k) -  \displaystyle\sum_{i = 1}^N \mathfrak{l} (\lambda_i - \lambda_k) \cdot \chi ( Q_k (s) - Q_i (s) ) \Bigg| \le 5B \mathfrak{M}^{1/2} (\log N)^{15}.
		\end{flalign*} 	
		
		\noindent Subtracting this estimate at $s = t$ from that at $s = 0$ yields on $\mathsf{E}_1$ that
		\begin{flalign*}
		 	\Bigg| \displaystyle\sum_{i: Q_i (t) < Q_k (t)} & \mathfrak{l} (\lambda_i - \lambda_k) - \displaystyle\sum_{i: Q_i (0) < Q_k (0)} \mathfrak{l} (\lambda_i - \lambda_k) \\ 
			& - \displaystyle\sum_{i=1}^N \mathfrak{l} (\lambda_i - \lambda_k) \cdot \big( \chi ( Q_k (t) - Q_i (t) ) - \chi ( Q_k (0) - Q_i (0)) \big) \Bigg| \le 10B \mathfrak{M}^{1/2} (\log N)^{15}.
		\end{flalign*} 
		
		\noindent It therefore remains to show for each $s \in \{ 0, t \}$ that, with overwhelming probability, 
		\begin{flalign}
		\label{qisqks2} 
		\begin{aligned}
			 \Bigg| \displaystyle\sum_{i: Q_i (s) < Q_k (s)} \mathfrak{l} ( \lambda_k - \lambda_i) - \displaystyle\sum_{i: Q_i (s) < Q_k (s)} \log |\lambda_k - \lambda_i| \Bigg| \le 1.
		\end{aligned}
		\end{flalign} 
		
		To that end, recalling \Cref{adelta}, define the event $\mathsf{E}_2 = \mathsf{SEP}_{\bm{L}(0)} (e^{-(\log N)^2})$. By \Cref{l0eigenvalues}, $\mathsf{E}_2$ holds with overwhelming probability. We thus restrict to $\mathsf{E}_1 \cap \mathsf{E}_2$ in what follows. To verify \eqref{qisqks2}, it suffices to show that 
		\begin{flalign*}
			 \displaystyle\sum_{i = 1}^N \big| \mathfrak{l}(\lambda_k - \lambda_i) - \log |\lambda_k - \lambda_i| \big| \le 1.
		\end{flalign*} 
		
		\noindent This follows from the fact that $| \mathfrak{l} (\lambda_k - \lambda_i) - \log |\lambda_k - \lambda_i| | \le \mathfrak{d}^2 |\lambda_k - \lambda_i|^{-2} \le e^{-(\log N)^2}$ (as $|\lambda_k - \lambda_i| \ge e^{-(\log N)^2}$, by our restriction to $\mathsf{E}_2$).
	\end{proof}

	\subsection{Matrix Bounds} 
	
	\label{EstimateMatrix} 
	
	In this section we define and discuss properties of a certain family of matrices that will be useful in analyzing the asymptotic scattering relation \Cref{ztlambda}. These matrices are provided by the following definition. Observe that \eqref{testimateb} below imposes a more stringent constraint on $T$ than in \Cref{lbetaeta}; we will remove it when proving \Cref{vestimate} in \Cref{ProofV}. While we will not impose this in what follows, it will be useful to think of $B = \mathcal{O}(1)$ and $\mathfrak{M} \sim T$.
	
	\begin{assumption} 
		
	\label{chi2}
	
	Adopt \Cref{mchi}. Assume that 
	\begin{flalign}
		\label{testimateb} 
		B \in [10, N^{1/500}], \quad T \in [B^4 (\log N)^{60}, N^{1/10}], \quad \text{and} \quad \mathfrak{M} \in [B^4 (\log N)^{60}, T].   
	\end{flalign} 
	
	\noindent Further fix indices $k_1, k_2 \in \llbracket N_1, N_2 \rrbracket$ satisfying
	\begin{flalign}
		\label{nk1k2} 
		N_1 + T (\log N)^7 \le k_1 \le -N(\log N)^{-10}; \quad N (\log N)^{-10} \le k_2 \le N_2 - T(\log N)^7.
	\end{flalign} 
	
	\noindent For any integers $\ell, m \in \llbracket k_1, k_2 \rrbracket$ with 
	\begin{flalign*} 
		k_1 \le \ell \le -N (\log N)^{-9}; \qquad  N (\log N)^{-9} \le m \le k_2,
	\end{flalign*}
	
	\noindent and $N$-tuples $\bm{\Lambda} = (\Lambda_{N_1}, \Lambda_{N_1+1}, \ldots , \Lambda_{N_2}) \in \mathbb{R}^N$ and $\bm{\mathfrak{Q}} = (\mathfrak{Q}_{N_1}, \mathfrak{Q}_{N_1+1}, \ldots , \mathfrak{Q}_{N_2}) \in \mathbb{R}^N$, define the matrix $\bm{S} = \bm{S}_{\bm{\Lambda}; \bm{\mathfrak{Q}}}^{\llbracket \ell, m \rrbracket} = [S_{ij}] = [S_{ij;\bm{\Lambda};\bm{\mathfrak{Q}}}] \in \SymMat_{\llbracket \ell, m \rrbracket}$ by for any $i, j \in \llbracket \ell, m \rrbracket$ setting 
	\begin{flalign}
		\label{sij2} 
		\begin{aligned} 
		 S_{ij} & = \Bigg( 2 \sum_{k=\ell}^{m} \mathfrak{l} (\Lambda_j -\Lambda_k) \cdot \chi' (\mathfrak{Q}_j - \mathfrak{Q}_k) + 1 \Bigg) \cdot \mathbbm{1}_{i=j} - 2 \mathfrak{l} (\Lambda_j  - \Lambda_i) \cdot \chi' ( \mathfrak{Q}_j - \mathfrak{Q}_i) \\
		 & \qquad + T^3 \cdot (\mathbbm{1}_{i \le \ell+T^2} + \mathbbm{1}_{i \ge m + T^2}) \cdot  \mathbbm{1} _{i=j}.
		 \end{aligned} 
	\end{flalign} 
	
	\end{assumption}
	
	Let us briefly explain the origin of the matrix $\bm{S}$ given by \Cref{chi2}. Consider the asymptotic scattering relation \eqref{lambdaq2}; ignore the error; and formally differentiate it in $t$. This yields a system of linear equations for $Q_j' (t)$, whose coefficients are essentially the entries of $\bm{S}$ (except for the $(i,j)$ coordinates with $i = j \in \llbracket \ell, \ell+T^2 \rrbracket \cup \llbracket m-T^2,  m \rrbracket$, in which case $S_{ij}$ is $T^3$ larger; these additional boundary terms are for convenience in the proofs and will not affect the asymptotics). To solve for the $Q_j' (t)$, one must therefore invert the matrix $\bm{S}$, and estimate its inverse. While we do not know how to do this in general (see \Cref{smatrixchi} below), it can be done under the following assumption \eqref{sumchi} on its entries. 
	
	\begin{lem} 
		
		\label{sk1k2inverse} 
		
		Adopt \Cref{chi2}, and fix any $\varepsilon \in (0, 1)$. Suppose for each $j \in \llbracket \ell, m \rrbracket$ that
		\begin{flalign}
			\label{sumchi} 
			|S_{jj}| \ge (2 + \varepsilon) \displaystyle\sum_{i=\ell}^{m} | \mathfrak{l} (\Lambda_j - \Lambda_i) \cdot \chi' (\mathfrak{Q}_j - \mathfrak{Q}_i) | +  \varepsilon.
		\end{flalign}
		
		\begin{enumerate} 
			\item The matrix $\bm{S}_{\bm{\mathfrak{Q}};\bm{\Lambda}}^{\llbracket \ell, m \rrbracket}$ is invertible, with inverse $\bm{R} = \bm{R}_{\bm{\mathfrak{Q}}; \bm{\Lambda}}^{\llbracket \ell, m \rrbracket} = [R_{ij}] \in \SymMat_{\llbracket \ell, m \rrbracket}$.
			\item Let $\bm{v} = (v_{\ell}, v_{\ell+1}, \ldots , v_{m}) \in \mathbb{R}^{m-\ell+1}$, and denote $\bm{R} \bm{v} = \bm{w} = (w_{\ell}, w_{\ell+1}, \ldots , w_{m})$. Then, for any $i \in \llbracket \ell, m \rrbracket$ and $U \in \mathbb{R}_{\ge 1} \cup \{ \infty \}$, we have 
			\begin{flalign}
				\label{wivi} 
				|w_i| \le \varepsilon^{-1} \cdot \displaystyle\max_{k:|\mathfrak{Q}_k - \mathfrak{Q}_i| \le U \mathfrak{M}} |v_i| + \varepsilon^{-1}  e^{-\varepsilon U / 8} \cdot \displaystyle\max_{k \in \llbracket \ell, m \rrbracket} |v_k|.
			\end{flalign} 
			
		\end{enumerate} 
		
	\end{lem} 
	
	\begin{proof}
		
		By \eqref{sij2}, the bound \eqref{sumchi} guarantees that the matrix $\bm{S}^{\llbracket \ell, m \rrbracket}$ is strictly diagonally dominant and is thus invertible. It remains to establish the second statement of the lemma.
		
		 To that end, set $\delta = \varepsilon / 2$, and let $i_0 \in \llbracket \ell, m \rrbracket$ denote an index such that $|w_k| \le (1+\delta) \cdot |w_{i_0}|$, whenever $k \in \llbracket \ell, m \rrbracket$ satisfies $|\mathfrak{Q}_k - \mathfrak{Q}_{i_0}| \le \mathfrak{M}$. Then, we have 
		\begin{flalign}
			\label{vi0wk} 
			\begin{aligned} 
			|v_{i_0}| & \ge  |w_{i_0}| \cdot |S_{i_0 i_0}|  - 2 \displaystyle\sum_{k = \ell}^{m} | \mathfrak{l} (\Lambda_k - \Lambda_i) \cdot \chi' (\mathfrak{Q}_{i_0} - \mathfrak{Q}_k) \cdot w_k | \\ 
			& \ge \varepsilon |w_{i_0}| + 2(1+\delta) \displaystyle\sum_{k=\ell}^{m} | \mathfrak{l} (\Lambda_{i_0} - \Lambda_k) \cdot \chi' (\mathfrak{Q}_{i_0} - \mathfrak{Q}_k) | \cdot ( |w_{i_0}| - (1+\delta)^{-1} \cdot |w_k| ) \ge \varepsilon |w_{i_0}|,
			\end{aligned} 
		\end{flalign} 
		
		\noindent where in the first statement we used \eqref{sij2}; in the second we used \eqref{sumchi}, with the fact that $2 \delta = \varepsilon$; and in the third we used the fact that $|w_k| \le (1+\delta) \cdot |w_{i_0}|$ whenever $|\mathfrak{Q}_{i_0} - \mathfrak{Q}_k| \le \mathfrak{M}$ (and that $\chi' (\mathfrak{Q}_{i_0} - \mathfrak{Q}_k) = 0$ whenever $|\mathfrak{Q}_{i_0} - \mathfrak{Q}_k| > \mathfrak{M}$, by \Cref{mchi}).
		
		Now assume to the contrary that \eqref{wivi} does not hold, and set $i_1 = i$. Then, by the $i_0 = i_1$ case of \eqref{vi0wk}, we either have that $|w_{i_1}| \le \varepsilon^{-1} \cdot |v_{i_1}|$ or that there exists some index $i_2 \in \llbracket \ell, m \rrbracket$ with $|\mathfrak{Q}_{i_2} - \mathfrak{Q}_{i_1}| \le \mathfrak{M}$ such that $|w_{i_2}| \ge (1+\delta) \cdot |w_{i_1}|$. The former would imply \eqref{wivi} and therefore cannot hold, so the latter must. Applying \eqref{vi0wk} again, now at $i_0 = i_2$, it follows that either $|w_{i_1}| \le |w_{i_2}| \le \varepsilon^{-1} \cdot |v_{i_2}|$ or that there exists some index $i_3 \in \llbracket \ell, m \rrbracket$ with $|\mathfrak{Q}_{i_3} - \mathfrak{Q}_{i_1}| \le \mathfrak{M} + |\mathfrak{Q}_{i_3} - \mathfrak{Q}_{i_2}| \le 2 \mathfrak{M}$ such that $|w_{i_3}| \ge (1+\delta) \cdot |w_{i_2}| \ge (1+\delta)^2 \cdot |w_{i_1}|$. The former again would imply \eqref{wivi} and therefore cannot hold, so the second does. Repeating in this way, and setting $K = \lfloor U \rfloor + 1$, it follows that there exist an index  $i_K \in \llbracket \ell, m \rrbracket$ satisfying $|w_{i_K}| \ge (1 + \delta)^{\lfloor U \rfloor} \cdot |w_{i_1}|$. Hence,
		\begin{flalign*}
			|w_i| \le (1+\delta)^{-\lfloor U \rfloor} \cdot \displaystyle\max_{k \in \llbracket \ell, m \rrbracket} |w_k| & \le \varepsilon^{-1} (1+\delta)^{-U / 2} \cdot \displaystyle\max_{k \in \llbracket \ell, m \rrbracket} |v_k|  \le \varepsilon^{-1} e^{-\varepsilon U / 8} \cdot \displaystyle\max_{k \in \llbracket \ell, m \rrbracket} |v_k|,
		\end{flalign*} 
		
		\noindent where the first bound follows from the fact that $|w_{i_K}| \ge (1+\delta)^{\lfloor U \rfloor} \cdot |w_{i_1}|$; the second from \eqref{vi0wk} (with the $i_0$ there equal to $k$ for which $|w_k|$ is maximal, which implies \eqref{vi0wk}); and the third from the fact that $1+\delta \ge e^{\delta/2} = e^{\varepsilon/4}$ for $\delta \in (0, 1)$. Thus, \eqref{wivi} holds, which is a contradiction; this confirms the lemma.
	\end{proof}
	
	To use \Cref{sk1k2inverse}, we will need to verify the bound \eqref{sumchi} when the $\mathfrak{Q}_j$ are close to the $Q_{\varphi_0^{-1}(j)}$. For general values of $\theta$ (such as if $\beta$ is very large in comparison to $\theta^{-1}$), it can be verified that this bound might in fact be false. However, \Cref{beta0theta} will be used to show that it does hold if $\theta$ is sufficiently small.

		\begin{lem} 
			
			\label{qjt} 
			
			There exists a constant $c>0$ such that the below holds. Adopt \Cref{chi2}, and fix a real number $s \in [0, T]$. The following holds with probability at least $1 - c^{-1} e^{-c(\log N)^2}$. Let $\theta_0 = \theta_0 (\beta) > 0$ be the real number from \Cref{beta0theta}, and suppose that $\theta \in (0, \theta_0)$. Assume for each $j \in \llbracket \ell, m \rrbracket$ that $\Lambda_j = \lambda_{\varphi_0^{-1} (j)}$, and that
			\begin{flalign}
				\label{lambdajqj}  
				\displaystyle\max_{j \in \llbracket \ell, m \rrbracket} | \mathfrak{Q}_j - Q_{\varphi_0^{-1} (j)} (s) | \le (\log N)^{-10} B^{-1} \mathfrak{M}.
			\end{flalign} 
			
			\noindent Then, \eqref{sumchi} holds with the $\varepsilon$ there equal to $(\log N)^{-1}$ here.
			
		\end{lem}
		
		\begin{proof}

			We begin by introducing several events on whose intersection we will be able to verify \eqref{sumchi}. Recalling \Cref{adelta}, define the events $\mathsf{E}_1 = \mathsf{BND}_{\bm{L}} (\log N)$ and
			\begin{flalign}
				\label{e16} 
				\begin{aligned} 
				\mathsf{E}_2 & = \bigcap_{j = \ell}^{m} \Bigg\{ \bigg| 2 \displaystyle\sum_{k=N_1}^{N_2} \mathfrak{l} (\Lambda_j - \Lambda_k ) \cdot \chi' ( Q_{\varphi_0^{-1} (j)} (s) - Q_{\varphi_0^{-1} (k)} (s) ) \\
				& \qquad \qquad \qquad \qquad \qquad \qquad - 2 \alpha^{-1} \displaystyle\int_{-\infty}^{\infty} \mathfrak{l} (\Lambda_j - x) \varrho (x) d x \bigg| \le B \mathfrak{M}^{-1/2} (\log N)^{16} \Bigg\}; \\
				\mathsf{E}_3 & = \bigcap_{j=\ell}^{m} \Bigg\{ \bigg| 2 \displaystyle\sum_{k=N_1}^{N_2} | \mathfrak{l} (\Lambda_j-\Lambda_k) \cdot \chi' ( Q_{\varphi_0^{-1} (j)} (s) - Q_{\varphi_0^{-1} (k)} (s) ) | \\ 
				& \qquad \qquad \qquad \qquad \qquad \qquad - 2 \alpha^{-1} \displaystyle\int_{-\infty}^{\infty} \big| \mathfrak{l} (\Lambda_j-x) \big| \varrho (x) dx \bigg| \le B \mathfrak{M}^{-1/2} (\log N)^{16} \Bigg\}; \\
				\mathsf{E}_4 & = \bigcap_{j=1}^{N} \{ | \varphi_j (s) - \varphi_j (0) | \le T (\log N)^2 \}; \qquad \mathsf{E}_5 = \bigcap_{\substack{i,j \in \llbracket N_1, N_2 \rrbracket \\ |i-j| \ge T(\log N)^5}} \bigg\{ | q_i (s) - q_j (s) | \ge \displaystyle\frac{|\alpha|}{2} \cdot |i-j| \bigg\}; \\
				\mathsf{E}_6 & = \bigcap_{i, j \in \llbracket \ell, m \rrbracket} \big\{ \big| q_{\varphi_0^{-1} (j)} (s) - q_{\varphi_0^{-1} (i)} (s) - \alpha (\varphi_0^{-1} (j) -\varphi_0^{-1} (i) ) \big| \le | \varphi_0^{-1}(i) - \varphi_0^{-1}(j) |^{1/2} \cdot (\log N)^2 \big\}.
				\end{aligned} 
			\end{flalign}
			
			\noindent Observe by \Cref{l0eigenvalues} (and \Cref{centert}, with \eqref{nk1k2}, to verify \eqref{js}) that $\mathsf{E}_1$ is overwhelmingly probable. By the $(F,G; A,B,D,S) = (1, \chi'; 1, (\log N)^3, B\mathfrak{M}^{-1}, \mathfrak{M})$ and $f \in \{ \mathfrak{l}, |\mathfrak{l}|\}$ cases of \Cref{concentrationh3} (with \Cref{mchi} to verify its hypotheses), and the fact that 
			\begin{flalign*}
				\displaystyle\int_{-\infty}^{\infty} \chi' (\alpha q) dq = \alpha^{-1} \cdot (\chi (\mathfrak{M}) - \chi (-\mathfrak{M})) = \alpha^{-1},
			\end{flalign*} 
			
			\noindent $\mathsf{E}_2 \cap \mathsf{E}_3$ is overwhelmingly probable. By \Cref{centert}, $\mathsf{E}_4$ is overwhelmingly probable and, by \Cref{qijsalpha}, $\mathsf{E}_5$ is also overwhelmingly probable. By \Cref{qijsalpha}, with the fact that on $\mathsf{E}_4$ we have $\varphi_s (\varphi_0^{-1} (k)) \in \llbracket N_1 + T(\log N)^3, N_2 - T(\log N)^3 \rrbracket$ for $k \in \llbracket \ell, m \rrbracket$ (by \eqref{nk1k2}, as $|\varphi_s (\varphi_0^{-1} (k)) - k| \le 2T (\log N)^2$ for such $k$ on $\mathsf{E}_4$), $\mathsf{E}_6$ is overwhelmingly probable. Therefore, by a union bound, we restrict to $\mathsf{E} = \bigcap_{i=1}^6 \mathsf{E}_i$ below.
			
			First observe if $\ell + T^2 \le j \le m - T^2$ and $k \notin \llbracket \ell, m \rrbracket$ that 
			\begin{flalign}
				\label{qjqk} 
				| Q_{\varphi_0^{-1}(j)} (s) - Q_{\varphi_0^{-1} (k)} (s) | \ge \displaystyle\frac{\alpha}{4} \cdot T^2  > \mathfrak{M},
			\end{flalign}
			
			\noindent where in the first inequality we used the fact \eqref{qjs2} that $Q_{\varphi_0^{-1} (i)} (s) = q_{\varphi_s (\varphi_0^{-1} (i))} (s)$, our restriction to the event $\mathsf{E}_5$, and the bound 
			\begin{flalign*} 
				|\varphi_s(\varphi_0^{-1} (j)) - \varphi_s (\varphi_0^{-1} (k)) | \ge |j-k| - 4T (\log N)^2 \ge T^2 - 4T(\log N)^2 \ge \displaystyle\frac{1}{2} \cdot T^2,
			\end{flalign*} 
			
			\noindent which holds  by our restriction to the event $\mathsf{E}_4$, the fact that $|j-k| \ge T^2$ for $j \in \llbracket \ell + T^2, m - T^2 \rrbracket$ and $k \notin \llbracket \ell, m \rrbracket$, and \eqref{testimateb}; in the second inequality, we used the fact (recall \Cref{mchi}) that $\mathfrak{M} \le T$. It follows for $\ell + T^2 \le j \le m - T^2$ that 
			\begin{flalign}
				\label{sjjlower} 
				\begin{aligned} 
					|S_{jj}| & = \Bigg| 2 \displaystyle\sum_{k=\ell}^m \mathfrak{l} (\Lambda_j-\Lambda_k) \cdot \chi' (\mathfrak{Q}_j-\mathfrak{Q}_k) + 1 \Bigg| - 10B \mathfrak{M}^{-1} (\log N)^2  \\
					&  \ge \Bigg| 2 \displaystyle\sum_{k=\ell}^{m} \mathfrak{l} (\Lambda_j - \Lambda_k) \cdot \chi' \big( Q_{\varphi_0^{-1} (j)} (s) - Q_{\varphi_0^{-1}(k)} (s) \big) + 1 \Bigg|  \\
					& \qquad \qquad  - 10B \mathfrak{M}^{-2} (\log N)^2 \cdot 2B^{-1} \mathfrak{M} (\log N)^{-10} \cdot \mathfrak{M} (\log N)^5 - 10B \mathfrak{M}^{-1} (\log N)^2 \\
					&  \ge \Bigg| 2 \displaystyle\sum_{k=N_1}^{N_2} \mathfrak{l} (\Lambda_j - \Lambda_k) \cdot \chi' \big( Q_{\varphi_0^{-1} (j)} (s) - Q_{\varphi_0^{-1}(k)} (s) \big) + 1 \Bigg|  \\
					& \qquad \qquad  - 10B \mathfrak{M}^{-2} (\log N)^2 \cdot 2B^{-1} \mathfrak{M} (\log N)^{-10} \cdot \mathfrak{M} (\log N)^5 - 10B \mathfrak{M}^{-1} (\log N)^2 \\
					& \ge \Bigg| 2 \alpha^{-1} \displaystyle\int_{-\infty}^{\infty} \mathfrak{l} (\Lambda_j-x) \varrho (x) dx + 1 \Bigg| - (\log N)^{-2} \ge 2 \alpha^{-1} \displaystyle\int_{-\infty}^{\infty} |\mathfrak{l} (\Lambda_j - x)| \varrho(x) + \displaystyle\frac{1}{4},
				\end{aligned} 
			\end{flalign}
			
			\noindent where in the first statement we used \eqref{sij2}; in the second we used \eqref{lambdajqj}, the fact that $\mathfrak{l} (\Lambda_j-\Lambda_k) \le 5 (\log N)^2$ (by our restriction to $\mathsf{E}_1$), the fact that $|\chi'' (q)| \le B\mathfrak{M}^{-2}$ for all $q \in \mathbb{R}$ (by the second statement of \Cref{mchi}), and the fact that there are at most $\mathfrak{M} (\log N)^5$ indices $k \in \llbracket N_1, N_2 \rrbracket$ for which $|Q_{\varphi_0^{-1}(j)}(s) - Q_{\varphi_0^{-1}(k)}(s)| \le \mathfrak{M}$ (by our restriction to $\mathsf{E}_5 \cap \mathsf{E}_6$) and thus $\chi' (Q_{\varphi_0^{-1}(j)}(s)-Q_{\varphi_0^{-1}(k)}(s)) \ne 0$; in the third we used with \eqref{qjqk} and the fact that $\supp \chi' \subseteq [-\mathfrak{M},\mathfrak{M}]$ to sum over all $k \in \llbracket N_1, N_2 \rrbracket$ (and not only over $k \in \llbracket \ell, m \rrbracket$); and in the fourth we used our restriction to $\mathsf{E}_2$, with \eqref{testimateb}; and in the fourth we used \Cref{beta0theta}, with the fact that $\theta \in (0, \theta_0)$.
					
			Reversing the reasoning in \eqref{sjjlower} (and using our restriction to $\mathsf{E}_3$ in place of that to $\mathsf{E}_2$), we find for $\varepsilon = (\log N)^{-1}$ that  
			\begin{flalign}
				\label{integral0} 
				\begin{aligned} 
				2 \alpha^{-1} \displaystyle\int_{-\infty}^{\infty} |\mathfrak{l}(\Lambda_j-x)| \varrho (x) dx + \displaystyle\frac{1}{4} & \ge 2 \displaystyle\sum_{k=N_1}^{N_2} |\mathfrak{l}(\Lambda_j-\Lambda_k)| \cdot \chi' (Q_{\varphi_0^{-1} (j)} (s) - Q_{\varphi_0^{-1} (k)} (s) ) + \displaystyle\frac{1}{5} \\ 
				& \ge 2 \displaystyle\sum_{k=N_1}^{N_2} |\mathfrak{l}(\Lambda_j-\Lambda_k)| \cdot \chi' (\mathfrak{Q}_j-\mathfrak{Q}_j) + \displaystyle\frac{1}{6} \\
				& \ge (2 + \varepsilon) \displaystyle\sum_{k=N_1}^{N_2} |\mathfrak{l}(\Lambda_j-\Lambda_k)| \cdot \chi' (\mathfrak{Q}_j-\mathfrak{Q}_j) + \varepsilon.
				\end{aligned}
			\end{flalign}
			
			\noindent Here, in the last bound we used the fact that there exists a constant $C>1$ such that
			\begin{flalign*}
				2 \displaystyle\sum_{k=N_1}^{N_2} |\mathfrak{l}(\Lambda_j-\Lambda_k)| \cdot \chi'(\mathfrak{Q}_j-\mathfrak{Q}_k) \le 2 \alpha^{-1} \displaystyle\int_{-\infty}^{\infty} |\mathfrak{l}(\Lambda_j-x)| \varrho(x) dx + \displaystyle\frac{1}{12} < C,
			\end{flalign*}
			
			\noindent where the first inequality holds by the first two bounds of \eqref{integral0} and the second holds since $\varrho$ is bounded and has exponential decay (by \Cref{rhoexponential}). Thus, if $\ell + T^2 \le j \le m - T^2$, then \eqref{sjjlower} and \eqref{integral0} confirm \eqref{sumchi} at $\varepsilon = (\log N)^{-1}$.
			
			 Now suppose that $j \notin \llbracket \ell+T^2,m-T^2 \rrbracket$. In this case, the third statement in \eqref{sjjlower} is no longer guaranteed to hold, as we do not necessarily have $\chi' (Q_{\varphi_0^{-1} (j)} (s) - Q_{\varphi_0^{-1} (k)} (s)) = 0$ for all $k \notin \llbracket \ell, m \rrbracket$. So, we instead make use of the last term of order $T^3$ on the right side of the definition \eqref{sij2} of $S_{ij}$.  To do so, first observe that there are at most $3T (\log N)^2$ indices $k \in \llbracket \ell, m \rrbracket$ for which $|\mathfrak{Q}_j - \mathfrak{Q}_k| \le \mathfrak{M}$ (and thus for which $\chi' (\mathfrak{Q}_j - \mathfrak{Q}_k) \ne 0$), as for $|\varphi_0^{-1}(j)-\varphi_0^{-1}(i)| \ge T (\log N)^5$ we have by \eqref{lambdajqj} and our restriction to $\mathsf{E}_5$ that 
			 \begin{flalign*}
			 	|\mathfrak{Q}_j-\mathfrak{Q}_k| & \ge | Q_{\varphi_0^{-1}(j)} (s) - Q_{\varphi_0^{-1}(k)}(s) | - \mathfrak{M}  \ge \displaystyle\frac{|\alpha|}{2} \cdot| \varphi_0^{-1} (j)-\varphi_0^{-1}(i) |  - \mathfrak{M} \ge \displaystyle\frac{|\alpha|}{2} \cdot T (\log N)^4 > \mathfrak{M}.
			 \end{flalign*}
			 
			 \noindent Therefore,
			\begin{flalign*}
				S_{jj} & \ge T^3 -  2  \displaystyle\sum_{k=\ell}^{m} |\mathfrak{l} (\Lambda_j-\Lambda_k)| \cdot \chi' (\mathfrak{Q}_j-\mathfrak{Q}_k) \\
				& \ge T^3 - 5 (\log N)^2 \cdot 3T (\log N)^5   \\
				& \ge \displaystyle\frac{1}{2} \cdot T^3 \ge 3 \cdot 5 (\log N)^2 \cdot 3T(\log N)^5 \cdot T + 1  \ge 3 \displaystyle\sum_{k=\ell}^{m} |\mathfrak{l}(\Lambda_j-\Lambda_k)| \cdot \chi'(\mathfrak{Q}_j-\mathfrak{Q}_k) + 1,
			\end{flalign*}
			
			\noindent where the first bound holds by \eqref{sij2}; the second and fifth hold by the bounds $|\mathfrak{l}(\Lambda_j-\Lambda_k)| \le 5 (\log N)^2$ (by our restriction to $\mathsf{E}_1$) and $|\chi'(\mathfrak{Q}_j-\mathfrak{Q}_k)| \le B \mathfrak{M}^{-1} \le 1$ (in view of \eqref{testimateb}); and the third and fouth by \eqref{testimateb} and the fact that $N$ is sufficiently large. This again verifies \eqref{sumchi} at $\varepsilon = 1 > (\log N)^{-1}$, thereby establishing the lemma.
		\end{proof}

		\begin{rem} 
			
			\label{smatrixchi} 
			
			Our reason for imposing $\theta \le \theta_0$ below will be to apply \Cref{qjt}. By \Cref{sk1k2inverse}, this verifies that, when the $(\mathfrak{Q}_k)$ are close to the $(Q_{\varphi_0^{-1} (k)})$, the matrix $\bm{S}$ is likely invertible, with eigenvalues essentially bounded away from $0$. One might hope that this statement more generally holds, for some choice of $\chi$ satisfying \Cref{mchi}, without imposing the constraint \eqref{sumchi}. If this were true, we would not need to assume that $\theta$ is sufficiently small in the below analysis.  
			
		\end{rem}

	\section{Proxy Dynamics and Comparisons}
	
	\label{DynamicsQ}
	
	In this section we use the results of \Cref{InverseAsymptotic} to prescribe certain dynamics $(\mathfrak{Q}_k (t))$, which (as mentioned in \Cref{Estimate2}) will serve as a ``proxy'' for the eigenvalue location dynamics $(Q_k (t))$. We first define these proxy dynamics by \Cref{qjt2} in \Cref{Linear0}; we then prove that they are close to the $(Q_j(t))$ dynamics through \Cref{q2q}, in \Cref{NearQ} and \Cref{ProofW}. Throughout this section, we adopt \Cref{chi2}. 
	
	\subsection{Proxy Dynamics}
	
	\label{Linear0}

	As indicated following \Cref{chi2}, we would like to proceed by differentiating the asymptotic scattering relation \eqref{lambdaq2} in $t$. However, the error on the right side of that bound is not necessarily differentiable. So we will instead introduce a ``proxy'' dynamic $(\mathfrak{Q}_j(t))$, in which that error is not present, and show that they are close to $( Q_j (t) )$. We will not be able to verify this approximation for the extreme (first and last few) indices, so it will be convenient to define $( \mathfrak{Q}_j (t) )$ on time intervals of the form $t \in [i,i+1]$, reducing the number of indices $j$ each time $i$ increments. Recall we adopt \Cref{chi2} thoughout.

 	\begin{definition}
 		
 		\label{qjt2} 
 		
 		For each integer $i \in \llbracket 0, T \rrbracket$, denote $\ell_i = k_1 + i T^3$ and $m_i = k_2 - iT^3$. Further define $\bm{\Lambda} = (\Lambda_{k_1}, \Lambda_{k_1+1}, \ldots , \Lambda_{k_2})$ by $\Lambda_j = \lambda_{\varphi_0^{-1} (j)}$ for each $j \in \llbracket k_1, k_2 \rrbracket$. Then define the $N$-tuple $\bm{\mathfrak{Q}}(s) = ( \mathfrak{Q}_{N_1} (s), \mathfrak{Q}_{N_1+1} (s), \ldots , \mathfrak{Q}_{N_2} (s) ) \in \mathbb{R}^N$ as follows. At $s = 0$, set $\mathfrak{Q}_j (0) = Q_{\varphi_0^{-1} (j)} (0) = q_j (0)$ for each $j \in \llbracket N_1, N_2 \rrbracket$. Now suppose $\bm{\mathfrak{Q}}(s)$ has been defined for all $s \in [0, i]$, for some index $i \in \llbracket 0, T \rrbracket$. 
 		
 			\begin{enumerate} 
 				\item If $j \notin \llbracket \ell_i, m_i \rrbracket$, then set $\mathfrak{Q}_j (s) = \mathfrak{Q}_j (i)$ for each $s \ge i$.   
 				\item If instead $j \in \llbracket \ell_i, m_i \rrbracket$, then for $s \in [i, i+1]$ let $\mathfrak{Q}_j (s)$ be the unique solution to the system of ordinary differential equations 
 		\begin{flalign}
 			\label{derivativeqjs}
 			\mathfrak{Q}_j' (s) = \displaystyle\sum_{k=\ell_i}^{m_i} R_{jk;s}^{\llbracket \ell_i, m_i \rrbracket} \cdot \Lambda_k,
 		\end{flalign}
 		
 		\noindent where $\bm{R}_s^{\llbracket \ell_i, m_i \rrbracket} = [ R_{jk;s}^{\llbracket \ell_i, m_i \rrbracket} ] = ( \bm{S}_{\bm{\Lambda}; \bm{\mathfrak{Q}}(s)}^{\llbracket \ell_i, m_i \rrbracket} )^{-1}$, with initial data $\mathfrak{Q}_j (i)$ determined by the fact that $\mathfrak{Q}_j (s)$ has been defined at $s=i$. Here, we assume under the above definition that $\bm{S}_{\bm{\Lambda};\bm{\mathfrak{Q}}(r)}^{\llbracket \ell_i, m_i \rrbracket}$ is invertible if $r \in [i, i+1]$, so that the Picard--Lindel\"{o}f theorem guarantees that \eqref{derivativeqjs} has a unique solution. If $\bm{S}_{\bm{\Lambda};\bm{\mathfrak{Q}}(r)}^{\llbracket \ell_i, m_i \rrbracket}$ is not invertible for some $r \in [i, i+1]$, we set $\mathfrak{Q}_j (s) = \mathfrak{Q}_j (i)$ whenever $s \in [i, i+1]$. 
 		
 		\end{enumerate} 
 		
 	\end{definition} 
 	
 	\begin{rem} 
 		
 		\label{derivative2qjs} 
 		
 		Observe that \eqref{derivativeqjs} is equivalent, upon multplying by $\bm{S}_{\bm{\Lambda}; \bm{\mathfrak{Q}}(s)}^{\llbracket \ell_i, m_i \rrbracket}$ (using the facts that $\mathfrak{l}$ and $\chi'$ are even) to the equation 
 		\begin{flalign}
 			\label{derivativeqjs2} 
 			\begin{aligned} 
 			\Lambda_j & = \mathfrak{Q}_j' (s) \cdot \Bigg( 2 \displaystyle\sum_{k=\ell_i}^{m_i} \mathfrak{l}(\Lambda_j-\Lambda_k) \cdot \chi' (\mathfrak{Q}_j (s) -\mathfrak{Q}_k (s) ) + 1 + T^3 \cdot (\mathbbm{1}_{j \le \ell_i +T^2} + \mathbbm{1}_{j \ge m_i - T^2}) \Bigg) \\
 			& \qquad - 2 \displaystyle\sum_{k=\ell_i}^{m_i} \mathfrak{Q}_k' (s) \cdot \mathfrak{l}(\Lambda_j-\Lambda_k) \cdot \chi' ( \mathfrak{Q}_j (s) -\mathfrak{Q}_k (s) ). 
 			\end{aligned} 
 		\end{flalign} 
 	\end{rem}

 	\subsection{Comparison Between $\bm{Q}$ and $\bm{\mathfrak{Q}}$}
 	
 	\label{NearQ} 
 	
 	 In this section we establish the following proposition indicating that $\mathfrak{Q}_j (s)$ approximates $Q_{\varphi_0^{-1}(j)} (s)$ with high probability. Recall we adopt \Cref{chi2} (and use the notation from \Cref{qjt2}) throughout.
 	 
 	\begin{prop}
 		
 		\label{q2q} 
 		
 		The following holds with overwhelming probability. Let $i \in \llbracket 0, T \rrbracket$ be an integer, and let $\theta_0 = \theta_0 (\beta) > 0$ denote the real number from \Cref{beta0theta}; assume that $\theta \in (0, \theta_0)$. 
 		
 		\begin{enumerate} 
 			\item For any $r \in [i,i+1]$, the matrix $\bm{S}_{\bm{\Lambda};\bm{\mathfrak{Q}}(r)}^{\llbracket \ell_i, m_i \rrbracket}$ is invertible.
 			\item For any $j \in \llbracket \ell_i, m_i \rrbracket$ and $t \in [0, i+1]$, we have $| \mathfrak{Q}_j (t) - Q_{\varphi_0^{-1} (j)} (t) | \le B \mathfrak{M}^{1/2} (\log N)^{20}$ and $|\mathfrak{Q}_j'(t)| \le (\log N)^3$.
 		\end{enumerate} 
 	\end{prop} 
 	 		
 		We begin by, for each $j \in \llbracket N_1, N_2 \rrbracket$ and $s \in [0, T]$, setting
 		\begin{flalign}
 			\label{wjs} 
 			w_j (s) = \mathfrak{Q}_j (t) - Q_{\varphi_0^{-1} (j)} (t).
 		\end{flalign} 
 		
 		\noindent Further let $\mathfrak{T}$ be the supremum over all $S \in [0, T]$ such that we have both
 		\begin{flalign}
 			\label{qjderivativen10}
 			\begin{aligned}
 			& |\mathfrak{Q}_j' (t)| \le (\log N)^3, \qquad \qquad  \quad \text{for each $t \in (0, S)$ and $j \in \llbracket \ell_{\lfloor S  \rfloor}, m_{\lfloor S  \rfloor} \rrbracket$}; \\
 			& |w_j (t)| \le B \mathfrak{M}^{1/2} (\log N)^{20}, \qquad \text{for each $t \in (0, S)$ and $j \in \llbracket \ell_{\lfloor S  \rfloor}, m_{\lfloor S  \rfloor} \rrbracket$}.
 			\end{aligned} 
 		\end{flalign} 
 		
 		\noindent The following two results quickly imply \Cref{q2q}. In the below, for any $r \in [0, T]$, we write $\bm{S}_{\bm{\Lambda};\bm{\mathfrak{Q}}(r)}^{\llbracket \ell_i, m_i \rrbracket} = [S_{ij} (r)]$ if the index $i \in \llbracket 0, T \rrbracket$ is given.
 		
 		\begin{lem} 
 			
 		\label{sr}  
 		
 		The following holds with overwhelming probability. For any $i \in \llbracket 0, \mathfrak{T}  \rrbracket$, we have 
			\begin{flalign}
				\label{sjjs2} 
				S_{jj} (s) \ge (1+(2 \log N)^{-1} ) \displaystyle\sum_{\substack{k \in \llbracket \ell_i, m_i \rrbracket \\ k \ne j}} |S_{jk} (s)| + (2\log N)^{-1}, \quad \text{for all $(j,s) \in \llbracket \ell_i, m_i \rrbracket \times [0, \mathfrak{T}]$.}
			\end{flalign} 
			
			\noindent Furthermore, the matrix $\bm{S}_{\bm{\Lambda};\bm{\mathfrak{Q}}(s)}^{\llbracket \ell_i,m_i\rrbracket}$ is invertible for each $s \in [0, \mathfrak{T}]$.
		
 		\end{lem} 
 		
 		\begin{prop} 
 			
 		\label{tt}

 		With overwhelming probability, we have $\mathfrak{T} = T$. 
 		
 		\end{prop} 
 		
 		\begin{proof}[Proof of \Cref{q2q}]
 			
 			This follows from (the last statement of) \Cref{sr}, \Cref{tt}, and the definition of $\mathfrak{T}$ as the supremum over all $S \in [0, T]$ satisfying \eqref{qjderivativen10}.
 		\end{proof} 
 		
 		\begin{proof}[Proof of \Cref{sr}]
 			
 			Throughout this proof, recalling \Cref{adelta}, we restrict to the event $\mathsf{E} = \mathsf{BND}_{\bm{L}} (\log N)$, which we may by \Cref{l0eigenvalues}. It suffices to verify \eqref{sjjs2} since, by \Cref{sk1k2inverse}, it implies that $\bm{S}_{\bm{\Lambda};\bm{\mathfrak{Q}}(r)}^{\llbracket \ell_i, m_i \rrbracket}$ is invertible. To that end, let $\mathcal{T} \subseteq [0, \mathfrak{T}]$ denote a $N^{-20}$-mesh of $[0, \mathfrak{T}]$, and define the events 
 		\begin{flalign*} 
 			\mathsf{F}_i (s) = \bigcap_{j=\ell_i}^{m_i} \bigg\{ |S_{jj} (s)| \ge (1+(\log N)^{-1} ) \displaystyle\sum_{\substack{k \in \llbracket \ell_i, m_i \rrbracket \\ k \ne j}} |S_{jk} (s)| + (\log N)^{-1} \bigg\}; \qquad \mathsf{F} = \bigcap_{i=0}^{\lfloor \mathfrak{T} \rfloor} \bigcap_{s \in \mathcal{T}} \mathsf{F}_i (s),
 		\end{flalign*} 
 		
 		\noindent for any $i \in \llbracket 0, \mathfrak{T} \rrbracket$ and $s \in [0,\mathfrak{T}]$, where we have abbreviated the $(j,k)$ entry of $\bm{S}_{\bm{\Lambda};\bm{\mathfrak{Q}}(s)}^{\llbracket \ell_i,m_i \rrbracket}$ by $S_{jk}(s)$. By the definition of $\mathfrak{T}$, \eqref{lambdajqj} holds for any $j \in \llbracket \ell_i, m_i \rrbracket$ and $i \in \llbracket 0,  \mathfrak{T} \rrbracket$, so \Cref{qjt} with a union bound yields that $\mathsf{F}$ is overwhelmingly probable. Hence, we also restrict to $\mathsf{F}$.  
 		
 		For any $s \in [0,\mathfrak{T}]$, there exists an $s_0 \in \mathcal{T}$ with $|s-s_0| \le N^{-20}$. Observe for $j, k \in \llbracket \ell_i, m_i \rrbracket$ we have $|\mathfrak{l}(\Lambda_j - \Lambda_k)| \le 5 (\log N)^2$ (by our restriction to $\mathsf{E}$) and 
 		\begin{flalign*} 
 			\big| \chi' (\mathfrak{Q}_j(s)-\mathfrak{Q}_k (s) ) - \chi' (\mathfrak{Q}_j(s_0)-\mathfrak{Q}_i(s_0) ) \big| \le 2B \mathfrak{M}^{-2} \cdot (\log N)^3 \cdot |s-s_0| \le N^{-10},
 		\end{flalign*} 
 		
 		\noindent by the second statement in \Cref{mchi}, the first statement in \eqref{qjderivativen10}, and the facts that $B \le T \le N$ and $|s-s_0| \le N^{-20}$. Hence, we have for such $(j,k)$ that $|S_{jk} (s) - S_{jk} (s_0)| \le N^{-5}$, by \eqref{sij2}. Together with our restriction to $\mathsf{F} (s_0)$, this confirms \eqref{sjjs2}.
 		\end{proof} 
 		
 		We next show \Cref{tt} through a continuity argument, by verifying the following variant of \eqref{qjderivativen10}, in which the second inequality there is replaced by a stronger one; we will prove it in \Cref{ProofW} below. For every $i \in \llbracket 0, T \rrbracket$, we denote $\mathfrak{T}_i = \min \{ \mathfrak{T}, i+1 \}$.

 		\begin{prop} 
 			
 			\label{wjt} 
 			
 			The following holds with overwhelming probability. For any indices $i \in \llbracket 0, \mathfrak{T}  \rrbracket$ and $j \in \llbracket \ell_i, m_i \rrbracket$, and real number $t \in [0, \mathfrak{T}_i - 1]$, we have 
 			\begin{flalign}
 				\label{qjt3} 
 				 |w_j(t)| \le 3B \mathfrak{M}^{1/2} (\log N)^{19},
 			\end{flalign} 
 		\end{prop}

 		\begin{proof}[Proof of \Cref{tt}]
 			
 			We first claim that the following holds with overwhelming probability. For any indices $i \in \llbracket 0, \mathfrak{T}  \rrbracket$ and $j \in \llbracket \ell_i, m_i \rrbracket$, and real number $t \in [0, \mathfrak{T}_i]$, we have $|\mathfrak{Q}_j'(t)| \le 2 (\log N)^2$. To that end, fix $i \in \llbracket 0, \mathfrak{T}  \rrbracket$ and, recalling \Cref{adelta}, we restrict to the event $\mathsf{BND}_{\bm{L}} (\log N)$, which we may by \Cref{l0eigenvalues}. Then, we have 
 			\begin{flalign}
 				\label{qjt4} 
 				\displaystyle\max_{j \in \llbracket \ell_i, m_i \rrbracket} \displaystyle\sup_{t \in [0, i]} |\mathfrak{Q}_j'(t)| \le 2 \log N \cdot \displaystyle\max_{\Lambda \in \bm{\Lambda}} |\Lambda| \le 2 (\log N)^2,
 			\end{flalign}
 			
 			\noindent where the first bound follows from \eqref{derivativeqjs} and the $U = \infty$ case of \Cref{sk1k2inverse} (using \eqref{sjjs2} to verify its hypothesis \eqref{sumchi} at $\varepsilon = (2 \log N)^{-1}$); the second follows from our restriction to $\mathsf{BND}_{\bm{L}} (\log N)$.

 		Throughout the remainder of this proof, we restrict to the event $\mathsf{E}_1$ on which \Cref{centert} holds; the event $\mathsf{E}_2$ on which \Cref{centerdistance} holds; to the event $\mathsf{E}_3$ on which \Cref{sr} holds; to the event $\mathsf{E}_4$ on which \eqref{qjt4} holds; and to the event $\mathsf{E}_5$ on which \Cref{wjt} holds. 
 		
 		Now, fix $i \in \llbracket 0, \mathfrak{T}  \rrbracket$ and $j \in \llbracket \ell_i, m_i \rrbracket$. It suffices to show that there exists some $\mathfrak{T}_i' > \mathfrak{T}_i$ for which $|\mathfrak{Q}_j'(t)| \le (\log N)^3$ and $|w_j (t)| \le B \mathfrak{M}^{1/2} (\log N)^{20}$, as then this would contradict the maximality of $\mathfrak{T}$. For the former, in view of \eqref{qjt4} (and our restriction to $\mathsf{E}_4$), we have for each $t \in [0, \mathfrak{T}_i]$ that $|\mathfrak{Q}_j' (t)| \le 2 (\log N)^2$. Observe that $\mathfrak{Q}_j'$ is half-continuous at $t = \mathfrak{T}_i$, since by \Cref{sr} (and our restriction to $\mathsf{E}_3$) the matrix $\bm{S}_{\bm{\Lambda}; \bm{Q}(\mathfrak{T}_i)}^{\llbracket \ell_i, m_i \rrbracket}$ satisfies \eqref{sjjs2} and is thus invertible. It follows that there exists some $\mathfrak{T}_i' > \mathfrak{T}_i$ such that $|\mathfrak{Q}_j' (t)| \le (\log N)^3$ for each $t \in [0, \mathfrak{T}_i']$, for sufficiently large $N$.
 		
 		For the latter, we may assume that the above $\mathfrak{T}_i' \in (\mathfrak{T}_i, \mathfrak{T}_i + 1)$. It is quickly verified that $\varphi_t (\varphi_0^{-1} (j)) \in \llbracket N_1 + T(\log N)^4, N_2 - T(\log N)^4 \rrbracket$ for all $t \in [0, T]$, using the fact that $j \in \llbracket k_1, k_2 \rrbracket$, \eqref{nk1k2}, and \Cref{centert} (with our restriction to $\mathsf{E}_1$); thus, \Cref{centerdistance} (with our restriction to $\mathsf{E}_2$) applies with the $\varphi$ there equal to $\varphi_t (\varphi_0^{-1} (j))$ here. Hence, for any $t \in [\mathfrak{T}_i, \mathfrak{T}_i']$, we have 
 		\begin{flalign*}
 			|w_j (t)| & \le |w_j (\mathfrak{T}_i - 1)| + (t - \mathfrak{T}_i + 1) (\log N)^3 + \big| q_{\varphi_t (\varphi_0^{-1} (j))} (t) - q_{\varphi_{\mathfrak{T}_i-1} (\varphi_0^{-1} (j))} (\mathfrak{T}_i-1) \big| \\
 			& \le 3B \mathfrak{M}^{1/2} (\log N)^{19} + 2(\log N)^3 + 4(\log N)^4 \le B \mathfrak{M}^{1/2} (\log N)^{20},
 		\end{flalign*}
 		
 		\noindent which establishes the lemma. Here, in the first inequality we used the definitions \eqref{wjs} of $w_j$ and \eqref{qjs2} of $Q_j$, with the fact that $|\mathfrak{Q}_j' (s)| \le (\log N)^3$ for $s \in [0, \mathfrak{T}_i']$; in the second we used \eqref{qjt3} (with our restriction to $\mathsf{E}_5$), the fact that $t - \mathfrak{T}_i \le 1$, and the second statement of \Cref{centerdistance} (with our restriction to $\mathsf{E}_2$); and in the third we used the $N$ is sufficiently large.
 		 \end{proof}

 		\subsection{Proof of \Cref{wjt}}
 			
 		\label{ProofW}
 		
 		In this section we establish \Cref{wjt}. Throughout, we adopt the notation of \Cref{NearQ} and further fix an integer $i \in \llbracket 0, \mathfrak{T} \rrbracket$; an index $j \in \llbracket \ell_i, m_i \rrbracket$; and a real number $t \in [0, \mathfrak{T}_i-1]$. Additionally, recalling \Cref{adelta}, we restrict to the event $\mathsf{E}_1 = \mathsf{BND}_{\bm{L}(0)} (\log N)$ throughout this proof, which we may by \Cref{l0eigenvalues}. We further restrict to the event $\mathsf{E}_2$ on which \Cref{qijsalpha} holds; to the event $\mathsf{E}_3$ on which \Cref{centert} holds; and to the event $\mathsf{E}_4$ on which \Cref{sr} holds.
 		
 		The following lemma indicates that $|Q_{\varphi_0^{-1} (j)} - Q_{\varphi_0^{-1} (k)}|$ and $|\mathfrak{Q}_j - \mathfrak{Q}_k|$ are large if $|j-k|$ is large. We will frequently use the fact that for sufficiently large $N$ we have 
 		\begin{flalign}
 			\label{kestimate} 
 			\varphi_s (\varphi_{s'}^{-1} (k)) \in \llbracket N_1 + T(\log N)^6, N_2 - T (\log N)^6 \rrbracket, \quad \text{for any $k \in \llbracket k_1, k_2 \rrbracket$ and $s, s' \in [0, T]$}.
 		\end{flalign}
 		
 		\noindent as follows quickly from \eqref{nk1k2} and \Cref{centert} (with our restriction to $\mathsf{E}_3$). 
 		
 		\begin{lem}
 		
 		\label{qjqk2} 
 		
 		Fix a real number $s \in [0, \mathfrak{T}_i]$ and an index $k \in \llbracket k_1, k_2 \rrbracket$ satisfying $|\varphi_s (\varphi_0^{-1} (j)) - \varphi_s (\varphi_0^{-1} (k))| \ge \mathfrak{M} (\log N)^5$. For $N$ sufficiently large, we have
 		\begin{flalign}
 			\label{qjqk3} 
 			\big| Q_{\varphi_0^{-1} (j)} (s) - Q_{\varphi_0^{-1} (k)} (s) \big| > 3 \mathfrak{M}; \qquad |\mathfrak{Q}_j (s) - \mathfrak{Q}_k (s)| > 2 \mathfrak{M}.
 		\end{flalign}
 		\end{lem} 
 		
 		\begin{proof} 
 			
 			To establish the first bound in \eqref{qjqk3}, observe that
 			\begin{flalign}
 				\label{qjqk4} 
 				\begin{aligned} 
 				\big| Q_{\varphi_0^{-1}(j)} (& s) - Q_{\varphi_0^{-1}(k)} (s) \big| \\
 				& = \big| q_{\varphi_s (\varphi_0^{-1} (j))} (s) - q_{\varphi_s(\varphi_0^{-1} (k))} (s) \big| \\
 				& \ge |\alpha| \cdot \big| \varphi_s (\varphi_0^{-1} (j)) - \varphi_s (\varphi_0^{-1} (k)) \big| - \big| \varphi_s (\varphi_0^{-1} (j)) - \varphi_s (\varphi_0^{-1} (k)) \big|^{1/2} \cdot (\log N)^2 > 3 \mathfrak{M},
 				\end{aligned} 
 			\end{flalign} 
 			
 			\noindent where in the first statement we used \eqref{qjs2}; in the second we used \eqref{qiqjs4} (and our restriction to $\mathsf{E}_2$, with \eqref{kestimate} to verify its assumption); and in the third we used the bound $| \varphi_s (\varphi_0^{-1} (j)) - \varphi_s (\varphi_0^{-1} (k)) | \ge \mathfrak{M} (\log N)^5$. To establish the second bound in \eqref{qjqk3}, observe that 
 			\begin{flalign*}
 				|\mathfrak{Q}_j (s) - \mathfrak{Q}_k (s)| \ge \big| Q_{\varphi_0^{-1} (j)} (s) - Q_{\varphi_0^{-1} (k)} (s) \big| - |w_j (s)| - |w_k (s)| \ge 3 \mathfrak{M} - 2 B \mathfrak{M}^{1/2} (\log N)^{20} > 2 \mathfrak{M},
 			\end{flalign*}	
 			
 			\noindent where in the first statement we used the definition \eqref{wjs} of $w_j$; in the second we used the first inequality \eqref{qjqk3} with the second inequality in \eqref{qjderivativen10}; and in the third we used \eqref{testimateb}.
 		\end{proof} 
 		
 		\begin{proof}[Proof of \Cref{wjt}]
 			
 			By \Cref{sr} (with our restriction to $\mathsf{E}_4$), for each $i' \in \llbracket 0, \mathfrak{T} \rrbracket$ and $s \in [i', \mathfrak{T}_{i'}]$, the matrix $\bm{S}_{\bm{\Lambda};\bm{\mathfrak{Q}}(s)}^{\llbracket \ell_{i'}, m_{i'} \rrbracket}$ is invertible. Therefore, by \Cref{qjt2}, \eqref{derivativeqjs2} holds for the $i$ there equal to any $i' \in \llbracket 0, \mathfrak{T} \rrbracket$, and with the $(j, s)$ there equal to any element of $\llbracket \ell_{i'}, m_{i'} \rrbracket \times [i', \mathfrak{T}_{i'}]$. Thus, for such $(j, s)$, we have
 		 \begin{flalign}
 		 	\label{lambdajsum2}
 		 	\begin{aligned}
 		 \Lambda_j & = \mathfrak{Q}_j'(s) \cdot (1 + T^3 \cdot (\mathbbm{1}_{j \le \ell_{i'} + T^2} + \mathbbm{1}_{j \ge m_{i'} - T^2}) ) \\
 		 & \qquad +  2 \displaystyle\sum_{k=\ell_{i'}}^{m_{i'}} ( \mathfrak{Q}_j'(s) - \mathfrak{Q}_k'(s) ) \cdot \mathfrak{l} (\Lambda_j - \Lambda_k) \cdot \chi' ( \mathfrak{Q}_j (s) - \mathfrak{Q}_k (s) ).
 		 \end{aligned} 
 		 \end{flalign}
 		 
 		 \noindent Now set $i_0 = i-1$, and denote $\ell_i' = \ell_i - \lfloor T^3/2 \rfloor$ and $m_i' = m_i + \lfloor T^3/2 \rfloor$; in this way, we have $\llbracket \ell_i, m_i \rrbracket \subseteq \llbracket \ell_i', m_i' \rrbracket \subseteq \llbracket \ell_{i_0}, m_{i_0} \rrbracket$. Let us assume that $i' \in \llbracket 0, i_0 \rrbracket$ and that $(j, s) \in \llbracket \ell_i', m_i' \rrbracket \times [i', \mathfrak{T}_{i'}]$. Since $j \in \llbracket \ell_i', m_i' \rrbracket$, and $\ell_i' \ge \ell_{i'} + T^3/2$ and $m_i' \le m_{i'} + T^3/2$, we have $\mathbbm{1}_{j \le \ell_{i'}+ T^2} + \mathbbm{1}_{j \ge m_{i'} - T^2} = 0$. Also, since for any $k \notin \llbracket \ell_{i_0}, m_{i_0} \rrbracket$ we have by \Cref{centert} (with our restriction to $\mathsf{E}_3$) and \eqref{testimateb} that 
 		 \begin{flalign*} 
 		 	\big| \varphi_s (\varphi_0^{-1} (j)) - \varphi_s (\varphi_0^{-1} (k)) \big| \ge |k-j| - 4T (\log N)^2 > T^2 - 4T(\log N)^2 > \mathfrak{M} (\log N)^5,
 		 \end{flalign*} 
 		 
 		 \noindent it follows from \Cref{qjqk2}, and the inclusion $\supp \chi \subseteq [-\mathfrak{M}, \mathfrak{M}]$ (by \Cref{mchi}), that $\chi' (\mathfrak{Q}_j (s) - \mathfrak{Q}_k (s)) = 0$ whenever $k \notin \llbracket \ell_{i_0}, m_{i_0} \rrbracket$. Inserting these into \eqref{lambdajsum2} yields for any $(j, s) \in \llbracket \ell_i', m_i' \rrbracket \times [0, \mathfrak{T}_{i_0}]$ that
 		 	 \begin{flalign}
 		 	 	\label{lambdajsum3} 
 		 	\Lambda_j & = \mathfrak{Q}_j'(s) +  2 \displaystyle\sum_{k=\ell_{i_0}}^{m_{i_0}} ( \mathfrak{Q}_j'(s) - \mathfrak{Q}_k'(s) ) \cdot \mathfrak{l} (\Lambda_j - \Lambda_k) \cdot \chi' ( \mathfrak{Q}_j (s) - \mathfrak{Q}_k (s) ).
 		 \end{flalign}
 		 
 		  Integrating \eqref{lambdajsum3} over $s \in [0, t]$ for any fixed $t \in [0, \mathfrak{T}_{i_0}]$ yields for any $j \in \llbracket \ell_i', m_i' \rrbracket$ that
 		 \begin{flalign*}
 		 	\Lambda_j t & = \mathfrak{Q}_j (t) - \mathfrak{Q}_j (0)  + 2 \displaystyle\sum_{k=\ell_{i_0}}^{m_{i_0}} \mathfrak{l}(\Lambda_j-\Lambda_k) \cdot \big( \chi ( \mathfrak{Q}_j (t) - \mathfrak{Q}_k (t) ) - \chi ( \mathfrak{Q}_j (0) - \mathfrak{Q}_k(0) ) \big).
 		 \end{flalign*}

 		\noindent Subtracting this from \eqref{lambdaq2} (with the $k$ there equal to $\varphi_0^{-1} (j)$ here); using \Cref{qjqk2} with the facts that $\supp \chi \subseteq [-\mathfrak{M}, \mathfrak{M}]$ and that $(\ell_{i_0}, m_{i_0}) = (\ell_i' - \lceil T^3 / 2 \rceil, m_i' + \lceil T^3/ 2 \rceil)$ to restrict the sum there to over $i = \varphi_0^{-1} (k)$ with $k \in \llbracket \ell_{i_0}, m_{i_0} \rrbracket$ (after also using \eqref{qiqjs5}, with our restriction to $\mathsf{E}_2$, to restrict to $k \in \llbracket k_1, k_2 \rrbracket$); and using the fact that $\mathfrak{Q}_k (0) = Q_{\varphi_0^{-1} (k)} (0)$ then yields
 		\begin{flalign*}
 			\Bigg| w_j (t&) \cdot ( 1 + T^3 \cdot (\mathbbm{1}_{j \le \ell_{i_0} +T^2} + \mathbbm{1}_{j \ge m_{i_0} - T^2}) ) \\
 			 & + 2 \displaystyle\sum_{k=\ell_{i_0}}^{m_{i_0}} \mathfrak{l} (\Lambda_j-\Lambda_k) \cdot \big( \chi ( \mathfrak{Q}_j (t) - \mathfrak{Q}_k (t) ) - \chi \big( Q_{\varphi_0^{-1}(j)} (t) - Q_{\varphi_0^{-1} (k)} (t) \big) \big) \Bigg|  \le B\mathfrak{M}^{1/2} (\log N)^{16},
 		\end{flalign*}
		
		\noindent where we also used the fact that $\mathbbm{1}_{j \le \ell_{i_0} + T^2} + \mathbbm{1}_{j \ge m_{i_0} - T^2} = 0$, since $j \in \llbracket \ell_i', m_i' \rrbracket$. Taylor expanding $\chi'$, it follows that
	\begin{flalign}
		\label{wj2} 
		\begin{aligned} 
		\Bigg| w_j & (t) \cdot ( 1 + T^3 \cdot (\mathbbm{1}_{j \le \ell_{i_0}+T^2} + \mathbbm{1}_{j \ge m_{i_0} -T^2}) ) \\
		&  + 2 \displaystyle\sum_{k=\ell_{i_0}}^{m_{i_0}}  \mathfrak{l}(\Lambda_j - \Lambda_k) \cdot ( w_j (t) - w_k (t) ) \cdot \chi' ( \mathfrak{Q}_j (t) - \mathfrak{Q}_k (t) ) \Bigg| \\ 
		&  \le 10 B \mathfrak{M}^{-2} (\log N)^2 \displaystyle\sum_{k=\ell_{i_0}}^{m_{i_0}} w_k (t)^2 \cdot \big( \mathbbm{1}_{|\mathfrak{Q}_j(t)-\mathfrak{Q}_k(t)|\le\mathfrak{M}} + \mathbbm{1}_{|Q_{\varphi_0^{-1}(j)}(t)-Q_{\varphi_0^{-1}(k)} (t)| \le \mathfrak{M}} \big) \\
		& \qquad + B \mathfrak{M}^{1/2} (\log N)^{16} \\
		&  \le B \mathfrak{M}^{1/2} (\log N)^{16} + 10 B \mathfrak{M}^{-2} (\log N)^2 \cdot (B \mathfrak{M}^{1/2} (\log N)^{20})^2 \cdot 6 \mathfrak{M} (\log N)^5  \le B \mathfrak{M}^{1/2} (\log N)^{18}.
		\end{aligned} 
	\end{flalign}
	
	\noindent Here, in the first bound we used \Cref{mchi} and the fact that $\mathfrak{l} (\Lambda_j-\Lambda_k) \le 5 (\log N)^2$ (by our restriction to $\mathsf{E}_1$); in the second we used \eqref{qjderivativen10} and the fact that there are at most $6 \mathfrak{M} (\log N)^5$ indices $k \in \llbracket \ell_{i_0}, m_{i_0} \rrbracket$ for which the summand in \eqref{wj2} is nonzero (by \Cref{qjqk2} and the fact that $\supp \chi' \subseteq [-\mathfrak{M}, \mathfrak{M}]$); and in the third we used \eqref{testimateb}.
	
	The bound \eqref{wj2} applies if $j \in \llbracket \ell_i', m_i' \rrbracket$. If $j \in \llbracket \ell_{i_0}, m_{i_0} \rrbracket \setminus \llbracket \ell_i', m_i' \rrbracket$, then we have 
	\begin{flalign}
		\label{wj3}
		\begin{aligned} 
		\Bigg| w_j & (t) \cdot ( 1 + T^3 \cdot (\mathbbm{1}_{j \le \ell_{i_0}+T^2} + \mathbbm{1}_{j \ge m_{i_0} -T^2}) ) \\
		&  + 2 \displaystyle\sum_{k=\ell_{i_0}}^{m_{i_0}}  \mathfrak{l}(\Lambda_j - \Lambda_k) \cdot ( w_j (t) - w_k (t) ) \cdot \chi' ( \mathfrak{Q}_j (t) - \mathfrak{Q}_k (t)) \Bigg| \\
		& \qquad \le \big(T^3 + 1 + 4 \cdot 5 (\log N)^5 \cdot B \mathfrak{M}^{-1} \cdot 3 \mathfrak{M} (\log N)^5 \big) \cdot \displaystyle\max_{j \in \llbracket \ell_{i_0}, m_{i_0} \rrbracket} |w_j (t)| \le T^4.
		\end{aligned}
	\end{flalign}
	
	\noindent Here, in the first bound we used the fact that $|\mathfrak{l} (\Lambda_j - \Lambda_k)| \le 5 (\log N)^2$ (by our restriction to $\mathsf{E}_1$), that $|\chi'(q)| \le B \mathfrak{M}^{-1}$; and that there are at most $3 \mathfrak{M} (\log N)^5$ indices $k$ for which $\chi' (\mathfrak{Q}_j (t) - \mathfrak{Q}_k (t)) \ne 0$ by \Cref{qjqk2} and the fact that $\supp \chi' \subseteq [-\mathfrak{M}, \mathfrak{M}]$. In the second, we used \eqref{qjderivativen10} with \eqref{testimateb}. 
	
	 Denoting $\bm{w} = (w_{\ell_{i_0}}(t), w_{\ell_{i_0}+1}(t), \ldots , w_{m_{i_0}}(t)) \in \mathbb{R}^{m_{i_0}-\ell_{i_0}+1}$, we find from \eqref{wj2}, \eqref{wj3}, and \eqref{sij2} that 
	\begin{flalign*} 
		\bm{S}_{\bm{\Lambda};\bm{\mathfrak{Q}}(t)}^{\llbracket \ell_{i_0}, m_{i_0} \rrbracket} \bm{w} = \bm{u}, \qquad \text{where} \qquad |u_k| \le B \mathfrak{M}^{1/2} (\log N)^{18} + T^4 \cdot (\mathbbm{1}_{j \le \ell_i - \lfloor T^3/2 \rfloor} + \mathbbm{1}_{j \ge m_i + \lfloor T^3/2 \rfloor}),
	\end{flalign*} 	
		
	\noindent and we have denoted $\bm{u} = (u_{\ell_{i_0}}, u_{\ell_{i_0}+1}, \ldots , u_{m_{i_0}})$. Recalling from \eqref{sjjs2} that $\bm{S}_{\bm{\Lambda};\bm{\mathfrak{Q}}(t)}^{\llbracket \ell_i, m_i \rrbracket}$ satisfies \eqref{sumchi} with the $\varepsilon$ there equal to $(2 \log N)^{-1}$ here, it follows from \eqref{wivi}  (with the $U$ there equal to $\mathfrak{M}^{-1} T^2  \ge T$ here) that 
	\begin{flalign}
		\label{wkt} 
		\begin{aligned} 
		\displaystyle\max_{k \in \llbracket \ell_i, m_i \rrbracket} |w_k(t)| & \le 2 \log N \cdot \displaystyle\max_{k : |\mathfrak{Q}_j (t) - \mathfrak{Q}_k (t)| \le \lfloor T^3 / 4 \rfloor} |u_k| + 2 \log N \cdot e^{-T / 16 \log N} \cdot T^4 \\
		& \le 2B \mathfrak{M}^{1/2} (\log N)^{19} + N^{-1} \le 3B \mathfrak{M}^{1/2} (\log N)^{19},
		\end{aligned} 
	\end{flalign}
		
	\noindent where we also used the fact that $|j-k| \le \lfloor T^3 / 4 \rfloor$ whenever $|\mathfrak{Q}_j (t) - \mathfrak{Q}_k (t)| \le T^2$ (by reasoning entirely analogous to that used in the proof of \Cref{qjqk2}). This confirms \eqref{qjt3} for all $t \in [0, \mathfrak{T}_{i_0}]$ and thus (as $i_0 = i-1$) for all $t \in [0, \mathfrak{T}_i - 1]$, thereby establishing the proposition.
	\end{proof}

	\section{Analysis of the Proxy Dynamics} 
	
	\label{EquationAnalyze}
	
	In this section we prove \Cref{vestimate}, by analyzing the proxy dynamics $(\mathfrak{Q}_j)$ from \Cref{qjt2}. We begin in \Cref{EquationQ0} by showing as \Cref{lambdajq} that the effective velocities $v_{\eff} (\Lambda_j)$ approximately satisfy the relation \eqref{derivativeqjs2} defining the $(\mathfrak{Q}_j')$. We then use this to approximate the derivatives of $\mathfrak{Q}_j$ by the $v_{\eff} (\Lambda_j)$ in \Cref{EquationQ}, through \Cref{qv}. This quickly yields \Cref{vestimate} under more restrictive hypotheses (\Cref{q}), which we remove in \Cref{ProofV}.

	\subsection{Approximate Relation for Effective Velocities} 
	
	\label{EquationQ0} 
	
	In this section we establish the following lemma, which indicates that setting $\mathfrak{Q}_j'(s) \approx v_{\eff} (\Lambda_j)$ approximately yields a solution of \eqref{derivativeqjs2} (though with the $\mathfrak{Q}_j$ there replaced by $Q_{\varphi_0^{-1} (j)}$ here). Throughout this section, we adopt \Cref{chi2}, recall \Cref{qjt2}, let $\theta_0 = \theta_0 (\beta) > 0$ denote the constant from \Cref{beta0theta}, and assume that $\theta \in (0, \theta_0)$.
	
	\begin{lem} 
		
	\label{lambdajq} 
	
	The following holds with overwhelming probability. Let $s \in [0, T]$ be a real number, and let $j, \ell, m \in \llbracket N_1, N_2 \rrbracket$ be indices satisfying 
	\begin{flalign}
		\label{jinterval} 
		j, \ell, m \in \llbracket N_1 + T (\log N)^6, N_2 - T(\log N)^6 \rrbracket; \qquad j \in \llbracket \ell + T (\log N)^5, m - T (\log N)^5 \rrbracket.
	\end{flalign} 
	
	\noindent Then, we have that 
	\begin{flalign}
		\label{vlambda0} 
		\begin{aligned} 
		\Bigg| & \Lambda_j -  v_{\eff} (\Lambda_j) \cdot \bigg( 2 \displaystyle\sum_{k=\ell}^{m} \mathfrak{l} (\Lambda_j - \Lambda_k) \cdot \chi' \big( Q_{\varphi_0^{-1} (j)} (s) - Q_{\varphi_0^{-1}(k)} (s) \big) + 1 \bigg) \\
		& \qquad \qquad   + 2 \displaystyle\sum_{k=\ell}^m v_{\eff} (\Lambda_k) \cdot \mathfrak{l} (\Lambda_j - \Lambda_k) \cdot \chi' \big( Q_{\varphi_0^{-1} (j)} (s) - Q_{\varphi_0^{-1}(k)} (s) \big) \Bigg|  \le B \mathfrak{M}^{-1/2} (\log N)^{21}.
		\end{aligned} 
	\end{flalign} 
	
	\end{lem} 
	
	We first show the below variant of \Cref{lambdajq} addressing a fixed time $s \in [0,T]$.

	\begin{lem} 
		
		\label{lambdajq2} 
		
		There exists a constant $c>0$ such that the below holds. Fix a real number $s \in [0, T]$. The following holds with probability at least $1 - c^{-1} e^{-c(\log N)^2}$. For any indices $j, \ell, m \in \llbracket N_1, N_2 \rrbracket$ satisfying \eqref{jinterval}, we have 
		\begin{flalign}
			\label{vlambda02} 
			\begin{aligned} 
				\Bigg| & \Lambda_j -  v_{\eff} (\Lambda_j) \cdot \bigg( 2 \displaystyle\sum_{k=\ell}^{m} \mathfrak{l} (\Lambda_j - \Lambda_k) \cdot \chi' \big( Q_{\varphi_0^{-1} (j)} (s) - Q_{\varphi_0^{-1}(k)} (s) \big) + 1 \bigg) \\
				& \qquad \quad   + 2 \displaystyle\sum_{k=\ell}^m v_{\eff} (\Lambda_k) \cdot \mathfrak{l} (\Lambda_j - \Lambda_k) \cdot \chi' \big( Q_{\varphi_0^{-1} (j)} (s) - Q_{\varphi_0^{-1}(k)} (s) \big) \Bigg| \le \displaystyle\frac{B}{2} \cdot \mathfrak{M}^{-1/2} (\log N)^{21}.
			\end{aligned} 
		\end{flalign} 
		
	\end{lem}

	\begin{proof}
		
		Throughout this proof, recalling \Cref{adelta}, we restrict to the event $\mathsf{E}_1  = \mathsf{BND}_{\bm{L}(0)} (\log N)$, as we may by \Cref{l0eigenvalues}. We further restrict to the event $\mathsf{E}_2$ on which \Cref{qijsalpha} holds, and to the event $\mathsf{E}_3$ on which \Cref{centert} holds.
		
		We will further restrict to a fourth event $\mathsf{E}_4$, which will be obtained by applying \Cref{concentrationh3} with the $(f, G)$ there equal to $(\mathfrak{l}, \chi')$ here, and the $F$ there equal to either $v_{\eff}$ or $1$ here. Then, we may take the parameters $(S, B)$ in \Cref{fgab} to be $(\mathfrak{M}, B)$ here (by \Cref{mchi}), and the parameter $A$ there to be $(\log N)^2$ here (by \Cref{derivativev}). Moreover, we may take the parameter $D$ in \eqref{f1d} to be $(\log N)^3$, by the definition \eqref{functionl} of $\mathfrak{l}$. Therefore, since \Cref{mchi} implies that $\int_{-\infty}^{\infty} \chi'(\alpha q) dq = \alpha^{-1}$, \Cref{concentrationh3} yields an overwhelmingly probable event $\mathsf{E}_4$, on which the following holds. Whenever $k \in \llbracket 1, N \rrbracket$ satisfies $N_1 + T(\log N)^5 \le \varphi_s (k) \le N_2 - T (\log N)^5$, we have
		\begin{flalign}
			\label{2q0} 
			\Bigg| 2 \displaystyle\sum_{i=1}^N \mathfrak{l} (\lambda_k - \lambda_i) \cdot \chi' \big( Q_k (s) - Q_i (s) \big) - 2 \alpha^{-1} \displaystyle\int_{-\infty}^{\infty} \mathfrak{l} (\lambda_k - \lambda) \varrho (\lambda) d \lambda \Bigg| \le 	B \mathfrak{M}^{1/2} (\log N)^{18},
		\end{flalign} 
		
		\noindent and 
		\begin{flalign}
			\label{3q0} 
			\begin{aligned} 
			\Bigg| 2 \displaystyle\sum_{i=1}^N v_{\eff} (\lambda_i) \cdot \mathfrak{l} (\lambda_k - \lambda_i) \cdot \chi' \big( Q_k (s) - Q_i (s) \big) - 2 \alpha^{-1} \displaystyle\int_{-\infty}^{\infty} v_{\eff} (\lambda) \mathfrak{l} (& \lambda_k - \lambda) \varrho (\lambda) d \lambda \Bigg| \\
			& \le 	B \mathfrak{M}^{1/2} (\log N)^{18}.
			\end{aligned} 
		\end{flalign}
		
		\noindent We additionally further restrict $\mathsf{E}_4$ in the below.
		
		Let us first approximate the integrals appearing in \eqref{2q0} and \eqref{3q0}. To that end, observe that there exist constants $c_2>0$ and $C_1>1$ such that, for sufficiently large $N$,
		\begin{flalign}
			\label{integral1} 
			\begin{aligned} 
			\Bigg| \displaystyle\int_{-\infty}^{\infty} \mathfrak{l}(x-y) v_{\eff} (y & ) \varrho(y) dy - \displaystyle\int_{-\infty}^{\infty} \log |x-y| v_{\eff} (y) \varrho (y) dy \Bigg| \\
			& \le C_1 \displaystyle\int_{-\infty}^{\infty} \big|\mathfrak{l}(x-y) - \log |x-y| \big| \cdot  (|y|+1) e^{-c_2 y^2} dy \le e^{-(\log N)^2}.
			\end{aligned} 
		\end{flalign}
		
		\noindent Here, the first inequality follows from the fact that $|v_{\eff} (x)| \le C_2 (|x|+1)$ for some $C_2 > 1$ (by \Cref{derivativev}), with the exponential decay of $\varrho$ (from \Cref{rhoexponential}), and the second follows from the fact that 
		\begin{flalign*} 
			\big| \mathfrak{l}(x) - \log |x| \big| \le 2 |\log x|  \cdot \mathbbm{1}_{|x| \le e^{-2(\log N)^2}} + e^{-2(\log N)^2},
		\end{flalign*} 
		
		\noindent which is a quick consequence of the definition \eqref{functionl} of $\mathfrak{l}$. By entirely analogous reasoning, we also find for sufficiently large $N$ that  
		\begin{flalign}
			\label{integral2} 
			\Bigg| \displaystyle\int_{-\infty}^{\infty} \mathfrak{l} (x-y) \varrho (y) dy - \displaystyle\int_{-\infty}^{\infty} \log |x-y| \varrho(y) dy \Bigg| \le e^{-(\log N)^2}.
		\end{flalign} 
		
		We next	insert \eqref{integral1} and \eqref{integral2} into \eqref{2q0} and \eqref{3q0}, respectively, and additionally change variables from $(i, k)$ to $(\varphi_0^{-1} (k), \varphi_0^{-1} (j))$. This yields  
			\begin{flalign}
			\label{4q0} 
			\begin{aligned} 
			\Bigg| 2 \displaystyle\sum_{k=N_1}^{N_2} \mathfrak{l} (\Lambda_j - \Lambda_k) \cdot \chi' \big( Q_{\varphi_0^{-1}(j)} (s) - Q_{\varphi_0^{-1}(k)} (s) \big) - 2\alpha^{-1}  \displaystyle\int_{-\infty}^{\infty} &  \log |\Lambda_j - \lambda| \varrho (\lambda) d \lambda \Bigg| \\
			& \le 2 B \mathfrak{M}^{1/2} (\log N)^{18},
			\end{aligned} 
		\end{flalign} 
		
		\noindent and 
		\begin{flalign}
			\label{5q0} 
			\begin{aligned} 
				\Bigg| 2 \displaystyle\sum_{k=N_1}^{N_2} v_{\eff} (\Lambda_k) \cdot \mathfrak{l} (& \Lambda_j - \Lambda_k) \cdot \chi' \big( Q_{\varphi_0^{-1} (j)} (s) - Q_{\varphi_0^{-1} (k)} (s) \big) \\
				& \qquad - 2\alpha^{-1} \displaystyle\int_{-\infty}^{\infty} \log | \Lambda_j - \lambda| v_{\eff} (\lambda)  \varrho (\lambda ) d \lambda \Bigg|  \le 2B \mathfrak{M}^{1/2} (\log N)^{18},
			\end{aligned} 
		\end{flalign}
		
		\noindent where we used the fact that $N_1 + T(\log N)^5 \le \varphi_s (\varphi_0^{-1} (j)) \le N_2 - T(\log N)^5$ (by \eqref{jinterval}, since \Cref{centert} with our restriction to $\mathsf{E}_3$ implies that have $|\varphi_s (\varphi_0^{-1} (j)) - j| \le 2T (\log N)^2$). Combining these bounds, and using the fact from \Cref{derivativev} (with our restriction to $\mathsf{E}_1$) that $|v_{\eff} (\lambda)| \le (\log N)^2$ for sufficiently large $N$ and any $\lambda \in \eig \bm{L}$, yields  
		\begin{flalign}
			\label{lambdajv} 
			\begin{aligned} 
			\Bigg| & \Lambda_j -  v_{\eff} (\Lambda_j) \cdot \bigg( 2 \displaystyle\sum_{k=N_1}^{N_2} \mathfrak{l} (\Lambda_j - \Lambda_k) \cdot \chi' \big( Q_{\varphi_0^{-1} (j)} (s) - Q_{\varphi_0^{-1}(k)} (s) \big) + 1 \bigg) \\
			& \qquad \qquad \qquad \qquad + 2 \displaystyle\sum_{k=N_1}^{N_2} v_{\eff} (\Lambda_k) \cdot \mathfrak{l} (\Lambda_j - \Lambda_k) \cdot \chi' \big( Q_{\varphi_0^{-1} (j)} (s) - Q_{\varphi_0^{-1}(k)} (s) \big) \Bigg| \\
			& \le \Bigg| \Lambda_j -  v_{\eff} (\Lambda_j) \cdot \bigg( 2\alpha^{-1}  \displaystyle\int_{-\infty}^{\infty} \log |\Lambda_j - \lambda| \varrho (\lambda) d \lambda + 1 \bigg) + 2 \alpha^{-1} \displaystyle\int_{-\infty}^{\infty} \log |\Lambda_j - \lambda| v_{\eff} (\lambda) \varrho (\lambda) d \lambda \Bigg| \\
			&  \qquad + B \mathfrak{M}^{-1/2} (\log N)^{21} = B \mathfrak{M}^{-1/2} (\log N)^{21}.
			\end{aligned} 
		\end{flalign}
		
		\noindent Here, in the last equality we used the fact that 
		\begin{flalign*}
			v_{\eff} & (\Lambda_j) \cdot \bigg( 2\alpha^{-1}  \displaystyle\int_{-\infty}^{\infty} \log |\Lambda_j - \lambda| \varrho (\lambda) d \lambda + 1 \bigg) - 2 \alpha^{-1} \displaystyle\int_{-\infty}^{\infty} \log |\Lambda_j - \lambda| v_{\eff} (\lambda) \varrho (\lambda) d \lambda \\
			&  = v_{\eff} (\Lambda_j) \cdot \big(\alpha^{-1} \cdot \bm{\mathrm{T}} \varrho (\Lambda_j) + 1 \big) - \alpha^{-1} \cdot \bm{\mathrm{T} \varrho} v_{\eff} (\Lambda_j) = (\theta^{-1} \cdot \bm{\varsigma_0^{\dr}} - \alpha^{-1} \cdot \bm{\mathrm{T} \varrho}) v_{\eff} (\Lambda_j) = \Lambda_j,
		\end{flalign*}
		
		\noindent where the first statement holds by the definition \eqref{operatort} of $\bm{\mathrm{T}}$; the second holds by the second statement in \eqref{rho2}; and the third holds by \Cref{vt}.
		
		It therefore remains to show that the sums over $k$ on the left side of \eqref{lambdajv} can be restricted to the interval $\llbracket \ell, m \rrbracket$. To that end, it suffices to verify that $\chi' ( Q_{\varphi_0^{-1} (j)} (s) - Q_{\varphi_0^{-1} (k)} (s) ) = 0$ if $j \in \llbracket \ell + T(\log N)^5, m - T(\log N)^5 \rrbracket$ and $k \notin \llbracket \ell, m \rrbracket$, and hence (as $\supp \chi' \subseteq [-\mathfrak{M}, \mathfrak{M}]$) to confirm for $|j-k| \ge T(\log N)^5$ that $|Q_{\varphi_0^{-1} (j)} (s) - Q_{\varphi_0^{-1} (k)} (s)| > \mathfrak{M}$. This follows very similarly to as in the proof of \eqref{qjqk4}, as a quick consequence of \Cref{qijsalpha} and \Cref{centert} (with our restriction to $\mathsf{E}_2 \cap \mathsf{E}_3$); further details are therefore omitted.
	\end{proof}
	
	Now we can establish \Cref{lambdajq}. 
	
	\begin{proof}[Proof of \Cref{lambdajq}]
		
		Throughout this proof, recalling \Cref{adelta}, we restrict to the event $\mathsf{E}_1 = \mathsf{BND}_{\bm{L}} (\log N)$, which we may by \Cref{l0eigenvalues}. We further restrict to the event $\mathsf{E}_2$ on which \Cref{centert} holds and to the event $\mathsf{E}_3$ on which \Cref{centerdistance} holds. Additionally, we restrict to the event $\mathsf{E}_4$ on which \eqref{qjqk3} holds for any $k \in \llbracket \ell, m \rrbracket$, as we may by \Cref{qjqk2}.\footnote{While \Cref{qjqk2} was claimed for indices $j, k \in \llbracket k_1, k_2 \rrbracket$, it is quickly verified from the deriviation in \eqref{qjqk4} that the first bound in \eqref{qjqk3} holds for $j$ satisfying \eqref{jinterval} and $k \in \llbracket \ell, m \rrbracket$.} For any $r \in [0, T]$, also let $\mathsf{E}_5 (r)$ denote the event on which \eqref{vlambda02} holds. Let $\mathcal{T} \subset [0, T]$ be an $N^{-1}$-mesh of $[0, T]$, and set $\mathsf{E}_5 = \bigcap_{r \in \mathcal{T}} \mathsf{E}_5 (r)$; we may also restrict to $\mathsf{E}_5$ by \Cref{lambdajq2}, with a union bound. 
		
		Now fix $s \in [0, T]$, and let $s_0 \in \mathcal{S}$ be such that $|s-s_0| \le N^{-1}$. Since (by our restriction to $\mathsf{E}_5$) \eqref{vlambda02} holds with the $s$ there equal to $s_0$ here, it suffices to show that 
		\begin{flalign*}
			& \displaystyle\max_{k \in \llbracket \ell, m \rrbracket} |v_{\eff} (\Lambda_k)| \cdot \displaystyle\max_{k \in \llbracket \ell, m \rrbracket} |\mathfrak{l} (\Lambda_j - \Lambda_k)| \\ 
			& \qquad  \times \displaystyle\sum_{k=\ell}^m \big| \chi' (Q_{\varphi_0^{-1} (j)} (s) - Q_{\varphi_0^{-1} (k)} (s) \big) - \chi' \big( Q_{\varphi_0^{-1} (j)} (s_0) - Q_{\varphi_0^{-1} (k)} (s_0) \big) \big|  \le B \mathfrak{M}^{-1/2} (\log N)^{20}.
		\end{flalign*}
		
		\noindent Due to our restriction to $\mathsf{E}_1$, we have for any $k \in \llbracket N_1, N_2 \rrbracket$ that $|\mathfrak{l} (\Lambda_j - \Lambda_k)| \le 5 (\log N)^2$ and $|v_{\eff} (\Lambda_k)| \le (\log N)^2$ (the latter by \Cref{derivativev}). Hence, it remains to show that 
		\begin{flalign}
			\label{sumq5} 
			\displaystyle\sum_{k=\ell}^m \big| \chi' (Q_{\varphi_0^{-1} (j)} (s) - Q_{\varphi_0^{-1} (k)} (s) \big) - \chi' \big( & Q_{\varphi_0^{-1} (j)} (s_0) - Q_{\varphi_0^{-1} (k)} (s_0) \big) \big|  \le B \mathfrak{M}^{-1/2} (\log N)^{15}.
		\end{flalign}
		
		Next, by \Cref{centert} (with our restriction to $\mathsf{E}_2$), we have for any $k \in \llbracket \ell, m \rrbracket$ that $|\varphi_s (\varphi_0^{-1} (k)) - k| \le 2T (\log N)^2$, and so $\varphi_s (\varphi_0^{-1} (k)) \in \llbracket \ell + 2T(\log N)^2, m - 2T (\log N)^2 \rrbracket \subseteq \llbracket N_1 + T(\log N)^5, N_2 - T(\log N)^5 \rrbracket$, by \eqref{jinterval}. Therefore, by \Cref{centerdistance} (with our restriction to $\mathsf{E}_2$) with the fact that $|s-s_0| \le N^{-1} \le 1/2$, we have for any $k \in \llbracket \ell, m \rrbracket$ that
		\begin{flalign*}
			\big| \big( Q_{\varphi_0^{-1} (j)} (s) - Q_{\varphi_0^{-1} (k)} (s) \big) - \big( Q_{\varphi_0^{-1} (j)} (s_0) - Q_{\varphi_0^{-1} (k)} (s_0) \big) \big| \le 3 (\log N)^4.
		\end{flalign*}
		
		\noindent Thus, by Taylor expanding $\chi'$ (and using \Cref{mchi}), we deduce
		\begin{flalign}
			\label{sumq4} 
			\begin{aligned} 
				\displaystyle\sum_{k=\ell}^m \big| \chi' \big( & Q_{\varphi_0^{-1} (j)} (s) - Q_{\varphi_0^{-1} (k)} (s) \big) - \chi' \big(  Q_{\varphi_0^{-1} (j)} (s_0) - Q_{\varphi_0^{-1} (k)} (s_0) \big) \big| \\ 
				& \le  3 B \mathfrak{M}^{-2} (\log N)^4 \displaystyle\sum_{k=\ell}^m \big( \mathbbm{1}_{|Q_{\varphi_0^{-1} (j)} (s) - Q_{\varphi_0^{-1} (k)} (s)| \le \mathfrak{M}} + \mathbbm{1}_{|Q_{\varphi_0^{-1} (j)} (s_0) - Q_{\varphi_0^{-1} (k)} (s_0)| \le \mathfrak{M}} \big). 
			\end{aligned} 
		\end{flalign}
		
		\noindent By the first inequality in \eqref{qjqk3} (with our restriction to $\mathsf{E}_4$), there are at most $3\mathfrak{M} (\log N)^5$ indices $k$ for which the summand on the right side of \eqref{sumq4} is nonzero, in which case it is at most equal to $2$. It follows that 
		\begin{flalign*} 
			\displaystyle\sum_{k=\ell}^m \big| \chi' \big( Q_{\varphi_0^{-1} (j)} (s) - Q_{\varphi_0^{-1} (k)} (s) \big) - \chi' \big(  Q_{\varphi_0^{-1} (j)} (s_0) - Q_{\varphi_0^{-1} (k)} (s_0) \big) \big| & \le 18B \mathfrak{M}^{-1} (\log N)^9 \\
			& \le B \mathfrak{M}^{-1/2} (\log N)^{15}
		\end{flalign*}
		
		\noindent where in the last inequality we used \eqref{testimateb}. This confirms \eqref{sumq5} and thus the lemma.	
	\end{proof} 
	
	\subsection{Derivative Approximation for $\mathfrak{Q}_j$}
	
	\label{EquationQ}
	
	In this section we use \Cref{lambdajq} to prove the following proposition indicating that the derivative $\mathfrak{Q}_j'$ of the proxy dynamics is close to $v_{\eff} (\Lambda_j)$, with high probability. 
	
	\begin{prop} 
		\label{qv} 
		
		Adopt \Cref{chi2}, recall \Cref{qjt2}, let $\theta_0 = \theta_0 (\beta) > 0$ denote the constant from \Cref{beta0theta}, and assume that $\theta \in (0, \theta_0)$. The following holds with overwhelming probability. For each $j \in \llbracket k_1 + T^5 , k_2 - T^5  \rrbracket$ and $s \in [0, T]$, we have 
		\begin{flalign*} 
			| \mathfrak{Q}_j' (s) - v_{\eff} (\Lambda_j) | \le B^2 \mathfrak{M}^{1/2} (\log N)^{32}.
		\end{flalign*} 
		
	\end{prop}

	\begin{proof}[Proof of \Cref{qv}]
		
		Throughout this proof, we assume for notational convenience that $T$ is an integer. Recalling \Cref{adelta}, we restrict to the event $\mathsf{E}_1 = \mathsf{BND}_{\bm{L}} (\log N)$ throughout, as we may by \Cref{l0eigenvalues}. We further restrict to the events $\mathsf{E}_2$ on which \Cref{qijsalpha} holds; $\mathsf{E}_3$ on which \Cref{centert} holds; $\mathsf{E}_4$ on which \Cref{q2q}, \Cref{sr}, and \Cref{tt} all hold; $\mathsf{E}_5$ on which \Cref{qjqk2} holds; and $\mathsf{E}_6$ on which \Cref{lambdajq} holds.
		
		First observe, since $\supp \chi' \subseteq [-\mathfrak{M}, \mathfrak{M}]$ (by \Cref{mchi}), \Cref{qjqk2} (and our restriction to $\mathsf{E}_5$) implies that $\chi' (\mathfrak{Q}_j (s) - \mathfrak{Q}_k (s)) = 0$ if $|j-k| \ge T^2$. Therefore, this in particular holds if $j \in \llbracket \ell_T + T^2, m_T - T^2 \rrbracket$ and $k \notin \llbracket \ell_T, m_T \rrbracket$ (where we recall the $(\ell_i, m_i)$ from \Cref{qjt2}). Moreover, by our restriction to $\mathsf{E}_4$, \eqref{derivativeqjs2} holds for any $i \in \llbracket 0, T \rrbracket$, for all $j \in \llbracket \ell_T, m_T \rrbracket$ and $s \in [i, i+1]$. Together, these two facts imply for any $j \in \llbracket \ell_T + T^2, m_T - T^2 \rrbracket$ and $s \in [0, T]$ that
		\begin{flalign}
			\label{lambdaqchi} 
			\begin{aligned} 
				\Lambda_j & = \mathfrak{Q}_j' (s) \cdot \Bigg( 2 \displaystyle\sum_{k=\ell_T}^{m_T} \mathfrak{l} (\Lambda_j - \Lambda_k) \cdot \chi' \big( \mathfrak{Q}_j (s) - \mathfrak{Q}_k (s) \big) + 1 \Bigg) \\
				& \qquad \qquad \qquad \qquad - 2 \displaystyle\sum_{k=\ell_T}^{m_T} \mathfrak{Q}_k' (s) \cdot \mathfrak{l} (\Lambda_j - \Lambda_k) \cdot \chi' \big( \mathfrak{Q}_j (s) - \mathfrak{Q}_k(s) \big).
			\end{aligned} 
		\end{flalign}
		
		By \Cref{lambdajq} (and our restriction to $\mathsf{E}_6$), $v_{\eff} (\Lambda_j)$ satisfies a similar equation, namely, we have for any $j \in \llbracket \ell_T, m_T \rrbracket$ (where we observe that $\llbracket \ell_T, m_T \rrbracket \subseteq \llbracket k_1, k_2 \rrbracket \subseteq \llbracket N_1 + 2T(\log N)^6, N_2 - 2T(\log N)^6 \rrbracket$ by \eqref{nk1k2} and \eqref{testimateb}) and $s \in [0, T]$ that
		\begin{flalign}
			\label{lambdavchi} 
			\begin{aligned} 
			\Bigg| & \Lambda_j - v_{\eff} (\Lambda_j) \cdot \Bigg( 2 \displaystyle\sum_{k=\ell_T}^{m_T} \mathfrak{l} (\Lambda_j - \Lambda_k) \cdot \chi' \big( Q_{\varphi_0^{-1} (j)} (s) - Q_{\varphi_0^{-1} (k)} (s) \big) + 1 \Bigg) \\
			& \qquad \quad + 2 \displaystyle\sum_{k=\ell_T}^{m_T} v_{\eff} (\Lambda_k) \cdot \mathfrak{l} (\Lambda_j - \Lambda_k) \cdot \chi' \big( Q_{\varphi_0^{-1} (j)} (s) - Q_{\varphi_0^{-1} (k)} (s) \big) \Bigg| \le B \mathfrak{M}^{-1/2} (\log N)^{21}.
			\end{aligned} 
		\end{flalign}
		
		\noindent A distinction between this bound and \eqref{lambdaqchi} is that the latter replaces $Q_{\varphi_0^{-1} (i)} (s)$ with $\mathfrak{Q}_i (s)$, so let us bound the difference between the associated terms. By the second property in \Cref{q2q} (with our restriction to $\mathsf{E}_4$), we have that $|\mathfrak{Q}_k(s) - Q_{\varphi_0^{-1}(k)}(s)| \le B \mathfrak{M}^{1/2} (\log N)^{20}$ for each $k \in \llbracket \ell_T, m_T \rrbracket$ and $s \in [0, T]$. So it follows for any such $k$, and $j \in \llbracket \ell_T + T^2, m_T - T^2 \rrbracket$, that  
		\begin{flalign}
			\label{chiq2} 
			\begin{aligned}
				\big| \chi' \big( \mathfrak{Q}_j (s) - \mathfrak{Q}_k (s &) \big) - \chi' \big( Q_{\varphi_0^{-1}(j)} (s) - Q_{\varphi_0^{-1}(k)} (s) \big) \big| \\
				& \le B^2 \mathfrak{M}^{-3/2} (\log N)^{20} \cdot \big( \mathbbm{1}_{|\mathfrak{Q}_j (s) - \mathfrak{Q}_k(s)| \le \mathfrak{M}} +  \mathbbm{1}_{|Q_{\varphi_0^{-1}(j)} (s) - Q_{\varphi_0^{-1}(k)} (s)| \le \mathfrak{M}} \big),
			\end{aligned} 
		\end{flalign}
		
		\noindent where we have also used \Cref{mchi}. By \Cref{qjqk2} (and our restriction to $\mathsf{E}_5$), there are at most $3 \mathfrak{M} (\log N)^5$ indices $k \in \llbracket \ell_T, m_T \rrbracket$ for which at least one indicator function on the right side of \eqref{chiq2} is nonzero. Moreover, our restriction to $\mathsf{E}_1$ implies that $|v_{\eff} (\Lambda_k)| \le (\log N)^2$ and $|\mathfrak{l} (\Lambda_j-\Lambda_k)| \le 5 (\log N)^2$ for each $k \in \llbracket \ell_T, m_T \rrbracket$ (using \Cref{derivativev} for the former). Inserting these bounds, with \eqref{chiq2}, into \eqref{lambdavchi} yields 
		\begin{flalign}
			\label{lambdajqjderivative}
			\begin{aligned}
				\Bigg| & \Lambda_j -  v_{\eff} (\Lambda_j) \cdot \bigg( 2 \displaystyle\sum_{k=\ell_T}^{m_T} \mathfrak{l} (\Lambda_j - \Lambda_k) \cdot \chi' ( \mathfrak{Q}_j (s) - \mathfrak{Q}_k (s)) + 1 \bigg) \\
				& \qquad \qquad \qquad \qquad + 2 \displaystyle\sum_{k=\ell_T}^{m_T} v_{\eff} (\Lambda_k) \cdot \mathfrak{l} (\Lambda_j - \Lambda_k) \cdot \chi' (\mathfrak{Q}_j (s) - \mathfrak{Q}_k (s)) \Bigg| \\
				& \qquad \quad \le B \mathfrak{M}^{-1/2} (\log N)^{21} + 2 \cdot (\log N)^2 \cdot 5 (\log N)^2 \cdot 2B^2 \mathfrak{M}^{-3/2} (\log N)^{20} \cdot 3 \mathfrak{M} (\log N)^5 \\
				&  \qquad \quad \le 61 B^2 \mathfrak{M}^{-1/2} (\log N)^{29}.
			\end{aligned} 
		\end{flalign} 
		
		Now, fix $s \in [0,T]$; denote $\mathfrak{w}_j (r) = \mathfrak{Q}_j' (s) - v_{\eff} (\Lambda_j)$ for each $j \in \llbracket \ell_T, m_T \rrbracket$; and set $\bm{\mathfrak{w}} = (\mathfrak{w}_{\ell_T}, \mathfrak{w}_{\ell_T+1}, \ldots , \mathfrak{w}_{m_T})$. Denoting $\bm{\mathfrak{Q}} = (\mathfrak{Q}_{N_1} (s), \mathfrak{Q}_{N_1+1}(s), \ldots , \mathfrak{Q}_{N_2} (s))$ for each $s \in [0, T]$, abbreviate the matrix $\bm{S} = \bm{S}_{\bm{\Lambda}; \bm{\mathfrak{Q}}}^{\llbracket \ell_T, m_T \rrbracket}$ (from \Cref{chi2}). Let $\bm{S} \bm{\mathfrak{w}} = \bm{\mathfrak{u}}$, where $\bm{\mathfrak{u}} = (\mathfrak{u}_{\ell_T}, \mathfrak{u}_{\ell_T+1}, \ldots , \mathfrak{u}_{m_T})$. We claim for each $j \in \llbracket \ell_T, m_T \rrbracket$ that 
		\begin{flalign}
			\label{uj} 
			|\mathfrak{u}_j| \le B^2 \mathfrak{M}^{-1/2} (\log N)^{30} +  \mathbbm{1}_{j \notin \llbracket \ell_T + T^2, m_T - T^2 \rrbracket} \cdot N^6.
		\end{flalign}
		
		If $j \in \llbracket \ell_T + T^2, m_T - T^2 \rrbracket$, we deduce by subtracting \eqref{lambdajqjderivative} from \eqref{lambdaqchi} that
		\begin{flalign*}
			|\mathfrak{u}_j| & = \Bigg|  \mathfrak{w}_j \cdot \bigg( 2 \displaystyle\sum_{k=\ell_T}^{m_T} \mathfrak{l} (\Lambda_j - \Lambda_k) \cdot \chi' \big( Q_{\varphi_0^{-1} (j)} (s) - Q_{\varphi_0^{-1}(k)} (s) \big) + 1 \bigg) \\
			& \qquad \qquad - 2 \displaystyle\sum_{k=\ell_T}^{m_T} \mathfrak{w}_k \cdot \mathfrak{l} (\Lambda_j - \Lambda_k) \cdot \chi' \big( Q_{\varphi_0^{-1} (j)} (s) - Q_{\varphi_0^{-1}(k)} (s) \big) \Bigg| \le B^2 \mathfrak{M}^{-1/2} (\log N)^{30},
		\end{flalign*}
		
		\noindent where we have recalled that $\bm{\mathfrak{u}} = \bm{S} \bm{\mathfrak{w}}$ and the definition of $\bm{S}$ from \Cref{chi2}; this confirms \eqref{uj} these $j$. For the remaining $j \in \llbracket \ell_T, m_T \rrbracket$, we have for sufficiently large $N$ that 
		\begin{flalign*}
			|\mathfrak{u}_j| & \le |\mathfrak{w}_j| \cdot (T^3 + 1) +  2 \displaystyle\sum_{k=\ell_T}^{m_T} (|\mathfrak{w}_j| + |\mathfrak{w}_k|) \cdot |\mathfrak{l} (\Lambda_j - \Lambda_k)| \cdot | \chi' ( \mathfrak{Q}_j - Q_k  ) | \\
			& \le  4N^4 \cdot \displaystyle\sup_{q \in \mathbb{R}} |\chi'(q)| \cdot \displaystyle\max_{k \in \llbracket \ell_T, m_T \rrbracket} |\mathfrak{w}_k| \cdot \displaystyle\max_{k \in \llbracket \ell_T, m_T \rrbracket} |\mathfrak{l} (\Lambda_j-\Lambda_k)| \le N^5 \cdot \displaystyle\max_{k \in \llbracket \ell_T, m_T \rrbracket} |\mathfrak{w}_k| \le N^6,
		\end{flalign*}
		
		\noindent where in the first statement we recalled the definition \eqref{sij2} of $\bm{S}$; in the second we used the facts that $T \le N$ and $m_T - \ell_T + 1 \le N$; in the third we used the second statement in \Cref{mchi}, \eqref{testimateb}, and the fact that $|\mathfrak{l} (\Lambda_j-\Lambda_k)| \le 5 (\log N)^2$ (by our restriction to $\mathsf{E}_1$); and in the fourth we used the fact that $|\mathfrak{w}_k| \le |\mathfrak{Q}_j' (s)| + |v_{\eff} (\Lambda_j)| \le 4 (\log N)^3 \le N$ (as we have $|\mathfrak{Q}_j' (s)| \le (\log N)^3$, by \Cref{q2q} and our restriction to $\mathsf{E}_4$, and $|v_{\eff} (\Lambda_j)| \le (\log N)^2$, by \Cref{derivativev} and our restriction to $\mathsf{E}_1$). This establishes \eqref{uj} in general.
		
		By \Cref{sr} (with our restriction to $\mathsf{E}_4$), \eqref{sjjs2} holds, thereby enabling us to apply \Cref{sk1k2inverse}, with the $\varepsilon$ there equal to $(2 \log N)^{-1}$ here. The $U = \mathfrak{M}^{-1} T^2$ case of \eqref{wivi} then yields for $j \in \llbracket \ell_T + 2T^3, m_T - 2T^3 \rrbracket$ that
		\begin{flalign}
			\label{qj2} 
			\begin{aligned} 
				\big|\mathfrak{Q}_j'(s) - v_{\eff} (\Lambda_j) \big| = |\mathfrak{w}_j| & \le 2 \log N \cdot \displaystyle\max_{k: |\mathfrak{Q}_j - \mathfrak{Q}_k| \le T^2} |\mathfrak{u}_k| + 2 \log N \cdot e^{-T / 16 \log N} \cdot N^6 \\
				& \le 2 \log N \cdot \displaystyle\max_{k: |j - k| \le T^3} |\mathfrak{u}_k| + N^{-1} \le B^2 \mathfrak{M}^{-1/2} (\log N)^{32},
			\end{aligned} 
		\end{flalign}
		
		\noindent where we also used \eqref{uj} and the fact that $|\mathfrak{Q}_j - \mathfrak{Q}_k| \le T^3$ implies that $|j-k| \le T^2$, as quickly follows from \Cref{qijsalpha} and \Cref{centert}, with our restriction to $\mathsf{E}_2 \cap \mathsf{E}_3$ (entirely analogously to in the proof of \Cref{qjqk2}). This, together with the fact that $\llbracket \ell_T+ 2T^3, m_T - 2T^3 \rrbracket \subseteq \llbracket k_1 + T^5, k_2 - T^5 \rrbracket$ (as $(\ell_T, m_T) = (k_1 + T^4, k_2 - T^4)$), implies the proposition.
	\end{proof} 
	
	A quick consequence of \Cref{qv} is the following corollary; it indicates that \Cref{vestimate} holds under more restrictive hypotheses \eqref{jinterval2} than \eqref{n1n2zetat} and \eqref{j0}.

	\begin{cor}
		
		\label{q}
		
		Adopt \Cref{lbetaeta}; let $\theta_0 = \theta_0 (\beta) > 0$ denote the constant from \Cref{beta0theta}; and suppose that $\theta \in (0, \theta_0)$. The following holds with overwhelming probability. Let $j \in \llbracket 1, N \rrbracket$ be an index, and assume that 
		\begin{flalign} 
			\label{jinterval2} 
			10^8 \cdot (\log N)^{60} \le T \le N^{1/10}, \quad \text{and} \quad N_1 + 2 T^5 \le \varphi_0 (j) \le N_2 - 2 T^5,
		\end{flalign} 
		
		\noindent we have
		\begin{flalign*} 
			\displaystyle\sup_{t \in [0, T]} \big|Q_j (t) - Q_j (0) - t v_{\eff} (\lambda_j) \big| \le T^{1/2} (\log N)^{33}.
		\end{flalign*} 		 
	\end{cor} 
	
	\begin{proof}

		This will follow from \Cref{qv} and \Cref{q2q}, where we recall the notation from \Cref{qjt2} throughout this proof. We must first set the parameters $(B, \mathfrak{M}, k_1, k_2)$ implicit in the statements of those results. So, let us set $B = 100$ (which guarantees the existence of $\chi$ as in \Cref{mchi}); $\mathfrak{M} = T$; and $k_1 = N_1 + T^2$ and $k_2 = N_2 - T^2$. Under this setup, observe that \eqref{testimateb} and \eqref{nk1k2} hold (due to \eqref{jinterval2}); we therefore adopt \Cref{chi2} and define the proxy dynamics $\bm{\mathfrak{Q}}(s)$ as in \Cref{qjt2}. For the remainder of this proof, we restrict to the event $\mathsf{E}$ on which both \Cref{q2q} and \Cref{qv} hold (for the above choices of parameters). 
		
		Then by \Cref{qv}, for any $j_0 \in \llbracket N_1 + 2T^5, N_2 - 2T^5 \rrbracket \subseteq \llbracket k_1 + T^5, k_2 - T^5 \rrbracket$ and $t \in [0,T]$, we have 
		\begin{flalign}
			\label{q11q} 
			\displaystyle\max_{t \in [0,T]} \big|\mathfrak{Q}_{j_0} (t) - \mathfrak{Q}_{j_0} (0) - t v_{\eff} (\lambda_{\varphi_0^{-1} (j_0)}) \big| \le T \cdot B^2 \mathfrak{M}^{-1/2} (\log N)^{32} = 10^4 \cdot T^{1/2} (\log N)^{32},
		\end{flalign}
		
		\noindent where we have recalled from \Cref{qjt2} that $\Lambda_{j_0} = \lambda_{\varphi_0^{-1} (j_0)}$. Moreover, by \Cref{q2q} we have for $j_0 \in \llbracket N_1 + 2T^5, N_2 - 2T^5 \rrbracket \subseteq \llbracket k_1 + T^4, k_2 - T^4 \rrbracket \subseteq \llbracket \ell_{\lfloor T \rfloor}, m_{\lfloor T \rfloor} \rrbracket$ that 
		\begin{flalign}
			\label{q22q} 
			\displaystyle\max_{t \in [0,T]} \big| \mathfrak{Q}_{j_0} (t) - Q_{\varphi_0^{-1} (j_0)} (t) \big| \le B\mathfrak{M}^{1/2} (\log N)^{20} = 100 T^{1/2} (\log N)^{20}.
		\end{flalign}
		
		\noindent Combining \eqref{q11q} and \eqref{q22q}, and setting $j_0 = \varphi_0 (j)$, we deduce the corollary.
	\end{proof}

	\subsection{Proof of \Cref{vestimate}}
	
	\label{ProofV} 
	
	\begin{proof}[Proof of \Cref{vestimate}]
		
		First observe that if $T \le 10^8 \cdot (30 \log N)^{60}$ then we have $T^{1/2} (\log N)^{35} \ge (T+1) (\log N)^4$, and so the theorem follows from the $(t, t'; \lambda; \varphi, \varphi') = (0, t; \lambda_j; \varphi_0 (j), \varphi_t (j))$ case of \Cref{centerdistance}. Therefore, we assume that $T \ge 10^8 \cdot (30 \log N)^{60}$ in what follows. 
		
		We will show for any fixed $t \in [0,T]$ that, with overwhelming probability, 
		\begin{flalign}
			\label{qjt0} 
		 \big| Q_j (t) - Q_j(0) - tv_{\eff} (\lambda_j) \big| \le T^{1/2} (\log N)^{34}.
		\end{flalign} 
		
		\noindent Let us first verify that this is sufficient to confirm the theorem. Recalling \Cref{adelta}, restrict to the event $\mathsf{F}_1 = \mathsf{BND}_{\bm{L}} (\log N)$ from \Cref{adelta}, which we may by \Cref{l0eigenvalues}. We also restrict to the event $\mathsf{F}_2$ on which \Cref{qijsalpha} holds and to the event $\mathsf{F}_3$ on which \Cref{centerdistance} holds.
		
		Let $\mathsf{F}_4 (t)$ denote the event on which \eqref{qjt0} holds; let $\mathcal{T} \subseteq [0, T]$ denote an $N^{-2}$-mesh of $[0, T]$; and let $\mathsf{F}_4 = \bigcap_{s \in \mathcal{T}} \mathsf{F}_4 (s)$. Restrict to $\mathsf{F}_4$, and fix $t \in [0, T]$. Then, there exists $s \in \mathcal{S}$ for which $|t-s| \le N^{-2}$, meaning that 
		\begin{flalign*}
			\big| Q_j (t) - Q_j (0) - tv_{\eff} (\lambda_j) \big| & \le \big| Q_j (s) - Q_j (0) - sv_{\eff} (\lambda_j) \big| + |t-s| \cdot |v_{\eff} (\lambda_j)| + | Q_j (t) - Q_j (s) | \\
			& \le  T^{1/2} (\log N)^{34} + N^{-2} \cdot (\log N)^2 + | Q_j (t) - Q_j (s) | \\
			& \le T^{1/2} (\log N)^{34} + N^{-1} + 2 (\log N)^4 \le T^{1/2} (\log N)^{35}.
		\end{flalign*} 
		
		\noindent Here, the second inequality follows from our restriction to $\mathsf{F}_4$ and the fact that $|v_{\eff} (\lambda_j)| \le (\log N)^2$ (by \Cref{derivativev} and our restriction to $\mathsf{F}_1$); the third from \Cref{centerdistance} (with our restriction to $\mathsf{F}_3$), as $|t-s| \le 1$; and the fourth from the fact that $N$ is sufficiently large. This establishes the lemma, so it suffices to show that \eqref{qjt0} holds with overwhelming probability. 
		
		To that end we first apply \Cref{q} on a Toda lattice (at thermal equilibrium) on a larger interval, and then use comparison estimates (such as \Cref{a2p2}, \Cref{lleigenvalues2}, and \Cref{lleigenvalues}) to approximate the original Toda lattice by the enlarged one. This will proceed similarly to in the proof of \cite[Theorem 8.5]{LC} given \cite[Theorem 8.2]{LC}. To implement it, first let $\tilde{N}_1 \le \tilde{N}_2$ be integers satisfying 
		\begin{flalign}
			\label{nn} 
			\tilde{N}_1 + N^{10} \le N_1 \le N_2 \le \tilde{N}_2 - N^{10}; \qquad  T^{15} \le \tilde{N} \le N^{30},
		\end{flalign} 
		
		\noindent where $\tilde{N} = \tilde{N}_2 - \tilde{N}_1 + 1$. Let $(\tilde{\bm{a}} (s); \tilde{\bm{b}} (s)) \in \mathbb{R}^{\tilde{N}} \times \mathbb{R}^{\tilde{N}}$ denote the Flaschka variables for a Toda lattice on $\llbracket \tilde{N}_1, \tilde{N}_2 \rrbracket$; letting $\tilde{\bm{a}} (s) = (\tilde{a}_{\tilde{N}_1} (s), \tilde{a}_{\tilde{N}_1+1} (s), \ldots , \tilde{a}_{\tilde{N}_2} (s))$ and $\tilde{\bm{b}} (s) = (\tilde{b}_{\tilde{N}_1} (s), \tilde{b}_{\tilde{N}_1+1} (s), \ldots , \tilde{b}_{\tilde{N}_2}(s))$, they satisfy $\tilde{a}_{\tilde{N}_2} (s) = 0$, and \eqref{derivativepa} holds for each $(j, t) \in \llbracket \tilde{N}_1, \tilde{N}_2 \rrbracket \times \mathbb{R}_{\ge 0}$. We sample the initial data $(\tilde{\bm{a}}(0);\tilde{\bm{b}}(0))$ according to the thermal equilibrium $\mu_{\beta,\theta;\tilde{N}-1,\tilde{N}}$ of \Cref{mubeta2}; we couple $(\tilde{\bm{a}}(0);\tilde{\bm{b}}(0))$ with $(\bm{a}(0);\bm{b}(0))$ so that $(\tilde{a}_i(0),\tilde{b}_i(0)) = (a_i(0), b_i(0))$ for all $i \in \llbracket N_1, N_2 - 1 \rrbracket$. 
		
		For any $s \in \mathbb{R}_{\ge 0}$, denote the Lax matrix associated with $(\tilde{\bm{a}}(s);\tilde{\bm{b}}(s))$ (as in \Cref{matrixl}) by $\tilde{\bm{L}}(s) = [\tilde{L}_{ij}(s)] \in \SymMat_{\tilde{N}\times\tilde{N}}$, and set $\eig \tilde{\bm{L}}(s) = (\tilde{\lambda}_1, \tilde{\lambda}_2, \ldots , \tilde{\lambda}_{\tilde{N}})$. Setting $\tilde{\zeta} = e^{-100 (\log \tilde{N})^{3/2}}$, for each $s \in \mathbb{R}_{\ge 0}$ let $\tilde{\varphi}_s : \llbracket 1, \tilde{N} \rrbracket \rightarrow \llbracket \tilde{N}_1, \tilde{N}_2 \rrbracket$ denote a $\zeta$-localization center bijection for $\tilde{\bm{L}}(s)$. Further let $(\tilde{\bm{p}} (s); \tilde{\bm{q}}(s)) \in \mathbb{R}^{\tilde{N}} \times \mathbb{R}^{\tilde{N}}$ denote the Toda state space variables associated with $(\tilde{\bm{a}}(s); \tilde{\bm{b}}(s))$, as in \Cref{Open}, where we have indexed the $\tilde{N}$-tuples $\tilde{\bm{p}} (s) = (\tilde{p}_{\tilde{N}_1} (s), \tilde{p}_{\tilde{N}_1+1} (s), \ldots , \tilde{p}_{\tilde{N}_2} (s))$ and $\tilde{\bm{q}} (s) = (\tilde{q}_{\tilde{N}_1} (s), \tilde{q}_{\tilde{N}_1+1} (s), \ldots , \tilde{q}_{\tilde{N}_2} (s))$. For each $s \in \mathbb{R}_{\ge 0}$ and $i \in \llbracket 1, \tilde{N} \rrbracket$, denote $\tilde{Q}_i (s) = \tilde{q}_{\tilde{\varphi}_s (i)} (s)$.
		
		We next restrict to seven events.\footnote{Here, when we define an event on which some result (such as \Cref{centert}) holds with the $\bm{L}$ there equal to $\tilde{\bm{L}}$, we also implicitly replace $(N_1, N_2, \bm{q})$ in that result with $(\tilde{N}_1, \tilde{N}_2, \tilde{\bm{q}})$.} Recalling \Cref{adelta}, we first restrict to the event $\mathsf{E}_1 = \bigcap_{r \ge 0} \mathsf{BND}_{\bm{L}(r)} (\log N) \cap \mathsf{BND}_{\tilde{\bm{L}}(r)} (\log N)$, as we may by \Cref{l0eigenvalues}. Further restrict to the event $\mathsf{E}_2$ on which $a_i (0) \ge e^{-(\log N)^2}$ for each $i \in \llbracket N_1, N_2 - 1 \rrbracket$, which we may by the explicit density of the $(a_i)$ from \Cref{mubeta2}. Also restrict to the event $\mathsf{E}_3$ on which \Cref{qijsalpha} holds, with the $\bm{q}(s)$ there equal to both $\bm{q}(s)$ and $\tilde{\bm{q}}(s)$ here. Moreover restrict to the event $\mathsf{E}_4$ on which \Cref{centert} and \Cref{centerdistance} both hold, with the $(\bm{L};\varphi_j)$ there equal to both $(\bm{L}(0); \varphi_0 (j))$ and $(\tilde{\bm{L}}(0); \tilde{\varphi}_0 (j))$ here. Additionally restrict to the event $\mathsf{E}_5$ on which \Cref{q} holds, with the $\bm{L}(s)$ there equal to $\tilde{\bm{L}}(s)$ here; observe that $10^8 \cdot (\log \tilde{N})^{60} \le 10^8 \cdot (30 \log N)^{60} \le T \le N \le \tilde{N}^{1/10}$, verifying the first estimate in its assumption \eqref{jinterval2}. 
		
		To define the sixth event, set $K = T (\log N)^2$, and observe by \Cref{ltl0} that there exists random matrices $\bm{M} = [M_{ij}] \in \SymMat_{\llbracket N_1, N_2 \rrbracket}$ and $\tilde{\bm{M}} = [\tilde{M}_{ij}] \in \SymMat_{\llbracket \tilde{N}_1, \tilde{N}_2 \rrbracket}$ with the same laws as $\bm{L}(0)$ and $\tilde{\bm{L}}(0)$, respectively, and an overwhelmingly probable event $\mathsf{E}_6$, on which we have (as $K \ge T \ge 5 (\log N)^3$) that
		\begin{flalign}
			\label{lm11} 
			\displaystyle\max_{i,j \in \llbracket N_1+K, N_2-K \rrbracket} | M_{ij} - L_{ij} (t) | \le e^{-(\log N)^3}; \quad \displaystyle\max_{i,j \in \llbracket N_1+K, N_2-K \rrbracket} | \tilde{M}_{ij} - \tilde{L}_{ij} (t) | \le e^{-(\log N)^3}.
		\end{flalign} 
		
		\noindent We may further assume on $\mathsf{E}_6$ that $M_{i,i+1}, \tilde{M}_{i,i+1} \ge e^{-(\log N)^2}$ for each $i$, by the explicit densities of these entries from \Cref{mubeta2}. We restrict to $\mathsf{E}_6$ in what follows. To define the seventh event, observe by \Cref{a2p2} (with the $A$ there equal to $\log N$ here, using our restriction to $\mathsf{E}_1$) that 
	\begin{flalign}
		\label{l11} 
		\displaystyle\sup_{t \in [0, T]} \displaystyle\max_{i \in \llbracket N_1 + K, N_2 - K \rrbracket} \big( |a_i (t) - \tilde{a}_i (t)| + |b_i (t) - \tilde{b}_i (t)| \big) \le e^{-(\log N)^3}. 
	\end{flalign}  	
		
		\noindent We may therefore (by \eqref{l11} and \eqref{lm11}) further restrict to the event $\mathsf{E}_7$ on which \Cref{lleigenvalues2} and \Cref{lleigenvalues} both hold, with the $(\delta; \mathcal{D})$ there equal to $(3e^{-(\log N)^3}; \llbracket \tilde{N}_1, \tilde{N}_2 \rrbracket \setminus \llbracket N_1+K, N_2-K \rrbracket)$ here, and the $(\bm{L}, \tilde{\bm{L}})$ equal to any of $(\bm{L}(0), \tilde{\bm{L}}(0))$, $(\bm{M},\bm{L}(t))$, $(\tilde{\bm{M}},\tilde{\bm{L}}(t))$, and $(\tilde{\bm{M}},\bm{M})$ here (viewing $\bm{M}$ as a $\tilde{N} \times \tilde{N}$ matrix by setting $M_{ij} = 0$ if $(i,j) \in \llbracket \tilde{N}_1, \tilde{N}_2 \rrbracket \setminus \llbracket N_1, N_2 \rrbracket$, and similarly for $\bm{L}(s)$). 
		
		Now, by \Cref{q} (and our restriction to $\mathsf{E}_5$), we have for any index $j \in \llbracket 1, \tilde{N} \rrbracket$ satisfying $\tilde{N}_1 + 2T^5 \le \tilde{\varphi}_0 (j) \le \tilde{N}_2 - 2 T^5$ that 
		\begin{flalign}
			\label{q2j0} 
			\displaystyle\sup_{t \in [0,T]} \big| \tilde{Q}_j (t) - \tilde{Q}_j (0) - t v_{\eff} (\tilde{\lambda}_j) \big| \le T^{1/2} (\log \tilde{N})^{33} \le T^{1/2} \cdot (30 \log N)^{33}.
		\end{flalign}
		
		\noindent We must therefore approximate $\tilde{Q}_j (s)$ by $Q_j (s)$ and $\tilde{\lambda}_j$ by $\lambda_j$, which we will do using \Cref{a2p2}, \Cref{lleigenvalues2}, and \Cref{lleigenvalues} (with our restriction to $\mathsf{E}_7$). 
		  
		To do so, fix $j \in \llbracket 1, N \rrbracket$ with 
		\begin{flalign}
			\label{0j}  
			N_1 + 3T(\log \tilde{N})^4 \le \varphi_0 (j) \le N_2 - 3T (\log \tilde{N})^4,
		\end{flalign} 
		
		\noindent which satisfies \eqref{j0}. By \Cref{centert} (with our restriction to $\mathsf{E}_4$), we have 
		\begin{flalign}
			\label{0j1} 
			N_1 + 2 T(\log \tilde{N})^4 \le \varphi_t (j) \le N_2 - 2 T(\log \tilde{N})^4.
		\end{flalign} 
		
		\noindent By \Cref{lleigenvalues} (with our restriction to $\mathsf{E}_7$), there exists a constant $c_2>0$ and an eigenvalue $\mu \in \eig \bm{M}$ satisfying the following properties. We have that $|\mu - \lambda_j| \le c_2^{-1} e^{-c_2 (\log N)^3}$ and $\varphi_t (j)$ is a $N^{-1} \zeta$-localization center of $\mu$ with repect to $\bm{M}$. Again by \Cref{lleigenvalues}, there exists an eigenvalue $\tilde{\mu} \in \eig \tilde{\bm{M}}$ such that $|\mu - \tilde{\mu}| \le c_2^{-1} e^{-c_2 (\log N)^3}$ and $\varphi_t (j)$ is an $N^{-2} \zeta$-localization center of $\tilde{\mu}$ with respect to $\tilde{\bm{M}}$. By \Cref{lleigenvalues2} (and our restriction to $\mathsf{E}_7$), there exists an index $\tilde{j} \in \llbracket 1, \tilde{N} \rrbracket$ such that $|\tilde{\mu} - \tilde{\lambda}_{\tilde{j}}| \le c_2^{-1} e^{-c_2 (\log N)^3}$ and $\varphi_t (j)$ is an $N^{-3} \zeta$-localization center of $\tilde{\lambda}_{\tilde{j}}$ with respect to $\tilde{\bm{L}}(t)$. By \Cref{centerdistance} (with our restriction to $\mathsf{E}_4$), we therefore have $|\varphi_t (j) - \tilde{\varphi}_t (\tilde{j})| \le (\log \tilde{N})^4$; similarly, $|\varphi_0 (j) - \tilde{\varphi}_0 (\tilde{j})| \le (\log \tilde{N})^4$. Combining the above estimates yields
		\begin{flalign}
			\label{lambdalambdaj} 
			|\lambda - \tilde{\lambda}_{\tilde{j}}| \le 3c_2^{-1} e^{-c_2 (\log N)^3}; \quad | \varphi_0 (j) - \tilde{\varphi}_0 (\tilde{j}) | \le (\log \tilde{N})^4; \quad | \varphi_t (j) - \tilde{\varphi}_t (\tilde{j}) | \le (\log \tilde{N})^4.
		\end{flalign} 
		
		 By \eqref{lambdalambdaj}, \eqref{0j}, and \eqref{0j1}, it follows that $N_1 + T(\log \tilde{N})^4 \le \tilde{\varphi}_s (\tilde{j}) \le N_2 - T(\log \tilde{N})^4$ for each $s \in \{ 0, t \}$. Thus, since $\llbracket N_1 + 3T (\log \tilde{N})^4, N_2 -3T(\log \tilde{N})^4 \rrbracket \subseteq \llbracket \tilde{N}_1 + 2T^5, \tilde{N}_2 - 2T^5 \rrbracket$ (by \eqref{nn}), \eqref{q2j0} implies
		\begin{flalign}
			\label{q1} 
			\displaystyle\sup_{t \in [0,T]} \big| \tilde{q}_{\tilde{\varphi}_t(\tilde{j})} (t) - \tilde{q}_{\tilde{\varphi}_0(\tilde{j})} (0) - t v_{\eff} (\tilde{\lambda}_{\tilde{j}}) \big| \le 30^{33} \cdot T^{1/2} (\log N)^{33}.
		\end{flalign}
		
		\noindent Observe by \Cref{derivativev} and \eqref{lambdalambdaj} that 
		\begin{flalign}
		\label{q22} 
			| tv_{\eff} (\lambda_j) - tv_{\eff} (\tilde{\lambda}_{\tilde{j}}) | \le T \cdot (\log N)^2 \cdot 3c_2^{-1} e^{-c_2 (\log N)^3} \le 1.
		\end{flalign} 
		
		\noindent For any $i \in \llbracket N_1+K,N_2-K \rrbracket$ and $s \in \{ 0, T \}$, further observe by \eqref{l11} and \eqref{q00} (with the fact that $a_i (s) \ge e^{-2(\log N)^2}$, by our restriction to $\mathsf{E}_2 \cap \mathsf{E}_6$) that
		\begin{flalign}
			\label{q3} 
			| q_i (s) - \tilde{q}_i (s) | \le 2N \cdot e^{-(\log N)^3}  \cdot 2 e^{2(\log N)^2} \le 1.
		\end{flalign} 
		
		\noindent Hence, for any $j \in \llbracket N_1 + T(\log N)^5, N_2 - T(\log N)^5 \rrbracket$, we have (as $(\log N)^5 > 3 (\log \tilde{N})^4$, by \eqref{nn})  
		\begin{flalign*}
			\big| Q_j (t) - Q_j (0) - t v_{\eff} (\lambda_j) \big| & \le \big| \tilde{Q}_{\tilde{j}} (t) - \tilde{Q}_{\tilde{j}} (0) - t v_{\eff} (\tilde{\lambda}_{\tilde{j}}) \big| + | t v_{\eff} (\lambda_j) - t v_{\eff} (\tilde{\lambda}_{\tilde{j}}) | \\
				& \qquad + | q_{\tilde{\varphi}_0(\tilde{j})} (0) - \tilde{q}_{\tilde{\varphi}_0 (\tilde{j})} (0) | + | q_{\tilde{\varphi}_t(\tilde{j})} (t) - \tilde{q}_{\tilde{\varphi}_t (\tilde{j})} (t) | \\
				& \qquad + | q_{\varphi_0(j)} (0) - q_{\tilde{\varphi}_0(\tilde{j})} (0) | + | q_{\varphi_t (j)} (t) - q_{\tilde{\varphi}_t (\tilde{j})} (t) | \\
				& \le 30^{33} \cdot T^{1/2} (\log N)^{33} + 3 + 2(\log N)^5 \le T^{1/2} (\log N)^{34}.
		\end{flalign*}
		
		\noindent Here, in first bound, we used the definition \eqref{qjs2} of $Q_j$. In the second, we used \eqref{q1} to bound the first term, \eqref{q22} to bound the second, \eqref{q3} to bound the third and fourth, and \eqref{qiqjs4} with \eqref{lambdalambdaj} (using our restriction to $\mathsf{E}_3$) to bound the fifth and sixth. In the third bound, we used that $N$ is sufficiently large. This establishes the theorem.
	\end{proof}
	
	\appendix

	\section{Proofs of Results From \Cref{TProperty}}
	
	\label{ProofT}

	\subsection{Proof of \Cref{rhoexponential} and \Cref{rhorho}}
	
	\label{ProofIntegral}

	In this section we establish \Cref{rhoexponential} and \Cref{rhorho}. We first establish the latter but, before doing so, we require the following lemma, indicating that the density $\varrho$ is bounded (independently of $\theta$, though this will not be needed until the proof of \Cref{beta0theta} in \Cref{IntegralEstimate} below). 
	
	\begin{lem} 
		
		\label{1thetarho} 
		
		There exists a constant $C = C(\beta) > 1$ (that is independent of $\theta>0$) such that $\varrho (x) < C$ for all $x \in \mathbb{R}$.
		
	\end{lem} 
	
	\begin{proof} 
		
		Throughout, we adopt the notation from \Cref{lbetaeta}. By \Cref{lf}, $\varrho$ is the limiting empirical spectral distribution of the random Lax matrix $\bm{L}$, sampled under thermal equilibrium. Due to this, the lemma follows from the Wegner estimate \cite[Theorem 4.1]{R}, whose hypotheses are verified by the fact (from \Cref{mubeta2} and \Cref{matrixl}) that the density of any diagonal entry $L_{ii}$ of $\bm{L}$, conditional on all of the other entries of $\bm{L}$, is bounded above (independently of $\theta$).
	\end{proof}

	\begin{proof}[Proof of \Cref{rhorho}]  
		
		By \cite[Theorem 3.15]{SERL}, these exists a certain constant $Z_{\theta} > 0$ such that 
		\begin{flalign}
			\label{rhox2} 
			\log \varrho_{\beta} (x) = 2 \theta \displaystyle\int_{-\infty}^{\infty} \log |x-y| \varrho_{\beta} (y) dy - \displaystyle\frac{\beta x^2}{2} - \log Z_{\theta},
		\end{flalign}
		
		\noindent and, by \cite[Section 3]{SERL}, this constant $Z_{\theta}$ is explicitly given by 
		\begin{flalign}
			\label{ztheta} 
			Z_{\theta} = \displaystyle\lim_{N\rightarrow \infty} Z_{N; \theta/N} \cdot Z_{N-1; \theta/N}^{-1}
		\end{flalign}
		
		\noindent Here, $Z_{N; v}$ is the partition function of the Gaussian $2v$ Ensemble scaled by $\beta^{1/2}$, namely,
		\begin{flalign*}
			Z_{N;v} = \displaystyle\int_{\mathbb{R}^N} \exp \bigg( -\displaystyle\frac{\beta}{2} \sum_{j=1}^N \lambda_j^2 \bigg)  \displaystyle\prod_{1 \le i < j \le N} |\lambda_i - \lambda_j|^{2v} \displaystyle\prod_{j=1}^N d \lambda_j,
		\end{flalign*} 
		
		\noindent which is equal to (see \cite[Equation (2)]{MME})
		\begin{flalign}
			\label{znab} 
			Z_{N;v} = (2 \pi)^{N/2} \beta^{-v \binom{N}{2} - N/2} \displaystyle\prod_{j=1}^N \displaystyle\frac{\Gamma (1+vj)}{\Gamma (1+v)}.
		\end{flalign}
		
		\noindent Together, \eqref{znab} and \eqref{ztheta} yield $Z_{\theta} = (2\pi)^{1/2} \beta^{-\theta-1/2} \cdot \Gamma (\theta + 1)$. By \eqref{rhox2}, this implies upon adding $\log \theta$ to both sides that
		\begin{flalign*}
			\log ( \theta \cdot \varrho_{\beta} (x) ) & = 2\theta \displaystyle\int_{-\infty}^{\infty} \log |x-y| \varrho_{\beta} (y) dy + \displaystyle\frac{1}{2} \cdot \big( (2 \theta + 1) \log \beta - \log (2\pi) - \beta x^2 \big) \\
			& \qquad  - \log \Gamma(\theta+1) + \log \theta.
		\end{flalign*} 
		
		\noindent Differentiating both sides with respect to $\theta$ yields
		\begin{flalign}
			\label{theta0rho} 
			\displaystyle\frac{\varrho(x)}{\theta \cdot \varrho_{\beta} (x)} = 2 \displaystyle\int_{-\infty}^{\infty} \log|x-y| \varrho (y) dy + \log \beta - \displaystyle\frac{\Gamma'(\theta+1)}{\Gamma(\theta+1)} + \displaystyle\frac{1}{\theta},
		\end{flalign}
		
		\noindent where we used the fact that 
		\begin{flalign*}
			\partial_{\theta} \Bigg( 2\theta \displaystyle\int_{-\infty}^{\infty} \log |x-y| \varrho_{\beta} (y) dy \Bigg) = 2 \displaystyle\int_{-\infty}^{\infty} \log |x-y| \varrho (y) dy,
		\end{flalign*} 
	
		\noindent which holds since $\varrho = \partial_{\theta} (\theta \varrho_{\beta})$ from \Cref{frho}; since $\varrho$ is bounded by \Cref{1thetarho}; and since $\varrho$ has subpolynomial tails by \Cref{fcn}. 
		
		The equality \eqref{theta0rho}, with the fact that $\Gamma' (\theta+1) =  \Gamma (\theta) + \theta \cdot \Gamma' (\theta)$ (as $\Gamma(\theta+1) = \theta \cdot \Gamma (\theta)$), gives   
		\begin{flalign*}
			\displaystyle\frac{\varrho (x)}{\varrho_{\beta} (x)} = 2 \theta \displaystyle\int_{-\infty}^{\infty} \log |x-y| \varrho(y) dy + \theta \log \beta - \theta \cdot \displaystyle\frac{\Gamma'(\theta)}{\Gamma (\theta)}.
		\end{flalign*} 
		
		\noindent Together with the definitions \eqref{operatort} of $\bm{\mathrm{T}}$ and \eqref{alpha} of $\alpha$, this establishes the lemma.
	\end{proof}
	
	\begin{proof}[Proof of \Cref{rhoexponential}]
		
		The first two bounds in \eqref{estimate0} are due to \cite[Lemma 2.2]{RET}, so we must show the third. To that end, by \Cref{rhorho} and the first estimate in \eqref{estimate0}, it suffices to show that there exists a constant $C > 1$ such that $|\bm{\mathrm{T}} \varrho (x)| \le C(|x|+1)$, for all $x \in \mathbb{R}$. This follows from the fact that there exist constants $C_1 > 1$ and $C_2 > 1$ such that 
		\begin{flalign*}
			| \bm{\mathrm{T}} \varrho (x)| & = 2 \Bigg| \displaystyle\int_{-\infty}^{\infty} \log |x-y| \varrho(y) dy \Bigg| \\
			& \le 2  \displaystyle\int_{x-1}^{x+1} \big| \log |x-y| \big|  \varrho(y) dy + 2(|x|+1)  \displaystyle\int_{|y-x| > 1} (|x|+1)^{-1} \cdot \big| \log |x-y| \big| \cdot \varrho (y) dy  \\
			& \le C_1 + 2(|x|+1)  \displaystyle\int_{|y-x| > 1} (|x|+1)^{-1} \cdot \big| \log |x-y| \big| \cdot \varrho (y) dy   \le C_1 + C_2 (|x|+1),
		\end{flalign*}
	
		\noindent where the first and second statements hold by \eqref{operatort}, and the third holds by \Cref{1thetarho}; and the fourth holds by \Cref{fcn} and the bound $(|x|+1)^{-1} \cdot |\log |x-y|| \le C_3 (|y|+1)$ for some constant $C_3 > 1$, whenever $|x-y| \ge 1$.		
	\end{proof}

	\subsection{Proof of \Cref{alphat}}
	
	\label{ProofBound}
	
	In this section we establish \Cref{alphat}. This will proceed by expressing $\bm{\mathrm{T}} \varrho (x) + \alpha$ in terms of the resolvent of a random Lax matrix, and using known estimates on the latter. To implement the former, we require the following two lemmas; the first is a general expression for a certain entry of the resolvent of a tridiagonal matrix (essentially due to \cite{DSRL}, though we provide its quick proof here), and the second is a probabilistic interpretation for $\alpha$.

	\begin{lem}[{\cite[Equation 3]{DSRL}}]
		
		\label{gn1n2}
		
		Fix integers $N_1 \le N_2$, and let $\bm{M} = [M_{ij}] \in \SymMat_{\llbracket N_1, N_2 \rrbracket}$ denote a symmetric, tridiagonal real matrix. For any $E \in \mathbb{R} \setminus \eig \bm{M}$, define $\bm{G}(E) = [G_{ij}(E)] \in \Mat_{\llbracket N_1, N_2 \rrbracket}$ by $\bm{G}(E) = (\bm{M}-E \cdot \Id)^{-1}$. Then, we have 
		\begin{flalign*}
			|G_{N_1 N_2} (E)| = \displaystyle\prod_{i=N_1}^{N_2-1} |M_{i,i+1}| \cdot \displaystyle\prod_{\mu \in \eig \bm{M}} |\mu-E|^{-1}.
		\end{flalign*}
		
	\end{lem} 
	
	\begin{proof} 
	
	Set $N = N_2 - N_1 + 1$. Let $\bm{C}(E) = [ C_{ij} (E) ]$ denote the cofactor matrix of $\bm{M} - E \cdot \Id$, and observe that $G_{N_1 N_2} (E) = (-1)^{N+1} \cdot C_{N_1 N_2} (E) \cdot (\det \bm{M} - E \cdot \Id)^{-1}$. Since removing the row of index $N_1$ and column of index $N_2$ from $\bm{M}$ yields an upper triangular $(N-1) \times (N-1)$ matrix with diagonal entries $(M_{i,i+1})_{N_1 \le i < N_2}$, we deduce that $C_{N_1 N_2} (E) = \prod_{i=N_1}^{N_2 - 1} M_{i,i+1}$. Hence,
	\begin{flalign*}
			G_{N_1 N_2} (E) & = (-1)^{N+1} \cdot C_{N_1 N_2} (E) \cdot (\det \bm{M} - E \cdot \Id)^{-1} \\
			& = (-1)^{N+1} \cdot \displaystyle\prod_{i=N_1}^{N_2-1} M_{i,i+1} \cdot \displaystyle\prod_{\mu \in \eig \bm{M}} (\mu - E)^{-1},
	\end{flalign*}
	
	\noindent which confirms the lemma.
	\end{proof}

	\begin{lem}[{\cite[Lemma 3.11]{LC}}]
		
		\label{aintegral} 
		
		Let $\mathfrak{a} > 0$ be a random variable with law $\mathbb{P} [\mathfrak{a} \in da] = 2 \beta^{\theta} \cdot \Gamma(\theta)^{-1} \cdot a^{2\theta-1} e^{-\beta a^2} da$. Denoting $\mathfrak{a} = e^{-\mathfrak{r}/2}$, we have that $\mathbb{E} [\mathfrak{r}] = \alpha$.
	\end{lem}

	We further require the following result indicating that the off-diagonal entries in the resolvent of a Lax matrix, of the Toda lattice under the thermal equilibrium, decay exponentially. 
	
	\begin{lem}[{\cite[Theorem 4]{LRBM}}]
		
		\label{gijexponential}
		
		Adopt \Cref{lbetaeta}. For any real number $s \in (0, 1)$, there exists a constant $c = c(s) > 0$ such that the following holds. For any $E \in \mathbb{R} \setminus \eig \bm{L}$, denote $\bm{G} (E) = [G_{ij} (E) ] = (\bm{L} - E)^{-1}$. We have  
		\begin{flalign}
			\label{gijs}
			\displaystyle\sup_{E \in \mathbb{R}} \mathbb{E} \big[ |G_{ij} (E) |^s \big] \le c^{-1} e^{-c|i-j|}.
		\end{flalign}  
	\end{lem}

	Now we can establish \Cref{alphat}.

	\begin{proof}[Proof of \Cref{alphat}]
		
		Let $\varepsilon \in (0, 1)$ be a real number, and define the interval $\mathcal{I}_{\varepsilon} = [x-\varepsilon, x + \varepsilon]$, so that $|\mathcal{I}_{\varepsilon}| = 2 \varepsilon$. Adopt \Cref{lbetaeta} and, for any $E \in \mathbb{R} \setminus \eig \bm{L}$, denote the resolvent $\bm{G} (E) = [ G_{ij} (E) ] = (\bm{L}-E)^{-1}$. Observe that there exists a constant $c > 0$ such that 
		\begin{flalign}
			\label{clambdan}
			\begin{aligned} 
				-cN  \ge \displaystyle\frac{4}{|\mathcal{I}_{\varepsilon}|} \displaystyle\int_{\mathcal{I}_{\varepsilon}} \log \mathbb{E} \big[ | G_{N_1 N_2} (E) |^{1/2} \big] dE & \ge \mathbb{E} \Bigg[ \displaystyle\frac{2}{|\mathcal{I}_{\varepsilon}|} \displaystyle\int_{\mathcal{I}_{\varepsilon}} \log | G_{N_1 N_2} (E) | dE \Bigg] \\
				& = 2 \displaystyle\sum_{j=N_1}^{N_2-1} \mathbb{E} [\log L_{j,j+1}] -  \mathbb{E} \Bigg[ \displaystyle\sum_{j=1}^N \displaystyle\frac{2}{|\mathcal{I}_{\varepsilon}|} \displaystyle\int_{\mathcal{I}_{\varepsilon}} \log |\lambda_i - E| dE \Bigg],
			\end{aligned}
		\end{flalign}
		
		\noindent where the first statement follows from \Cref{gijexponential}; the second from the concavity of $\log x$; and the third from \Cref{gn1n2}. By \Cref{lbetaeta}, \Cref{matrixl}, and \Cref{mubeta2}, if we denote $a_j = L_{j,j+1}$, then $a_j$ has law $\mathbb{P} [a_j \in da] = 2 \beta^{\theta} \cdot \Gamma (\theta)^{-1} \cdot a^{2\theta-1} e^{-\beta a^2} da$. Hence, by \Cref{aintegral}, we have $2 \cdot \mathbb{E} [\log L_{j,j+1}] = -\alpha$. Inserting this into \eqref{clambdan} gives
		\begin{flalign}
			\label{alphaf} 
			\alpha + 2 \cdot \mathbb{E} \Bigg[  \displaystyle\frac{1}{N} \displaystyle\sum_{j=1}^N f(\lambda_j) \Bigg] \ge c, \qquad \text{where} \qquad f(\lambda) = |\mathcal{I}_{\varepsilon}|^{-1} \displaystyle\int_{\mathcal{I}_{\varepsilon}} \log |\lambda - E| dE.
		\end{flalign}
		
		Next observe that there exists a constant $C = C (x, \varepsilon) > 0$ such that $|f(\lambda)| \le C \log (\lambda^2 + 2) \le C(\lambda^2 + 2)$. Therefore, \Cref{lf} and \Cref{fcn} together imply that 
		\begin{flalign*}
			\displaystyle\lim_{N \rightarrow \infty} \mathbb{E} \Bigg[ \displaystyle\frac{1}{N} \displaystyle\sum_{j=1}^N f(\lambda_j) \Bigg] = \displaystyle\int_{-\infty}^{\infty} f(\lambda) \varrho (\lambda) d \lambda = |\mathcal{I}_{\varepsilon}|^{-1} \displaystyle\int_{\mathcal{I}_{\varepsilon}} \displaystyle\int_{-\infty}^{\infty} \log |\lambda - E| \varrho (\lambda) d \lambda dE.
		\end{flalign*} 
		
		\noindent Inserting this into \eqref{alphaf} and letting $\varepsilon$ tend to $0$ gives
		\begin{flalign*}
			c \le \alpha + \displaystyle\lim_{\varepsilon \rightarrow 0} \displaystyle\frac{2}{|\mathcal{I}_{\varepsilon}|} \displaystyle\int_{-\infty}^{\infty} \displaystyle\int_{\mathcal{I}_{\varepsilon}} \log |\lambda-E| \varrho (\lambda) dE d \lambda = \alpha + 2 \displaystyle\int_{-\infty}^{\infty} \log |x-\lambda| \varrho (\lambda) d \lambda,
		\end{flalign*} 
		
		\noindent which upon recalling the definition of $\bm{\mathrm{T}}$ from \eqref{operatort} yields the lemma.
	\end{proof}

	\begin{rem} 
		
		\label{trhoc} 
		
		The constant $c$ in \Cref{alphat} does not depend on $\theta$, as long as $ \theta$ is uniformly bounded. More precisely, for any fixed $\Theta>0$, the constant $c>0$ \Cref{alphat} only depends on $\Theta$, as long as $\theta \in (0, \Theta)$. Indeed, it is quickly verified from the proof \cite[Theorem 4]{LRBM} (see \cite[Equations (2.46) and (2.47)]{LRBM}) that the constant $c$ from \Cref{gijexponential} only depends on the off-diagonal entries $(L_{ij})_{i\ne j}$ of $\bm{L}$ through a lower bound on $\mathbb{P}[|L_{ij}| \le \tau]$, for some fixed $\tau > 1$. Since this probability is uniformly bounded below independently of $\theta$, as long as $\theta \in (0, \Theta)$, the $c$ from \Cref{gijexponential} does not depend on $\theta \in (0,\Theta)$. Thus, the constant $c>0$ in the proof of \Cref{alphat} does not, either.
		
	\end{rem}

	\subsection{Proofs of \Cref{trhobounded} and \Cref{trho}}
	
	\label{ProofT0}
	
	In this section we show \Cref{trhobounded} and \Cref{trho}. Define the function $\sigma : \mathbb{R} \rightarrow \mathbb{R}$ and subspace $\mathcal{H}_{0} \subseteq \mathcal{H}$ by for each $x \in \mathbb{R}$ setting 
	\begin{flalign}
		\label{sigmatheta} 
		\sigma(x) = \theta \cdot \bm{\mathrm{T}} \varrho (x) + \theta \alpha; \qquad \mathcal{H}_0 = \{ f \in \mathcal{H}: \langle f, \varsigma_0 \rangle_{\varrho} = 0 \},
	\end{flalign} 
	
	\noindent and observe that $\sigma \in \mathcal{H}$. We begin by proving \Cref{trhobounded}.

	\begin{proof}[Proof of \Cref{trhobounded}]
		
		Fix $f \in \mathcal{H}$. Observe that there exists a constant $C>1$ such that 
		\begin{flalign*}
			\| \bm{\mathrm{T} \varrho_{\beta}} f \|_{\mathcal{H}}^2 & = 4 \displaystyle\int_{-\infty}^{\infty} \bigg| \displaystyle\int_{-\infty}^{\infty} \log |x-y|  \cdot \varrho_{\beta} (y) f(y) dy \bigg|^2 \varrho (x) dx \\
			& \le 4 \displaystyle\int_{-\infty}^{\infty} | f(y) |^2 \varrho (y) dy \cdot  \displaystyle\int_{-\infty}^{\infty}  \Bigg| \displaystyle\int_{-\infty}^{\infty} (\log |x-y|)^2 \sigma (y)^{-1} \varrho_{\beta} (y) dy \Bigg| \cdot \varrho (x) dx \le C \cdot \| f \|_{\mathcal{H}}^2 ,
		\end{flalign*}
		
		\noindent where in the first statement we used \eqref{betaproduct}; in the second we used the fact that $\varrho (x) = \sigma (x) \cdot \varrho_{\beta} (x)$ by \Cref{rhorho} and \eqref{sigmatheta}; and in the third we used the facts that $\varrho_{\beta}$ is bounded with subexponential decay (by \Cref{rhoexponential}) and $\sigma (y) > c$ for some real number $c>0$ (by \Cref{alphat} and \eqref{sigmatheta}), again with \eqref{betaproduct}. This establishes the lemma.
	\end{proof}

	Next define $\bm{\mathrm{S}} = \theta^{-1} \bm{\sigma} - \bm{\mathrm{T} \varrho}$, whose domain is dense, as it contains any bounded, compactly supported function in $\mathcal{H}$. Thus, $\bm{\mathrm{S}}$ is an unbounded, self-adjoint operator on $\mathcal{H}$; denote its domain by $\mathcal{H}'$. The following lemma essentially indicates that $\bm{\mathrm{S}}$ acts on $\mathcal{H}_0$.

	\begin{lem} 
		
		\label{gbetat}
		
		For any $f \in \mathcal{H}_0 \cap \mathcal{H}'$, we have that $\bm{\mathrm{S}} f \in \mathcal{H}_0$.
	\end{lem} 
	
	\begin{proof}
		
		The lemma holds due to the sequence of equalities,
		\begin{flalign*}
			\langle \bm{\mathrm{S}} f, \varsigma_0 \rangle_{\varrho} & = \theta^{-1} \displaystyle\int_{-\infty}^{\infty} \sigma (x) f(x) \varrho(x)  dx - \displaystyle\int_{-\infty}^{\infty} \bm{\mathrm{T} \varrho} f(x) \cdot \varrho (x) dx \\ 
			& = \displaystyle\int_{-\infty}^{\infty} ( \bm{\mathrm{T}} \varrho (x) + \alpha ) \cdot f(x) \varrho (x) dx - 2 \displaystyle\int_{-\infty}^{\infty} \displaystyle\int_{-\infty}^{\infty} \log|x-y| \cdot \varrho (y) f(y) \cdot \varrho (x) dy dx \\
			& = \displaystyle\int_{-\infty}^{\infty} ( \bm{\mathrm{T}} \varrho (x) + \alpha ) \cdot f(x) \varrho (x) dx - \displaystyle\int_{-\infty}^{\infty} \bm{\mathrm{T}} \varrho (y) \cdot f(y) \varrho (y) dy = \alpha \displaystyle\int_{-\infty}^{\infty} f(x) \varrho (x) dx = 0.
		\end{flalign*}
		
		\noindent Here, in the first statement we used \eqref{betaproduct} and the definition of $\bm{\mathrm{S}}$; in the second we used the definitions \eqref{sigmatheta} of $\sigma$ and \eqref{operatort} of $\bm{\mathrm{T}}$; in the third we interchanged the order of integration in the second integral and again used \eqref{operatort}; and in the fourth and fifth we used that $f \in \mathcal{H}_0$. 
	\end{proof}
		
	The following (standard) lemma indicates that $\bm{\mathrm{T} \varrho}$ is nonpositive on $\mathcal{H}_0$. 
	
	\begin{lem} 
		
		\label{tg0} 
		
		For any bounded, compactly supported function $f \in \mathcal{H}_0$, we have that $\langle \bm{\mathrm{T} \varrho} f, f \rangle_{\varrho} \le 0$. 
	
	\end{lem} 
	
	\begin{proof}
		
		Denoting $g(x) = f(x) \varrho (x)$ for each $x \in \mathbb{R}$, observe that 
		\begin{flalign}
			\label{trhof}
			\langle \bm{\mathrm{T}} \bm{\varrho} f, f \rangle_{\varrho} = \displaystyle\int_{-\infty}^{\infty} \overline{f(x)} \varrho (x) \bm{\mathrm{T}} \bm{\varrho} f(x) dx = 2 \displaystyle\int_{-\infty}^{\infty} \displaystyle\int_{-\infty}^{\infty} \overline{g(x)} g(y) \log |x-y| dx dy.
		\end{flalign}
		
		\noindent It suffices to show the right side of \eqref{trhof} is nonpositive, which will closely follow the proof of \cite[Lemma 2.6.2(d)]{IRM}. Observe in the below that $g$ is bounded; compactly supported; and (as $f \in \mathcal{H}_0$) satisfies $\int_{-\infty}^{\infty} g(x) dx = 0$. 
		
		To that end, observe for any real number $u \in \mathbb{R} \setminus \{ 0 \}$ that  
		\begin{flalign*}
			\log |u| = \displaystyle\int_1^{|u|} z^{-1} dz = \displaystyle\int_{1}^{|u|} (2z)^{-1} \displaystyle\int_0^{\infty} e^{-w/2} dw dz & = \displaystyle\frac{1}{2}  \displaystyle\int_0^{\infty} t^{-2} \displaystyle\int_1^{|u|}  z e^{-z^2/2t} dz dt \\ 
			&  = \displaystyle\frac{1}{4} \displaystyle\int_0^{\infty} t^{-2} \displaystyle\int_1^{u^2} e^{-v/2t} dv dt \\
			&  = \displaystyle\int_0^{\infty} (2t)^{-1} \cdot (e^{-1/2t} - e^{-u^2/2t}) dt,
		\end{flalign*} 
		
		\noindent where the first two statements follow from performing the integration; the third from changing variables $w = t^{-1} z^2$; the fourth from changing variables $v = z^2$; and the fifth from performing the integration. Therefore, 
		\begin{flalign*}
			\displaystyle\int_{-\infty}^{\infty} \displaystyle\int_{-\infty}^{\infty} \overline{g(x)} g(y) \log |x-y| dx dy = \displaystyle\int_0^{\infty} (2t)^{-1} \displaystyle\int_{-\infty}^{\infty} \displaystyle\int_{-\infty}^{\infty} (e^{-1/2t} - e^{-|x-y|^2/2t}) \cdot \overline{g(x)} g(y) dx dy dt,
		\end{flalign*} 
		
		\noindent where interchanging the order of integration is justified by the facts that $e^{-1/2t} - e^{-|x-y|^2/2t}$ is of order $t^{-1}$ as $t$ tends to $\infty$ and that it decays exponentially in $|x-y|^2/t$ as $t$ tends to $0$. Since $\int_{-\infty}^{\infty} g(x) dx = 0$, it follows that
		\begin{flalign*}
			\displaystyle\int_{-\infty}^{\infty} \displaystyle\int_{-\infty}^{\infty} & \overline{g(x)} g(y) \log |x-y| dx dy \\
			& = -\displaystyle\int_0^{\infty} (2t)^{-1} \displaystyle\int_{-\infty}^{\infty} \displaystyle\int_{-\infty}^{\infty} e^{-(x-y)^2/2t} \overline{g(x)} g(y) dx dy \\
			& = - \displaystyle\int_0^{\infty} (8t\pi)^{-1/2} \displaystyle\int_{-\infty}^{\infty} \displaystyle\int_{-\infty}^{\infty} \displaystyle\int_{-\infty}^{\infty} e^{\mathrm{i} r(x-y) - tr^2/2} \overline{g(x)} g(y) dr dx dy dt \\
			& = -\displaystyle\int_0^{\infty} (8t\pi)^{-1/2} \displaystyle\int_{-\infty}^{\infty} e^{-tr^2/2} \Bigg| \displaystyle\int_{-\infty}^{\infty} e^{-\mathrm{i} ry} g(y) dy \Bigg|^2 dr dt \le 0,
		\end{flalign*}
		
		\noindent which establishes the lemma. 
	\end{proof} 
	
	We next show that $\bm{\mathrm{S}}$ is a bijection from $\mathcal{H}'$ to $\mathcal{H}$, from which we will deduce \Cref{trho}.
	
	\begin{cor}
	
	\label{cs0}
	
	The operator $\bm{\mathrm{S}} : \mathcal{H}' \rightarrow \mathcal{H}$ is a bijection. 
	\end{cor} 
	
	\begin{proof} 
		
		We begin by observing that
		\begin{flalign} 
			\label{salpha} 
			\bm{\mathrm{S}}  \varsigma_0 = \theta^{-1} \cdot \sigma (x) - \bm{\mathrm{T}} \varrho (x)  = \alpha \cdot \varsigma_0,
		\end{flalign}
		
		\noindent where the first statement holds by the definition of $\bm{\mathrm{S}}$, and the second from that \eqref{sigmatheta} of $\sigma$. Hence, $\mathbb{C} \cdot \varsigma_0$ is a one-dimensional eigenspace of $\bm{\mathrm{S}}$ that is not in its kernel, so $\bm{\mathrm{S}}$ acts bijectively on it. 
		
		Therefore, it remains to show that $\bm{\mathrm{S}}: \mathcal{H}_0 \cap \mathcal{H}' \rightarrow \mathcal{H}_0 \cap \mathcal{H}$ is a bijection, as $\mathcal{H}_0 \cap \mathcal{H}'$ is the orthogonal complement of $\mathbb{C} \cdot \varsigma_0$ in $\mathcal{H}'$. By the spectral theorem for unbounded operators, namely, \cite[Proposition 5.13]{USOS} with \cite[Proposition 3.18]{USOS}, it suffices to show that there exists a constant $c>0$ such that $\bm{\mathrm{S}} - c$ is nonnegative as an operator on $\mathcal{H}_0 \cap \mathcal{H}'$. To that end, it suffices to show that $\langle (\bm{\mathrm{S}} - c) f, f \rangle_{\varrho} \ge 0$ for any $f \in \mathcal{H}_0 \cap \mathcal{H}'$, for $c>0$ sufficiently small. We may assume that $f$ is bounded and compactly supported, as such functions are dense in $\mathcal{H}'$. Then, we have for $c$ sufficiently small that
		\begin{flalign*}
			\langle (\bm{\mathrm{S}}-c) f, f \rangle_{\varrho} = \theta^{-1} \displaystyle\int_{-\infty}^{\infty} ( \sigma(x) - c \theta) \cdot | f(x) |^2 \varrho (x) dx - \langle \bm{\mathrm{T} \varrho} f, f \rangle_{\varrho} \ge 0, 
		\end{flalign*} 
		
		\noindent where the first statement follows from the definition of $\bm{\mathrm{S}}$, and the second from \Cref{tg0} and the fact (by \Cref{alphat} and the definition \eqref{sigmatheta} of $\sigma$) that $\sigma(x) > \theta c$.
	\end{proof}

	\begin{proof}[Proof of \Cref{trho}]
		
		Observe by the definition of $\bm{\mathrm{S}}$, the definition \eqref{sigmatheta} of $\sigma$, and \Cref{rhorho} that $\bm{\mathrm{S}} = (\theta^{-1} - \bm{\mathrm{T} \varrho_{\beta}}) \bm{\sigma}$. Define $\mathcal{G} = \{ \bm{\sigma} f : f \in \mathcal{H}' \}$. Since \Cref{alphat} yields a constant $c>0$ such that $\sigma(x)>c$ for all $x \in \mathbb{R}$, the multiplication operator $\bm{\sigma} : \mathcal{H}' \rightarrow \mathcal{G}$ is a bijection. Together with the fact from \Cref{cs0} that $\bm{\mathrm{S}} = (\theta^{-1} - \bm{\mathrm{T} \varrho_{\beta}}) \bm{\sigma}: \mathcal{H}' \rightarrow \mathcal{H}$ is a bijection, it follows that the operator $(\theta^{-1} - \bm{\mathrm{T}} \bm{\varrho_{\beta}}) : \mathcal{G} \rightarrow \mathcal{H}$ is a bijection. Hence, it admits an inverse $(\theta^{-1} - \bm{\mathrm{T} \varrho_{\beta}})^{-1} : \mathcal{H} \rightarrow \mathcal{G}$, which is a bijection. By \Cref{trhobounded}, we have $\mathcal{H} \subseteq \mathcal{G}$, so to confirm that $(\theta^{-1} - \bm{\mathrm{T} \varrho_{\beta}}) : \mathcal{H} \rightarrow \mathcal{H}$ is a bijection it suffices to show that $\mathcal{G} \subseteq \mathcal{H}$. 
		
		To that end, fix $g \in \mathcal{G}$,  so that $g = \bm{\sigma} f$ for some $f \in \mathcal{H}' \subseteq \mathcal{H}$ and $(\theta^{-1} - \bm{\mathrm{T} \varrho_{\beta}}) g \in \mathcal{H}$. We must show that $g \in \mathcal{H}$, or equivalently that $\bm{\mathrm{T} \varrho_{\beta}} g \in \mathcal{H}$. Since $\varrho_{\beta} g = \varrho f$ (as $\sigma (x) \cdot \varrho_{\beta} (x) = \varrho (x)$, by \Cref{rhorho} and \eqref{sigmatheta}), this follows similarly to in the proof of \Cref{trhobounded}, from the estimates
		\begin{flalign*}
			\| \bm{\mathrm{T} \varrho} f \|_{\mathcal{H}}^2 & = 4 \displaystyle\int_{-\infty}^{\infty} \bigg| \displaystyle\int_{-\infty}^{\infty} \log |x-y|  \cdot \varrho (y) f(y) dy \bigg|^2 \varrho (x) dx \\
			& \le 4 \displaystyle\int_{-\infty}^{\infty} | f(y) |^2 \varrho (y) dy \cdot  \displaystyle\int_{-\infty}^{\infty}  \Bigg| \displaystyle\int_{-\infty}^{\infty} (\log |x-y|)^2  \varrho (y) dy \Bigg| \cdot \varrho (x) dx \le C \cdot \| f \|_{\mathcal{H}}^2 ,
		\end{flalign*}
	
		\noindent for some constant $C>1$, where we used the fact that $\varrho$ is bounded and has subexponential decay in the last bound (from \Cref{rhoexponential}).
	\end{proof} 
	
	\begin{rem} 
		
		\label{sigmaalpha0} 
		
		Observe by \eqref{salpha} and the fact that $\bm{\mathrm{S}} = (\theta^{-1} - \bm{\mathrm{T} \varrho_{\beta}}) \bm{\sigma}$ that $(\theta^{-1} - \bm{\mathrm{T} \varrho_{\beta}}) \sigma = 0$ if $\alpha = 0$. Thus, \Cref{trho} is false if $\alpha = 0$.
		
	\end{rem}

	\subsection{Proofs of \Cref{rho0} and \Cref{vt}}
	
	\label{Proofrho} 
	
	In this section we establish \Cref{rho0} and \Cref{vt}. We first require the following lemma.

	\begin{lem} 
		
		\label{hbetarhof} 
		
		Let $f \in \mathcal{H}$ be a function.
		
		\begin{enumerate} 
			
			\item We have that $(\theta^{-1} - \bm{\varrho_{\beta}} \bm{\mathrm{T}}) \bm{\varrho_{\beta}} (\theta^{-1}  - \bm{\mathrm{T}} \bm{\varrho_{\beta}})^{-1} f = \bm{\varrho_{\beta}} f$. 
			\item  If $(\theta^{-1} - \bm{\varrho_{\beta}} \bm{\mathrm{T}}) \bm{\varrho_{\beta}} f = 0$, then $f = 0$. 
			
		\end{enumerate} 
		
	\end{lem} 
	
	\begin{proof} 
		The first statement of the lemma follows from the fact that $(\bm{g} - \bm{\varrho}_{\beta} \bm{\mathrm{T}}) \bm{\varrho_{\beta}} = \bm{\varrho_{\beta}} (\bm{g} - \bm{\mathrm{T}} \bm{\varrho_{\beta}})$. The second follows from the fact that $(\theta^{-1} - \bm{\varrho_{\beta} \mathrm{T}}) \bm{\varrho_{\beta}} = \bm{\varrho_{\beta}} (\theta^{-1} - \bm{\mathrm{T} \varrho_{\beta}})$; the fact that $\varrho_{\beta} (x) \ne 0$ for all $x \in \mathbb{R}$ (by \Cref{frho}); and the invertibility of $(\theta^{-1} - \bm{\mathrm{T} \varrho_{\beta}}): \mathcal{H} \rightarrow \mathcal{H}$, by \Cref{trho}.
	\end{proof}

	\begin{proof}[Proof of \Cref{rho0}]
		
		The first statement of \eqref{rho2} implies the second, by \Cref{rhorho}, so it suffices to establish the former. To that end, observe by \Cref{rhorho} that, for any $x \in \mathbb{R}$, we have 
		\begin{flalign*}
			\theta^{-1} \cdot \varrho (x) -  \varrho_{\beta} (x) \cdot \bm{\mathrm{T}} \varrho (x) = \alpha \cdot \varrho_{\beta} (x),
		\end{flalign*}
		
		\noindent or equivalently 
		\begin{flalign}
			\label{theta2t}
			(\theta^{-1} - \bm{\varrho_{\beta}} \bm{\mathrm{T}}) \varrho = \alpha \cdot \varrho_{\beta}.
		\end{flalign} 
		
		\noindent Further observe by the $f=\varsigma_0$ case of the first statement of \Cref{hbetarhof} (and \Cref{fdr}) that $(\theta^{-1} -\bm{\varrho_{\beta}} \bm{\mathrm{T}}) \bm{\varrho_{\beta}} \varsigma_0^{\dr} = \varrho_{\beta}$. Together with \eqref{theta2t}, this yields 
		\begin{flalign*} 
			(\theta^{-1} - \bm{\varrho_{\beta}} \bm{\mathrm{T}}) \bm{\varrho_{\beta}} (\varsigma_0^{\dr} - \alpha^{-1} \cdot \varrho_{\beta}^{-1} \cdot \varrho) = 0.
		\end{flalign*} 
		
		\noindent By the second statement of \Cref{hbetarhof}, with the fact that $\varrho_{\beta}^{-1} \cdot \varrho \in \mathcal{H}$ (which is a quick consequence of \Cref{rhorho}), this gives $\varrho = \alpha \cdot \varsigma_0^{\dr} \cdot \varrho_{\beta}$, yielding the first equality in \eqref{rho2} and thus the corollary. 
	\end{proof}

	\begin{proof}[Proof of \Cref{vt}]   
		
		This follows from the equalities 
		\begin{flalign*}
			 (\theta^{-1} \cdot \bm{\varsigma_0^{\dr}} - \alpha^{-1} \cdot \bm{\mathrm{T}} \bm{\varrho}) v_{\eff}  = \theta^{-1} \varsigma_1^{\dr} - \alpha^{-1} \cdot \bm{\mathrm{T}} \bm{\varrho} \cdot v_{\eff} & = \theta^{-1} \varsigma_1^{\dr} - \bm{\mathrm{T}} \bm{\varrho_{\beta}} \bm{\varsigma_0^{\dr}} \cdot v_{\eff} \\
			& = \theta^{-1} \varsigma_1^{\dr} - \bm{\mathrm{T}} \bm{\varrho_{\beta}} \cdot \varsigma_1^{\dr} = (\theta^{-1} - \bm{\mathrm{T}} \bm{\varrho_{\beta}}) \varsigma_1^{\dr} = \varsigma_1.
		\end{flalign*}
		
		\noindent Here, the first equality follows from the fact (by \Cref{v}) that $v_{\eff} = (\varsigma_0^{\dr})^{-1} \cdot \varsigma_1^{\dr}$; the second from the fact (by the first statement in \eqref{rho2}) that $\alpha^{-1} \cdot \varrho = \varrho_{\beta} \cdot \varsigma_0^{\dr}$; the third again from the fact that $v_{\eff} = (\varsigma_0^{\dr})^{-1} \cdot \varsigma_1^{\dr}$; and the fourth and fifth from the definition \Cref{fdr} of $\varsigma_1^{\dr}$.
	\end{proof}

	\subsection{Proof of \Cref{xf2}, \Cref{derivativefg}, and \Cref{derivativev}}
	
	\label{DerivativefProof} 
	
	\begin{proof}[Proof of \Cref{xf2}]
		
		As in \eqref{sigmatheta}, denote $\sigma (x) = \theta \cdot \bm{\mathrm{T}} \varrho (x) + \theta \alpha$, which is positive and bounded away from $0$ by \Cref{alphat}. Further denote $g(x) = f^{\dr} (x)$ and $h(x) = \sigma(x)^{-1} \cdot g(x)$. Since $(\theta^{-1} - \bm{\mathrm{T} \varrho_{\beta}}) g = f$, we have
		\begin{flalign}
			\label{2g} 
			g(x) = \theta f(x) + 2 \theta \displaystyle\int_{-\infty}^{\infty} \log |x-y| \varrho_{\beta} (y) g(y) dy.
		\end{flalign}
		
		\noindent To bound the integral on the right side of \eqref{2g}, observe that
		\begin{flalign*}
			\Bigg| \displaystyle\int_{-\infty}^{\infty} \log |x-y| \varrho_{\beta} (y) g(y) dy \Bigg| \le \Bigg( \displaystyle\int_{-\infty}^{\infty} |h(y)|^2 \varrho (y) dy \Bigg)^{1/2} \Bigg( \displaystyle\int_{-\infty}^{\infty} (\log |x-y|)^2 \varrho (y) dy \Bigg)^{1/2},
		\end{flalign*}
		
		\noindent where we have used the fact from \Cref{rhorho} that $\varrho (y) = \sigma(y) \cdot \varrho_{\beta} (y)$ for each $y \in \mathbb{R}$. By the boundedness and exponential decay of $\varrho$ (from \Cref{rhoexponential}), there exists $C_1 > 1$ such that 
		\begin{flalign*}
			\Bigg| \displaystyle\int_{-\infty}^{\infty} (\log |x-y|)^2 \varrho (y) dy \Bigg| \le C_1 (\log (|x|+2) )^2,
		\end{flalign*}
		
		\noindent from which we deduce
		\begin{flalign*}
			\Bigg| \displaystyle\int_{-\infty}^{\infty} \log |x-y| \varrho_{\beta} (y) g(y) dy \Bigg| \le  C_1 \| h \|_{\mathcal{H}} \cdot \log (|x|+2).
		\end{flalign*}

		\noindent Together with \eqref{2g}, this yields, for any $x \in \mathbb{R}$,
		\begin{flalign}
			\label{gh} 
				|g(x) - \theta f(x)| \le 2 \theta C_1 \| h \|_{\mathcal{H}} \cdot \log (|x|+2). 
		\end{flalign}
		
		We next bound $\| h \|_{\mathcal{H}}$. To do so, recall the operator $\bm{\mathrm{S}} = \theta^{-1} \bm{\sigma} - \bm{\mathrm{T} \varrho}$ from \Cref{ProofT0}. By \Cref{rhorho} (with the definition of $\sigma$), we have $\bm{\mathrm{S}} = (\theta^{-1} - \bm{\mathrm{T} \varrho_{\beta}}) \bm{\sigma}$, and so it follows since $h = \sigma^{-1} g = \sigma^{-1} f^{\dr}$ that $\bm{\mathrm{S}} h = (\theta^{-1} - \bm{\mathrm{T} \varrho_{\beta}}) f^{\dr} = f$. Recalling the space $\mathcal{H}_0$ from \eqref{sigmatheta}, denote $h = h_0 + h_1$, where $h_0 \in \mathcal{H}_0$ and $h_1 = \langle h, \varsigma_0 \rangle_{\varrho} \cdot \varsigma_0 \in \mathbb{C} \cdot \varsigma_0$. Then, for some $c \in (0,  \alpha^2)$, we have
		\begin{flalign}
			\label{festimate0} 
			\begin{aligned} 
			\| f \|_{\mathcal{H}}^2 = \| \bm{\mathrm{S}} h \|_{\mathcal{H}}^2 & = \| \bm{\mathrm{S}} h_0 \|_{\mathcal{H}}^2 + \alpha^2 \cdot \| h_1 \|_{\mathcal{H}}^2 \\
			& \ge \langle \bm{\mathrm{S}} h_0, h_0 \rangle_{\varrho}^2 \cdot \|h_0 \|_{\mathcal{H}}^{-2} +  \alpha^2 \cdot \| h_1 \|_{\mathcal{H}}^2 \\
			& = \big( \theta^{-1}  \langle \sigma h_0,  h_0 \rangle_{\varrho}  - \langle \bm{\mathrm{T} \varrho} h_0, h_0 \rangle_{\varrho} \big)^2 \cdot \| h_0 \|_{\mathcal{H}}^{-2}  +  \alpha^2 \cdot \| h_1 \|_{\mathcal{H}}^2 \\
			& \ge  \theta^{-2}  \langle \sigma h_0,  h_0 \rangle_{\varrho}^2 \cdot \| h_0 \|_{\mathcal{H}}^{-2}  +  \alpha^2 \cdot \| h_1 \|_{\mathcal{H}}^2 \\
			& \ge c \cdot \langle h_0, h_0 \rangle_{\varrho} + \alpha^2 \cdot \| h_1 \|_{\mathcal{H}}^2 \ge c \| h \|_{\mathcal{H}}^2,
			\end{aligned} 
		\end{flalign} 
		
		\noindent where in the first statement we used the fact that $f = \bm{\mathrm{S}} h$; in the second we used the fact that $\mathcal{H}_0 \subseteq \mathcal{H}$ is the orthogonal complement of $\mathbb{C} \cdot \varsigma_0 \subseteq \mathcal{H}$, and from \eqref{salpha} that $\varsigma_0$ is an eigenfunction of $\bm{\mathrm{S}}$ with eigenvalue $\alpha$; in the third and fourth we used the definition of $\bm{\mathrm{S}}$; in the fifth we used the fact from \Cref{tg0} that $\bm{\mathrm{T} \varrho}$ is nonpositive on $\mathcal{H}_0$; in the sixth we used the fact from \Cref{alphat} that $\theta^{-1} \cdot \sigma (x) > c$ for all $x \in \mathbb{R}$; and in the seventh we again used that $\mathcal{H}_0 \subseteq \mathcal{H}$ is the orthogonal complement of $\mathbb{C} \cdot \varsigma_0 \subseteq \mathcal{H}$. Inserting \eqref{festimate0} into \eqref{gh} yields a constant $C_2 > 1$ such that, for any $x \in \mathbb{R}$,
		\begin{flalign}
			\label{fh0} 
			|f^{\dr} (x)| = |g(x)| \le C_2 \cdot |f(x)| + C_2 \| f \|_{\mathcal{H}}  \cdot \log (|x|+2),
		\end{flalign}
		
		\noindent proving the lemma.
	\end{proof}

	\begin{proof}[Proof of \Cref{derivativefg}]
		
		Denote $g(x) = f^{\dr} (x)$ for each $x \in \mathbb{R}$. Then, observe that 
		\begin{flalign}
			\label{gderivativef}
			g' = (f' + \bm{\mathrm{T} \varrho_{\beta}'} g)^{\dr}.
		\end{flalign} 
		
		\noindent Indeed, differentiating the relation $g(x) = \theta f(x) + \theta \cdot \bm{\mathrm{T} \varrho_{\beta}} g (x)$ in $x$ (and recalling \eqref{operatort}) yields 
		\begin{flalign*}
			g'(x) = \theta f'(x) + 2 \theta \displaystyle\int_{-\infty}^{\infty} \log |y| \cdot \partial_x \big( g(x-y) \varrho_{\beta} (x-y) \big) dy = \theta f' + \theta \bm{\mathrm{T} \varrho_{\beta}} g' + \theta \bm{\mathrm{T} \varrho_{\beta}'} g,
		\end{flalign*} 
		
		\noindent from which we obtain \eqref{gderivativef}. Applying \Cref{xf2}, it follows there exists a constant $C_1>1$ such that, for any $x \in \mathbb{R}$, 
		\begin{flalign}
			\label{gderivativexestimate0} 
			|g'(x)| \le C_1 \cdot ( |f'(x)| + |\bm{\mathrm{T} \varrho_{\beta}'} g(x)| ) + C_1 \cdot ( \| f' \|_{\mathcal{H}} + \| \bm{\mathrm{T} \varrho_{\beta}'} g \|_{\mathcal{H}} ) \cdot \log (|x|+2).
		\end{flalign}
		
		We must now bound $\bm{\mathrm{T} \varrho_{\beta}'} g$. To that end, observe from \eqref{operatort} that 
		\begin{flalign}
			\label{tderivativeg} 
			\begin{aligned} 
			|\bm{\mathrm{T} \varrho_{\beta}'} g(x)| & \le 2  \displaystyle\int_{-\infty}^{\infty} \log |x-y| \cdot \varrho_{\beta}' (y) |g(y)| dy  \\
			& \le 2 \Bigg( \displaystyle\int_{-\infty}^{\infty} |g(y)|^2 \varrho (y) dy \Bigg)^{1/2} \Bigg( \displaystyle\int_{-\infty}^{\infty} (\log |x-y|)^2 \varrho (y)^{-1} \varrho_{\beta}' (y)^2 dy \Bigg)^{1/2}.
			\end{aligned}
		\end{flalign} 
		
		\noindent From \Cref{rhoexponential}, \Cref{rhorho}, and \Cref{alphat}, there exists constants $C_2 > 1$ and $C_3 >1$ that 
		\begin{flalign*}
			|\varrho_{\beta}' (y)| \le C_2 (|y|+1) \cdot \varrho_{\beta} (y) \le C_3 (|y|+1) \cdot \varrho (y).
		\end{flalign*}
		
		\noindent Inserting this into \eqref{tderivativeg} (and using the boundedness and exponential decay of $\varrho$, from \Cref{rhoexponential}) yields for some constants $C_4 > 1$ and $C_5 > 1$ that
		\begin{flalign*} 
			|\bm{\mathrm{T} \varrho_{\beta}'} g(x)| \le C_4 \| g \|_{\mathcal{H}} \cdot \Bigg( \displaystyle\int_{-\infty}^{\infty} (\log |x-y|)^2 (|y|+1) \varrho (y) dy \Bigg)^{1/2} \le C_5 \| g \|_{\mathcal{H}} \cdot \log (|x|+2).
		\end{flalign*} 
		
		\noindent By \Cref{xf2}, we have $\| g \|_{\mathcal{H}} \le C_6 \| f \|_{\mathcal{H}}$ for some $C_6 > 1$, so it follows for some $C_7 > 1$ that 
		\begin{flalign*}
			|\bm{\mathrm{T} \varrho_{\beta}'} g(x)| \le C_7 \| f \|_{\mathcal{H}} \cdot \log (|x|+2); \qquad	\| \bm{\mathrm{T} \varrho_{\beta}'} g \|_{\mathcal{H}} \le C_7 \| f \|_{\mathcal{H}}.
		\end{flalign*}
		
		\noindent Applying these bounds in \eqref{gderivativexestimate0} then yields the lemma.
	\end{proof}

	\begin{proof}[Proof of \Cref{derivativev}]
		
		Fix $x \in [-A, A]$. Observe that there exists a constant $C_1>1$ such that 
		\begin{flalign*} 
			| v_{\eff} (x) | \le | \varsigma_0^{\dr} (x) |^{-1} \cdot | \varsigma_1^{\dr} (x) | \le C_1 (A + \log A) \le 2C_1 A,
		\end{flalign*} 
		
		\noindent where  in the first statement we used \Cref{v}; in the second we used \Cref{positive0} and \Cref{xf2}; and in the third we used the fact that $A \ge 2$. This confirms the first statement of the corollary. Similarly, there exists a constant $C_2 > 1$ such that
		\begin{flalign*}
			| \partial_x v_{\eff} (x) | & \le | \varsigma_0^{\dr} (x) |^{-1} \cdot | \partial_x \varsigma_1^{\dr} (x) | + | \varsigma_0^{\dr} (x) |^{-2} \cdot | \varsigma_1^{\dr} (x) | \cdot | \partial_x \varsigma_0^{\dr} (x) | \\
			& \le C_2 \log A + C_2 (\log A)^2 + C_2 A \log A \le 2C_2 A \log A,
		\end{flalign*} 
		
		\noindent where in the first statement we used \Cref{v}; in the second we used \Cref{xf2}, \Cref{derivativefg}, and \Cref{positive0}; and in the third we used the fact that $A \ge 2$.  This confirms the second statement of the corollary.
	\end{proof}

	\subsection{Proof of \Cref{beta0theta}}
	
	\label{IntegralEstimate}

	\begin{proof}[Proof of \Cref{beta0theta}]
		
		By \Cref{1thetarho} and the explicit form of $\mathfrak{l} (x) = \log (x^2 + \mathfrak{d}^2)/2$, there exists a constant $C = C(\beta) > 1$ such that
		\begin{flalign}
			\label{integral00} 
			\displaystyle\sup_{\lambda \in \mathbb{R}} \displaystyle\int_{\lambda-1}^{\lambda+1} |\mathfrak{l} (x-\lambda)| \varrho (x) dx \le C.
		\end{flalign}
		
		\noindent Now, observe that $\Gamma' (\theta) \cdot \Gamma (\theta)^{-1}$ tends to $-\infty$ as $\theta$ tends to $0$. Therefore, there exists a constant $\theta_0 = \theta_0 (\beta) > 0$ such that, for any $\theta \in (0, \theta_0)$, we have 
		\begin{flalign*}
			\alpha = \log \beta - \Gamma' (\theta) \cdot \Gamma (\theta)^{-1} > 8C.
		\end{flalign*}
		
		\noindent Together with \eqref{integral00}, this implies for $\theta \in (0, \theta_0)$ that 
		\begin{flalign*}
			4 \alpha^{-1} \cdot \displaystyle\sup_{\lambda \in \mathbb{R}} \displaystyle\int_{\lambda-1}^{\lambda+1} |\mathfrak{l} (x-\lambda)| \varrho (x) d x  < \displaystyle\frac{1}{2}.
		\end{flalign*} 
	
		\noindent Hence, since $\mathfrak{l}(x-\lambda) \ge 0$ whenever $|x-\lambda| \ge 1$, it follows for $\theta \in (0, \theta_0)$ and any $\lambda \in \mathbb{R}$ that
		\begin{flalign*}
			\Bigg| 2 \alpha^{-1} \displaystyle\int_{-\infty}^{\infty} \mathfrak{l} (x-\lambda) \varrho (x) d x + 1 \Bigg| & \ge  2 \alpha^{-1} \displaystyle\int_{|x-\lambda| > 1} \mathfrak{l} (x-\lambda) \varrho (x) d x + 1 -  2 \alpha^{-1} \displaystyle\int_{\lambda-1}^{\lambda+1} |\mathfrak{l} (x-\lambda)| \varrho (x) d x  \\
			& >  2 \alpha^{-1} \displaystyle\int_{-\infty}^{\infty} |\mathfrak{l}(x-\lambda)| \varrho (x) dx  + \frac{1}{2},
		\end{flalign*}
		
		\noindent thereby establishing the lemma.
	\end{proof}

	\section{Proof of \Cref{centerdistance}} 
	
	\label{Proofst} 
	
	In this section we prove \Cref{centerdistance}. We adopt the notation and assumptions of that lemma throughout. We further call a unit vector $\bm{v} = (v_{N_1}, v_{N_1+1}, \ldots , v_{N_2}) \in \mathbb{R}^N$ nonnegatively normalized if $v_j > 0$, where $j \in \llbracket N_1, N_2 \rrbracket$ is the minimal index such that $v_j \ne 0$. For each index $j \in \llbracket 1, N \rrbracket$ and real number $s \ge 0$, let $\bm{u}_j (s) = (u_j (N_1;s), u_j (N_1+1;s), \ldots , u_j (N_2; s))$ denote the nonnegatively normalized, unit eigenvector of $\bm{L}(s)$ with eigenvalue $\lambda_j$. 
	
	We will use the following result showing approximate uniqueness for localization centers of $\bm{L}(t)$.

	\begin{lem}[{\cite[Proposition 2.9]{LC}}]
		\label{centerdistance2} 
		
		Adopt \Cref{lbetaeta}, but assume more generally that $\zeta \ge N^3 e^{-200 (\log N)^{3/2}}$. The following holds with overwhelming probability. Fix $t \in [0, T]$; $\lambda \in \eig \bm{L}$; and $\zeta$-localization centers $\varphi, \tilde{\varphi} \in \llbracket N_1, N_2 \rrbracket$ of $\lambda$ with respect to $\bm{L}(t)$. If $N_1 + T (\log N)^3 \le \varphi \le N_2 - T (\log N)^3$, then $|\varphi - \tilde{\varphi}| \le (\log N)^3$. 
		
	\end{lem}

	\begin{proof}[Proof of \Cref{centerdistance}]
		
		The proof of this lemma is similar to that of \cite[Lemma 5.2]{LC}. Recalling \Cref{adelta}, we restrict to the event $\mathsf{E}_1 =  \bigcap_{s \ge 0} \mathsf{BND}_{\bm{L}(s)} (\log N)$, as we may by \Cref{l0eigenvalues}. We further restrict to the events $\mathsf{E}_2$ on which \Cref{qijsalpha} holds and $\mathsf{E}_3$ on which \Cref{centerdistance2} holds.  
		
		Next, we recall a fact concerning the evolution of the Lax matrix $\bm{L}(s)$. For each $s \in \mathbb{R}$, define the tridiagonal skew-symmetric matrix $\bm{P} (s) = [ P_{ij} (s) ] \in \Mat_{\llbracket N_1, N_2 \rrbracket}$ as follows. For each $i \in \llbracket N_1, N_2 - 1 \rrbracket$, set $P_{i,i+1} (s) = a_i (s)/2$ and $P_{i+1,i} (s) = -a_i (s)/2$; for all $(i, j) \in \llbracket N_1, N_2 \rrbracket^2$ not of the above form, set $P_{i,j} (s) = 0$. For any real number $s \in \mathbb{R}_{\ge 0}$, further let $\bm{V}(s) = [ V_{ij} (s) ] \in \Mat_{\llbracket N_1, N_2 \rrbracket}$ satisfy the ordinary differential equation $\partial_s \bm{V}(s) = \bm{P}(s) \cdot \bm{V}(s)$, with initial data $\bm{V}(0) = \Id$; the existence of such a matrix $\bm{V}(s)$ follows from the Picard--Lindel\"{o}f theorem. Similarly, fixing $r \ge 0$, for any real number $s \ge r$, let $\bm{V}(r;s) = [\bm{V} (r;s)] \in \Mat_{\llbracket N_1, N_2 \rrbracket}$ satisfy the ordinary differential equation $\partial_s \bm{V} (r; s) = \bm{P} (s) \cdot \bm{V} (r; s)$, with initial data $\bm{V}(r;r) = \bm{V}(r)$. Observe that $\bm{V} (0;r) = \bm{V}(r)$. For any $(i, j) \in \llbracket N_1, N_2 \rrbracket^2$, the $(i,j)$-entry of $\bm{V}(r;s)$ is more explicitly given by 
		\begin{flalign}
			\label{vijs}
			V_{ij} (r;s) = V_{ij} (r) + \displaystyle\sum_{k=1}^{\infty}  \displaystyle\sum_{i_1=N_1}^{N_2} \cdots \displaystyle\sum_{i_k=N_1}^{N_2}  \displaystyle\int_r^s \cdots \displaystyle\int_r^s \mathbbm{1}_{s_1 > s_2 > \cdots > s_k} \cdot V_{i_k j} (r) \displaystyle\prod_{h=1}^k P_{i_{h-1}, i_h} (s_h) ds_h,
		\end{flalign} 
		
		\noindent where we have set $i_0 = i$. Indeed, it is quickly verified under this choice that $\bm{V}(r;r) = \bm{V}(r)$ and $\partial_s \bm{V}(r;s) = \bm{P}(s) \cdot \bm{V}(r;s)$, the latter by differentiating both sides of \eqref{vijs} with respect to $s$. Then, by \cite[Section 2]{FMPL}, we have $\bm{V} (r; s)^{-1} \cdot \bm{L}(s) \cdot \bm{V}(r; s) = \bm{L}(r)$. 
		
		This implies that $\bm{L}(s) = \bm{V}(r; s) \cdot \bm{L}(r) \cdot \bm{V}(r; s)^{\mathsf{T}}$, as $\bm{V}(r; s)$ is orthogonal (since $\bm{V}(0) = \Id$, $\partial_s \bm{V}(r; s) = \bm{P}(s) \cdot \bm{V}(r; s)$, and $\bm{P}(s)$ is skew-symmetric). Hence, letting $\bm{U}(s) = [ U_{ij} (s) ] \in \Mat_{N \times N}$ denote matrix of eigenvectors of $\bm{L}(s)$, whose $(i,j)$-entry is given by $U_{ij}(s) = u_j (i; s)$ for each $(i, j) \in \llbracket N_1, N_2 \rrbracket \times \llbracket 1, N \rrbracket$, we have $\bm{U}(s) = \bm{V}(r; s) \cdot \bm{U}(r)$. In particular,
		\begin{flalign}
			\label{vijsu0} 
			u_j (i; s) = \displaystyle\sum_{k = N_1}^{N_2} V_{ik}(r; s) \cdot u_j (k; r).
		\end{flalign}
		
		Now observe whenever $|i-j| \ge 20 (s-r) \log N$ that 
		\begin{flalign}
			\label{u0vijs}
			| V_{ij} (r; s) |\le \displaystyle\sum_{k = |i-j|}^{\infty} \displaystyle\frac{(s-r)^k}{k!} \cdot (2 \log N)^k \le \displaystyle\sum_{k=|i-j|}^{\infty} \bigg( \displaystyle\frac{2e (s-r) \log N}{k} \bigg)^k \le \displaystyle\sum_{k=|i-j|}^{\infty} e^{-k} \le 2e^{-|i-j|}.
		\end{flalign}
		
		\noindent Here, in the first inequality we used \eqref{vijs}, with the facts that each $| P_{ij} (s_h) | \le \log N$ (as we have restricted to the event $\mathsf{E}_1$), that $P_{ij} = 0$ whenever $|i-j| \ne 1$ (meaning that there are at most two choices for each $i_h$ that gives rise to a nonzero summand in \eqref{vijs}), and that $|V_{i'j'} (r)| \le 1$ for all $i',j' \in \llbracket N_1, N_2 \rrbracket$ (as $\bm{V}(r)$ is orthogonal); in the second we used the bound $k! \ge (e^{-1} k)^k$ for each $k \ge 0$; in the third we used the bound $2k^{-1} e(s-r) \log N \le e^{-1}$ for $k \ge |i-j| \ge 20 (s-r) \log N$; and in the fourth we performed the sum. 
		
		Now we can establish the lemma. We assume that $t' \ge t$, as the proof when $t' < t$ is entirely analogous. Let $j \in \llbracket 1, N \rrbracket$ be such that $\lambda = \lambda_j$, and denote $\Delta = (|t-t'| + 1) (\log N)^3$ and $\zeta_0 = N^3 e^{-200 (\log N)^{3/2}}$. Assuming to the contrary that $|\varphi - \varphi'| > (|t-t'|+2) (\log N)^3 = \Delta + (\log N)^3$, the first statement of the lemma follows from the bounds 
		\begin{flalign*}
			| u_j (\varphi'; t') | & \le \displaystyle\sum_{k = N_1}^{N_2} \mathbbm{1}_{|k-\varphi'| \ge \Delta} \cdot | V_{\varphi' k}(t; t') | + \displaystyle\sum_{k=N_1}^{N_2} \mathbbm{1}_{|k-\varphi| > (\log N)^3} \cdot | u_j (k; t) | \\
			& \le 2 \displaystyle\sum_{k = N_1}^{N_2} \mathbbm{1}_{|k-\varphi'| \ge \Delta} \cdot e^{-|k-\varphi'|} + N \zeta_0 \le N^5 e^{-200 (\log N)^{3/2}} < \zeta,
		\end{flalign*}
		
		\noindent which contradicts the fact that $\varphi'$ is a $\zeta$-localization center of $\bm{u}_j (t')$. Here, the first inequality follows from \eqref{vijsu0}, together with the facts that $| u_j (k; 0) | \le 1$ (as $\bm{u}_j$ is a unit vector) and $| V_{mk} (t; t') | \le 1$ (as $\bm{V} (t; t')$ is orthogonal); the second follows from \eqref{u0vijs} (with the fact that $\Delta > 20 (t'-t) \log N $ for sufficiently large $N$) and the fact from \Cref{centerdistance2} (and our restriction to $\mathsf{E}_3$) that $k$ is not a $\zeta_0$-localization center for $\bm{u}_j (t)$ if $|k-\varphi| > (\log N)^3$; and the third and fourth follow from performing the sum and recalling that $\Delta \ge (\log N)^3$, that $\zeta_0 = N^3 e^{-200 (\log N)^{3/2}}$, and that $\zeta \ge e^{-150 (\log N)^{3/2}}$.
		
		Therefore, since $\varphi \in \llbracket N_1 + T (\log N)^4, N_1 - T (\log N)^4 \rrbracket$, we have $\varphi' \in \llbracket N_1 + T(\log N)^3, N_2 - T(\log N)^3 \rrbracket$, so \Cref{qijsalpha} applies with the $(i, j)$ there equal to $(\varphi, \varphi')$ here. The second statement of the lemma then follows from the estimates  
		\begin{flalign*}
			|q_{\varphi} (t) - q_{\varphi'} (t')| & \le  |q_{\varphi} (t) - q_{\varphi} (t')| + |q_{\varphi} (t') - q_{\varphi'} (t')| \\
			& \le |t-t'| \cdot \displaystyle\sup_{s \in [0, T]} |b_{\varphi} (s)| + \alpha \cdot |\varphi - \varphi'| + |\varphi - \varphi'|^{1/2} (\log N)^2 \\
			& \le |t-t'| \cdot \log N + 2 \alpha (|t-t'|+2) (\log N)^{7/2} \le (|t-t'|+1) (\log N)^4,
		\end{flalign*}
		
		\noindent where the second inequality holds by \eqref{qtpt}, \eqref{abr}, and \eqref{qiqjs4} (with our restriction to $\mathsf{E}_2$); the third holds from our restriction to $\mathsf{E}_1$ and the first part of the lemma; and the fourth holds since $N$ is sufficiently large.		
	\end{proof}

	\section{Proof of \Cref{xfsum}}
	
	\label{Proofxf}
	
	In this section we establish \Cref{xfsum}. Setting $|\mathcal{I}| = n$, we assume throughout that $m = n$ and $\mathcal{J}_i = \{ i \}$ for each $i \in \llbracket 1, n \rrbracket$, as the proof is entirely analogous in the general case. Further set $A_k = \Infl_{x_k} (F;p)$ for each $k \in \llbracket 1, n \rrbracket$. We begin with the following lemma that exhibits a set $\mathcal{Y} \subseteq \mathbb{R}^n$ that $\bm{x}$ is likely to lie in, on which $F$ changes by a bounded amount upon perturbing a given coordinate. In what follows, for any integer $k \in \llbracket 1, n \rrbracket$ and $n$-tuple $\bm{w} = (w_1, w_2, \ldots , w_n) \in \mathbb{R}^n$, we let $\bm{w}^{(k)} = (w_1, w_2, \ldots , w_{k-1}, w_{k+1}, \ldots , w_n) \in \mathbb{R}^{n-1}$ denote the $(n-1)$-tuple obtained by removing $w_k$ from $\bm{w}$. 
	
	\begin{lem} 
		
		\label{yprobability}
		
		There exists a subset $\mathcal{Y} \subseteq \mathbb{R}^n$ with $\mathbb{P}[\bm{x} \in \mathcal{Y}] \ge 1 - 2np^{1/2}$, such that for each $k \in \llbracket 1, n \rrbracket$ we have
		\begin{flalign*}
			| F(\bm{u}) - F(\bm{v})| \le A_k, \qquad \text{for any $\bm{u}, \bm{v} \in \mathcal{Y}$ with $\bm{u}^{(k)} = \bm{v}^{(k)}$}.
		\end{flalign*}

	\end{lem} 
	
	\begin{proof} 
		
		Fix $k \in \llbracket 1, n \rrbracket$. Define the $n$-tuple $\bm{y} = (y_1, y_2, \ldots , y_n) \in \mathbb{R}^n$ of mutually independent random variables, by setting $\bm{y}^{(k)} = \bm{x}^{(k)}$, and setting $y_k$ to be a random variable with the same law as $x_k$ that is independent from $\bm{x} \cup \bm{y}^{(k)}$. Observe by \Cref{xf} that $\mathbb{P} [ |F(\bm{x}) - F(\bm{y})| \ge A_k ] \le p$. Therefore, a Markov bound yields subsets $\mathcal{Y}_{k,1} \subseteq \mathbb{R}^{n-1}$ and $\mathcal{Y}_{k,2} (\bm{w}^{(k)}) \subseteq \mathbb{R}$ for each $\bm{w}^{(k)} \in \mathcal{Y}_{k,1}$, satisfying the following two properties. First, for each $\bm{w}^{(k)} \in \mathcal{Y}_{k,1}$, we have  
		\begin{flalign}
			\label{yk12} 
			\mathbb{P}[ \bm{x}^{(k)} \subseteq \mathcal{Y}_{k,1}] \ge 1 - p^{1/2}, \qquad \text{and} \qquad \mathbb{P} [ x_k \subseteq \mathcal{Y}_{k,2} (\bm{w}^{(k)}) ] \ge 1 - p^{1/2}.
		\end{flalign}
		
		\noindent Second, denoting $\mathcal{Y}_k = \{ \bm{w} = (w_1, w_2, \ldots , w_n) \in \mathbb{R}^n : \bm{w}^{(k)} \in \mathcal{Y}_{k,1}, w_k \in \mathcal{Y}_{k,2} (\bm{w}^{(k)}) \}$, we have 
		\begin{flalign}
			\label{akfuv} 
			| F (\bm{u}) - F(\bm{v}) | \le A_k, \qquad \text{for any $\bm{u}, \bm{v} \in \mathcal{Y}_k$ such that $\bm{u}^{(k)} = \bm{v}^{(k)}$}.
		\end{flalign} 
		
		Now set $\mathcal{Y} = \bigcap_{k=1}^n \mathcal{Y}_k$. By \eqref{yk12} and the independence between $\bm{x}^{(k)}$ and $x_k$, we have that $\mathbb{P}[ \bm{x} \in \mathcal{Y}_k] \ge (1 - p^{1/2})^2 \ge 1 - 2p^{1/2}$ for each $k \in \llbracket 1, n \rrbracket$. Hence, a union bound yields $\mathbb{P}[\bm{x} \in \mathcal{Y}] \ge 1 - 2n p^{1/2}$, which verifies the first estimate in the lemma. The second follows from \eqref{akfuv}.
	\end{proof} 
	
	Now we can establish \Cref{xfsum}.
	
	\begin{proof}[Proof of \Cref{xfsum}]
		
		Let $\mathcal{Y} \subseteq \mathbb{R}^n$ denote the set from \Cref{yprobability}, and define the function $G: \mathbb{R}^n \rightarrow \mathbb{R}$ by setting $G(\bm{x}) = F(\bm{x}) \cdot \mathbbm{1}_{\bm{x} \in \mathcal{Y}}$, for each $\bm{x} \in \mathcal{Y}$. Then, for any $j \in \llbracket 1, n \rrbracket$ and $\bm{u}, \bm{v} \in \mathbb{R}^n$ with $\bm{u}^{(j)} = \bm{v}^{(j)}$, we have  $( G(\bm{u}) - G(\bm{v}) )^2 \le A_j^2$, by \Cref{yprobability}. Hence, \cite[Theorem 12]{CI} gives 
		\begin{flalign*}
			\mathbb{P} \big[ \big| G(\bm{x}) - \mathbb{E} [ G(\bm{x}) ] \big| \ge RS^{1/2} \big] \le 2e^{-R^2/4}.
		\end{flalign*}
		
		\noindent Moreover, by \Cref{yprobability}, we have $\mathbb{P} [ F(\bm{x}) \ne G(\bm{x}) ] \le 2np^{1/2}$, and so it follows that 
		\begin{flalign}
			\label{fgprobability}
			\mathbb{P} \big[ \big| F(\bm{x}) - \mathbb{E} [ G(\bm{x}) ] \big| \ge RS^{1/2} \big] \le 2np^{1/2} + 2e^{-R^2/4}.
		\end{flalign}
		
		\noindent Additionally, we have 
		\begin{flalign*}
			\big| \mathbb{E} [ G(\bm{x}) - F (\bm{x}) ] \big| \le \mathbb{E} [ | F(\bm{x}) | \cdot \mathbbm{1}_{\bm{x} \notin \mathcal{Y}} ] \le \mathbb{P} [\bm{x} \notin \mathcal{Y}]^{1/2} \cdot \mathbb{E} [ F(\bm{x})^2 ]^{1/2} \le U(2np^{1/2})^{1/2},
		\end{flalign*}
		
		\noindent where we used the definition of $G$, \Cref{yprobability}, and the definition of $U$ from \eqref{sdelta}. Upon insertion into \eqref{fgprobability}, this yields the lemma.	
	\end{proof}

	\section{Proof of \Cref{sumlambdaiexpectation}}
	
	\label{ProofLambdaH} 
	
	In this section we prove \Cref{sumlambdaiexpectation}; we adopt the notation of \Cref{ProofH} throughout.

	An approximation of the linear functional $\sum_{i=1}^N H(\lambda_i)$ of $\bm{L}$ is provided by \Cref{lf}, but without an effective error. To remedy this, we will ``embed'' $\bm{L}$ into a much larger matrix $\bm{\mathfrak{L}}$; compare expectations of their linear functionals; and apply \Cref{lf} to $\bm{\mathfrak{L}}$. In what follows, we abbreviate $\bm{a}(0) = \bm{a} = (a_{N_1}, a_{N_1+1},\ldots , a_{N_2-1})$ and $\bm{b}(0) = \bm{b} = (b_{N_1}, b_{N_1+1}, \ldots , b_{N_2})$. 
	
	Let $\mathfrak{N}_1 < \mathfrak{N}_2$ be integers with $\mathfrak{N}_1 < N_1$ and $N_2 < \mathfrak{N}_2$. Set $\mathfrak{N} = \mathfrak{N}_2 - \mathfrak{N}_1 + 1$; we will take $\mathfrak{N}$ sufficiently large so that \eqref{hsumnu} below holds. Sample $(\bm{\mathfrak{a}}; \bm{\mathfrak{b}}) \in \mathbb{R}^{\mathfrak{N}-1} \times \mathbb{R}^{\mathfrak{N}}$ according to the thermal equilibrium $\mu_{\beta,\theta;\mathfrak{N}-1,\mathfrak{N}}$ (from \Cref{mubeta2}), and denote $\bm{\mathfrak{a}} = (\mathfrak{a}_{\mathfrak{N}_1}, \mathfrak{a}_{\mathfrak{N}_1+1}, \ldots , \mathfrak{a}_{\mathfrak{N}_2-1})$ and $\bm{b} = (\mathfrak{b}_{\mathfrak{N}_1}, \mathfrak{b}_{\mathfrak{N}_1+1}, \ldots , \mathfrak{b}_{\mathfrak{N}_2})$; we couple $(\bm{\mathfrak{a}}; \bm{\mathfrak{b}})$ with $(\bm{a}; \bm{b})$ so that $\mathfrak{a}_i = a_i$ for each $i \in \llbracket N_1, N_2-1 \rrbracket$ and $\mathfrak{b}_i = b_i$ for each $i \in \llbracket N_1, N_2 \rrbracket$. Define $\bm{\mathfrak{L}} = [\mathfrak{L}_{ik}] \in \SymMat_{\llbracket \mathfrak{N}_{1}, \mathfrak{N}_{2} \rrbracket}$ (as in \Cref{matrixl}) by setting $\mathfrak{L}_{i,i} = \mathfrak{b}_i$ for each $i \in \llbracket \mathfrak{N}_{1}, \mathfrak{N}_{2} \rrbracket$; setting $\mathfrak{L}_{i,i+1} = \mathfrak{L}_{i+1,i} = \mathfrak{a}_i$ for each $i \in \llbracket \mathfrak{N}_{1}, \mathfrak{N}_{2}-1 \rrbracket$; and setting $\mathfrak{L}_{i,k} = 0$ for all $(i,k) \in \llbracket \mathfrak{N}_{1}, \mathfrak{N}_{2} \rrbracket^2$ not of the above form. By \Cref{lf}, we may choose $\mathfrak{N}$ sufficiently large (relative to $N$ and $H$) so that 
	\begin{flalign}
		\label{hsumnu}
		\begin{aligned}  
			& \mathbb{P} \Bigg[ \bigg| \displaystyle\frac{1}{\mathfrak{N}} \cdot \displaystyle\sum_{\nu \in \eig \bm{\mathfrak{L}}} H(\nu) - \displaystyle\int_{-\infty}^{\infty} H(\lambda) \varrho(\lambda) d \lambda \bigg| \ge \displaystyle\frac{A}{N} \Bigg] \le e^{-(\log N)^2}; \\
			& \mathbb{P} \big[ \# \{ \nu \in \eig \bm{\mathfrak{L}} : |\nu| \ge \log N \} \ge N^{-1} \cdot \mathfrak{N} \big] \le e^{-(\log N)^2},
		\end{aligned}
	\end{flalign}
	
	\noindent where in the last bound we used the superexponential decay of $\varrho$ from \Cref{rhoexponential}.
	
	For any $z \in \mathbb{C}$, denote the resolvents $\bm{G}(z) = [G_{ik} (z)] = (\bm{L}-z)^{-1}$ and $\bm{\mathfrak{G}} (z) = [\mathfrak{G}_{ik} (z)] = (\bm{\mathfrak{L}} - z)^{-1}$. The following lemma compares the diagonal entries of $\bm{G}$ and $\bm{\mathfrak{G}}$.  
	
	\begin{lem} 
		
		\label{glambda2}
		
		There exists a constant $c>0$ such that the following holds with probability at least $1 - c^{-1} e^{-c(\log N)^2}$. Set $\eta = e^{-(\log N)^3}$, and define $\Omega = \{ z \in \mathbb{C}: -N \le \Real z \le N, \eta \le \Imaginary z \le 1 \}$. Then, for any complex number $z \in \Omega$ and integer $i \in \llbracket N_1 + (\log N)^4, N_2 - (\log N)^4 \rrbracket$, we have
		\begin{flalign*}
			| \mathfrak{G}_{ii} (z) - G_{ii} (z) | \le c^{-1} e^{-c (\log N)^4}.
		\end{flalign*}
	\end{lem}

	The proof of \Cref{glambda2} will make use of the below estimate on resolvents of perturbations of random Lax matrices, as in the context of \Cref{lmatrixl}.

	\begin{lem}[{\cite[Lemma 5.4]{LC}}]
		
		\label{lgdifference}
		
		There exists a constant $c \in (0,1)$ such that the following holds with probability at least $1 - c^{-1} e^{-c(\log N)^2}$. Adopt \Cref{lmatrixl} and, for any $z \in  \mathbb{C}$, denote $\bm{G}(z) = [ G_{ij} (z) ] = (\bm{L} - z)^{-1}$ and $\tilde{\bm{G}} (z) = [ \tilde{G}_{ij} (z) ] = (\tilde{\bm{L}} - z)^{-1}$. Let $\eta \in [\delta, 1]$ be a real number; and define the set $\Omega = \{ z \in \mathbb{C} : -N \le \Real z \le N, \eta \le \Imaginary z \le 1 \}$. For any $i, j \in \llbracket N_1, N_2 \rrbracket$, we have 
		\begin{flalign}
			\label{gkke} 
			\displaystyle\sup_{z \in \Omega} | G_{ij} (z) - \tilde{G}_{ij} (z) | \le e^{ (\log N)^2} \eta^{-2} (\delta^{1/4} + e^{-c \dist (i, \mathcal{D}) - c \dist (j, \mathcal{D})}).
		\end{flalign}
		
	\end{lem}

	\begin{proof}[Proof of \Cref{glambda2}]
		
		This lemma would follow from \Cref{lgdifference} except for the fact that, if it were used directly, the $N$ there must be $\mathfrak{N}$ here, and this would cause the prefactor $e^{(\log \mathfrak{N})^2}$ in \eqref{gkke} to be too large. To circumvent this, we use \Cref{lgdifference} inductively.
		
		Fix an integer $r \ge 1$ and a triple of integers $(\mathfrak{N}_{0;m}, \mathfrak{N}_{1;m}, \mathfrak{N}_{2;m})$ for each $m \in \llbracket 0, r \rrbracket$, satisfying the following properties. First, we have that $\mathfrak{N}_{0;m} = \mathfrak{N}_{2;m} - \mathfrak{N}_{1;m} + 1 \ge 1$ and $r \le \log \mathfrak{N}$. Second, we have $(\mathfrak{N}_{0;0}, \mathfrak{N}_{1;0}, \mathfrak{N}_{2;0}) = (\mathfrak{N}, \mathfrak{N}_1, \mathfrak{N}_2)$ and $(\mathfrak{N}_{0;r}, \mathfrak{N}_{1;r}, \mathfrak{N}_{2;r}) = (N, N_1, N_2)$. Third, for each $m \in \llbracket 0, r-1 \rrbracket$, we have 
		\begin{flalign}
			\label{nm} 
			\mathfrak{N}_{1;m} < \mathfrak{N}_{1;m+1} < \displaystyle\frac{\mathfrak{N}_{1;m}}{10} < 0 < \displaystyle\frac{\mathfrak{N}_{2;m}}{10} < \mathfrak{N}_{2;m+1} < \mathfrak{N}_{2;m}.
		\end{flalign} 
		
		\noindent It is quickly verified that such integers exist, since $\mathfrak{N}_1 < N_1 < 0 < N_2 < \mathfrak{N}_2$. 
		
		For each $m \in \llbracket 0, r \rrbracket$, define $\bm{\mathfrak{L}}_m = [\mathfrak{L}_{i,k;m}] \in \SymMat_{\llbracket \mathfrak{N}_{1;m}, \mathfrak{N}_{2;m} \rrbracket}$ (as in \Cref{matrixl}) by setting $\mathfrak{L}_{i,i;m} = \mathfrak{b}_i$ for $i \in \llbracket \mathfrak{N}_{1;m}, \mathfrak{N}_{2;m} \rrbracket$; setting $\mathfrak{L}_{i,i+1;m} = \mathfrak{L}_{i+1,i;m} = \mathfrak{a}_i$ for $i \in \llbracket \mathfrak{N}_{1;m}, \mathfrak{N}_{2;m}-1 \rrbracket$; and setting $\mathfrak{L}_{i,k;m} = 0$ for $(i, k) \in \llbracket \mathfrak{N}_{1;m}, \mathfrak{N}_{2;m} \rrbracket^2$ not of the above form. For any $z \in \mathbb{C}$, denote $\bm{\mathfrak{G}}_m (z) = [ \mathfrak{G}_{i,k;m} (z) ] = (\bm{\mathfrak{L}}_m - z)^{-1}$. Then, $\bm{\mathfrak{L}}_0 = \bm{\mathfrak{L}}$; $\bm{\mathfrak{L}}_r = \bm{L}$; $\bm{\mathfrak{G}}_0 (z) = \bm{\mathfrak{G}} (z)$; and $\bm{\mathfrak{G}}_r (z) = \bm{G} (z)$. 
		
		For each $m \in \llbracket 0, r \rrbracket$, define the set $\Omega_m = \{ z \in \mathbb{C} : -\mathfrak{N}_{0;m} \le \Real z \le \mathfrak{N}_{0;m}, \eta \le \Imaginary z \le 1 \}$. We next apply \Cref{lgdifference}, with the $(N, \delta; \bm{L})$ there equal to $(\mathfrak{N}_{m-1}; 0; \bm{\mathfrak{L}}_{m-1})$ here, and the $\tilde{\bm{L}}$ there given by the extension by $0$ of $\bm{\mathfrak{L}}_m$ to $\llbracket \mathfrak{N}_{1;m-1}, \mathfrak{N}_{2;m-1} \rrbracket$ here, meaning that we set $\mathfrak{L}_{i,k;m} = 0$ if $i$ or $k$ is in $\llbracket \mathfrak{N}_{1;m-1}, \mathfrak{N}_{2;m-1} \rrbracket \setminus \llbracket \mathfrak{N}_{1;m}, \mathfrak{N}_{2;m} \rrbracket$; observe that \Cref{l0eigenvalues} verifies the assumption \eqref{ll}, with the $\mathcal{D}$ there equal to $\llbracket \mathfrak{N}_{1;m-1}, \mathfrak{N}_{2;m-1} \rrbracket \setminus \llbracket \mathfrak{N}_{1;m}, \mathfrak{N}_{2;m} \rrbracket$ here. In this way, the $\Omega$, $G_{ii} (z)$, and $\tilde{G}_{ii}(z)$ from \Cref{lgdifference} become $\Omega_m$, $\mathfrak{G}_{ii;m-1} (z)$, and $\mathfrak{G}_{ii;m}(z)$ here, respectively. Thus, there exists a constant $c_1 > 0$ and an event $\mathsf{G}_m$ with $\mathbb{P} [ \mathsf{G}_m^{\complement} ] \le c_1^{-1} e^{-c_1(\log \mathfrak{N}_{0;m})^2}$, such that the following holds on $\mathsf{G}_m$. For any complex number $z \in \Omega_m$ and integer $i \in \llbracket \mathfrak{N}_{1;m} + (\log \mathfrak{N}_{0;m})^4, \mathfrak{N}_{2;m} - (\log \mathfrak{N}_{0;m})^4 \rrbracket$ (meaning that $\dist (i, \llbracket \mathfrak{N}_{1;m-1}, \mathfrak{N}_{2;m-1} \rrbracket \setminus \llbracket \mathfrak{N}_{1;m}, \mathfrak{N}_{2;m} \rrbracket) \ge (\log \mathfrak{N}_{0;m})^4$), we have 
		\begin{flalign}
			\label{gim1m} 
			| \mathfrak{G}_{ii;m-1} (z) - \mathfrak{G}_{ii;m} (z) | \le c_1^{-1} \eta^{-2} \cdot e^{(\log \mathfrak{N}_{0;m-1})^2 -c_1 (\log \mathfrak{N}_{0;m})^4}. 
		\end{flalign}
		
		 Thus, denoting $\mathsf{G} = \bigcap_{m=1}^r \mathsf{G}_m$, we deduce that 
		\begin{flalign*} 
			\mathbb{P} [\mathsf{G}^{\complement}] \le \displaystyle\sum_{m=1}^r \mathbb{P} [\mathsf{G}_m^{\complement}] \le c_1^{-1} \sum_{m=1}^r e^{-c_1 (\log \mathfrak{N}_{0;m})^2} \le 2c_1^{-1} e^{-c_1 (\log N)^2/2},
		\end{flalign*} 
	
		\noindent where the last inequality holds by \eqref{nm}. Moreover, by summing the estimates \eqref{gim1m}, we deduce that there exists a constant $c>0$ so that, for all $z \in \bigcap_{m=1}^r \Omega_m = \Omega$ and $i \in \llbracket N_1 + (\log N)^4, N_2 - (\log N)^4 \rrbracket$, we have 
		\begin{flalign*} 
			|\mathfrak{G}_{ii}(z) - G_{ii} (z)| \le \displaystyle\sum_{m=1}^r | \mathfrak{G}_{ii;m-1} (z) - \mathfrak{G}_{ii;m} (z) \big| & \le  c_1^{-1} \eta^{-2} \displaystyle\sum_{m=1}^r e^{(\log \mathfrak{N}_{0;m-1})^2 - c_1 (\log \mathfrak{N}_{0;m})^4} \\
			& \le c^{-1} e^{-c (\log N)^4},
		\end{flalign*} 
	
		\noindent where the last inequality again holds by \eqref{nm} (with the fact that $\eta = e^{-(\log N)^3}$). This establishes the lemma.
	\end{proof}

	We next have the following lemma, expressing linear functionals of $\bm{L}$ and $\bm{\mathfrak{L}}$ through their resolvent entries. We then deduce \Cref{sumlambdaiexpectation} as a consequence.

	\begin{lem} 
		
		\label{estimatehlambdanu}
		
		There exists a constant $c>0$ such that the following holds with probability at least $1 - c^{-1} e^{-c(\log N)^2}$. Letting $n_1' = n_1 + (\log N)^5$; $n_2' = n_2 - (\log N)^5$; $\mathfrak{N}_1' = \mathfrak{N}_1 + (\log \mathfrak{N})^5$; $\mathfrak{N}_2' = \mathfrak{N}_2 - (\log \mathfrak{N})^5$; and $\eta = e^{-(\log N)^3}$, we have  
		\begin{flalign}
			\label{sumh}
			\begin{aligned}  
				& \Bigg| \displaystyle\sum_{i=n_1}^{n_2}  H(\Lambda_i) - \displaystyle\frac{1}{\pi} \displaystyle\int_{-\log N}^{\log N} H(E) \displaystyle\sum_{i=n_1'}^{n_2'} \Imaginary G_{ii} (E + \mathrm{i} \eta) dE \Bigg| \le 6A (\log N)^5; \\
				& \Bigg| \displaystyle\frac{N}{\mathfrak{N}} \displaystyle\sum_{\nu \in \eig \bm{\mathfrak{L}}} H(\nu) - \displaystyle\frac{N}{\mathfrak{N} \pi} \displaystyle\int_{-\log N}^{\log N} H(E) \displaystyle\sum_{i=\mathfrak{N}_1'}^{\mathfrak{N}_2'} \Imaginary \mathfrak{G}_{ii} (E + \mathrm{i} \eta) dE \Bigg| \le 7A (\log N)^5.
			\end{aligned} 
		\end{flalign}
	\end{lem} 
	
	\begin{proof}
		
		Define the function $\tilde{H} : \mathbb{R} \rightarrow \mathbb{R}$ by setting $\tilde{H}(\lambda) = H(\lambda) \cdot \mathbbm{1}_{|\lambda| \le \log N}$ for each $\lambda \in \mathbb{R}$. Throughout, we restrict to the event $\mathsf{BND}_{\bm{L}} (\log N)$, which we may do by \Cref{l0eigenvalues}. Then we claim that, with probability at least $1 - c^{-1} e^{-c(\log N)^2}$,
		\begin{flalign}
			\label{sumh2} 
			\begin{aligned} 
				& \Bigg| \displaystyle\frac{1}{n} \displaystyle\sum_{i=n_1}^{n_2}  \tilde{H} (\Lambda_i) - \displaystyle\frac{1}{\pi n} \displaystyle\int_{-\infty}^{\infty} \tilde{H}(E) \displaystyle\sum_{i=n_1'}^{n_2'} \Imaginary G_{ii} (E + \mathrm{i} \eta) dE \Bigg| \le \displaystyle\frac{6A}{n} \cdot (\log N)^5; \\
				& \Bigg| \displaystyle\frac{1}{\mathfrak{N}} \displaystyle\sum_{\nu \in \eig \bm{\mathfrak{L}}} \tilde{H}(\nu) - \displaystyle\frac{1}{\pi \mathfrak{N}} \displaystyle\int_{-\infty}^{\infty} \tilde{H} (E) \displaystyle\sum_{i=\mathfrak{N}_1'}^{\mathfrak{N}_2'} \Imaginary \mathfrak{G}_{ii} (E + \mathrm{i} \eta) dE \Bigg| \le \displaystyle\frac{6A}{\mathfrak{N}} \cdot (\log \mathfrak{N})^5.
			\end{aligned} 
		\end{flalign}
		
		\noindent Then the first bound in \eqref{sumh2} yields the first statement in \eqref{sumh}, since $H(\lambda) = \tilde{H}(\lambda)$ for each $\lambda \in \eig \bm{L}$ (by our restriction to $\mathsf{BND}_{\bm{L}} (\log N)$); since $H(E) = \tilde{H}(E)$ whenever $|E| \le \log N$; and since $\tilde{H}(E) = 0$ whenever $|E| > \log N$. Moreover, multiplying the second bound in \eqref{sumh2} by $N$ and using the facts that $N \le \mathfrak{N}$; that $\tilde{H} (\nu) \ne H(\nu)$ only if $|\nu| \ge \log N$; that $| \tilde{H}(\nu) | \le A$ for such $\nu$; and that there are at most $N^{-1} \cdot \mathfrak{N}$ many such $\nu$ with probability at least $1 - e^{-(\log N)^2}$ (by the second bound in \eqref{hsumnu}), we deduce the second statement in \eqref{sumh}. 
		
		Therefore, it suffices to verify \eqref{sumh2}. The proofs of both bounds there are entirely analogous (observe if $(n_1, n_2) = (N_1, N_2)$ then the second bound of \eqref{sumh2} takes a similar form to the first one, with the $\bm{L}$ there replaced by $\bm{\mathfrak{L}}$), so we only show the former. To that end, first observe that
		\begin{flalign*}
			G_{ii} (E+\mathrm{i}\eta) = \displaystyle\sum_{k=1}^N \displaystyle\frac{u_k (i)^2}{\lambda_k - E - \mathrm{i} \eta},
		\end{flalign*}
		
		\noindent for any $i \in \llbracket N_1, N_2 \rrbracket$, and so
		\begin{flalign}
			\label{h2} 
			\begin{aligned} 
				\displaystyle\frac{1}{\pi} \displaystyle\int_{-\infty}^{\infty} \tilde{H}(E) \displaystyle\sum_{i=n_1'}^{n_2'} \Imaginary G_{ii} (E+\mathrm{i}\eta) dE &= \displaystyle\frac{1}{\pi} \displaystyle\int_{-\infty}^{\infty} \tilde{H}(E) \displaystyle\sum_{k=1}^N \displaystyle\sum_{i=n_1'}^{n_2'} u_k (i)^2 \cdot \Imaginary (\lambda_k - E - \mathrm{i} \eta)^{-1} dE.
			\end{aligned} 
		\end{flalign} 
		
		It will next be useful to restrict the sum over $k$ on the right side of \eqref{h2}. To do so, we have by \cite[Equation (5.1)]{LC} and a union bound that there exists a constant $c>0$ and an event $\mathsf{G}$ with $\mathbb{P} [ \mathsf{G}^{\complement} ] \le c^{-1} e^{-c(\log N)^2}$ such that, on $\mathsf{G}$, we have
		\begin{flalign}
			\label{gku} 
			| u_k (i) | \le e^{-c(\log N)^5}, \qquad \text{whenever $| i - \varphi_0(k) | \ge (\log N)^5$}.
		\end{flalign}  
		
		\noindent We restrict to $\mathsf{G}$ in what follows. Then, denoting $\mathcal{K} = \{ \varphi_0^{-1} (m) : m \in \llbracket n_1, n_2 \rrbracket \}$, we have 
		\begin{flalign*}
			\begin{aligned} 
				\Bigg|  \displaystyle\int_{-\infty}^{\infty} & \tilde{H}(E) \displaystyle\sum_{k=1}^N \displaystyle\sum_{i=n_1'}^{n_2'}  u_k (i)^2 \cdot \Imaginary (\lambda_k - E - \mathrm{i} \eta)^{-1} dE - \displaystyle\int_{-\infty}^{\infty} \tilde{H}(E) \displaystyle\sum_{k=n_1}^{n_2} \Imaginary (\Lambda_k - E - \mathrm{i} \eta)^{-1} dE \Bigg|  \\
				& = \Bigg|  \displaystyle\int_{-\infty}^{\infty} \tilde{H}(E) \displaystyle\sum_{k=1}^N \displaystyle\sum_{i=n_1'}^{n_2'} u_k (i)^2 \cdot \Imaginary (\lambda_k - E - \mathrm{i} \eta)^{-1} dE \\
				& \qquad    - \displaystyle\int_{-\infty}^{\infty} \tilde{H}(E) \displaystyle\sum_{k \in \mathcal{K}} \displaystyle\sum_{i=N_1}^{N_2} u_k (i)^2 \cdot \Imaginary (\lambda_k - E - \mathrm{i} \eta)^{-1} dE \Bigg| \\ 
				& \le \Bigg| \displaystyle\int_{-\infty}^{\infty} \tilde{H}(E) \displaystyle\sum_{i=n_1'}^{n_2'} \displaystyle\sum_{k \notin \mathcal{K}} u_k (i)^2 \cdot \Imaginary (\lambda_k - E - \mathrm{i} \eta)^{-1} dE \Bigg| \\
				& \qquad    + \Bigg| \displaystyle\int_{-\infty}^{\infty} \tilde{H} (E) \displaystyle\sum_{k \in \mathcal{K}} \displaystyle\sum_{i \notin \llbracket n_1', n_2' \rrbracket} u_k (i)^2 \cdot \Imaginary (\lambda_k - E - \mathrm{i} \eta)^{-1} dE \Bigg|,
			\end{aligned} 
		\end{flalign*}
		
		\noindent where we used the orthonormality of the $(\bm{u}_k)$, the fact that $\Lambda_k = \lambda_{\varphi_0^{-1} (k)}$, and the definition of $\mathcal{K}$. Setting $n_1'' = n_1 - (\log N)^5$ and $n_2'' = n_2 + (\log N)^5$, it follows that 
		\begin{flalign}
			\label{h7} 
			\begin{aligned} 
			\Bigg|  \displaystyle\int_{-\infty}^{\infty} & \tilde{H}(E) \displaystyle\sum_{k=1}^N \displaystyle\sum_{i=n_1'}^{n_2'}  u_k (i)^2 \cdot \Imaginary (\lambda_k - E - \mathrm{i} \eta)^{-1} dE - \displaystyle\int_{-\infty}^{\infty} \tilde{H}(E) \displaystyle\sum_{k=n_1}^{n_2} \Imaginary (\Lambda_k - E - \mathrm{i} \eta)^{-1} dE \Bigg|  \\
			&  \le \Bigg| \displaystyle\int_{-\infty}^{\infty} \tilde{H}(E) \displaystyle\sum_{k \in \mathcal{K}} \displaystyle\sum_{i \in \llbracket n_1'', n_2'' \rrbracket \setminus \llbracket n_1', n_2' \rrbracket} u_k (i)^2 \cdot \Imaginary (\lambda_k - E - \mathrm{i} \eta)^{-1} dE \Bigg| \\
			& \qquad    + N^2  \displaystyle\int_{-\infty}^{\infty} | \tilde{H}(E) | \cdot \displaystyle\max_{i \in \llbracket N_1', N_2' \rrbracket} \displaystyle\max_{k: |i - \varphi_0(k)| \ge (\log N)^5} u_k (i)^2 \cdot \Imaginary (\lambda_k - E - \mathrm{i} \eta)^{-1} dE,
			\end{aligned} 
		\end{flalign}
		
		\noindent where we used the fact that $| i - \varphi_0 (k) | \ge (\log N)^5$ if $i \in \llbracket n_1', n_2' \rrbracket$ and $k \notin \mathcal{K}$, or if $k \in \mathcal{K}$ and $i \notin \llbracket n_1'', n_2'' \rrbracket$ (by the definition of $\mathcal{K}$). Moreover,
		\begin{flalign}
			\label{h8} 
			\begin{aligned} 
				N^2  \displaystyle\int_{-\infty}^{\infty} | \tilde{H}(E) | \cdot \displaystyle\max_{i \in \llbracket N_1', N_2' \rrbracket} \displaystyle\max_{k: |i - \varphi_0(k)| \ge (\log N)^5} & u_k (i)^2 \cdot \Imaginary (\lambda_k - E - \mathrm{i} \eta)^{-1} dE \\
				& \le 2AN^3 \eta^{-1} e^{-c(\log N)^5} \le Ae^{-(\log N)^3},
			\end{aligned} 
		\end{flalign}
		
		\noindent where in the first statement we used the facts that $\supp \tilde{H} \subseteq [-\log N, \log N] \subseteq [-N, N]$, that $| \tilde{H} (E) | \le A$ for all $E \in \mathbb{R}$, and that \eqref{gku} holds (by our restriction to $\mathsf{G}$); in the second, we used the fact that $\eta = e^{-(\log N)^3}$. Additionally, we have
		\begin{flalign}
			\label{h1}
			\begin{aligned} 
				&  \Bigg| \displaystyle\int_{-\infty}^{\infty} \tilde{H}(E) \displaystyle\sum_{k \in \mathcal{K}} \displaystyle\sum_{i \in \llbracket n_1'', n_2'' \rrbracket \setminus \llbracket n_1', n_2' \rrbracket} u_k (i)^2 \cdot \Imaginary (\lambda_k - E - \mathrm{i} \eta)^{-1} dE \Bigg| \\
				& \qquad \qquad \le 4 A (\log N)^5 \cdot \displaystyle\max_{k \in \llbracket 1, N \rrbracket} \displaystyle\int_{-\infty}^{\infty} \Imaginary (\lambda_k - E - \mathrm{i} \eta)^{-1}  dE = 4 \pi A (\log N)^5,
			\end{aligned} 
		\end{flalign}
		
		\noindent where in the first statement we used the facts that $| \tilde{H}(E) | \le A$ for each $E \in \mathbb{R}$, that there are most $4 (\log N)^5$ indices $i \in \llbracket n_1'', n_2'' \rrbracket \setminus  \llbracket n_1', n_2' \rrbracket$, and that the $\bm{u}_k$ are orthonormal; in the second, we used the fact that $\int_{-\infty}^{\infty} \Imaginary (\Lambda_k - E - \mathrm{i} \eta)^{-1} dE = \pi$. 
		
		Combining \eqref{h2}, \eqref{h7}, \eqref{h8}, and \eqref{h1}, we deduce that 
		\begin{flalign}
			\label{h9}
			\begin{aligned} 
				\Bigg| \displaystyle\sum_{i=n_1}^{n_2} \tilde{H} (\Lambda_i) - \displaystyle\frac{1}{\pi} & \displaystyle\int_{-\infty}^{\infty} \tilde{H}(E) \displaystyle\sum_{i=n_1'}^{n_2'} \Imaginary G_{ii} (E + \mathrm{i} \eta) dE \Bigg| \\
				& \le \Bigg| \displaystyle\sum_{i=n_1}^{n_2} \tilde{H} (\Lambda_i) - \displaystyle\frac{1}{\pi} \displaystyle\int_{-\infty}^{\infty} \tilde{H}(E) \displaystyle\sum_{k=n_1'}^{n_2'} \Imaginary (\Lambda_k - E - \mathrm{i} \eta)^{-1} dE \Bigg| +5A(\log N)^5 \\
				& = \Bigg| \displaystyle\sum_{k=n_1}^{n_2} \displaystyle\frac{\eta}{\pi} \displaystyle\int_{-\infty}^{\infty} \displaystyle\frac{( \tilde{H}(\Lambda_i) - \tilde{H}(E) ) dE}{(\Lambda_k - E)^2 + \eta^2} \Bigg| + 5A (\log N)^5,
			\end{aligned}
		\end{flalign}
		
		\noindent where in the last statement we used the identity $\eta \int_{-\infty}^{\infty} ( (\Lambda_k - E)^2 + \eta^2 )^{-1} dE = \pi$ for each $k \in \llbracket n_1, n_2 \rrbracket$.  Now, denoting $\eta' = e^{-(\log N)^{5/2}}$, we have
		\begin{flalign*} 
			\Bigg| \displaystyle\sum_{k=n_1}^{n_2} & \displaystyle\frac{\eta}{\pi} \displaystyle\int_{-\infty}^{\infty} \displaystyle\frac{( \tilde{H} (\Lambda_k) - \tilde{H}(E) ) dE}{(\Lambda_k-E)^2 + \eta^2} \Bigg| \\
			&  \le \displaystyle\sum_{k=n_1}^{n_2} \displaystyle\frac{\eta}{\pi} \Bigg( \displaystyle\int_{|E-\Lambda_k| \le \eta'} \displaystyle\frac{| \tilde{H}(\Lambda_k) - \tilde{H}(E) |}{(\Lambda_k - E)^2 + \eta^2} + 2A \displaystyle\int_{|E-\Lambda_k| > \eta'} \displaystyle\frac{dE}{(\Lambda_k - E)^2 + \eta^2} \Bigg) \\
			& \le \displaystyle\sum_{k=n_1}^{n_2} \displaystyle\frac{\eta}{\pi} \Bigg( Ae^{-(\log N)^2} \displaystyle\int_{|E-\Lambda_k| \le \eta'} \displaystyle\frac{dE}{(\Lambda_k-E)^2 + \eta^2} + 4A \eta'^{-1} \Bigg) \\
			& \le AN (e^{-(\log N)^2} + 4 \eta'^{-1} \eta ) \le 5 AN e^{-(\log N)^2},
		\end{flalign*} 
		
		\noindent where in the first bound we used the fact that $| \tilde{H}(\Lambda) | \le A$ for all $\Lambda \in \mathbb{R}$; in the second we used \eqref{hxhy2}; in the third we evaluated the integral; and in the fourth we used the fact that $\eta'^{-1} \eta \le e^{-(\log N)^2}$. This, with \eqref{h9}, verifies the first bound in \eqref{sumh2}. As mentioned above, the second is shown analogously; this establishes the lemma.
	\end{proof}

	\begin{proof}[Proof of \Cref{sumlambdaiexpectation}]
		
		We adopt the notation for the parameters $\eta = e^{-(\log N)^3}$, $(n_1', n_2')$, and $(\mathfrak{N}_1', \mathfrak{N}_2')$ from \Cref{estimatehlambdanu} throughout, and also set $n' = n_2'-n_1'+1$ and $\mathfrak{N}' = \mathfrak{N}_2' - \mathfrak{N}_1' + 1$. We may assume that $(\mathfrak{N}_1, \mathfrak{N}_2)$ are such that there exists an integer $\mathfrak{K} \ge 1$ such that $\mathfrak{N}' = \mathfrak{K} \cdot n'$. 
		
		Let us first show for any indices $i_1, i_2 \in \llbracket \mathfrak{N}_1', \mathfrak{N}_2' \rrbracket$ with $i_2 - i_1 +1 = n'$ that
		\begin{flalign}
			\label{h4}
			\Bigg| \mathbb{E} \bigg[ \displaystyle\int_{-\log N}^{\log N} H(E) \displaystyle\sum_{i=i_1}^{i_2} \Imaginary \mathfrak{G}_{ii} (E + \mathrm{i} \eta) dE \bigg] - \mathbb{E} \bigg[ \displaystyle\int_{-\log N}^{\log N} H(E) \displaystyle\sum_{i=n_1'}^{n_2'} \Imaginary G_{ii} (E + \mathrm{i} \eta) dE \bigg] \Bigg| \le A.
		\end{flalign}
		
		\noindent Denote $I = i_1 - n_1 = i_2 - n_2$ (where the last equality holds since $i_2 - i_1 = n' - 1 = n_2 - n_1)$. By \Cref{glambda2} (shifting the row and column indices of $\bm{\mathfrak{L}}$ there by $I$), there exists a constant $c>0$ such that we can couple $\bm{G}$ and $\bm{\mathfrak{G}}$ such that the following holds. With probability at least $1 - c^{-1} e^{-c(\log N)^2}$, we have $| G_{ii} (E + \mathrm{i} \eta) - \mathfrak{G}_{i+I,i+I} (E + \mathrm{i} \eta) | \le c^{-1} e^{-c (\log N)^2}$ for each $i \in \llbracket n_1', n_2' \rrbracket$ and $E \in [-N, N]$. Summing over $i \in \llbracket n_1', n_2' \rrbracket$; multiplying by $H(E)$; using the fact that $| H(E) | \le A$ for all $E \in \mathbb{R}$; and integrating over $E \in [-\log N, \log N]$ then gives with probability at least $1 - c^{-1} e^{-c(\log N)^2}$ that 
		\begin{flalign*}
			\Bigg| \displaystyle\int_{-\log N}^{\log N} H(E) \displaystyle\sum_{i=i_1}^{i_2} \mathfrak{G}_{ii} (E + \mathrm{i} \eta) dE - \displaystyle\int_{-\log N}^{\log N} H(E) \displaystyle\sum_{i=n_1'}^{n_2'} G_{ii} (E + \mathrm{i} \eta) dE \Bigg| \le c^{-1} AN e^{-c(\log N)^2}.
		\end{flalign*} 
		
		\noindent Therefore, \eqref{h4} follows from taking expectations of both sides (and again using the facts that $| H(E) | \le A$ and $\int_{-\infty}^{\infty} \Imaginary \mathfrak{G}_{ii} (E+\mathrm{i}\eta) dE = \pi$ for all $E \in \mathbb{R}$).
		
		Averaging \eqref{h4} over $\mathfrak{K} = n'^{-1} \cdot \mathfrak{N}'$ disjoint intervals $\llbracket i_1, i_2 \rrbracket$ covering $\llbracket \mathfrak{N}_1', \mathfrak{N}_2' \rrbracket$, we deduce
		\begin{flalign}
			\label{h5} 
			\begin{aligned} 
				\Bigg| \mathbb{E} \bigg[ \displaystyle\frac{1}{\mathfrak{K}} \displaystyle\int_{-\log N}^{\log N} H(E) \displaystyle\sum_{i=\mathfrak{N}_1'}^{\mathfrak{N}_2'} \Imaginary \mathfrak{G}_{ii} (E + \mathrm{i} \eta) dE \bigg] - \mathbb{E} \bigg[ \displaystyle\int_{-\log N}^{\log N} H(E) \displaystyle\sum_{i=n_1'}^{n_2'} \Imaginary G_{ii} (E + \mathrm{i} \eta) dE \bigg] \Bigg| \le A. 
			\end{aligned}  
		\end{flalign}
		
		\noindent By \Cref{estimatehlambdanu}, taking the expectation of \eqref{sumh}, and using the fact that $| H(\lambda) | \le A$ for all $\lambda \in \mathbb{R}$ (and also the fact that $2 \mathfrak{K} \cdot N \ge 2 \mathfrak{K} \cdot n' \ge 2 \mathfrak{N}' \ge \mathfrak{N}$), we obtain
		\begin{flalign*}
			& \Bigg| \mathbb{E} \bigg[ \displaystyle\sum_{i=n_1}^{n_2} H(\Lambda_i) \bigg] - \mathbb{E} \bigg[ \displaystyle\frac{1}{\pi} \displaystyle\int_{-\log N}^{\log N} H(E) \displaystyle\sum_{i=n_1'}^{n_2'} \Imaginary G_{ii} (E + \mathrm{i} \eta) dE \bigg] \Bigg| \le 7A(\log N)^5; \\
			& \Bigg| \mathbb{E} \bigg[ \displaystyle\frac{1}{\mathfrak{K}} \displaystyle\sum_{\nu \in \eig \bm{\mathfrak{L}}} H(\nu) \bigg] - \mathbb{E} \bigg[ \displaystyle\frac{1}{\mathfrak{K} \pi} \displaystyle\int_{-\log N}^{\log N} H(E) \displaystyle\sum_{i=\mathfrak{N}_1'}^{\mathfrak{N}_2'} \Imaginary \mathfrak{G}_{ii} (E + \mathrm{i} \eta) d E \bigg] \Bigg| \le 15 A (\log N)^5.
		\end{flalign*}
		
		\noindent Together with \eqref{h5}, this yields 
		\begin{flalign}
			\label{h6}
			\Bigg| \mathbb{E} \bigg[ \displaystyle\sum_{i=n_1}^{n_2} H(\Lambda_i) \bigg] - \mathbb{E} \bigg[ \displaystyle\frac{1}{\mathfrak{K}} \displaystyle\sum_{\nu \in \eig \bm{\mathfrak{L}}} H(\nu) \bigg] \Bigg| \le 23 A (\log N)^5.
		\end{flalign}
		
		\noindent Further multiplying the first bound in \eqref{hsumnu} by $\mathfrak{N} \mathfrak{K}^{-1} \le 2N$; using the bound
		\begin{flalign*} 
			\bigg| \displaystyle\frac{\mathfrak{N}}{\mathfrak{K}} - n \bigg| = n \cdot \bigg| \displaystyle\frac{\mathfrak{N}}{\mathfrak{N}'} \cdot \displaystyle\frac{n'}{n} - 1 \bigg| \le n \cdot \big| \big( ( 1 + 3 n^{-1} (\log N)^5 \big)^2 - 1 \big| \le 7 (\log N)^5
		\end{flalign*}
		
		\noindent as $n - 2(\log N)^5 \le n' \le n$ and $\mathfrak{N} - 2 (\log \mathfrak{N})^5 \le \mathfrak{N}' \le \mathfrak{N}$; using the fact that $| H(\lambda) | \le A$ for all $\lambda \in \mathbb{R}$; and taking expectations further yields 
		\begin{flalign*} 
			\Bigg| \mathbb{E} \bigg[ \displaystyle\frac{1}{\mathfrak{K}} \displaystyle\sum_{\nu \in \eig \bm{\mathfrak{L}}} H(\nu) \bigg] - n \displaystyle\int_{-\infty}^{\infty} H(\lambda) \varrho(\lambda) d \lambda \Bigg| \le 7A(\log N)^5 + 2A + 2ANe^{-(\log N)^2} \le 8A(\log N)^5.
		\end{flalign*} 
		
		\noindent Together with \eqref{h6}, this gives 
		\begin{flalign*}
			\Bigg| \mathbb{E} \bigg[ \displaystyle\sum_{i = n_1}^{n_2} H(\Lambda_i) \bigg] - n \displaystyle\int_{-\infty}^{\infty} H(\lambda) \varrho(\lambda) d \lambda \Bigg| \le 31 A (\log N)^5,
		\end{flalign*}
		
		\noindent which yields the lemma.
	\end{proof}

\end{document}